\newif\iffullversion
\fullversiontrue
\pdfoutput=1

\newif\ifdraft
\newif\ifanonymous

\drafttrue
\documentclass{article}
\usepackage[utf8]{inputenc}
\usepackage[T1]{fontenc}
\usepackage{amsmath,amssymb,paralist,stmaryrd,mathpartir,mathtools,tabularx,xcolor,centernot,tabto}
\usepackage{tocloft}

\usepackage[backend=bibtex,maxbibnames=100]{biblatex}

\usepackage{makeidx}\makeindex
\usepackage[font=footnotesize,labelfont=bf]{caption}
\usepackage{graphicx}
\usepackage[section]{placeins}
\usepackage{idxlayout}
\usepackage{yfonts}
\usepackage{listings}
\usepackage{longtable}
\usepackage{declmath}
\usepackage{version}
\usepackage{xr}
\usepackage{relsize}
\usepackage[framemethod=tikz]{mdframed}
\expandafter\ifx\csname opt@preview.sty\endcsname\relax
\usepackage[a4paper,margin=1in]{geometry}
\fi
\usepackage{circuits}
\usepackage[pdfborderstyle={/S/U/W 0.3}]{hyperref}
\usepackage[delaytext,envcountsame,multiautoref]{latexhacks}

\usetikzlibrary{decorations.pathreplacing}

\newcommand\mytilde{\raise.17ex\hbox{$\scriptstyle\sim$}}

\makeatletter

\input{macros}

\DeclareUnicodeCharacter{212D}{\textswab C}
\DeclareUnicodeCharacter{1D529}{\textswab l}
\DeclareUnicodeCharacter{1D51E}{\textswab a}
\DeclareUnicodeCharacter{2260}{\ensuremath{\neq}}
\DeclareUnicodeCharacter{27E6}{\ensuremath{\llbracket}}
\DeclareUnicodeCharacter{27E7}{\ensuremath{\rrbracket}}
\DeclareUnicodeCharacter{2297}{\ensuremath{\otimes}}
\DeclareUnicodeCharacter{22C5}{\ensuremath{\cdot}}
\DeclareUnicodeCharacter{2264}{\ensuremath{\leq}}
\DeclareUnicodeCharacter{2293}{\ensuremath{\sqcap}}

\makeatletter
\DeclareCaptionFormat{myformat}{#1#2#3\leavevmode\leaders\hrule height .1pt\hfill\kern\z@}
\newcommand\lemmaautorefname{Lemma}

\mdfdefinestyle{theoremstyle}{%
  linecolor=black,
  backgroundcolor=black!10!white,
  frametitlebelowskip=0pt,
  innerleftmargin=5pt,
  innerrightmargin=5pt,
}

\mdtheorem[style=theoremstyle]{definition}{Definition}
\mdtheorem[style=theoremstyle]{lemma}{Lemma}
\mdtheorem[style=theoremstyle]{corollary}{Corollary}

\newtheorem{theorem}{Theorem}
\newtheorem{conjecture}{Conjecture}

\newcounter{claimstep}
\newtheorem{claim}{Claim}[claimstep]

\AtBeginDocument{
  \ifx\PreviewMacro\undefined\else
  \PreviewEnvironment*{figure}
  \PreviewMacro*\footnote
  \PreviewMacro[!]\boardpic
  \PreviewMacro*\lstinputlisting
  \PreviewMacro[{[]{}{}{}}]\RULE
  \fi}

\bibliography{language2}

\newcommand\PROOFREAD[1]{\TODOQ{proofread\ifx!#1!\else #1\fi}}

\ifanonymous
\newcommand\anonymouserror{{\Huge\bfseries ANONYMOUSXXX}\error}
\newcommand\nonanonymous[1]{}
\else
\newcommand\anonymouserror{}
\newcommand\nonanonymous[1]{#1}
\fi

\newcommand\giturl[1]{\anonymouserror
  \url{https://raw.githubusercontent.com/dominique-unruh/qrhl-tool/master/#1}}

\newenvironment{citemize}{\begin{itemize}}{\end{itemize}}

\allowdisplaybreaks

\begin{document}

\title{Quantum Relational Hoare Logic}
\ifanonymous
\author{}
\else
\author{Dominique Unruh\\\small University of Tartu}
\fi

\maketitle

\thispagestyle{empty}

\begin{abstract}
  We present a logic for reasoning about pairs of
  interactive quantum programs -- quantum relational Hoare logic
  (qRHL). This logic follows the spirit of probabilistic relational
  Hoare logic (Barthe et al.~2009) and allows us to formulate how the
  outputs of two quantum programs relate given the relationship of
  their inputs. Probabilistic RHL was used extensively for
  computer-verified security proofs of classical cryptographic
  protocols. Since pRHL is not suitable for analyzing quantum
  cryptography, we present qRHL as a replacement, suitable for the
  security analysis of post-quantum cryptography and quantum
  protocols.  The design of qRHL poses some challenges unique to the
  quantum setting, e.g., the definition of equality on quantum
  registers. Finally, we implemented a tool for verifying proofs in
  qRHL and developed several example security proofs in it.
\end{abstract}

\tableofcontents

\section{Introduction}
\label{sec:intro}

Handwritten security proofs of cryptographic protocols are inherently
error prone -- mistakes may and most likely will happen, and they can
stay unnoticed for years. (The most striking examples are probably the
OAEP construction \cite{BeRo_94} whose proof went through a number of
fixes \cite{JC:Shoup02,C:FOPS01,JC:FOPS04} until it was finally
formally proven in \cite{RSA:BGLZ11} after years of industrial use, and the
PRF/PRP switching lemma which was a standard textbook example for many
years before it was shown that the standard proof is flawed
\cite{EC:BelRog06}.) This undermines the purpose of said proofs. In
order to enable high trust in our cryptographic protocols,
computer-aided verification is the method of choice. Methods for the
security analysis of protocols have a long history in the formal
methods community, starting with the seminal paper by Dolev and Yao
\cite{dolev-yao}. However, most methods were designed using so-called
symbolic abstractions that abstract away from the details of the
cryptographic algorithms by treating messages as terms with
well-defined adversarial deduction rules. There is no formal
justification for this (except in cases where ``computational
soundness'' proofs exist, see \cite{AbRo_02} and follow-up work). To
avoid this heuristic, or for reasoning about the basic cryptographic
constructions themselves (e.g., when designing an encryption scheme),
we need verification methods that allow us to reason about
cryptographic protocols in a fine-grained way, treating computations
as what they are: time-bounded probabilistic computations (we call
this the ``computational model'').

In recent years, a number of logics and tools have been presented for
reasoning about cryptographic protocols in the computational
model. For example, the CryptoVerif tool \cite{cryptoverif,cryptoverif-web} employs a
heuristic search using special rewriting rules to simplify a protocol
into a trivial one. The CertiCrypt tool \cite{certicrypt,certicrypt-web} and its
easier to use EasyCrypt successor \cite{easycrypt,fosad-easycrypt} allow us to reason
about pairs of programs (that represent executions of cryptographic
protocols) using  ``probabilistic relational Hoare
logic'' (pRHL). Proofs in the EasyCrypt tool are developed manually,
using backwards reasoning with a set of special tactics for reasoning
about pRHL judgments.

All those tools only consider classical cryptography. In recent years,
interest in quantum cryptography, both in research and in industry,
has surged. Quantum cryptography is concerned with the development of
protocols that are secure even against attacks using quantum computers,
to ensure that we can securely use cryptography in the future
(post-quantum cryptography). And quantum cryptography is concerned
with the design of quantum protocols that make active use of quantum
mechanics in their communication, such as quantum key distribution
protocols (an area pioneered by \cite{Bennett:1984:Quantum}).
While these are two different concerns involving different technologies,
  there is a large overlap at least in the analysis and proof techniques.

Can we formalize quantum cryptographic proofs using the existing
tools? For quantum protocols, the answer is clearly no, because the
tools do not even allow us to encode the quantum operations performed
in the protocol. Yet, when it comes to post-quantum security, the
protocols themselves are purely classical and can be modeled in tools
such as EasyCrypt, only the adversaries (which stay abstract anyway)
cannot. Thus, even though existing tools were not designed with
post-quantum security in mind, it is conceivable that the logics
underlying those tools happen to be sound in the quantum setting.  The
logic might happen to only contain deduction rules that also happen to
be valid with respect to quantum adversaries. Unfortunately, at least
for the EasyCrypt tool, we show that this is not the case. We proved  in EasyCrypt
the security of a simple protocol (an execution of the CHSH game
\cite{chsh}) which we know not to be secure in the
quantum setting. (For more details, see below in \autoref{sec:intro.qrhl}.)

Thus the question arises: can we design tools for the verification of
quantum cryptography? (Both for post-quantum security proofs, and for
actual quantum protocols.) What underlying logics should we use? In
this paper, we set out to answer that question.

Our main contribution is the development of a logic for relating
quantum programs, quantum relational Hoare logic (qRHL). qRHL
generalizes the probabilistic logic pRHL used by EasyCrypt and
CertiCrypt, and allows us to reason about the relationship between
quantum programs. We have implemented a theorem prover for qRHL
programs, and have analyzed quantum teleportation and a simple
post-quantum secure protocol with it, to demonstrate the applicability
of qRHL to the formal verification of quantum cryptography.

We now describe the contributions in more detail. For this, we first
review classical RHL. (Followed by an overview over quantum RHL,
  of our tools and experiments, related work, and the outline for the
  rest of the paper.)

\subsection{Classical RHL}

\paragraph{Deterministic RHL.} The simplest case of RHL is deterministic
RHL \cite{benton04relational}. Deterministic RHL allows us to relate two
deterministic programs. In our context, a program is always a program
in some imperative programming language, operating on a finite set of
variables. Consider two programs $\bc_1$ and $\bc_2$, operating on
variables $\xx,\yy,\dots$  An RHL judgment is an expression
$\rhl A{\bc_1}{\bc_2}B$ where $A$ and $B$ are Boolean predicates involving the
variables from $\bc_1$ and $\bc_2$ (tagged with index $1$ and $2$,
respectively). For example $A$ might be $\xx_1=\xx_2$, and $B$ might
be $\xx_1>\xx_2$. Then $\rhl A{\bc_1}{\bc_2}B$ has the following
meaning: If the initial memories $m_1,m_2$ of $\bc_1$ and $\bc_2$ satisfy the predicate $A$, then the memories $m_1',m_2'$ after execution of $\bc_1$ and $\bc_2$, respectively, satisfy $B$.
\begin{definition}[Deterministic RHL, informal]\label{def:drhl.informal}
  $\rhl A{\bc_1}{\bc_2}B$ holds iff:
  For any memories $m_1,m_2$ such that $\denotee A{m_1m_2}=\true$, we
  have $\denotee B{m_1'm_2'}=\true$ where
  $m_1':=\denotc{\bc_1}(m_1)$ and $m_2':=\denotc{\bc_2}(m_2)$.
\end{definition}
In this definition, a memory $m_i$ is simply an assignment from
variables $\xx,\yy,\dots$ to their contents.  And
$\denotee{A}{\memuni{m_1}{m_2}}$ is the evaluation of the predicate
$A$ with the contents of memory $m_1$ assigned to the variables
$\xx_1,\yy_1,\dots$ and the contents of $m_2$ to variables
$\xx_2,\yy_2,\dots$ (That is, in a predicate $A$, we add indices
$1,2$ to the variables $\xx,\yy$ to distinguish the variables of
$\bc_1$ and $\bc_2$.)  And $\denotc{\bc_1}(m_1)$ represents the
content of the memory after running $\bc_1$ on initial memory
$m_1$. (That is, $\denotc{\bc_1}$ is the denotational semantics of a
deterministic programming language.)

For example,
$\rhl{\xx_1=\xx_2}{\assign\xx{\xx+1}}{\assign\xx{\xx-1}}{\xx_1>\xx_2}$
says that if the variable $\xx$ has the same value in memory $m_1$ and
$m_2$, and we increase $\xx$ in $m_1$ and decrease $\xx$ in $m_2$
(resulting in memories $m_1',m_2'$, respectively), then $\xx$ in $m_1'$
is greater than $\xx$ in $m_2'$.

Given this definition of RHL, we can then prove a suitable collection
of sound rules such as, e.g.,
{\smaller\inferrule{A\implies
  A'\\\rhl{A'}{\bc_1}{\bc_2}{B}}{\rhl{A}{\bc_1}{\bc_2}{B}}}.
These rules can then be used to prove complex relationships between
deterministic programs.

\paragraph{Probabilistic RHL.} Deterministic RHL as described in the
previous paragraph can reason only about deterministic programs. At
first glance, it may seem that it is easy to generalize RHL to the
probabilistic case (pRHL). However, there are some subtleties in the
definition of pRHL. Consider the program $\bc$ that assigns a
uniformly random bit to $\xx$. In this case, is the following pRHL
judgment true? $\rhl{\true}\bc\bc{\xx_1=\xx_2}$. That is, if we have
two programs that pick a random bit, should pRHL say that after the
execution, the bit is the same? At the first glace, one might expect
that the answer should be ``no'' because $\xx_1$ and $\xx_2$ are
chosen independently and thus equal only with probability
$\frac12$. However, if we define pRHL that way, we cannot express the
fact that the two programs $\bc$ do the same. We might argue that
$\xx_1=\xx_2$ should just mean that $\xx$ is chosen according to the
same distribution in both programs. The formalization of pRHL from
\cite{certicrypt} takes exactly this approach. They define pRHL as follows:
\begin{definition}[Probabilistic RHL, informal]\label{def:prhl.informal}
  $\rhl A{\bc_1}{\bc_2}B$ holds iff: For any memories $m_1,m_2$ such
  that $\denotee A{m_1m_2}=\true$, there exists a distribution $\mu'$
  on pairs of memories such that: $\denotee B{m_1'm_2'}=\true$ for all
  $(m_1',m_2')\in\suppd\mu'$, and $\mu'_1=\denotc{\bc_1}(m_1)$ and
  $\mu'_2=\denotc{\bc_2}(m_2)$.
\end{definition}
Here $\suppd\mu'$ is the support of $\mu'$, i.e., the set of all
$(m_1',m_2')$ that have non-zero probability according to $\mu'$. And
$\denotc{\bc_1}(m_1)$ is the probability distribution of the memory
after executing $\bc_1$ with initial memory $m_1$. And $\mu'_1$ and
$\mu'_2$ refer to the marginal distributions of $\mu'$ that return the
first and second memory, respectively.

What does this definition say, intuitively? It requires that there is
some hypothetical joint distribution $\mu'$ on pairs of memories (that
satisfy the postcondition $B$) such that the final memories of $\bc_1$
and $\bc_2$ are distributed according to the marginals of $\mu'$. In
other words, it requires that there is some way of choosing the
randomness of $\bc_1$ and $\bc_2$ in a joint way such that $B$ becomes
true. For example, when $\bc$ picks a uniformly random bit $\xx$, we
can chose $\mu'$ as the uniform distribution on $\{(0,0),(1,1)\}$
(where we write $0$ and $1$ for the memory that assigns $0$ or $1$,
respectively, to $\xx$). Then the two marginal distributions are the
uniform distribution, thus $\mu'_1=\denotc{\bc}(m_1)$ and
$\mu'_2=\denotc{\bc}(m_2)$, and we have that all
$(m_1,m_2)\in\suppd\mu'$ satisfy $B:=(\xx_1=\xx_2)$. This implies that
$\rhl{\true}\bc\bc{\xx_1=\xx_2}$. A curiosity is that the following
judgment is also true: $\rhl{\true}\bc\bc{\xx_1\neq\xx_2}$ (choose
$\mu'$ as the uniform distribution on $\{(0,1),(1,0)\}$). In
particular, the following rule does \emph{not} hold for pRHL (but it
does hold for deterministic RHL): {\smaller\inferrule{\rhl
  A{\bc_1}{\bc_2}{B_1}\\\rhl A{\bc_1}{\bc_2}{B_2}}{\rhl
  A{\bc_1}{\bc_2}{B_1\land B_2}}}.  Reference \cite{certicrypt} derives a number of
useful rules for pRHL. Using these rules, one can derive complex
relationships between probabilistic programs. For example, the
successful EasyCrypt tool \cite{easycrypt,fosad-easycrypt} uses these to prove the security of
various cryptographic schemes. 

pRHL is well-suited for formalizing cryptographic security proofs
because those proofs usually proceed by transforming the initial
protocol (formulated as a program, called a ``game'') in a sequence of
small steps to a trivial protocol, and the relationship between any
two of those games can be analyzed using pRHL.

\subsection{Quantum RHL}
\label{sec:intro.qrhl}

\paragraph{The need for a new logic.}
To illustrate that a new logic is needed, even if we consider only
classical protocols (but quantum adversaries, i.e., post-quantum
cryptography), we analyzed the CHSH game \cite{chsh,cleve04nonlocal}
in EasyCrypt. In this game, two potentially malicious parties Alice
and Bob may share a state but are not allowed to communicate. Alice
and Bob get uniformly random bits $x,y$,
respectively, and have to respond with bits $a,b$,
respectively. They win iff $a\oplus b=x\cdot y$.
In the classical setting, the probability of winning is at most
$\frac34$,
while in the quantum setting, it is known that the maximum winning
probability is $\cos^2(\pi/8)\approx 0.85$
\cite{tsirelson,cleve04nonlocal}.  The CHSH game can be seen as a
particularly simple case of a multi-prover proof system, see
\cite{cleve04nonlocal}.

In EasyCrypt, we show that no adversary (consisting of Alice and Bob)
can win the CHSH game with probability greater than $\frac34$.
Since a quantum adversary would be able to get a larger success
probability, this shows that the logic employed by EasyCrypt indeed
assumes that adversaries are classical. And if we use
EasyCrypt to reason about quantum adversaries, it will give unsound
results.

We stress that this is not a flaw in EasyCrypt. It was never claimed
that EasyCrypt models quantum adversaries. However, since the logic
used by EasyCrypt is not known to be complete, it could have been that
all rules hold ``by accident'' also in the quantum setting. The CHSH
example shows that this is not the case.

The code of the EasyCrypt development can be found at \cite{chsh-ec}.

In addition, we performed a similar proof (also of the the CHSH game)
in the CryptHOL framework \cite{crypthol}, with analogous results. The
resulting Isabelle theory
can be found at \cite{Chsh.thy}.

We would like to address one common fallacy in this
context: to prove post-quantum security, we could just use existing
tools (and proofs) and strengthen them simply by limiting the hardness
assumptions to quantum-safe hardness assumptions (e.g., we do not use
the hardness of factoring numbers, but, for example, the hardness of
certain lattice problems as our underlying assumption).
Unfortunately, this does not work.  The underlying argument ``if $X$
is proven secure classically based on assumption $Y$,
and $Y$
is quantum-safe, then $X$
is post-quantum secure'' is a common fallacy.  Examples of protocols
with classical security proofs that are insecure against quantum
adversaries, even when instantiated with quantum hard assumptions, can
be found, e.g., in \cite{qpok-imposs}. Also, the CHSH game above
constitutes a degenerate counterexample to this fallacy: It can be
seen as a multiplayer cryptographic protocol\footnote{Where security
  means that the adversary cannot win with probability greater $3/4$.}
that needs no assumptions, thus ``hardening'' the assumptions would
not change anything, so the above fallacy would tell us that CHSH
itself is quantum-secure. Which it is not.

\paragraph{Quantum programs.} In order to develop a quantum variant of
pRHL, we first need to fix the language used for describing quantum
programs. We make the following design choices: Our language is a
simple imperative programming language (similar to the language pWhile
underlying pRHL) with both classical and quantum variables.  We model
only classical control (i.e., conditions in if- and while-statements
depend only on classical variables). This is sufficient for our use
case of modeling quantum cryptography, since protocols in that field
are typically described by semi-formal pseudocode with classical
control, and with mixed classical and quantum operations.   The language supports classical
random sampling,\footnote{Although random sampling can be encoded by a
  combination of unitaries and measurements, random sampling is so
  common in cryptography that we prefer to include it explicitly.}
application of unitaries, projective measurements. For simplicity, we
do not model procedure calls. However, we expect that adding
procedures would not lead to any difficulties. Since we do not include any
advanced language features, the semantics of the language are easily
defined: The denotational semantics $\denotc\bc$ of a quantum program
$\bc$ is described by a function that maps quantum states (the state
of the memory before execution) to quantum states (the state of the
memory after execution). Here quantum states are described by density
operators (mixed states) since the language contains probabilistic
elements (random sampling, measurement outcomes); density operators
intuitively model probability distributions of quantum states.
The details of the language and its semantics are given in 
\autoref{sec:lang}.

\paragraph{Quantum RHL -- predicates.} To extend pRHL to the quantum setting, we
need to give meaning to judgments of the form $\rhl A\bc\bd B$ where
$\bc,\bd$ are quantum programs. For this, we need to answer two
questions: How are the predicates (pre- and post-conditions) $A,B$
formalized (i.e., what mathematical structure describes $A,B$)? And
what is the semantics of the overall qRHL judgment $\rhl A\bc\bd B$?

In classical RHL (both the deterministic and pRHL), predicates are
Boolean predicates on pairs of memories $m_1,m_2$ (cf.~the notation introduced
after \autoref{def:drhl.informal}). Equivalently, we can see a
predicate~$A$ as a set of pairs of memories, and a distribution~$\mu'$ over
pairs of memories (cf.~\autoref{def:prhl.informal}) satisfies~$A$ iff
$\suppd\mu'\subseteq A$.

For formalizing quantum predicates $A$, there are two natural
choices:
\begin{compactenum}[(a)]
\item\label{mixed-set} $A$
  could be represented as a set of density
  operators (mixed states)\footnote{%
    Density operators are a formalism for
      describing quantum states where a quantum state is represented
      by a positive Hermitian operator. This formalism allows us to model
      probabilistic mixtures of quantum states (roughly speaking, probability distributions
      over quantum states). See, e.g., \cite[Chapter 2.4]{nielsenchuang-10year}.
    } (possibly subject to some natural closure
  properties, e.g., closure under multiplication with a scalar). And
  then $\rho$ satisfies $A$ iff $\rho\in A$.
\item\label{pure-set} $A$
  could be represented as a set of pure states\footnote{%
    Pure states are states that are represented as vectors in a
      Hilbert space. This more elementary formalism cannot represent
      probability distributions. Pure states can be seen as special
      cases of density operators (of rank 1).
      See, e.g., \cite[Chapter 2.2]{nielsenchuang-10year}.} (possibly subject to some
  natural closure properties, e.g., a subspace of the Hilbert space of
  pure quantum states). Then $\rho$
  satisfies $A$
  iff $\suppo\rho\subseteq A$.
  Here $\suppo\rho$
  is the support of $\rho$.
  (Intuitively, $\rho$
  is a probability distribution over states in $\suppo\rho$.
  For a formal definition see the preliminaries.)
\end{compactenum}

Obviously, approach \eqref{mixed-set} is the more general
approach. However, in this work, we will model predicates as subspaces
of pure states, following approach \eqref{pure-set}.
The reason for this is that we would like two properties to be satisfied:
\begin{compactitem}
\item If $\rho_1$ and $\rho_2$ both satisfy $A$, then $p_1\rho_1+p_2\rho_2$ satisfies $A$.

  For example, consider a program $\bc$
  that runs $\bc_1$
  with probability $p_1$,
  and $\bc_2$
  with probability $p_2$.
  Assume that $\bc_1$
  and $\bc_2$
  both have final states $\rho_1,\rho_2$
  that satisfy $A$.
  Then we would expect that $\bc$
  also has a final state that satisfies $A$.
  This is guaranteed if $p_1\rho_1+p_2\rho_2$
  satisfies $A$.
\item If $p_1\rho_1+p_2\rho_2$
  satisfies $A$ (for $p_1,p_2>0$), then $\rho_1$ and $\rho_2$ both satisfy $A$.

  For example, consider a program $\langif b\bc\bd$,
  and an initial state $\rho=p_1\rho_1+p_2\rho_2$
  where $\rho_1$
  satisfies the condition $b$,
  and $\rho_2$
  satisfies $\lnot b$.
  Let $A$
  be some predicate that $\rho$
  satisfies.  (E.g., the $\rho_1,\rho_2$
  could differ in the value of a classical variable that the predicate
  $A$
  does not even talk about.)  Then we would expect that after
  branching depending on $b$,
  at the beginning of the execution of $\bc$
  and $\bd$,
  respectively, $A$
  is still satisfied. That is, $\rho_1$
  and $\rho_2$
  still satisfy $A$.
  (This is the basis of, for example, the \refrule{If1} in our logic.)
\end{compactitem}
The only sets that satisfy both conditions are those represented by a
subspace (at least in the finite dimensional case, see
\autoref{sec:pred.subspace} for a proof).

That is, a predicate $A$ is a subspace of the Hilbert space of
quantum states, and each $\psi\in A$ is a pure quantum state over two
memories (i.e., a superposition of states $\basis{}{m_1,m_2}$ where
$m_1,m_2$ are classical memories that assign classical values to all
variables of $\bc$ and $\bd$, respectively).

The precise definition of predicates will look somewhat different
because of a special treatment of classical variables in our
predicates. This is, however, just for notational convenience and will
be explained in \autoref{sec:predicates}.

\medskip

Note that our reasoning above introduced one implicit assumption:
Namely that we want to model predicates for which ``$\rho$
satisfies $A$''
is either true or false.  Instead, it would be possible to have
predicates that are fulfilled to a certain amount (say between $0$
and $1$).
For example, the quantum predicates modeled by D'Hondt and Panangaden
\cite{dhondt06weakest} are Hermitian operators, and a state $\rho$
satisfies a predicate $A$
to the amount $\tr A\rho$.
Hoare judgments then encode the fact that the postcondition is
satisfied at least as much as the precondition. We have opted not to
pursue this approach in this paper (in particular because the
resulting predicates are harder to understand, and we would also
deviate more strongly from the design choices made in pRHL and
EasyCrypt), but it would be interesting future work to explore such
predicates in the relational setting.

\medskip

To actually work with predicates, we introduce some syntactic
sugar. For example, $\CL{b}$
is the space of all states where the classical variables satisfy the
predicate $b$.
(E.g., $\CL{\xx_1=\xx_2+1}$
relates two classical variables.) This allows us to express all
classical predicates with minimal overhead. And $A\cap B$
is simply the intersection of two predicates, which is the analogue to
classical conjunction $\land$.

\paragraph{Quantum RHL -- judgments.} We now explain our definition
of qRHL judgments. The first attempt would be to generalize the
classical definition (\autoref{def:prhl.informal}) to the quantum case
by replacing classical objects by their quantum counterparts. For
example, the natural quantum analogue for a probability distribution
over some set $X$ is a density operator on a Hilbert space with basis
$X$. To get this analogue, we first restate \autoref{def:prhl.informal}
in an equivalent way:
\begin{definition}[Probabilistic RHL, informal]\label{def:prhl.informal2}
  $\rhl A{\bc_1}{\bc_2}B$ holds iff: For any distribution $\mu$ on
  pairs of memories with $\suppd\mu\subseteq A$, there exists a
  distribution $\mu'$ on pairs of memories with
  $\suppd\mu'\subseteq B$ such that $\mu'_1=\denotc{\bc_1}(\mu_1)$ and
  $\mu'_2=\denotc{\bc_2}(\mu_2)$.
\end{definition}
Here -- in slight abuse of notation -- we consider the predicate $A$ as the
set of all pairs of memories that satisfy this predicate.  And
$\mu_1,\mu_2,\mu'_1,\mu'_2$ are the first and second marginals of the
distributions $\mu$ and $\mu'$, respectively.  And we lift the
denotational semantics $\denotc{\cdot}$ to distributions in a natural
way: $\denotc{\bc_1}(\mu_1)$ is the distribution of the memory after
the execution of $\bc_1$ if the memory was distributed according to
$\mu_1$ before the distribution.

In comparison with 
\autoref{def:prhl.informal},
instead of deterministic initial
memories $m_1,m_2$, we allow distributions $\mu$ of pairs of initial
memories. It is easy to see that this definition is equivalent to
\autoref{def:prhl.informal}. However, the new definition has the
advantage of being more symmetric (we use distributions both for the
input and the output), and that it is more obvious how to transfer it
to the quantum case.

If we restate this definition in the quantum case in the natural way,
we get:
\begin{definition}[Quantum RHL, first try, informal]\label{def:qrhl.first.informal}
  \symbolindexmark{\rhlfirst}%
  Let $A,B$ be quantum predicates (represented as subspaces, see above).

  $\rhlfirst A{\bc_1}{\bc_2}B$ holds iff: For any density operator $\rho$ over
  pairs of memories with $\suppo\rho\subseteq A$, there exists a
  density operator $\rho'$ over pairs of memories with
  $\suppo\rho'\subseteq B$ such that $\rho'_1=\denotc{\bc_1}(\rho_1)$ and
  $\rho'_2=\denotc{\bc_2}(\rho_2)$.
\end{definition}
Here we replace the distribution $\mu$ by a density operator
$\rho$. That is, we consider a quantum system where we have two
memories, and the system can be in arbitrary superpositions and
probabilistic mixtures of classical assignments to the variables in
these memories. (Formally, a memory is a function mapping variables
to values, and we consider a quantum system with basis
$\basis{}{m_1,m_2}$, where $m_1,m_2$ range over all such functions.)

Then $\rho_1$ denotes the density operator resulting from tracing out
the second subsystem (i.e., the result of ``destroying'' the second
memory), and analogously for $\rho_2,\rho_1',\rho_2'$.

And $\suppo\rho$
is the support of $\rho$,
i.e., it contains all pure states that the mixed state $\rho$
is a mixture of. (Formally, the support also contains linear
combinations of those pure states, but that does not make a
difference in the present case, since $A$
and $B$
are closed under linear combinations anyway.)  Note that a state
$\rho$
with marginals $\rho_1,\rho_2$ is called a \emph{quantum coupling}%
\index{quantum coupling}%
\index{coupling!quantum} of $\rho_1,\rho_2$
(see, e.g., \cite{strassen}). So the definition could be rephrased
compactly as: ``If the initial states of $\bc,\bd$
have a quantum coupling with support in $A$,
then the final states have a quantum coupling with support in $B$.''
Similarly, our other variants of qRHL can be straightforwardly stated
in terms of quantum couplings.

At first glance, \autoref{def:qrhl.first.informal} seems to be a
reasonable analogue of pRHL in the quantum setting. At least, it seems
to convey the same intuition. However, when we try to derive a
suitable set of rules for qRHL judgments, we run into trouble. For
example, we were not able to prove the following frame rule:
\[
  \RULEX{Frame\textnormal{ (simplified)}}{\rhlfirst{A}\bc\bd{B}\\
  \text{variables in $R$ are disjoint from variables in $\bc,\bd,A,B$}}
  {\rhlfirst{A\cap R}\bc\bd{B\cap R}}
\]
This rule is crucial for compositional reasoning: It tells us that any
predicate that is independent of the programs under consideration
(e.g., one that is only important for a program executed before or after)
will be preserved. (Note that we were not able to \emph{prove} that
the frame rule does not hold for \autoref{def:qrhl.first.informal},
yet we conjecture that it does not.)

There are a number of natural variations of
\autoref{def:qrhl.first.informal} that can be tried, but for most of
them, we either failed to prove the frame rule, or we run into
problems when trying to prove rules involving predicates that state
the equality of quantum variables (see below). We discuss the various
unsuccessful approaches in \autoref{sec:alt.def}, since we believe that it is
important to understand why we chose the definition in this paper.

The approach that turned out to work (in the sense that we
can derive a useful set of rules), is to use the
\autoref{def:qrhl.first.informal}, with the only difference that we
quantify over \emph{separable} states $\rho$ and $\rho'$ only. That
is, we only consider states where the memories of the two programs are
not entangled.
(We stress that we do not restrict the entanglement \emph{within} the programs,
only \emph{between} the programs. In particular, there are no restrictions on algorithms
that use entanglement.)
That is, the definition of qRHL is:
\begin{definition}[Quantum RHL, informal]\label{def:qrhl.informal}
  Let $A,B$ be quantum predicates (represented as subspaces, see above).

  $\rhl A{\bc_1}{\bc_2}B$ holds iff: For any \emph{separable} density operator $\rho$ over
  pairs of memories with $\suppo\rho\subseteq A$, there exists a
  \emph{separable} density operator $\rho'$ over pairs of memories with
  $\suppo\rho'\subseteq B$ such that $\rho'_1=\denotc{\bc_1}(\rho_1)$ and
  $\rho'_2=\denotc{\bc_2}(\rho_2)$.
\end{definition}
Note that the only difference to \autoref{def:qrhl.first.informal} is
the occurrence of the word ``separable'' (twice).

With this definition (stated formally in \autoref{sec:qrhl},
\autoref{def:rhl}), we get a suitable set of rules for reasoning about
pairs of programs (frame rule, rules for predicates involving
equality, rules for reasoning about the different statements in
quantum programs, case distinction, sequential composition, etc.)

\paragraph{Relationship to classical pRHL.}
By
considering programs that contain no quantum instructions or quantum
variables, and by considering only pre-/postconditions of the form
$\CL{a}$,
we get a definition of a classical relational Hoare logic as a special
case. As it turns out, all three variants of qRHL (\autoref{def:qrhl.first.informal},
\autoref{def:qrhl.informal}, and a third approach discussed in \autoref{sec:alt.def}) specialize to the
same variant of pRHL, namely the original pRHL
(as in \autoref{def:prhl.informal}). (Shown in \autoref{sec:rel.classical}.)

\paragraph{On coupling proofs.} \cite{barthe15relational} and
\cite{barthe17coupling} elaborate on the connection between pRHL and
so-called probabilistic couplings. It is often asked whether couplings
can also be used as the basis for qRHL. The approach of
\cite{barthe15relational} uses the fact that pRHL \emph{by definition}
guarantees that there exists a probabilistic coupling for the output
distributions of the programs that are related via a pRHL judgment,
thus pRHL can be used to derive the existence of couplings with
certain properties. This observation carries over immediately to qRHL
because qRHL by definition implies the existence of quantum
couplings. (In fact, the use cases from \cite{barthe15relational} can
be expressed in qRHL as easily. An interesting question is, of course,
whether interesting purely quantum use cases can be found where this
method is fruitful.)
So in that sense, the answer is ``yes'', qRHL is based on couplings.

On the other hand \cite{barthe17coupling} defines
a \emph{new} logic in which an explicit product program is constructed
that produces the coupling. (This is then shown to be equivalent to
the original pRHL, but having explicit access to this program has
various advantages detailed in \cite{barthe17coupling}.)
Unfortunately, this approach most likely does not carry over to the
quantum setting. In fact, one of the qRHL variants
(\autoref{def:qrhl.uniform}) that we explored is the quantum
equivalent of this approach. We show that this approach is
incompatible with the notion of a quantum equality (see
below). Details are given in \autoref{sec:alt.def}.

\paragraph{Quantum equality.} Using the approach from the previous
paragraphs, we have a working definition of qRHL judgments, and a
formalism for describing pre- and postconditions as subspaces. We can
express arbitrary conditions on classical variables (in particular, we
can express that $\xx_1=\xx_2$ for variables $\xx_1,\xx_2$ of the left
and right program, respectively). However, we cannot yet conveniently
relate quantum variables in the left and the right program. That is, we cannot
express judgments such as, for example:
\[
\rhl{\qq_1\quanteq\qq_2}\bc\Skip{\qq'_1\quanteq\qq_2}
\]
where $\bc$ swaps $\qq_1$ and $\qq_1'$. Here $\quanteq$ denotes some
(yet to be formalized) equality notion on quantum states.

Note that even if our goal is to use qRHL only for security proofs in
post-quantum cryptography we need a notion of quantum equality: In
this setting, all protocols contain only classical variables, but the
adversary will contain quantum variables. In this context, typical
predicates (almost ubiquitous in classical EasyCrypt proofs) are of
the form:
$A\land \yy_1=\yy_2$ where
$\yy_1,\yy_2$ refers to the global state of the adversary in the
left/right program, and $A$ is some predicate on program variables.
In the post-quantum setting, $\yy_1,\yy_2$ will be
quantum variables, and we need to express quantum equality.

So, how can quantum equality be defined? In other words, what is the
formal meaning of $\qq_1\quanteq\qq_2$ as a quantum predicate where
$\qq_1,\qq_2$ are quantum variables (or, in the general case,
sequences of quantum variables)? If we do not limit ourselves to
predicates that are represented by subspaces (i.e., if a predicate can
be represented as an arbitrary set of density operators), several
natural equality notions exist:
\begin{compactenum}[(a)]
\item $\rho$ satisfies $\qq_1\quanteq\qq_2$ iff
  $\rho_{\qq_1}=\rho_{\qq_2}$ where $\rho_{\qq_1},\rho_{\qq_2}$ are
  the density operators resulting from tracing out all registers
  except $\qq_1,\qq_2$, respectively.
\item\label{swap.rho} $\rho$ satisfies $\qq_1\quanteq\qq_2$ iff $\rho$ is invariant
  under swapping $\qq_1$ and $\qq_2$. That is $\rho=U\rho U^\dagger$
  where $U$ is the unitary that swaps $\qq_1$ and $\qq_2$.
\item\label{swap.psi} $\rho$
  satisfies $\qq_1\quanteq \qq_2$
  iff $\rho$
  is a mixture of states that are fixed under swapping $\qq_1$
  and $\qq_2$.
  That is, for all $\psi\in\suppo\rho$,
  $U\psi=\psi$.\footnote{This
    is not the same as \eqref{swap.rho}, even though it looks
    similar. For example, if
    $\psi:=\frac1{\sqrt2}\basis{\qq_1\qq_2}{01}-\frac1{\sqrt2}\basis{\qq_1\qq_2}{10}$,
    and $\rho:=\proj\psi:=\adj\psi\psi$
    is the corresponding density operator, then
    $U\rho U^\dagger=\rho$,
    hence $\rho$
    satisfies $\qq_1\quanteq \qq_2$
    according to \eqref{swap.rho}, but $U\psi=-\psi\neq\psi$
    and $\psi\in\suppo\rho$,
    so $\rho$
    does not satisfy $\qq_1\quanteq \qq_2$
    according to \eqref{swap.psi}.  }
\end{compactenum}
Out of these three, only \eqref{swap.psi} is a predicate in our sense
(i.e., described by a subspace).  In fact,  we conjecture that
\eqref{swap.psi} is the \emph{only} notion of equality of quantum
variables that can be represented as a subspace and that has the
following property (we give partial formal evidence for this
in \autoref{sec:qeq.unique}):
\[
  \pb\rhl
  {\CL{X_1=X_2}\cap (Q_1\quanteq Q_2)}
  \bc\bc
  {\CL{X_1=X_2}\cap (Q_1\quanteq Q_2)}
\]

Here $X,Q$
are the classical and quantum variables of $\bc$,
respectively. (And $X_1,X_2,Q_1,Q_2$
are those variables indexed with $1$
and $2$,
respectively).  This rule says that, if all variables of $\bc$ on the left and
right are the same before running $\bc$
on both sides, then they are the same afterwards, too. This
is an essential property a notion of equality should satisfy.

Thus we get the following notion of quantum equality (formulated slightly
more generally than in~\eqref{swap.psi}):
\begin{definition}[Quantum equality, informal]
  The predicate
  $(\qq^{(1)}_1\qq_1^{(2)}\dots\ \quanteq\
  \qq^{(1)}_2\qq_2^{(2)}\dots)$ is the subspace $\{\psi:U\psi=\psi\}$
  where $U$
  is the unitary that swaps registers $\qq^{(1)}_1\qq_1^{(2)}\dots$
  with the registers $\qq^{(1)}_2\qq_2^{(2)}\dots$
\end{definition}
An example for a state satisfying $\qq_1\quanteq\qq_2$
would be $\psi_1\otimes\psi_2\otimes\phi$
where $\psi_1$
and $\psi_2$ are the same state on $\qq_1$ and $\qq_2$, respectively.

This quantum equality allows us to express, for example, relations of
the form
$(\text{``some predicate on classical variables''}\cap \qq_1\quanteq\qq_2)$
which will be enough for many proofs in post-quantum cryptography.
(We give an example of such a proof in \autoref{sec:enc.example}.)
More generally, however, we might wish for more expressivity. For
example, we might wish to express that if $\qq_1\quanteq \qq_2$,
and then we apply the Hadamard $H$
to $\qq_1$,
then afterwards $\qq_1$
equals ``$H$
applied to $\qq_2$''.
We generalize the definition of $\quanteq$
to cover such cases, too. Intuitively,
$(U_1\ \qq^{(1)}_1\qq_1^{(2)}\dots\ \quanteq\ U_2\
\qq^{(1)}_2\qq_2^{(2)}\dots)$ means that the variables
$\qq^{(1)}_1\qq_1^{(2)}\dots$,
when we apply the unitary $U_1$
to them, equal the variables $\qq^{(1)}_2\qq_2^{(2)}$
when we apply the unitary $U_2$
to them. We defer the details of the definition to
\autoref{sec:quantum.eq}.
We only mention that, for example, $U_1^{-1}\psi_1\otimes U_2^{-1}\psi_2\otimes\phi$
satisfies the predicate $U_1\qq_1\quanteq U_2\qq_2$.

\paragraph{Incompleteness.} While we prove the
soundness of all rules in this paper, we do not strive to achieve
completeness. (Incompleteness of pRHL was noted in
\cite{barthe17coupling} and carries over to our setting.)  In fact, we
cannot expect completeness (except for restricted classes of programs)
since equivalence even of programs is not semidecidable.

\subsection{Tool \& experiments}

We develop an interactive theorem prover for reasoning about qRHL.  The tool consists of
three main components: a ProofGeneral \cite{proofgeneral} frontend,
the core tool written in Scala, and an Isabelle/HOL \cite{isabelle}
backend for solving subgoals in the ambient logic (i.e., subgoals that
do not contain qRHL judgments). The ProofGeneral frontend eases the
interactive development of proofs. The Isabelle backend can 
also be accessed directly (via accompanying Isabelle theory files) to prove
verification conditions that are beyond the power of our core tool.

We stress that although we use Isabelle/HOL as a backend, this does
not mean that our tool is an LCF-style theorem prover (i.e., one that
breaks down all proofs to elementary mathematical proof steps).  All
tactics in the tool, and many of the simplification rules in
Isabelle are axiomatized (and backed by the proofs in this
paper).

We developed a number of examples in our tool:
\begin{citemize}
\item A security proof that the encryption scheme defined by
  $\mathrm{enc}(k,m):=m\oplus G(k)$
  is ROR-OT-CPA secure.\footnote{IND-OT-CPA is indistinguishability
      under chosen plaintext attacks (e.g., \cite[Def. 3.22]{katz2014introduction}) for one-time
      encryption.  ROR-OT-CPA is a different formulation of the
      same property.}
  (Where $G$ is a pseudorandom generator.)
\item A security proof that the same encryption scheme is IND-OT-CPA
  secure. These two proofs are typical examples of reasoning in
  post-quantum security where quantum variables occur only inside
  the adversary, and where we maintain the invariant $Q_1\quanteq Q_2$
  where $Q_1,Q_2$ are all quantum variables.
\item A simple example illustrating the interplay of quantum equality
  and the applications of unitaries.
\item A proof that quantum teleportation works.  That is, we show
  $\rhl{\qq_1=\qq_2}{\bc_{\mathit{teleport}}}\Skip{\qq'_1=\qq_2}$
  where $\bc_{\mathit{teleport}}$
  is the program that teleports a qubit from $\qq$
  to $\qq'$.
  This proof involves reasoning about quantum equality, measurements,
  unitary operations, initializations, as well as interaction between
  our tool and Isabelle/HOL for non-trivial subgoals.
\end{citemize}

See the manual included with the tool for a description of the architecture of the
tool, documentation of its use, and details on the examples.
The source code is published on GitHub \cite{github-source}, 
a binary distribution is available at \cite{qrhl-binary}.
Note that the tool currently (version 0.3) only supports the rules \rulerefx{Seq},
\rulerefx{Equal}, \rulerefx{Conseq}, \rulerefx{QrhlElimEq}, \rulerefx{Skip},
\rulerefx{Assign1}\textsc{/2},
\rulerefx{Sample1}\textsc{/2},
\rulerefx{JointSample},
\rulerefx{Case},
\rulerefx{QApply1}\textsc{/2},
\rulerefx{QInit1}\textsc{/2},
\rulerefx{Measure1}\textsc{/2}.

\paragraph{On advanced cryptographic proofs.}
 The simple examples presented above are, of course, only
preliminary indications of the usability of qRHL for quantum
cryptography. The study of complex quantum protocols (e.g., quantum
key distribution) is out of the scope of this first work. For post-quantum
cryptography (i.e., the security of classical protocols against
quantum adversaries), our examples studying an
encryption scheme indicate that derivations of qRHL judgments
involving classical programs (but quantum adversaries) are very
similar to those that would be performed in EasyCrypt (except, of
course, that the EasyCrypt tool is more mature). See the example in \autoref{sec:enc.example}.
An exception to this
are post-quantum security proofs which need to explicitly reason about
quantum properties (even in the pen-and-paper proof). Typically, these
are proofs involving rewinding (e.g., \cite{watrous-qzk,qpok}) or
quantum random oracles (e.g.,
\cite{boneh11quantumro,qro-nizk,zhandry:quantum.ibe}). As far as we know, proofs
involving rewinding have not even been modeled classically in
EasyCrypt, so it would be premature to speculate on how difficult
these would be in qRHL. Proof involving the quantum random oracle
model are currently very important, in particular in light of the NIST
competition for post-quantum cryptography.
Quantum random oracles are
challenging because the adversary has superposition access to the
oracle. \emph{Modeling} the quantum random oracle is easy in our
setting,\footnote{If $\mathbf h$
  is a classical variable of type $X\to Y$,
  then a random oracle can modeled by the program
  $\sample{\mathbf h}{\mathtt{uniform}}$
  (where the expression $\mathtt{uniform}$
  is the uniform distribution on $X\to Y$)
  for initialization, and $\Qapply{\xx,\yy}{U_{\mathbf h}}$
  (where $U_f$
  is the unitary $\basis{}{x,y}\to\basis{}{x,y\oplus f(x)}$)
  for queries.} the challenge lies in reasoning about it. So-called
history free reductions \cite{boneh11quantumro} (in which the random oracle is
replaced by a differently chosen but equally distributed function)
should be relatively easy, because in those cases the
distribution of the random oracle after replacement is the same as
before, so no quantum-specific reasoning is needed.
(We checked this by implementing a very small such proof in our tool, \texttt{random-oracle.qrhl},
where the adversary cannot distinguish between those two cases given one query to the random oracle.)
However, most post-quantum security proofs in the quantum random
oracle model need more advanced techniques, such as, e.g.,
semi-constant distributions \cite{zhandry:quantum.ibe}, one-way-to-hiding lemmas
\cite{qtc-jacm}, oracle indistinguishability \cite{zhandry12random}, etc. Whether
(and how) those techniques can be incorporated into qRHL-based proofs,
or whether the logic will need to be extended to accommodate them,
will be part of our future work.

\subsection{Related work}
Frameworks for verification of cryptographic protocols (in the
so-called computational model) include CryptoVerif \cite{cryptoverif},
CertiCrypt \cite{certicrypt}, EasyCrypt \cite{easycrypt}, FCF
\cite{FCF}, CryptHOL \cite{crypthol}, and Verypto
\cite{verypto}. However, all of these only target classical
cryptography and do not support quantum programs/protocols.

Hoare logics and weakest precondition calculi for quantum programs
have been presented by D'Hondt and Panangaden \cite{dhondt06weakest},
Chadha, Mateus and Sernadas \cite{chadha06reasoning}, Feng, Duan, Ji,
and Ying \cite{feng07proof}, Ying \cite{ying12floyd},
and Kakutani \cite{kakutani09logic}. However, these calculi do not allow us to reason about
the relation between different programs.

Those quantum Hoare logics fall into two flavours:
\cite{dhondt06weakest,feng07proof,ying12floyd} take a semantic
approach in which a pre-/postcondition is an arbitrary subspace (or
more generally, Hermitian operator to model expectations), while
\cite{chadha06reasoning,kakutani09logic} represent pre-/postcondition
as terms of a specific (fixed) structure. Our work follows the
semantic approach.

In independent work, Zhou, Ying, Yu, and Ying \cite{strassen} study
quantum couplings, and also mention (in passing) the applicability of
quantum couplings to quantum relational Hoare logics. Using their
notions of couplings, the resulting definition would probably be very
similar to \autoref{def:qrhl.first.informal} (which had problems with
the frame rule, see above). Unfortunately, their results do not seem
to give a method for proving the frame rule, and their Quantum
Strassen Theorem does not seem to work for quantum couplings with
separable states (as in \autoref{def:qrhl.informal}).

Existing approaches for computer-aided verification of the equivalence of quantum programs
focus either on the explicit calculation of the quantum state
resulting from an execution of the programs, e.g.,
\cite{ardeshir13equivalence, ardeshir14verification}, or on
bisimilarity of process calculi, e.g., \cite{kubota12application,
  kubota13automated, feng15toward}. Explicit calculation will work
only for finite-dimensional (and relatively small) systems, and is thus
inapplicable to computational cryptography (which has necessarily an
exponentially large state space).  Bisimilarity can show the exact
equivalence between two programs, but does not seem to scale well for
modeling computational cryptographic proofs (even in the classical
setting, we are not aware of any nontrivial computational security proof based
purely on bisimilarity).

Several papers verify the security of the BB84 quantum key
distribution (QKD) protocol. \cite{feng15toward, tavala11verification}
show security against a very simple, specific intercept-resend attack
(that measures all qubits in a random basis),
\cite{kubota12application, kubota13automated} verify one step of the
Shor-Preskill proof \cite{shor-preskill} for QKD. It is not clear
whether these approaches can scale to the full proof of QKD.

\subsection{Organization}
\autoref{sec:prelim} presents notation and elementary
definitions. \autoref{sec:lang} gives syntax and semantics of quantum
programs. \autoref{sec:predicates} defines what a quantum predicate is
and discusses constructions of quantum predicates (in particular
quantum equality, \autoref{sec:quantum.eq}).  \autoref{sec:qrhl}
defines quantum relational Hoare logic and derives reasoning rules
for qRHL. \autoref{sec:example} shows example derivations in our logic.
\autoref{sec:rule-proofs} contains the proofs of all rules.
\autoref{sec:alt.def} discusses some alternative approaches to
defining qRHL and their problems.  

An extended abstract of this paper appeared at POPL 2019
\cite{qrhl-popl}.

\section{Preliminaries}
\label{sec:prelim}

\paragraph{A note on notation.}
Many of the proofs in this paper are
notationally quite involved. This is partly due to the fact that there
are a number of isomorphic spaces (e.g., states over variables, states
over variables with index 1), and the isomorphisms need to be
explicitly specified to make the proofs logically consistent (there is
too much ambiguity to leave them implicit).  Similarly, some of the
rules also contain variable renamings and isomorphisms.  We recommend,
upon first reading, to ignore all these isomorphisms (specifically
$\idx1$,
$\idx2$,
$\Urename{\dots}$,
$\Erename{\dots}$),
i.e., assume that they are the identity.  Without these, the lemmas and proofs,
while not technically correct, will still contain the relevant ideas.

\paragraph{Basics.}
\symbolindexmark\setR$\setR$
are the reals, and \symbolindexmark\setC$\setC$
the complex numbers. \symbolindexmark\setRpos$\setRpos$
are the reals $\geq 0$.
\symbolindexmark\bool\symbolindexmark\true\symbolindexmark\false$\bool=\{\true,\false\}$
are the Booleans. \symbolindexmark\im$\im f$
is the image of a function (or operator). \symbolindexmark\dom$\dom f$
is the domain of $f$.
\symbolindexmark\upd$\upd fxy$
denote the function update.  Formally $\bigl(\upd fxy\bigr)(x)=y$
and $\bigl(\upd fxy\bigr)(x')=f(x')$
for $x\neq x$.
\symbolindexmark\restrict$\restrict fM$
is the restriction of $f$
to the set $M$.
\symbolindexmark\id$\id$ denotes the identity function.
\symbolindexmark\powerset$\powerset M$ is the powerset of $M$.
$\circ$ can refer both to the composition of functions or of relations
(i.e., \symbolindexmark\circrel$R\circrel R':=\{(x,z):\exists y.(x,y)\in R\land (y,z)\in R'\}$).

\paragraph{Variables.}
A \emph{program variable}%
\index{program variable}%
\index{variable!program} $\xx$
(short: variable) is an identifier annotated with a
\symbolindexmark{\typev}set $\typev \xx\neq\varnothing$,
and with a flag that determined whether the variable is
\emph{quantum}%
\index{quantum variable}%
\index{variable!quantum}%
\index{program variable!quantum} or \emph{classical}.%
\index{classical variable}%
\index{variable!classical}%
\index{program variable!classical} (In our semantics, for classical variables $\xx$ the type
$\typev\xx$
will be the set of all values a classical variable can store. Quantum
variables $\qq$ can store superpositions of values in $\typev\qq$.)

We will usually denote classical variables
with \symbolindexmark\xx\symbolindexmark\yy$\xx,\yy$
and quantum variables with \symbolindexmark\qq$\qq$.
Given a set $V$
of variables, we write \symbolindexmark\cl$\cl V$
for the classical variables in $V$
and \symbolindexmark\qu$\qu V$
for the quantum variables in $V$.

Given a set $V$ of variables, we write \symbolindexmark\types$\types V$ for the set of all
functions $f$ on $V$ with $f(\xx)\in \typev\xx$ for all $\xx\in V$. (I.e.,
the dependent product $\types V=\prod_{\xx\in V}\typev\xx$.)

Intuitively, $\types V$
is the set of all memories that assign a classical value to each variable in $V$.

For sets $X,Y$ of variables, we call $\sigma:X\to Y$ a \index{variable renaming}%
\index{renaming!variable}%
\emph{variable renaming} iff $\sigma$
is a bijection, and for all $\xx\in X$,
$\typev{\sigma(\xx)}=\typev \xx$
and $\sigma(\xx)$
is classical iff $\xx$
is classical.

Given a list $V=(\xx_1,\dots,\xx_n)$
of variables, \symbolindexmark\typel$\typel V:=\typev{\xx_1}\times\dots\times\typev{\xx_n}$.
Note that if $V$
is a list with distinct elements, and $V'$
is the set of those elements, then $\typel V$
and $\types{V'}$
are still not the same set, but their elements can be identified
canonically. Roughly speaking, for a list $V$,
the components of $m\in\typel V$
are indexed by natural numbers (and are therefore independent of the
names of the variables in $V$),
while for a set $V$,
the components of $m\in\types V$ are indexed by variable names.

For $m_1\in\types X$
and $m_2\in\types Y$
with disjoint $X,Y$,
let \symbolindexmark\memuni$\memuni{m_1m_2}$
be the union of the functions $m_1$ and $m_2$.
I.e., $(\memuni{m_1m_2})(\xx)=m_1(\xx)$
for $\xx\in X$ and $(\memuni{m_1m_2})(\xx)=m_2(\xx)$ for $\xx\in Y$.

Given disjoint sets $X,Y$
of variables, we write $XY$ for the union (instead of $X\cup Y$).

\paragraph{Distributions.}
Given a set $X$
(possibly uncountable), let $\ellone X$
denote the set of all functions $\mu:X\to\setRpos$.
%with $\sum_{x\in X}\mu(x)\leq 1$.
We call the elements of $\ellone X$
\index{distribution}\emph{distributions}.\footnote{That is, we do not
  consider distributions over arbitrary measurable spaces, only over
  discrete ones.}
% The \emph{weight}%
% \index{weight!distribution}%
% \index{distribution!weight}
% of a distribution $\mu$ is $\sum_{x\in X}\mu(x)$.
% We call $\mu$%
% \index{total!distribution}\index{distribution!total}\emph{total} if it has weight $1$.
The set of distributions with
$\sum_{x\in X}\mu(x)\leq1$ is denoted \symbolindexmark\distr$\distr{X}$
(subprobability distributions).%
\index{subprobability distribution}%
\index{distribution!subprobability}
We call a distribution
\index{total!distribution}\index{distribution!total}\emph{total} if
$\sum_{x\in X}\mu(x)=1$.
For a set $V$
of variables, let
\symbolindexmark\ellonev$\ellonev V:=\ellone{\types V}$ and
\symbolindexmark\distrv$\distrv V:=\distr{\types V}$.

For $\mu\in\ellone{X\times Y}$, let \symbolindexmark\marginal$\marginal1\mu$ be the first marginal of $\mu$, that is,
$\marginal1\mu(m_1):=\sum_{m_2\in Y}\mu(m_1,m_2)$.
And let $\marginal2\mu$ be the second marginal of $\mu$, that is,
 $\marginal2\mu(m_2):=\sum_{m_1\in X}\mu(m_1,m_2)$.

The support \symbolindexmark\suppd$\suppd\mu$
of a distribution $\mu$
is defined as $\suppd\mu:=\{x:\mu(x)>0\}$.
We write \symbolindexmark\pointdistr$\pointdistr a$
for the distribution with $\pointdistr a(a)=1$
and $\pointdistr a(b)=0$ for $a\neq b$.

% \paragraph{Relations.}
% We call a relation $R$
% \index{one-to-one!relation}\emph{one-to-one} iff for any $x$,
% there exists at most one $y$
% with $(x,y)\in R$,
% and for any $y$,
% there exists at most one $x$
% with $(x,y)\in R$. Let \symbolindexmark\leftim$\leftim R:=\{x:\exists y. (x,y)\in R\}$.
% Let \symbolindexmark\rightim$\rightim R:=\{y:\exists x. (x,y)\in R\}$.

\paragraph{Hilbert spaces.}
We assume an understanding of Hilbert spaces and of different
operators on Hilbert spaces. In the following, we only fix notation
and ambiguous terminology. We refer the reader to a standard textbook
on functional analysis, e.g.,
\cite{conway97functional}, for an introduction into the concepts used
here. A reader interested only in the finite dimensional case (i.e.,
where all variables are defined over finite domains) will find the
fact helpful that in that case, the sets $\bounded X$
and $\tracecl X$
defined below are simply the finite $\abs X\times \abs X$
matrices (with rows and columns indexed by $X$),
and $\adj A$
is the conjugate transpose of the matrix $A$.
Furthermore, to interpret the results in this paper, an understanding
of the basics of quantum computing is required, see a
textbook such as \cite{nielsenchuang-10year}, especially Chapters 1, 2, 4, and 8.

\medskip

The word \index{subspace}\emph{subspace} always refers to a topologically
closed subspace, and the word \index{subspace}\emph{basis} refers to a Hilbert basis.
The word \index{operator}\emph{operator} always refers to linear operators.

Given a set $X$
(possibly uncountable), let
$\elltwo X\subseteq\setC^X$\symbolindexmark{\elltwo}
be the Hilbert space of all functions $\psi:X\to\setC$
such that \symbolindexmark\norm$\norm\psi^2:=\sum_{x\in X}\abs{\psi(x)}^2$
exists.
(That is, $\elltwo X$ represents pure quantum states.)
Let \symbolindexmark\basis$\basis{}{x}:X\to\setC$
denote the function defined by  $\basis{}x(x):=1$,
$\basis{}{y}(x):=0$ for $y\neq x$. The vectors $\basis{}{x}$
with $x\in X$
form a basis of $\elltwo X$, the \index{basis!computational}%
\index{computational basis}\emph{computational basis}.  For
$M\subseteq \elltwo X$,
\index{span}\symbolindexmark\SPAN$\SPAN M$ denotes the smallest subspace of $\elltwo X$ containing $M$.
For subspaces $M_i$ of $\elltwo X$, we define $\sum_i M_i:=\SPAN\bigcup_i M_i$.
For a subspace $S\subseteq\elltwo X$, \symbolindexmark\orth$\orth S$ denotes the orthogonal complement.

For a set $V$
of variables, let \symbolindexmark\elltwov$\elltwov V:=\elltwo{\types V}$.
Intuitively, elements of $\elltwov V$
can be thought of as superpositions of different memories with
variables $V$, i.e., as the set of all pure quantum states
that a memory with variables $V$ can be in.
For added clarity, we write \symbolindexmark\basis$\basis Vm$
instead of $\basis{}m$
for the elements of the computational basis of $\elltwov V$.
We write \symbolindexmark\idv$\idv V$ for the identity on $\elltwov V$.

In this notation, we can define the tensor product on Hilbert spaces of
the form $\elltwov V$
as follows: For disjoint $V_1,V_2$, let \symbolindexmark\tensor$\elltwov {V_1}\tensor\elltwov {V_2}:=\elltwov{V_1V_2}$.
And  $\psi\in\elltwov{V_1}$, $\phi\in\elltwov{V_2}$, define $\psi\tensor\phi\in\elltwov{V_1V_2}$
by $(\psi\tensor\phi)(\memuni{m_1m_2}):=\psi(m_1)\phi(m_2)$.

Note that this definition of the tensor product is isomorphic to the
usual definition. However, our definition ``labels'' the factors with
variables (i.e., each variable labels a quantum register). In
particular, $\tensor$ as defined here is commutative and associative.

\paragraph{Operator spaces.}
Let \symbolindexmark\bounded$\bounded{X,Y}$
denote the set of bounded operators from $\elltwo X$
to $\elltwo Y$.
For a bounded operator $A$,
let \symbolindexmark\adj$\adj A$
denote the \index{adjoint}\emph{adjoint} of $A$.
Let $\bounded X:=\bounded{X,X}$.
For a set $V$
of variables, let \symbolindexmark\boundedv$\boundedv{V,W}:=\bounded{\types V,\types W}$,
i.e., the bounded operators from $\elltwov V$ to $\elltwov{W}$. And $\boundedv{V}:=\boundedv{V,V}$.

We write \symbolindexmark\idv$\idv V$
for the identity on $\boundedv V$.
(Note: we use the same symbol as for the identity on $\elltwov V$.)

We define \symbolindexmark\tensor$\boundedv{V_1,W_1}\tensor\boundedv{V_2,W_2}:=\boundedv{V_1V_2,W_1W_2}$,
and for $A\in\boundedv{V_1,W_1},B\in\boundedv{V_2,W_2}$, we define
$A\tensor B\in\boundedv{V_1V_2,W_1W_2}$ as
the unique bounded operator satisfying
$(A\tensor B)(\psi\tensor\phi):=A\psi\tensor B\phi$.

We say $A\leq B$ iff $B-A$ is self-adjoint and positive.

We define \symbolindexmark\boundedleq$\boundedleq{X,Y}$ as the bounded operators with operator norm $\leq 1$. And $\boundedleq{X}:=\boundedleq{X,X}$.

Let \symbolindexmark\iso$\iso{X,Y}\subseteq\bounded{X,Y}$
denote the set of all isometries from $\elltwo X$
to $\elltwo Y$,
that is, operators $U$
with $\adj UU=\id$.
(Note: we do not require $U\adj U=\id$,
i.e., we do not assume $U$ to be a unitary.) 
Let \symbolindexmark\isov$\isov{V,W}:=\iso{\types V,\types W}$.
Furthermore, $\iso X:=\iso{X,X}$ and $\isov V:=\isov{V,V}$.

Let \symbolindexmark\uni$\uni{X,Y}\subseteq\iso{X,Y}$ denote the set of all unitaries
from $\elltwo X$
to $\elltwo Y$,
that is, operators $U$
with $\adj UU=\id$ and $U\adj U=\id$.
Let \symbolindexmark\univ$\univ{V,W}:=\uni{\types V,\types W}$.
Furthermore, $\uni X:=\uni{X,X}$ and $\univ V:=\univ{V,V}$.

Let \symbolindexmark\tracecl$\tracecl{X}\subseteq\bounded{X}$
denote the set of all \index{trace-class}\emph{trace-class} operators,\footnote{Trace-class operators are operators for which
    the trace exists. In the finite dimensional case, every operator is trace-class.} and let
\symbolindexmark\tr$\tr A$
denote the \emph{trace}\index{trace} of $A$.
Let \symbolindexmark\tracepos$\tracepos X\subseteq\tracecl X$
denote the positive self-adjoint trace-class operators. (I.e.,
$A\in\tracepos X$ iff there exists $B\in\tracecl X$ with $A=\adj BB$.)
We abbreviate
\symbolindexmark\traceclv$\traceclv V:=\tracecl{\types V}$,
\symbolindexmark\traceposv$\traceposv V:=\tracepos{\types V}$.
(That is, $\traceposv V$ is the set of all \emph{mixed} quantum states
  that a memory with variables $V$ can be in.)

When we write $\sum_{i\in I}\rho_i$
for some $\rho_i\in\tracepos X$,
we do not assume a finite or countable set $I$.
$\sum_i\rho_i$
is defined as the least upper bound of $\sum_{i\in J}\rho_i$
where $J$ ranges over all finite subsets of $I$.
$\sum_i\rho_i\in\tracepos X$ exists iff $\sum_i\tr\rho_i$ exists.

A \index{projector}\emph{projector} is an operator $P\in\bounded X$
with $P^2=P=\adj P$.
(I.e., projectors always are orthogonal projectors.)  For
$\psi\in\elltwo X$,
let \symbolindexmark\proj$\proj\psi:=\psi\adj\psi$.
(Then $\proj\psi$ is a projector onto $\SPANO\psi$.)

A \emph{cq-operator}\index{cq-operator}\index{operator!cq-} is an operator
$\rho\in\traceclv V$ that can be written as
$\rho=\sum_m \proj{\basis{\cl V}{m}}\tensor\rho_m$ with
$\rho_m\in\traceclv{\qu V}$. (``cq'' stands for
``classical-quantum''.) Here $m$ ranges over $\types{\cl V}$.
(Intuitively, this means that the classical part of the memory is $m$
with probaility $\tr\rho_m$, and in this case, the quantum part of the
memory is in state $\rho_m/\tr\rho_m$.)  We write
\symbolindexmark{\traceclcq}$\traceclcq V\subseteq\traceposv V$ for the
set of all cq-operators.
And \symbolindexmark{\traceposcq}$\traceposcq V := \traceclcq V\cap\traceposv V$.

We call an operator $\rho\in\traceposv{V_1V_2}$
\emph{$(V_1,V_2)$-separable}\index{separable}
iff it can be written as $\sum_i\rho_i\tensor\rho_i'$
for some $\rho_i\in\traceposv{V_1}$
and $\rho_i'\in\traceposv{V_2}$.
We usually write just ``separable'' if $V_1,V_2$
are clear from the context.

For $\rho\in\tracepos X$,
the support \symbolindexmark\suppo$\suppo\rho$
is defined as follows: Let $P$
be the smallest projector such that $P\rho P=\rho$.
Then $\suppo\rho:=\im P$.
In particular, when $\rho=\sum_i\proj{\psi_i}$
for some $\psi_i\in\elltwo X$
(any $\rho\in\tracepos X$
can be decomposed in this way), we have
$\suppo\rho=\SPAN\{\psi_i\}_i$.

A \emph{superoperator}\index{superoperator} from $V$ to $W$ is a completely positive
trace-decreasing\footnote{I.e., $\tr\calE(\rho)\leq\tr\rho$.} linear map \symbolindexmark\calE$\calE:\traceclv V\to\traceclv W$.
A
\emph{cq-superoperator}\index{cq-superoperator}\index{superoperator!cq-} 
is a completely positive
trace-decreasing linear map $\calE:\traceclcq V\to\traceclcq W$.
A (cq-)\hskip0pt superoperator \emph{on} $V$ is a 
(cq-)\hskip0pt superoperator from $V$ to $V$.

% We call a superoperator or cq-superoperator
% \emph{total} on $A$\index{superoperator!total}\index{total!superoperator} iff
% it is trace-preserving, i.e., $\tr\calE(A)=\tr A$
% for all positive $A$.

For (cq-)superoperators $\calE$
from $V$
to $V'$ and $\calE'$ from $W$ to $W'$, let \symbolindexmark\tensor%
$\calE\otimes\calE'$
be the unique (cq-)\hskip0pt superoperator from $VW$ to $V'W'$ with
$(\calE\otimes\calE')(A\otimes B)=\calE(A)\otimes\calE'(B)$.

The trace $\tr:\traceclv V\to\setC=\traceclv{\varnothing}$
is in fact a superoperator (and a cq-superoperator). We can define the
\index{partial trace}\index{trace!partial}\emph{partial trace}
\symbolindexmark\partr$\partr{V}{W}:\traceclv{VW}\to\traceclv{V}$
as $\partr VW=\idv V\tensor\tr$
(where $\tr$
is the trace on $\traceclv W$). (Note: usually, the partial trace is written simply
$\tr_W$,
i.e., only the subsystem $W$  that
is removed is mentioned. However, for additional clarity, we
also mention the retained subsystem $V$
 in the symbol $\partr VW$.)
Note that $\partr VW$ is a superoperator and a cq-superoperator.

Let \symbolindexmark\Meas$\Meas DE$
denote the set of all projective measurements on $\elltwo E$
with outcomes in $D$.
Formally, $\Meas DE$
is the set of functions $M:D\to \bounded E$
such that $M(z)$
is a projector for all $z\in D$,
and $\sum_z M(z)\leq\id$.
We call a measurement $M$
\index{total!measurement}\index{measurement!total}\emph{total} iff
$\sum_z M(z)=\id$.

\paragraph{Variable-related canonical isomorphisms.}
For a variable renaming $\sigma:X\to Y$,
let \symbolindexmark{\Urename}$\Urename\sigma:\elltwov X\to\elltwov Y$
be the unitary mapping $\basis X m$
to $\basis Y{m\circ\sigma^{-1}}$.
We define the cq-superoperator $\Erename\sigma$
from $X$
to $Y$
by
\symbolindexmark{\Erename}$\Erename\sigma(\rho):=\Urename\sigma\rho\adj{\Urename\sigma}$.

For a list of distinct quantum variables $Q:=\qq_1,\dots,\qq_n$,
let the unitary \symbolindexmark\Uvarnames$\Uvarnames Q\in\uni{\typel Q,\types Q}$
be the canonical isomorphism from
$\elltwo{\typel Q}$
to $\elltwov Q$.
Formally, $\Uvarnames Q\basis{}{(q_1,\dots,q_n)}=\basis Qm$
where $m\in\types Q$ is given by $m(\qq_i):=q_i$.

\paragraph{Further notation.}
The following notation will be defined later in this paper (in
addition to the above and the syntax of programs in \autoref{sec:syn.prog}):
$\spaceat A\psi$
in \autoref{def:spaceat}.  $\denotee em$
(semantics of expressions) in \autoref{sec:expr}.
$\denotc\bc(\rho)$
(semantics of programs) in \autoref{sec:qsemantics}.  $\fv(e)$
(free variables of an expression) in \autoref{sec:expr}.  $\fv(\bc)$
(free variables of a program) in \autoref{sec:syn.prog}.
$\idx1,\idx2$
(add index $1$
or $2$
to each variable in an expression or program) in \autoref{sec:qrhl}.
$\lift AQ$ (operator/subspace $A$ lifted to variables $Q$) in \autoref{def:lift}.
$\CL{e}$
(classical predicates) in \autoref{def:cla}.  $\quanteq$
(quantum equality) in \multiautoref{def:quanteq,def:quanteq.simple}.  $\prafter e\bc\rho$
(probability of event $e$
after running $\bc$)
in \autoref{def:prafter}.  $\typee e$
(the type of an expression) in \autoref{sec:expr}.  $\restricte e\rho$
(restriction to state satisfying $e$)
in \autoref{sec:qsemantics}.  $\subst e\sigma$
(variable renaming applied to expression) and
$\substi e{e_1/\xx_1,\dots}$
(substitution applied to expression) in \autoref{sec:expr}.

\paragraph{Auxiliary lemmas.}

\begin{lemma}[Schmidt decomposition]\label{lemma:schmidt}
  Let $\psi\in\elltwov{XY}$.
  Then $\psi$
  can be decomposed as $\psi=\sum_i\lambda_i\psi_i^X\tensor\psi_i^Y$
  for some $\lambda_i>0$,
  orthonormal $\psi_i^X\in\elltwov X$,
  and orthonormal $\psi^Y_i\in\elltwov Y$.
\end{lemma}

\begin{proof} This lemma is well-known for separable Hilbert spaces (i.e., if
  $\types{XY}$ is countable).
  A generalization to the non-separable case is presented, e.g., in \cite{mo-schmidt}.
\end{proof}

\begin{lemma}\label{lemma:tensor.match}
  Assume $X,Y\subseteq V$,
  $X\cap Y=\varnothing$.
  Fix
  $\psi\in\elltwov{X},\psi'\in\elltwov{V\setminus
    X},\phi\in\elltwov{Y},\phi'\in\elltwov{V\setminus Y}$.  Assume
  $\psi\otimes\psi'=\phi\otimes\phi'\neq 0$.
  Then there exist $\eta\in\elltwov{V\setminus XY}$
  such that $\psi'=\phi\otimes\eta$ and $\phi'=\psi\otimes\eta$.
\end{lemma}

\begin{proof}
  Since $\psi\otimes\psi'=\phi\otimes\phi'\neq 0$, we have $\psi,\psi',\phi,\phi'\neq 0$.
  Let $\alpha:=\norm{\psi'}^2\neq 0$ and $\beta:=\norm{\phi}^2\neq 0$.
  Then
  \[
    \alpha\, \proj\psi = \partr{X}{V\setminus X}\proj{\psi\otimes\psi'}
    = \partr{X}{V\setminus X}\proj{\phi\otimes\phi'}
    = \partr{X}{V\setminus XY}
    \partr{V\setminus Y}{Y}
    \proj{\phi\otimes\phi'}
    = \partr{X}{V\setminus XY}
    \beta\, \proj{\phi'}.
  \]
  Hence $\phi'=\psi\otimes\eta$ for some $\eta\in\elltwov{V\setminus XY}$.
  Thus $\psi\otimes\psi'=\phi\otimes\phi'=\phi\otimes\psi\otimes\eta$. Since $\psi\neq0$, this implies $\psi'=\phi\otimes\eta$.
\end{proof}

\section{Quantum programs}
\label{sec:lang}

\subsection{Expressions}
\label{sec:expr}

In the context of this paper, an \index{expression}\emph{expression} $e$
is -- informally speaking -- an arbitrary mathematical formula, possibly containing free
\emph{classical} variables. (See \autoref{sec:prelim} for the definition of
variables.) We denote by \symbolindexmark\fv$\fv(e)$
the set of free variables occurring in the expression. For
$m\in\types X$
with $X\supseteq\fv(e)$,
let \symbolindexmark\denotee$\denotee{e}m$
denote the result of substituting every $\xx\in\fv(e)$
in $e$
by $m(\xx)$
and then evaluating the resulting expression (which has no free
variables any more).  Intuitively, this corresponds to evaluating the
expression $e$
on a memory $m$
that contains the content of the variables.
With every
expression~$e$,
we associate a set \symbolindexmark\typee$\typee e$
such that $\denotee em\in\typee e$
for all $m$.

Formally, an expression simply consists of a finite set
$\fv(e)\subseteq\cl V$,
a set $\typee e$,
and a function $\denotee e{}:\types V\to\typee e$
(with its argument written as subscript) such that
$\restrict m{\fv(e)}=\restrict{m'}{\fv(e)}\implies \denotee
em=\denotee e{m'}$.\footnote{This formalization is  more
    suitable for modeling expressions in a formal logical system
    (and thus for implementing in a theorem
    prover such as Isabelle or Coq). In this paper, we nevertheless write expressions
    as formulas with free variables because this leads to more
    readable formulas. E.g., we can simply write $e+f$
    if $e,f$
    are expressions, instead of having to write
    $\lambda m. (e(m)+f(m))$
    as we would have to if expressions were functions.}

  Note that expressions never contain quantum variables,\footnote{%
    At least not as free variables. Expressions may contain, e.g., the
    names of operators that are parameterized by quantum variables. For
    example, $\Uvarnames\qq\cdot\xx$
    (where $\typev\xx=\elltwo{\typev\qq}$) would be a valid expression with $\fv(e)=\{\xx\}$.
  } nor are they probabilistic.

% A word on notation: when we write, e.g., $g:=e+f$
% we mean the expression $g$
% with $\denotee{g}m = \denotee{e}m\odot\denotee fm$
% for all $m$.
% In particular, when $e$,
% $f$
% are written as formulas in some concrete syntax, then $g$
% is written $e+f$ in that syntax.

Given a variable renaming $\sigma$
and an expression $e$,
we write \symbolindexmark\subst$\subst e\sigma$
for the result of replacing free variables~$\xx$
in $e$
by $\sigma(\xx)$.
Formally, $\denotee{\subst e\sigma}m := \denotee{e}{m\circ\sigma}$.

And \symbolindexmark\substi$\substi e{e_1/\xx_1,\dots,e_n/\xx_n}$
denotes the result of simultaneously replacing all $\xx_i$
by $e_i$ in $e$. Formally,
\[
  \denotee{\substi e{e_1/\xx_1,\dots,e_n/\xx_n}}m
:=\denotee e{m(\xx_1:=\denotee{e_1}m,\dots,\xx_n:=\denotee{e_n}m)}.
\]

For $e$ with $\typee e=\bool$, we say that ``$e$ holds''%
\index{expressions holds}%
\index{holds!expression}
iff $\denotee{e}m=\true$ for all $m$.

\subsection{Syntax of programs}
\label{sec:syn.prog}

Programs in our language are constructed according to the following
grammar.  We will typically denote programs with
\symbolindexmark\bc$\bc$ or \symbolindexmark\bd$\bd$.

\symbolindexmark{\Qinit}%
\symbolindexmark{\Qapply}%
\symbolindexmark{\Qmeasure}%
\symbolindexmark{\assign}%
\symbolindexmark{\sample}%
\symbolindexmark{\Skip}%
\symbolindexmark{\while}%
\begin{align*}
\bc,\bd := {} &\Skip && \text{(no operation)} \\
& \assign \xx e && \text{(classical assignment)} \\
& \sample \xx e && \text{(classical sampling)}\\
& \langif{e}{\bc}{\bd} && \text{(conditional)}\\
& \while{e}{\bc} && \text{(loop)}\\
& \seq{\bc}{\bd} && \text{(sequential composition)} \\
& \Qinit{\qq_1\dots \qq_n}{e} && \text{(initialization of quantum registers)}\\
& \Qapply{e}{\qq_1\dots \qq_n} && \text{(quantum application)} \\
              & \Qmeasure{\xx}{\qq_1\dots \qq_n}{e}
                     && \text{(measurement)}
\end{align*}
In the sampling statement, $e$
evaluates to a distribution.  In the initialization of quantum
registers, $e$
evaluates to a pure quantum state, $\qq_1\dots\qq_n$
are jointly initialized to that state. In the quantum application, $e$
evaluates to an isometry that is applied to
$\qq_1\dots\qq_n$.
In the measurement, $e$
evaluates to a projective measurement, the outcome is
stored in $\xx$.
(Recall that an expression $e$
  can be an arbitrarily complex mathematical formula in the classical
  variables. So, e.g., an expression that describes an isometry could
  be something as simple as just $H$
  (here $H$ denotes the Hadamard transform),
  or something more complex such as, e.g., $H^{\xx}$, meaning $H$ is applied if $\xx=1$.)

A program is \emph{well-typed}%
\index{well-typed (program)}%
\index{program!well-typed} according to the following rules:
\begin{compactitem}
\item $\assign\xx e$ is well-typed iff $\typee e\subseteq\typev\xx$,
\item $\sample \xx e$ is well-typed iff $\typee e\subseteq\distr{\typev\xx}$,
\item $\langif{e}{\bc}{\bd}$ is well-typed iff $\typee e\subseteq\bool$ and $\bc,\bd$ are well-typed.
\item $\while{e}{\bc}$ is well-typed iff $\typee e\subseteq\bool$ and $\bc$ is well-typed.
\item $\seq{\bc}{\bd}$ is well-typed iff $\bc$ and $\bd$ are well-typed.
\item $\Qinit{Q}{e}$
  is well-typed iff
  $\typee e\subseteq\elltwo{\typel Q}$, and
   $\norm\psi = 1$
  for all $\psi\in\typee e$. 
\item $\Qapply{e}Q$ is well-typed iff $\typee e\subseteq\iso{\typel Q}$. 
\item $\Qmeasure{\xx}Q{e}$
  is well-typed iff
  $\typee e\subseteq\Meas {\typev\xx}{\elltwo{\typel Q}}$.
\end{compactitem}
In this paper, we will only consider well-typed programs. That is,
``program'' implicitly means ``well-typed program'', and all
derivation rules hold under the implicit assumption that the programs
in premises and conclusions are well-typed.

The set of all free variables in a program $\bc$
is denoted \symbolindexmark\fv$\fv(\bc)$
and consists of the classical variables $\fv(e)$
for all expressions $e$
occurring in $\bc$,
the classical variables $\xx$ in subterms
$\assign \xx e$, $\sample\xx e$, $\Qmeasure\xx Qe$,
and all quantum variables $Q$
in subterms $\Qinit{Q}{e}$,
$\Qapply{e}Q$, $\Qmeasure{\xx}Q{e}$.

\subsection{Classical semantics}
\label{sec:classical.sem}
\index{semantics!classical}
\index{classical!semantics}

We call a program \emph{classical}%
\index{classical!program}%
\index{program!(classical)} iff it does not contain the constructions
$\Qinit Q{e}$, $\Qapply{e}Q$, or
$\Qmeasure{\xx}Q{e}$.
For classical programs, we can define classical semantics as follows.

Fix some set $V$
of variables relative to which the semantics will be defined (in the
remainder of the paper, that set will always be clear from the
context, and it will always be called either $V$, $V_1$, or $V_2$).

Let $S$
be the set of all functions $\types V\to \setR$
such that $\sum_a\abs{ f(x)}\leq\infty$.
(Note that $\distrv V\subseteq S$.
In fact, $S$
is the smallest real Hilbert space containing $\distrv V$.)
We define \symbolindexmark\denotcl$\denotcl\bc:S\to S$ as the linear function satisfying the following rules:
\begin{align*}
  \denotcl{\Skip}(\mu) &:= \mu \\
  \denotcl{\assign\xx e}(\delta_m) &:= \pointdistr{\upd m\xx{\denotee{e}m}} \\
  \denotcl{\sample\xx e}(\pointdistr m) &:=
    \sum_{z\in\typev\xx}{\denotee{e}m}(z) \cdot \pointdistr{\upd m\xx z} \\
  \denotcl{\langif{e}{\bc_1}{\bc_2}}(\mu) &:=
                                            \denotcl{\bc_1}\pb\paren{\restricte e(\mu)} +
                                            \denotcl{\bc_2}\pb\paren{\restricte {\lnot e}(\mu)} \\
  \denotcl{\while e\bc}(\mu) &:=\sum_{i=0}^\infty
                               \restricte {\lnot e}\bigl((\denotcl\bc\circ\restricte e)^i(\mu)\bigr) \\
  \denotcl{\seq{\bc_1}{\bc_2}} &:= \denotcl{\bc_2}\circ\denotcl{\bc_1} \\
\end{align*}

Here \symbolindexmark\restricte$\restricte e(\mu)$
denotes the distribution $\mu$
over memories, restricted to those memories that make the expression
$e$
true. Formally, $\restricte e(\mu)=\nu$
with $\nu(m):=\mu(m)$
if $\denotee em=\true$ and $\nu(m):=0$ otherwise.

\begin{comment}
  \medskip\noindent
  The semantics of the language are given by the denotation
  $\denotc\bc:\traceclcq V\to\traceclcq V$
  of a program for some fixed set of variables.
That is, if the memory of a program is in state $\rho\in\traceclcq V$, then after running $\bc$ it is in state $\denotc\bc(\rho)$.
  The semantics are
  standard and straightforward (no subtleties such as quantum control
  etc. are involved), we defer them to \autoref{app:progs.wt.sem}.
\end{comment}

\subsection{Quantum semantics}
\label{sec:qsemantics}
\index{semantics!quantum}
\index{quantum!semantics}

Fix some set $V$
of variables relative to which the semantics will be defined (in the
remainder of the paper, that set will always be clear from the
context, and it will always be called either $V$, $V_1$, or $V_2$).

Given a program $\bc$
(with $\fv(\bc)\subseteq V$), we define its semantics
$\denotc\bc$\symbolindexmark{\denotc}
as a cq-superoperator on $\traceclcq V$.
 In the following, let
$\rho\in\traceclcq V$, $m\in\types{\cl V}$, and
$\rho_m\in\traceposv{\qu V}$. Note that specifying $\denotc\bc$ on
operators $\pointstate{\cl V}m\tensor\rho_m$ specifies $\denotc\bc$
on all $\rho\in\traceclcq V$, since $\rho$ can be written as an
infinite sum of $\pointstate{\cl V}m\tensor\rho_m$ (by definition of
$\traceclcq V$ and using the fact that any operator in
$\traceclv{\qu V}$ can be written as linear combination of four
$\rho_a\in\traceposv{\qu V}$).

\symbolindexmark\langif
\symbolindexmark\while
\symbolindexmark\seq
\symbolindexmark\assign
\symbolindexmark\sample
\symbolindexmark\Qinit
\symbolindexmark\Qapply
\symbolindexmark\Qmeasure
  \begin{align*}
  \denotc\Skip(\rho) &:= \rho \\
  \denotc{\assign\xx e}\pb\paren{\pointstate{\cl V}m\tensor\rho_m} &:=
                                                    \pointstate{\cl V}{\upd m\xx{ \denotee{e}{m}}}\tensor \rho_m \\
  \denotc{\sample\xx e}\pb\paren{\pointstate{\cl V}m\tensor\rho_m} &:=
                                                          \sum_{z\in\typev\xx}{\denotee{e}m}(z) \cdot \pb\pointstate{\cl V}{\upd m\xx z}\tensor \rho_m
  \\
    \denotc{\langif{e}{\bc}{\bd}}(\rho) &:=
    \denotc{\bc}( \restricte {e}(\rho) ) +
    \denotc{\bd}( \restricte {\lnot e}(\rho) )
                                           \\
  \denotc{\while e\bc}(\rho)
  &:=\sum_{i=0}^\infty
  \restricte {\lnot e}\bigl((\denotc\bc\circ\restricte {e})^i(\rho)\bigr)
     \\
    \denotc{\seq{\bc_1}{\bc_2}}&:=\denotc{\bc_2}\circ\denotc{\bc_1}
                                 \\
  \denotc{\Qinit Q{e}}\pb\paren{\pointstate{\cl V}m\tensor\rho_m}
                     &:= \pb\pointstate{\cl V}m \tensor
                       \partr{\qu V\setminus Q}{Q}{\rho_m}\otimes\pb\proj{
                       \Uvarnames Q\denotee em}
  \\
  \denotc{\Qapply{e}Q}\pb\paren{\pointstate{\cl V}m\tensor\rho_m}
                     &:= \pb\pointstate{\cl V}m \tensor U\rho_m\adj U\\
  &\qquad\qquad                     \text{where }
                       U := \Uvarnames Q\denotee{e}m \adj{\Uvarnames Q}\otimes\idv{\qu V\setminus Q}
  \\
  \denotc{\Qmeasure{\xx}{Q}e}\pb\paren{\pointstate{\cl V}m\tensor\rho_m}
                     &:=
                       \sum_{z\in\typev\xx} \!\!\!   \pb
                       \pointstate{\cl V}{\upd m\xx z} \tensor
                       P_z\rho_m P_z\\
  &\qquad\qquad                     \text{where }
                       P_z := \Uvarnames Q\denotee{e}m(z)\adj{\Uvarnames Q} \tensor \idv{\qu V\setminus Q}
\end{align*}

Here \symbolindexmark\restricte$\restricte e(\rho)$ is the cq-density operator $\rho$ restricted to the
parts where the expression $e$ holds. Formally, $\restricte e$ is the cq-superoperator on $V$ such that
\[\restricte e(\pointstate{\cl{V}}{m} \tensor \rho_m) :=
  \begin{cases}
    \pointstate{\cl{V}}m \tensor \rho_m & (\denotee{e}m=\true) \\
    0 & (\text{otherwise})
  \end{cases}
\]

\begin{definition}\label{def:prafter}
  Fix a program $\bc$, an expression $e$ with $\typee e=\bool$, and some $\rho\in\traceposcq V$.
  Then \symbolindexmark\prafter$\prafter e\bc\rho:=\sum_{m\text{ s.t.\ }\denotee em=\true}\tr\rho_m$ where
  $\denotc\bc(\rho)=:\sum_m\proj{\basis{\cl{V}}m}\otimes\rho_m$ with $\rho_m\in\traceposv{\qu V}$.
\end{definition}

That is, $\prafter e\bc\rho$
denotes the probability that $e$
(a condition on the classical variables) holds after running the
program $\bc$ on initial state $\rho$.

\begin{definition}
  For a set $X$
  of variables, we call $\bc$
  \index{program!local}\index{local!program}\emph{$X$-local} iff
  $\denotc\bc=\calE\otimes\idv{V\setminus X}$ for some cq-superoperator $\calE$ on $X$.
\end{definition}

\begin{lemma}\label{lemma:fv.local}
  If $\fv(\bc)\subseteq X$, then $\bc$ is $X$-local.
\end{lemma}

\begin{proof}
  By induction over the structure of $\bc$.
\end{proof}

\paragraph{Relation to classical semantics.}
For \emph{classical} variables $V$ and a distribution $\mu\in\distrv V$, \symbolindexmark{\qlift}\label{page:def:qlift}let
$\qlift\mu:=\sum_m\mu(m)\,\pb\pointstate Vm\in\traceposcq V$. Note that
$\qlift\mu$ is a bijection between $\ellonev V$ and $\traceposcq V$ (since $\qu V=\varnothing$).

\begin{lemma}[Relation to classical semantics]\label{lemma:classical.sem}
  Let $V$ be classical variables, and
  let $\bc$ be a classical program
  with $\fv(\bc)\subseteq V$.

  Then for all $\mu\in\distrv V$, we have
  \[
  \denotc\bc\bigl(\qlift\mu\bigr) =\pb \qlift{\denotcl\bc(\mu)}.
  \]
\end{lemma}

\begin{proof}
  By induction over the structure of a classical $\bc$. For each case in the induction,
  the statement
  $\denotc\bc\bigl(\qlift\mu\bigr) = \pb\qlift{\denotcl\bc(\mu)}$
  is verified by elementary computation.
\end{proof}

\section{Quantum predicates}
\label{sec:predicates}

\subsection{Defining predicates}

In this section, we discuss what a predicate is, i.e., how predicates
are mathematically represented.

In the introduction, we already gave a definition of
predicates. Namely, a predicate $A$
over variables $V$
is a subspace of $\elltwov{V}$,
and a cq-density operator $\rho\in\traceposcq{V}$
satisfies $A$
iff $\suppo\rho\subseteq A$.
Let us call a predicate according to this definition a \emph{simple
  predicate}.

While there is nothing wrong per se with simple predicates, it turns
out that we can optimize the representation of predicates somewhat.
For a simple predicate $A$,
let $A_m:=\{\psi\in\elltwov{\qu V}:\basis{\cl V}m\otimes\psi\in
A\}$. For any cq-density operator $\rho=\sum_m\proj{\basis{\cl V}m}\otimes\rho_m$, we then have
\begin{multline*}
  \rho\text{ satisfies }A
  \iff
  \suppo\rho\subseteq A
  \iff
  \forall m.\ \{\basis{\cl V}m\}\otimes\suppo\rho_m\subseteq A
  \iff
  \forall m.\ \rho_m\subseteq A_m.
\end{multline*}
This implies that a simple predicate $A$
is essentially described by the subspaces
$A_m\subseteq\elltwov{\qu V}$.\footnote{We
  say ``essentially'' here, because $A$
  is not uniquely determined by the $A_m$.
  But the set of cq-density operators that satisfy $A$
  are determined by the $A_m$.}

Thus, we will define predicates on $V$
not as subspaces of $\elltwov{V}$,
but as families of subspaces of $\elltwo{\qu V}$,
indexed by $\types{\cl V}$.

And in fact, a family indexed by $\types{\cl V}$
is most conveniently described as an expression (in the sense of
\autoref{sec:expr}) with free variables in $\cl V$. Namely,
if we represent a predicate as an expression $A$ with $\typee A\subseteq\elltwov{\qu V}$, we get the subspaces $A_m$ from above as~$\denotee Am$.

This leads  to the following definition:

\begin{definition}[Quantum predicates]\label{def:pred}
A (quantum) \emph{predicate}\index{predicate} $A$
over $V$
is an expression with $\fv(A)\subseteq\cl V$
and
$\typee A\subseteq\{S:S\text{ is a subspace of }\elltwov{\qu
  V}\}$.
\end{definition}

\begin{definition}[Satisfying predicates]\label{def:satisfy}
A state $\rho\in\traceposcq V$
\emph{satisfies}\index{satisfy!predicate}%
\index{predicate!satisfy} $A$ iff $\suppo\rho_m\subseteq \denotee Am$ for all $m$,
where $\rho=:\sum_{m\in\types{\cl V}}\pb\proj{\basis{\cl V}m}\tensor\rho_m$.
\end{definition}

This definition of predicates will make notation much more convenient
because it will be possible to express the properties of the quantum
variables in terms of the values of the classical variables.

  \paragraph{Examples of predicates.} We give some examples of
  predicates to illustrate the concepts involved. We only explain them
  intuitively here, the necessary definitions to understand them
  formally will be given in the following sections. Basically, an
  expression is a subspace of possible quantum states of the program's memory. In many
  cases, we want to say that a specific variable satisfies a
  predicate. E.g., we want to express that the state of variable $\qq$
  is in subspace $P$
  (where $P$
  is a subspace of $\elltwov\qq$).
  We will write this as $\lift P\qq$
  (and we will formally define this as a subspace later). I.e.,
  $\lift P\qq$
  means intuitively $\qq\in P$.
  Similarly, we can write $\lift{\SPAN\{\psi\}}\qq$
  to denote that $\qq$
  has state $\psi$.
  Since we can refer to classical variables in predicates, we can write, e.g.,
  $\lift{\SPAN\{\basis{}{\xx}\}}\qq$
  to denote that $\qq$
  has state $\basis{}{\xx}$ where $\xx$ is a classical variable.
  Or more advanced, 
  $\lift{\SPAN\{H^{\xx}\basis{}{0}\}}\qq$, where we apply $H$ to $\basis{}0$ iff $\xx=1$.

  We can combine predicates with logical connectives. In this paper we
  will use $+$
  to represent disjunction and $\cap$
  to represent conjunction (Birkhoff-von Neumann quantum
  logic).\footnote{But we stress that our definition of predicates
    allows us to use any other connectives, too, as long as the result is a subspace.}
  For example, 
  $\lift{\SPAN\{\basis{}{0}\}}\qq\cap
  \lift{\SPAN\{\basis{}{1}\}}\rr$ means that $\qq$ has state $\basis{}{0}$ and $\rr$ has state $\basis{}{1}$.

  Moreover, we often need to express facts about the classical
  variables alone. We have syntactic sugar $\CL{\dots}$
  for that. For example, $\CL{\xx=1}$
  means that the classical variable $\xx$
  has value $1$.
  We can write arbitrary formulas inside $\CL{\dots}$
  as long as they involve only classical variables. E.g.,
  $\CL{\forall z\in\xx.\ z\leq\yy}$
  would say that the $\xx$
  (a variable of type set) is upper-bounded by $\yy$.
  Note that predicates of the form $\CL\dots$
  can be combined with non-classical predicates using $+$ and $\cap$.

  Finally, we can even express that two quantum variables have the same
  content (formalized below) using expressions such as $\qq\quanteq\rr$.

  We stress that predicates are not restricted to what can be built
  from the above constructions. Any mathematical expression that
  describes a subspace is permissible. The above are merely a useful
  selection of common cases.

The following subsections will introduce various definitions and facts
about predicates. They are necessary for understanding the rules
defined later, but not for understanding the definition of qRHL
itself. A reader interested primarily in the latter may skip ahead to
\autoref{sec:qrhl} and come back later.

\subsection{Operations on predicates}

We now define a few elementary operations on predicates.

First, any operation that can be performed on subspaces is meaningful
on predicates, too. For example, the intersection $A\cap B$
is a the quantum analogue to the conjunction of classical predicates,
the sum $A+B$
is a quantum analogue to the disjunction of classical predicates, and
the orthogonal complement $\orth A$
is a quantum analogous to the negation of a classical predicate
(cf.~\autoref{lemma:cl.simps} below).

A word on notation: 
when we write $C:=A\odot B$
for some operation $\odot$,
we mean the predicate/expression $C$
with $\denotee Cm = \denotee{A}m\odot\denotee Bm$
for all $m$.
In particular, when $A$,
$B$
are written as formulas in some concrete syntax (with classical
variables as free variables), then $C$
is simply the formula written $A\odot B$ in that syntax.

\begin{lemma}\label{lemma:and.sat}
  A cq-density operator $\rho$ satisfies $A\cap B$ iff $\rho$ satisfies $A$ and $\rho$ satisfies $B$.
\end{lemma}

\begin{proof}
  Immediate from the definition.
\end{proof}

(Note that the analogue for $A+B$ does not hold.)

We can easily specify relations between predicates. A predicate $A$
implies another predicate $B$
iff $A\subseteq B$
for all assignments to the free variables of $A$ and $B$. That is:

\begin{definition}
  We say $A\subseteq B$ iff $\denotee Am\subseteq \denotee Bm$ for all $m$.
\end{definition}

\begin{lemma}\label{lemma:leq.sat}
  If $A\subseteq B$, and $\rho$ satisfies $A$, then $\rho$ satisfies $B$.
\end{lemma}

\begin{proof}
  Immediate from the definition.
\end{proof}

\medskip

Often we need to specify which quantum and classical variables a given
predicate refers to:
\begin{definition}[Locality]
  Fix variables $X\subseteq V$.
  A predicate $A$ on $V$
  is \index{local!predicate}\index{predicate!local}\emph{$X$-local}
  iff $\fv(A)\subseteq\cl{X}$,
  and for all $m\in\types{\cl X}$,
  $\denotee Am=S_m\tensor\elltwov{\qu V\setminus \qu X}$
  for some  subspace $S_m$ of $\elltwov{\qu X}$.
\end{definition}
Intuitively, a predicate is $X$-local
if we only need to look at classical and quantum variables in $X$
to decide whether the predicate is satisfied.

\medskip

\paragraph{Lifting of operators/subspaces.}
Often, we will need to express predicates like ``the state of
variables $\qq\in V$
is in subspace $S$''
where $\qq$
has type $T$
and $S\subseteq\elltwo{T}$.
Formally, this can be encoded as the predicate
$\Uvarnames \qq S\otimes\elltwov{\qu V\setminus \qq}$
(which is a subspace of $\elltwov{\qu V}$)
where $\Uvarnames \qq$
is the isomorphism $\elltwo{T}\to\elltwov{\qq}$
described in the preliminaries.  Since this notation is cumbersome, we
abbreviate it as $\lift S\qq$.

For example, if $\qq$
has type $\bit$,
then $\lift{\SPAN\{\basis{}0\}}{\qq}$
is the subspace spanned by all $\basis Vm$
with $m(\qq)=0$.
Informally, this is the predicate ``$\qq$ is in state $\basis{}0$.''

Similarly, we will often need to apply an operator
$A\in\bounded{T}$
to the (subsystem containing the) quantum variable $\qq$.
This is done by applying
$\Uvarnames \qq A\adj{\Uvarnames \qq}\otimes\idv{\qu V\setminus
  \qq}\in\boundedv{\qu V}$ to the overall
system. (Cf.~also the definition of the denotational
semantics of quantum application and measurements where similar
constructions occur.)
We abbreviate this as $\lift A\qq$.

For example, if $H$
is the Hadamard gate, and $\qq$
is of type $\bit$,
then applying $\lift H\qq$
to a state (in $\elltwov{\qu V}$) is the same as applying $H$
to the subsystem corresponding to variable $\qq$.

The following definition generalizes this to lists of variables:
\begin{definition}[Lifting]\label{def:lift}
  Fix a list $Q\subseteq\qu V$.
  Let $S\subseteq\elltwo{\typel Q}$
  be a subspace, and $A\in\bounded{\typel Q}$. Then:\symbolindexmark\lift
  \begin{equation*}
    \lift SQ  := \Uvarnames QS\otimes\elltwov{\qu
      V\setminus Q}
    \qquad\text{and}\qquad
    \lift AQ  := \Uvarnames QA\adj{\Uvarnames
               Q}\otimes\idv{\qu V\setminus Q}.
           \end{equation*}
\end{definition}
Here $V$ is a global set of variables, implicitly understood. For example, in the context of qRHL judgments (\autoref{sec:qrhl}),
this set will usually be $V_1V_2$, the union of the sets of all variables indexed with $1$ and $2$, respectively.

\paragraph{``Division'' of subspaces.}
Finally, we introduce one technical definition that will be needed to
express the precondition of quantum initialization statements (\ruleref{QInit1} below):
\begin{definition}\label{def:spaceat}
  Fix variables $W\subseteq V$,
  let $A\subseteq\elltwov{V}$
  be a subspace, and let $\psi\in\elltwov{W}$.
  Then
  \symbolindexmark\spaceat$\spaceat{A}{\psi}\subseteq\elltwov{V\setminus W}$ is defined by
    \[
      \phi\in\spaceat{A}{\psi}
      \qquad
      \text{iff}
      \qquad
      \phi\tensor\psi\in A.
    \]
  
\end{definition}
The notation $\spaceat{}{}$
is motivated by the fact that if $A=B\tensor\SPANO{\psi}$,
then $\spaceat A\psi=B$,
so $\spaceat{}{}$
is in a sense the division operation corresponding to the tensor
product $\tensor$.

The following simplification rule is useful for simplifying subgoals arising
from \ruleref{QInit1} (in combination with \ruleref{Conseq}):
\begin{lemma}\label{lemma:spacediv.leq}
  Let $A,B\subseteq\elltwov{\qu V}$ be subspaces. Let $Q\subseteq \qu V$ and $\psi\in\elltwov Q$.
  Assume that $A$ is $(\qu V\setminus Q)$-local.
  Then
  \begin{equation}
    \label{eq:spacediv.leq}
    A \subseteq (\spaceat{B}{\psi}) \tensor \elltwov Q
    \quad\iff\quad
    A \cap (\SPAN\{\psi\}\tensor \elltwov{\qu V\setminus Q}) \subseteq B.
  \end{equation}
\end{lemma}

\begin{proof}
  Without loss of generality we can assume that $\psi\neq0$,
  since for $\psi=0$
  both sides of \eqref{eq:spacediv.leq} are vacuously true.

  \medskip
  
  We first show the ``$\Longrightarrow$''
  direction. So assume that the lhs of \eqref{eq:spacediv.leq} holds, and
  fix
  $\phi\in A \cap (\SPAN\{\psi\}\tensor
  \elltwov{\qu V\setminus Q})$. We need to show $\phi\in B$.
  Since
  $\phi\in \SPAN\{\psi\}\tensor \elltwov{\qu V\setminus Q}$,
  we can decompose $\phi$
  as $\phi=\phi'\tensor\psi$
  with $\phi'\in\elltwov{\qu V\setminus Q}$.
  Since $\phi\in A$,
  from the lhs of \eqref{eq:spacediv.leq} we get
  $\phi'\tensor\psi=\phi\in (\spaceat{B}{\psi})
  \tensor \elltwov Q$.  Since $\psi\neq0$,
  this implies $\phi'\in \spaceat{B}{\psi}$.
  By definition of $\spaceat{}{}$,
  this means $\phi=\phi'\otimes\psi\in B$.
  This shows the ``$\Longrightarrow$'' direction.

  \medskip

  We show the ``$\Longleftarrow$''
  direction. So assume that the rhs of \eqref{eq:spacediv.leq} holds,
  and fix $\phi\in A$.
  We need to show
  $\phi\in (\spaceat{B}{\psi}) \tensor \elltwov Q$.
  Since $A$
  is $(\qu V\setminus Q)$-local,
  we can decompose $A$
  as $A=A'\otimes\elltwov Q$
  for some $A'\subseteq\elltwov{\qu V\setminus Q}$.
  Thus $\phi\in A$
  can be decomposed as $\phi=\sum_i\phi_{Ai}\otimes\phi_{Qi}$
  with $\phi_{Ai}\in A'$ and $\phi_{Qi}\in\elltwov{Q}$.

  For any $i$,
  we have: Since $\phi_{Ai}\in A'$,
  $\phi_{Ai}\otimes\psi\in A$.
  And
  $\phi_{Ai}\otimes\psi\in \SPAN\{\psi\}\otimes\elltwov{\qu V\setminus Q}$.  Thus the rhs of
  \eqref{eq:spacediv.leq} implies that
  $\phi_{Ai}\otimes\psi\in B$.
  By definition of $\spaceat{}{}$,
  this implies $\phi_{Ai}\in \spaceat B{\psi}$.
  Thus
  $\phi_{Ai}\otimes\phi_{Qi}\in (\spaceat{B}{\psi})
  \tensor \elltwov Q$.

  Since this holds for any $i$,
  and $ (\spaceat{B}{\psi}) \tensor \elltwov Q$
  is a subspace, we also have
  $\phi=\sum_i\phi_{Ai}\otimes\phi_{Qi}\in (\spaceat{B}{\psi}) \tensor \elltwov Q$.

  This shows the ``$\Longleftarrow$'' direction.
\end{proof}

\paragraph{Total programs.}
Sometimes we need to specify that a program is total, i.e., that it
terminates with probability $1$.
In case the totality holds only for certain initial states, the
following definition can be used:
\begin{definition}[Total programs]\label{def:prog.total}
  A program $\bc$
  is \index{total!program, on $A$}\emph{total on $A$}
  if for every $\rho\in\traceclcq{V}$
  that satisfies $A$, we have $\tr\denotc{\bc}(\rho)=\tr\rho$.
\end{definition}

\subsection{Classical predicates}

In many cases, we need predicates that only talk about classical
variables.  The following definition allows us to do so:
\begin{definition}[Classical predicates]\label{def:cla}
  We define \symbolindexmark\CL$\CL{b}$ for $b\in\bool$ as follows:
  $\CL{\true}:= \elltwov{\qu V}$
  and $\CL\false:=0$ (the trivial subspace).
  (Here we assume that the set $V$ of variables is clear from the context.)
\end{definition}

\begin{comment}
  For a Boolean expression $e$,
  $\CL{e}$
  basically encodes the classical predicate $e$
  on classical variables. We defer some additional lemmas formalizing
  this correspondence to \autoref{app:predicates}.
\end{comment}

The next lemma explains why $\CL\dots$
allows us to encode classical predicates:
\begin{lemma}\label{lemma:cl.rho}
  Fix an expression $e$  with $\typee e\subseteq\bool$ and $\fv(e)\subseteq \cl V$.
  Then $\CL e$ is a predicate. And for any 
  $\rho=\sum_m\proj{\basis{\cl V}m}\tensor\rho_m\in\traceposcq V$, we have:
  \par
  $\rho$ satisfies $\CL e$ iff $\denotee em=\true$ for all $m$ with $\rho_m\neq 0$.
\end{lemma}

\begin{proof}
  Immediate from the facts that $\suppo\rho_m\in\CL\false=0$
  iff $\rho_m=0$,
  and that $\suppo\rho_m\in\CL{\true}=\elltwov{\qu V}$
   for all $\rho_m$.
\end{proof}

Predicates of the form $\CL\dots$
also illustrate why intersection and sum of subspaces are the
analogues to conjunction and disjunction:
\begin{lemma}\label{lemma:cl.simps}
  For expressions $e$,
  $f$ with $\typee e,\typee f\subseteq\bool$, and expressions $g$, we have:
  \begin{compactenum}[(i)]
    \item\label{and} $\CL e\cap\CL f=\CL{e\land f}$.
    \item\label{or} $\CL e+\CL f=\CL{e\lor f}$.
    \item\label{not} $\orth{\CL e}=\CL{\lnot e}$.
    \item\label{forall} $\bigcap_{z\in g} \CL{e} = \CL{\forall z\in g.\ e}$.
      (Here we assume that $\typee g\subseteq\typee z$.
      $z$ may occur free in $e$.)
  \end{compactenum}
\end{lemma}

\begin{proof}
  To prove \eqref{and},
  we need to show $\denotee{\CL e\cap\CL f}m=\denotee{\CL{e\land f}}m$
  which follows by checking all four cases for
  $\denotee em,\denotee fm$. Analogously for \eqref{or} and \eqref{not}.
  
  For \eqref{forall}, we need to show
  $\denotee{\bigcap_{z\in g} \CL{e}}m = \denotee{\CL{\forall z\in g.\
      e}}m$.  We distinguish two cases. If there exists an
  $a\in\denotee gm$
  with $\denotee{e\{a/z\}}m=\false$,
  then both sides of the equation will be $0$.
  And otherwise, both sides will be $\elltwov{V}$.
\end{proof}

We can also use classical predicates to conveniently define what it
means for a program to have only readonly access to certain variables:
\begin{definition}[Readonly variables]
  Let $X$
  be a set of classical variables. A program is
  \index{readonly}\emph{$X$-readonly}
  iff for all $\xx\in X$,
  all $z\in\typev\xx$,
  all $\rho$
  that satisfy $\CL{\xx=z}$,
  we have that $\denotc\bc(\rho)$ satisfies $\CL{\xx=z}$.
\end{definition}
Of course, in most cases one will use a simple sufficient syntactic
  criterion: if no variable from $X$
  occurs on the lhs of an assignent, sampling, or measurement, then
  the program is $X$-readonly.

\subsection{Quantum equality}
\label{sec:quantum.eq}

One predicate that is very specific to the setting of qRHL is that of
quantum equality.  As discussed in the introduction
(\autoref{sec:intro.qrhl}), we need a predicate $Q_1\quanteq Q_2$ 
that encodes the intuitive fact that the quantum variables $Q_1$
contain the same state as the quantum variables $Q_2$.
If we want $Q_1\quanteq Q_2$
to be a predicate in our sense (i.e., a subspace), there seems to be only one
possible definition, namely, $Q_1\quanteq Q_2$
is the space of all quantum states that are invariant under swapping
the content of $Q_1$
and $Q_2$.
(The fact that this is the only possibility is a conjecture only,
but we give some formal evidence for this in
\autoref{sec:qeq.unique}.)

How do we write the unitary $U$
that swaps $Q_1$
and $Q_2$?
We cannot simply define $U$
by $U(\psi_1\otimes\psi_2)=\psi_2\otimes\psi_1$
for $\psi_1\in\elltwov{Q_1},\psi_2\in\elltwov{Q_2}$,
since $\elltwov{Q_1}$
and $\elltwov{Q_2}$
are different (but isomorphic) spaces.  The canonical isomorphism from
$\elltwov{Q_1}$
to $\elltwov{Q_2}$
is $\Uvarnames{Q_2}\adj{\Uvarnames{Q_1}}$,
and the canonical isomorphism from $\elltwov{Q_2}$
to $\elltwov{Q_1}$
is $\Uvarnames{Q_1}\adj{\Uvarnames{Q_2}}$.
We can then define the swap by
$U(\psi_1\otimes\psi_2)=\Uvarnames{Q_1}\adj{\Uvarnames{Q_2}}\psi_2\otimes\Uvarnames{Q_2}\adj{\Uvarnames{Q_1}}\psi_1$.
And since the tensor product is commutative (in our formalization
which uses variables to identify the different factors), we can write
this as:
$U(\psi_1\otimes\psi_2)=\Uvarnames{Q_2}\adj{\Uvarnames{Q_1}}\psi_1
\otimes \Uvarnames{Q_1}\adj{\Uvarnames{Q_2}}\psi_2$.  That is, the
swap of $Q_1$
and $Q_2$
is simply
$U=\Uvarnames{Q_2}\adj{\Uvarnames{Q_1}} \otimes
\Uvarnames{Q_1}\adj{\Uvarnames{Q_2}}$.  If $Q_1Q_2$
is not the set of all quantum variables $V$, we additionally need to
express that $U$ does not change the other variables, and thus we get
\[
  U := \underbrace{\Uvarnames{Q_2}\adj{\Uvarnames{Q_1}}}_{{}\in\univ{Q_1,Q_2}}
  \tensor
  \underbrace{\Uvarnames{Q_1}\adj{\Uvarnames{Q_2}}}_{{}\in\univ{Q_2,Q_1}}
  \tensor
  \underbrace{\idv{\qu V\setminus Q_1Q_2}}_{\hskip-1cm{}\in\univ{\qu V\setminus Q_1Q_2}\hskip-1cm}
  \in\univ{\qu V}.
\]

We can now define $(Q_1\quanteq Q_2)\subseteq \elltwov{\qu V}$ as the subspace fixed by $U$
(i.e., the set of all $\psi\in\elltwov{\qu V}$ with $U\psi=\psi$).

Before we do so (\autoref{def:quanteq.simple} below), however, let us
generalize the concept of quantum equality somewhat.  Consider the
following classical equality: $f(\xx_1)=g(\xx_2)$.
This means that if we apply $f$
to $\xx_1$,
we get the same value as when applying $g$
to $\xx_2$.
We would like to express something like this also using quantum
equality. Namely, we wish to express that $Q_1$,
when applying an isometry $U_1$,
has the same content as $Q_2$
when applying an isometry $U_2$.
(The quantum equality $Q_1\quanteq Q_2$
we discussed so far is a special case of this with $U_1=U_2=\id$.)
To model this, we define the swap somewhat differently, namely,
instead of $\Uvarnames{Q_1}$
(which simply maps $\psi_1\in\elltwov{Q_1}$
to $\elltwo{\typel{Q_1}}$),
we use the morphism $U_1\Uvarnames{Q_1}$
(which maps $\psi_1\in\elltwov{Q_1}$
to $\elltwo{\typel{Q_1}}$
and then applies $U_1$),
and instead of $\Uvarnames{Q_2}$ we use $U_2\Uvarnames{Q_2}$.
Thus we get the following operation $U'\in\boundedv{\qu V}$ instead of~$U$:
\[
  U' :=
  \underbrace{\Uvarnames{Q_2}\adj{U_2}U_1\adj{\Uvarnames{Q_1}}}_{{}\in\boundedv{Q_1,Q_2}}
  \tensor
  \underbrace{\Uvarnames{Q_1}\adj{U_1}U_2\adj{\Uvarnames{Q_2}}}_{{}\in\boundedv{Q_2,Q_1}}
  \tensor
  \underbrace{\idv{\qu V\setminus Q_1Q_2}}_{\hskip-1cm{}\in\boundedv{\qu V\setminus Q_1Q_2}\hskip-1cm}
  \in\boundedv{\qu V},
\]
and we define
$U_1Q_1\quanteq U_2Q_2$ to consist of all $\psi$ fixed by $U'$:
\begin{definition}[Quantum equality]\label{def:quanteq}
  Let $V$
  be a set of variables, and $Q_1,Q_2\subseteq\qu V$
  be disjoint lists of distinct quantum variables.  Let $Z$
  be a set.
  Let $U_1\in\bounded{\typel{Q_1},Z}$
  and $U_2\in\bounded{\typel{Q_2},Z}$.
  \par
  Then \symbolindexmark\quanteq
  $(U_1Q_1\quanteq U_2Q_2)\subseteq \elltwov{\qu V}$
  is defined as the subspace fixed by
  \[
    \underbrace{\Uvarnames{Q_2}\adj{U_2}U_1\adj{\Uvarnames{Q_1}}}_{{}\in\boundedv{Q_1,Q_2}}
    \tensor
    \underbrace{\Uvarnames{Q_1}\adj{U_1}U_2\adj{\Uvarnames{Q_2}}}_{{}\in\boundedv{Q_2,Q_1}}
    \tensor
    \underbrace{\idv{\qu V\setminus Q_1Q_2}}_{\hskip-1cm{}\in\boundedv{\qu V\setminus Q_1Q_2}\hskip-1cm}
    \in\boundedv{\qu V}.
  \]
\end{definition}
Note that we allow $U_1,U_2$
to be arbitary bounded operators (instead of isometries).  While we
cannot give an intuitive meaning to $U_1Q_1\quanteq U_2Q_2$
for non-isometries $U_1,U_2$,
it turns out to be convenient to allow non-isometries $U_1,U_2$
in the definition because some simplification rules for
$U_1Q_1\quanteq U_2Q_2$
can then be stated without extra premises (e.g., \autoref{lemma:qeq.move}), and
nothing is lost by the additional generality.

\medskip

From the general definition of $U_1Q_1\quanteq U_2Q_2$,
we can recover the special case $Q_1\quanteq Q_2$ that we started with by setting $U_1:=U_2:=\id$:
\begin{definition}\label{def:quanteq.simple}
  Let $Q_1,Q_2$
  be disjoint lists of distinct quantum variables
  with $\typel{Q_1}=\typel{Q_2}$.
  Then $(Q_1\quanteq Q_2):=(\id\,Q_1\quanteq \id\,Q_2)$
  where $\id$ is the identity on $\typel{Q_1}$.
\end{definition}

The following lemma gives a characterization of the quantum equality
on separable states. This characterization will make the proofs of
rules involving quantum equality much simpler.
\begin{lemma}\label{lemma:quanteq}
  Let $Q_1\subseteq\qu{ V_1}$,
  $Q_2\subseteq \qu{ V_2}$,
  $U_1\in\boundedleq{\typel{Q_1},Z}$,
  and $U_2\in\boundedleq{\typel{Q_2},Z}$ for some set $Z$.
  Let $\psi_1\in\elltwov{\qu{V_1}}$, $\psi_2\in\elltwov{\qu{V_2}}$ be normalized.

  If $\psi_1\tensor\psi_2\in(U_1Q_1\quanteq U_2Q_2)$
  then there exist normalized $\psi_1^Q\in\elltwov{Q_1}$,
  $\psi_1^Y\in\elltwov{\qu{V_1}\setminus Q_1}$,
  $\psi_2^Q\in\elltwov{Q_2}$,
  $\psi_2^Y\in\elltwov{\qu{V_2}\setminus Q_2}$
  such that $U_1\adj{\Uvarnames{Q_1}}\psi^Q_1= U_2\adj{\Uvarnames{Q_2}}\psi^Q_2$
  and $\psi_1=\psi^Q_1\tensor\psi^Y_1$
  and $\psi_2=\psi^Q_2\tensor\psi^Y_2$.

  If $U_1, U_2$ are isometries, the converse holds as well.
\end{lemma}

Note that this lemma assumes that $Q_1$
  and $Q_2$
  are in different subsystems that are not entangled (in states
  $\psi_1,\psi_2$,
  respectively); this will be the case in qRHL judgments as defined
  in \autoref{sec:qrhl} below. The lemma shows that in that case,
  $U_1Q_1\quanteq U_2Q_2$
  implies that the variables $Q_1$ and $Q_2$ are not entangled
  with any other variables.

\begin{proof}
  Throughout this proof, let $\Hat U_1:=U_1\adj{\Uvarnames{Q_1}}$
  and $\Hat U_2:=U_2\adj{\Uvarnames{Q_2}}$.
  Since $U_1,U_2$
  have operator norm $\leq1$, and $\Uvarnames{Q_1},\Uvarnames{Q_2}$ are unitaries, we have that $\Hat U_1,\Hat U_2$ have operator norm $\leq1$, namely $\Hat U_1\in\boundedleq{\types{Q_1},Z}$, $\Hat U_2\in\boundedleq{\types{Q_2},Z}$.
  
\medskip

  We first show the $\Rightarrow$-direction.
  
  Let $\psi_1=\sum_i \lambda_{1i}\, \psi_{1i}^Q\tensor \psi_{1i}^Y$
  be a Schmidt decomposition of $\psi_1$,
  i.e., $\lambda_{1i}>0$,
  and $\psi_{1i}^Q\in\elltwov{Q_1}$,
  $\psi_{1i}^Y\in\elltwov{\qu{V_1}\setminus Q_1}$
  are orthonormal. (Such a decomposition exists by
  \autoref{lemma:schmidt}.) Analogously, let
  $\psi_2=\sum_j \lambda_{2j}\, \psi_{2j}^Q\tensor \psi_{2j}^Y$
  be a Schmidt decomposition of $\psi_2$.
  We have $\sum\lambda_{1i}^2=1$
  and $\sum\lambda_{2j}^2=1$
  because $\psi_1$ and $\psi_2$ are normalized.

  Let $W\subseteq\elltwov{Q_1Q_2}$
  be the subspace fixed by
  $\adj{\Hat U_2}\Hat U_1\tensor\adj{\Hat U_1}\Hat U_2
  = \Uvarnames{Q_2}\adj{U_2}U_1\adj{\Uvarnames{Q_1}}
  \tensor
  \Uvarnames{Q_1}\adj{U_1}U_2\adj{\Uvarnames{Q_2}}
  $.
  Then $W\tensor\elltwov{\qu{V_1}\qu{V_2}\setminus Q_1Q_2}= (U_1Q_1\quanteq U_2Q_2)$ by definition of $\quanteq$.

  Since
  \[
    \sum_{i,j}\lambda_{1i}\lambda_{2j}\,
    \psi^Q_{1i}\tensor\psi^Q_{2j}\tensor\psi^Y_{1i}\tensor\psi^Y_{2j}
    =
    \psi_1\tensor\psi_2 \in (U_1Q_1\quanteq U_2Q_2) =
    W\tensor\elltwov{\qu{V_1}\qu{V_2}\setminus Q_1Q_2},
  \]
  and since all $\psi^Y_{1i}\tensor\psi^Y_{2j}\in\elltwov{\qu{V_1}\qu{V_2}\setminus Q_1Q_2}$
  are orthogonal and non-zero, it follows that all
  $\lambda_{1i}\lambda_{2j}\, \psi^Q_{1i}\tensor\psi^Q_{2j}\in W$.
  Since $\lambda_{1i}\lambda_{2j}>0$,
  we have $\psi^Q_{1i}\tensor\psi^Q_{2j}\in W$ for all $i,j$.

  This implies
  \[
    \adj{\Hat U_2}\Hat U_1\psi^Q_{1i} \tensor \adj{\Hat U_1}\Hat U_2\psi^Q_{2j}
    =
    (\adj{\Hat U_2}\Hat U_1\tensor\adj{\Hat U_1}\Hat U_2)(\psi^Q_{1i}\tensor\psi^Q_{2j})
    \starrel=
    \psi^Q_{1i}\tensor\psi^Q_{2j}
  \]
  where $(*)$
  uses $\psi^Q_{1i}\tensor\psi^Q_{2j}\in W$.
  Note that in the lhs, the first factor is in $\elltwov{Q_2}$,
  while in the rhs, the second factor is in $\elltwov{Q_2}$.
  (Recall that the tensor product is commutative in our formalism.) It
  follows that $\adj{\Hat U_2}\Hat U_1\psi^Q_{1i}=\alpha_{ij} \psi^Q_{2j}$ and $\adj{\Hat U_1}\Hat U_2\psi^Q_{2j}=\beta_{ij} \psi^Q_{1i}$
  for some $\alpha_{ij},\beta_{ij}\in \setC$.
  We have
  \[
    \abs{\alpha_{ij}}=\norm{\alpha_{ij} \psi^Q_{2j}}=\norm{\adj{\Hat U_2}\Hat U_1\psi^Q_{1i}}
    \starrel\leq
    \norm{\psi^Q_{1i}} = 1.
  \]
  Here $(*)$ uses that the operator norm of $\Hat U_1,\adj{\Hat U_2}$ is $\leq1$.
  Furthermore,
  \[
     \psi^Q_{1i}\tensor\psi^Q_{2j}
    =\adj{\Hat U_2}\Hat U_1\psi^Q_{1i} \tensor \adj{\Hat U_1}\Hat U_2\psi^Q_{2j}
    = \alpha_{ij}\psi^Q_{2j}\tensor\beta_{ij}\psi^Q_{1i}
    =\alpha_{ij}\beta_{ij}\cdot   \psi^Q_{1i}\tensor\psi^Q_{2j}.
  \]
  It follows that $\alpha_{ij}\beta_{ij}=1$.
 With $\abs{\alpha_{ij}}\leq1$, this implies $\abs{\alpha_{ij}}=1$ and $\beta_{ij}=1/\alpha_{ij}=\adj{\alpha_{ij}}$.

  Since $\adj{\Hat U_2}\Hat U_1\psi^Q_{1i}=\alpha_{ij} \psi^Q_{2j}$
  for all $j$
  and fixed $i$
  (there exists at least one index $i$
  since $\psi_1\neq0$),
  it follows that all $\psi^Q_{2j}$
  are colinear. Since they are also orthonormal, there can be only one index
  $j$.
  Let $\hat\psi^Q_2:=\psi^Q_{2j}$
  and $\hat\psi^Y_2:=\psi^Y_{2j}$.
  Analogously, there can be only one index $i$.
  Let $\psi^Q_1:=\psi^Q_{1i}$
  and $\psi^Y_1:=\psi^Y_{1i}$.
  Let $\alpha:=\alpha_{ij}$
  and $\beta:=\beta_{ij}=\adj\alpha$.
  
  Since
  $\psi^Q_{1i},\psi^Y_{1i},\psi^Q_{2j},\psi^Y_{2j}$
  are normalized, $\psi^Q_1$
  and $\psi^Y_1$ and $\hat\psi^Q_2$ and $\hat\psi^Y_2$ are all normalized.
And $\lambda_1=\sum\lambda_{1i}=1$ and $\lambda_2=\sum\lambda_{2j}=1$.

  We then have
  $\psi_1=\lambda_1\psi_1^Q\tensor\psi_1^Y=
\psi^Q_1\tensor\psi^Y_1$
  and $\psi_2=\lambda_2\hat\psi_2^Q\tensor\hat\psi_2^Y=
\hat\psi^Q_2\tensor\hat\psi^Y_2$
  and $\adj{\Hat U_2}\Hat U_1\psi^Q_1=\alpha\hat\psi^Q_2$ and
  $\adj{\Hat U_1}\Hat U_2\hat\psi^Q_2=\beta\psi^Q_1$.

  Let $P_1:=\Hat U_1\adj{\Hat U_1}$,
  $P_2:=\Hat U_2\adj{\Hat U_2}$.
 $P_1, P_2$ are positive operators with operator norm $\leq1$.
  Then
  \[
    \pb\norm{\Hat U_1\psi_1^Q}
    = \pb\norm{\Hat U_1 \adj{\Hat U_1}\Hat U_2\hat\psi^Q_2}
    = \pb\norm{ \Hat U_1 \adj{\Hat U_1}\ \Hat U_2\adj{\Hat U_2}\ \Hat U_1\psi^Q_1}
    = \pb\norm{ P_1P_2 \Hat U_1\psi^Q_1 }
    \leq \pb\norm{ P_2 \Hat U_1\psi^Q_1 }
\leq\pb\norm{\Hat U_1\psi_1^Q}.
  \]
Thus $\norm{\Hat U_1\psi_1^Q}=\norm{P_2\Hat U_1\psi_1^Q}$. Since the operator norm of $P_2$ is $\leq1$, this implies that $\Hat U_1\psi_1^Q$ is a linear combination of eigenvectors with eigenvalues of absolute value $1$. Since $P_2$ is positive, the only such eigenvalue is $1$. Thus $\Hat U_1\psi_1^Q$ is an eigenvector with eigenvalue $1$. Hence $P_2 \Hat U_1\psi_1^Q= \Hat U_1\psi_1^Q$. 
  Thus
  \begin{equation*}
    \alpha \Hat U_2\hat\psi_2^Q
    =
    \Hat U_2\adj{\Hat U_2}\Hat U_1\psi^Q_1
    =
    P_2\Hat U_1\psi^Q_1
    =
    \Hat U_1\psi^Q_1.
  \end{equation*}
  Let $\psi^Q_2:=\alpha\hat\psi^Q_2$
  and $\psi^Y_2:=\alpha^{-1}\hat\psi^Y_2$. (Since $\abs\alpha=1$, these are also normalized.)
  Then $\psi_1=\psi^Q_1\tensor\psi^Y_1$
  and $\psi_2=\psi^Q_2\tensor\psi^Y_2$ and $U_1\adj{\Uvarnames{Q_1}}\psi^Q_1= \Hat U_1\psi^Q_1=\alpha\Hat U_2 \hat\psi_2^Q=
\Hat U_2\psi^Q_2=U_2\adj{\Uvarnames{Q_2}}\psi^Q_2$.
  
  This shows the $\Rightarrow$-direction.

  \medskip

  We show the $\Leftarrow$-direction (the converse). In that case,
  $U_1,U_2$ are isometries, and thus $\Hat U_1,\Hat U_2$ as defined at
  the beginning of this proof are isometries as well. We have:
  \begin{align*}
    &\pb\paren{\adj{\Hat U_2}\Hat U_1\tensor\adj{\Hat U_1}\Hat U_2\tensor\idv{\qu{V_1}\qu{V_2}\setminus Q_1Q_2}}
    (\psi_1\tensor\psi_2)
    =
    \adj{\Hat U_2}\Hat U_1\psi^Q_1 \tensor
    \adj{\Hat U_1}\Hat U_2\psi^Q_2 \tensor
      \psi^Y_1\tensor\psi^Y_2 \\
    &\starrel=
      \adj{\Hat U_2}\Hat U_2\psi^Q_2 \tensor
      \adj{\Hat U_1}\Hat U_1\psi^Q_1 \tensor
      \psi^Y_1\tensor\psi^Y_2
      \starstarrel=
\psi^Q_2 \tensor
      \psi^Q_1 \tensor
      \psi^Y_1\tensor\psi^Y_2
    = \psi_1\tensor\psi_2.
  \end{align*}
  Here $(*)$ uses that $\Hat U_1\psi_1^Q=\Hat U_2\psi_2^Q$ by
  assumption of the lemma.  And $(**)$ uses that $\Hat U_1$ and
  $\Hat U_2$ are isometries, and hence $\adj{\Hat U_1}\Hat U_1=\id$ and
  $\adj{\Hat U_2}\Hat U_2=\id$.

  Thus $\psi_1\tensor\psi_2$
  is invariant under
  $\adj{\Hat U_2}\Hat U_1\tensor\adj{\Hat U_1}\Hat U_2\otimes\idv{\qu{V_1}\qu{V_2}\setminus Q_1Q_2}
  = \Uvarnames{Q_2}\adj{U_2}U_1\adj{\Uvarnames{Q_1}}
  \tensor
  \Uvarnames{Q_1}\adj{U_1}U_2\adj{\Uvarnames{Q_2}}\otimes\idv{\qu{V_1}\qu{V_2}\setminus Q_1Q_2}
  $.
By definition of $U_1Q_1\quanteq U_2Q_2$,
  this implies $\psi_1\tensor\psi_2\in( U_1Q_1\quanteq U_2Q_2)$.

  This shows the $\Leftarrow$-direction.
\end{proof}

For the special case $Q_1\quanteq Q_2$, the lemma can be additionally simplified:
\begin{corollary}\label{coro:quanteq}
  Let $Q_1=(\qq_1,\dots,\qq_n)\subseteq\qu{V_1}$
  and $Q_2=(\qq'_1,\dots,\qq'_n)\subseteq\qu{V_2}$
  be disjoint lists of distinct variables with $\typel{Q_1}=\typel{Q_2}$.
  Let $\psi_1\in\elltwov{\qu{V_1}}$,
  $\psi_2\in\elltwov{\qu{V_2}}$ be normalized.
  Let $\sigma:Q_2\to Q_1$
  be the variable renaming with $\sigma(\qq'_i)=\qq_i$.

  Then $\psi_1\tensor\psi_2\in(Q_1\quanteq Q_2)$
  iff there exist normalized $\psi_1^Q\in\elltwov{Q_1}$,
  $\psi_1^Y\in\elltwov{\qu{V_1}\setminus Q_1}$,
  $\psi_2^Q\in\elltwov{Q_2}$,
  $\psi_2^Y\in\elltwov{\qu{V_2}\setminus Q_2}$
  such that $\psi^Q_1=\Urename\sigma\psi^Q_2$
  and $\psi_1=\psi^Q_1\tensor\psi^Y_1$
  and $\psi_2=\psi^Q_2\tensor\psi^Y_2$.
\end{corollary}

\begin{proof}
  By definition, $(Q_1\quanteq Q_2)$
  is the same as $(\id\,Q_1\quanteq \id\,Q_2)$
  where $\id$
  is the identity on $\elltwo{\typel{Q_1}}=\elltwo{\typel{Q_2}}$.

  Thus by 
  \autoref{lemma:quanteq}, 
  $\psi_1\tensor\psi_2\in(Q_1\quanteq Q_2)$
  iff there are 
  normalized  $\psi_1^Q\in\elltwov{Q_1}$,
  $\psi_1^Y\in\elltwov{\qu{V_1}\setminus Q_1}$,
  $\psi_2^Q\in\elltwov{Q_2}$,
  $\psi_2^Y\in\elltwov{\qu{V_2}\setminus Q_2}$
  such that $\id\,\adj{\Uvarnames{Q_1}}\psi^Q_1=\id\,\adj{\Uvarnames{Q_2}}\psi^Q_2$
  and $\psi_1=\psi^Q_1\tensor\psi^Y_1$
  and $\psi_2=\psi^Q_2\tensor\psi^Y_2$.

  Furthermore,  $\id\,\adj{\Uvarnames{Q_1}}\psi^Q_1=\id\,\adj{\Uvarnames{Q_2}}\psi^Q_2$
  iff $\psi^Q_1=\Uvarnames{Q_1}\adj{\Uvarnames{Q_2}}\psi^Q_2$
  iff $\psi^Q_1=\Urename{\sigma}\psi^Q_2$
  (since $\Uvarnames{Q_1}\adj{\Uvarnames{Q_2}}=\Urename{\sigma}$).
\end{proof}

\paragraph{Simplification rules.} We prove three simplification rules
that are useful for rewriting predicates involving quantum equalities:

\begin{lemma}\label{lemma:qeq.move}
  For lists $Q_1$ and $Q_2$ of quantum variables,
  and operators $U_1\in\bounded{\typel{Q_1},Y}$
  and $U_2\in\bounded{\typel{Q_2},Z}$
  and $A\in\bounded{Y,Z}$, we have
  \[
    (AU_1)Q_1\quanteq U_2Q_2
    \qquad=\qquad
    U_1Q_1\quanteq(\adj AU_2)Q_2.
  \]
\end{lemma}
This is especially useful for canceling out terms, e.g., we get
$(AQ_1\quanteq AQ_2) = (\id Q_1\quanteq \adj AAQ_2) = (Q_1\quanteq Q_2)$
for isometries $A$.

\begin{proof}
  By definition,
  $
  (AU_1)Q_1\quanteq U_2Q_2$ is the set of states fixed by 
  \[
    S :=
    {\Uvarnames{Q_2}\adj{U_2}(AU_1)\adj{\Uvarnames{Q_1}}}
    \tensor
    {\Uvarnames{Q_1}\adj{(AU_1)}U_2\adj{\Uvarnames{Q_2}}}
    \tensor
    {\idv{\qu A\setminus Q_1Q_2}}.
  \]
and $    U_1Q_1\quanteq(\adj AU_2)Q_2$ the set of states fixed by
    \[
    T :=
    {\Uvarnames{Q_2}\adj{(\adj AU_2)}U_1\adj{\Uvarnames{Q_1}}}
    \tensor
    {\Uvarnames{Q_1}\adj{U_1}(\adj AU_2)\adj{\Uvarnames{Q_2}}}
    \tensor
    {\idv{\qu A\setminus Q_1Q_2}}.
  \]
Since $S=T$, 
we have $
\pb\paren{(AU_1)Q_1\quanteq U_2Q_2}
    =
\pb\paren{U_1Q_1\quanteq(\adj AU_2)Q_2}
  $.  
\end{proof}

\begin{lemma}\label{lemma:qeq.inside}
  For lists $Q_1$ and $Q_2$ of quantum variables,
  and operators $U_1\in\bounded{\typel{Q_1},Z}$
  and $U_2\in\bounded{\typel{Q_2},Z}$
  and $A_1\in\uni{\typel{Q_1}}$
  and $A_2\in\uni{\typel{Q_2}}$, we have
  \begin{equation}
    (\lift{A_1}{Q_1})
     \cdot (U_1Q_1\quanteq U_2Q_2)
    \qquad=\qquad
    (U_1\adj{A_1})Q_1\quanteq U_2Q_2
    \label{eq:v1.eq}
  \end{equation}
  and
  \begin{equation}
    (\lift{A_2}{Q_2})
    \cdot (U_1Q_1\quanteq U_2Q_2)
    \qquad=\qquad
    U_1Q_1\quanteq (U_2\adj{A_2})Q_2.
    \label{eq:v2.eq}
  \end{equation}
\end{lemma}
Recall that $\lift{}{}$ is the lifting from \autoref{def:lift}.

Predicates of this form occur, e.g., in the preconditions of
\ruleref{QApply1} below. ($\Uvarnames{Q_1}A_1\adj{\Uvarnames{Q_1}}$
is the operator $A_1$,
applied to quantum variables $Q_1$.)
With this rule, we can collect the effect of several unitary quantum
operations inside the quantum equality.

Notice that we require that $A_1,A_2$
are unitary.  We leave it as an open problem to generalize this rule
suitably to isometries.

\begin{proof}
  We first show
  \begin{equation}
    (\lift{A_1}{Q_1}) \cdot (U_1Q_1\quanteq U_2Q_2)
    \qquad\subseteq\qquad
    (U_1\adj{A_1})Q_1\quanteq U_2Q_2\label{eq:v1.leq}
  \end{equation}
  Fix some
  $\psi$ contained in the lhs of \eqref{eq:v1.leq}.  Then $\psi$
  can be written as
  $\psi= (\lift{A_1}{Q_1})\phi$
  for some $\phi\in (U_1Q_1\quanteq U_2Q_2)$.
  Let
  \begin{align*}
    S&:={\Uvarnames{Q_2}\adj{U_2}U_1\adj{\Uvarnames{Q_1}}}
       \tensor
       {\Uvarnames{Q_1}\adj{U_1}U_2\adj{\Uvarnames{Q_2}}}
       \tensor
       {\idv{\qu V\setminus Q_1Q_2}},\\
    T &:= {\Uvarnames{Q_2}\adj{U_2}U_1\adj{A_1}\adj{\Uvarnames{Q_1}}}
        \tensor
        {\Uvarnames{Q_1}A_1\adj{U_1}U_2\adj{\Uvarnames{Q_2}}}
        \tensor
        {\idv{\qu V\setminus Q_1Q_2}}
  \end{align*}
  Since $ (U_1Q_1\quanteq U_2Q_2)$ is the subspace fixed by $S$, we have $\phi=S\phi$.
  Thus
  \begin{align*}
    T\psi &=
            T (\lift{A_1}{Q_1})\phi =
            T \pb\paren{\Uvarnames{Q_1}A_1\adj{\Uvarnames{Q_1}}  \otimes\idv{V\setminus Q_1}}\phi \\
          &= (\Uvarnames{Q_2}\adj{U_2}U_1\adj{A_1}\adj{\Uvarnames{Q_1}}\Uvarnames{Q_1}A_1\adj{\Uvarnames{Q_1}}
            \tensor
            {\Uvarnames{Q_1}A_1\adj{U_1}U_2\adj{\Uvarnames{Q_2}}}
            \tensor
            {\idv{\qu V\setminus Q_1Q_2}})\phi
    \\
          &\starrel= (\Uvarnames{Q_2}\adj{U_2}U_1\adj{\Uvarnames{Q_1}}
            \tensor
            {\Uvarnames{Q_1}A_1\adj{U_1}U_2\adj{\Uvarnames{Q_2}}}
            \tensor
            {\idv{\qu V\setminus Q_1Q_2}})\phi
    \\
          &= (\idv{Q_2}
            \tensor
            \Uvarnames{Q_1}A_1\adj{\Uvarnames{Q_1}}
            \tensor
            {\idv{\qu V\setminus Q_1Q_2}})
            S\phi \\
          &= (\lift{A_1}{Q_1})S\phi
            = (\lift{A_1}{Q_1})\phi
           = \psi.
  \end{align*}
  Here $(*)$
  uses that $A_1$
  is a unitary, and hence $\adj{A_1}A_1=\id$.
  Since $ (U_1\adj{A_1})Q_1\quanteq U_2Q_2$
  is the subspace fixed by $T$, we have that $\psi$ is in the rhs of \eqref{eq:v1.leq}.
  Thus \eqref{eq:v1.leq} holds.

  \medskip

  Furthermore,
  \begin{align}
    (U_1\adj{A_1})Q_1\quanteq U_2Q_2
    &=
      (\lift{A_1}{Q_1})\cdot
      (\lift{\adj{A_1}}{Q_1})\cdot
      \cdot
      \pb\paren{(U_1\adj{A_1})Q_1\quanteq U_2Q_2}
    \notag\\
    &\eqrefrel{eq:v1.leq}\subseteq
      (\lift{A_1}{Q_1})\cdot
      \pb\paren{(U_1\adj{A_1}\adj{(\adj{A_1})})Q_1\quanteq U_2Q_2}
    \notag\\
    &=
      (\lift{A_1}{Q_1})\cdot
      \pb\paren{U_1Q_1\quanteq U_2Q_2}.
      \label{eq:v1.geq}
  \end{align}
  Here \eqref{eq:v1.leq} is applicable because $\adj{A_1}$
  is also unitary, and we showed \eqref{eq:v1.leq} for all bounded
  $U_1$ and unitary~$A_1$.

  From \eqref{eq:v1.leq} and \eqref{eq:v1.geq}, the equation
  \eqref{eq:v1.eq} follows.
  And \eqref{eq:v2.eq} follows by symmetry of $\quanteq$.  
\end{proof}

\begin{lemma}\label{lemma:qeq.span}
  Fix disjoint lists $Q_1$ and $Q_2$ of distinct quantum variables,
  and operators $U_1\in\bounded{\typel{Q_1},Z}$
  and $U_2\in\bounded{\typel{Q_2},Z}$
  and a vector $\psi\in\elltwov{\typel{Q_1}}$.
  Assume $\adj{U_1}U_2\adj{U_2}U_1\psi=\psi$ (e.g., if $U_1,\adj{U_2}$ are isometries)
  Then we have
  \begin{equation}
    (U_1Q_1\quanteq U_2Q_2) \cap
    \pb\paren{\lift{\SPAN\{\psi\}}{Q_1}}
    %\pb\paren{\SPAN\{\Uvarnames{Q_1}\psi\} \otimes \elltwov{\qu V\setminus Q_1}}
    %\\
    =
    \pb\paren{\lift{\SPAN\{\psi\}}{Q_1}}
    %\pb\paren{\SPAN\{\Uvarnames{Q_1}\psi\} \otimes \elltwov{\qu V\setminus Q_1}}
    \cap
    \pb\paren{\lift{\SPAN\{\adj{U_2}{U_1}\psi\}}{Q_2}}
    %\pb\paren{\SPAN\{\Uvarnames{Q_2}\adj{U_2}{U_1}\psi\} \otimes \elltwov{\qu V\setminus Q_2}}.
    \label{eq:qeq.span1}
  \end{equation}
%  and
%  \begin{multline}
  %   (U_1Q_1\quanteq U_2Q_2) \cap \pb\paren{\SPAN\{\Uvarnames{Q_2}\psi_2\} \otimes \elltwov{V\setminus Q_2}} \\
  %   =\qquad
  %   \pb\paren{\SPAN\{\Uvarnames{Q_1}\adj{U_2}{U_1}\psi_1\} \otimes \elltwov{V\setminus Q_1}}
  %   \cap
  %   \pb\paren{\SPAN\{\Uvarnames{Q_2}\psi_1\} \otimes \elltwov{V\setminus Q_2}}
  %   \label{eq:qeq.span2}
  % \end{multline}
\end{lemma}
This lemma tells us (in the special case $U_1=U_2=\id$)
that the following two statements are equivalent:
\begin{compactitem}
\item $Q_1$ and $Q_2$ have the same content, and $Q_1$ contains $\psi$.
\item $Q_1$ contains $\psi$ and $Q_2$ contains $\psi$.
\end{compactitem}

\begin{proof}
  We first show that the lhs of \eqref{eq:qeq.span1} is included in the rhs.
  W.l.o.g., we can assume $\psi\neq 0$ (otherwise the lhs and rhs are trivially $0$).
  Fix some nonzero $\phi\in(U_1Q_1\quanteq U_2Q_2) \cap
    \pb\paren{\lift{\SPAN\{\psi\}}{Q_1}}
%  \pb\paren{\SPAN\{\Uvarnames{Q_1}\psi\} \otimes \elltwov{\qu V\setminus Q_1}}
  $.

  Since
  $\phi\in
  {\lift{\SPAN\{\psi\}}{Q_1}}
%  \SPAN\{\Uvarnames{Q_1}\psi\}
%  \otimes \elltwov{\qu V\setminus Q_1}
  $, there is a $\phi'\in\elltwov{\qu V\setminus Q_1}$
such that $\phi=\Uvarnames{Q_1}\psi\otimes\phi'$ (by \autoref{def:lift}).
And since $0\neq \phi\in (U_1Q_1\quanteq U_2Q_2)$, we have
\begin{multline*}
  \overbrace{\Uvarnames{Q_2}\adj{U_2}U_1\psi}^{\in\elltwov{Q_2}}
  \otimes  
  \pB\paren{
    {\Uvarnames{Q_1}\adj{U_1}U_2\adj{\Uvarnames{Q_2}}}
    \tensor
    {\idv{\qu V\setminus Q_1Q_2}}}
  \phi'\\[-12pt]
=
  \pB\paren{\Uvarnames{Q_2}\adj{U_2}U_1\adj{\Uvarnames{Q_1}}
  \tensor
  {\Uvarnames{Q_1}\adj{U_1}U_2\adj{\Uvarnames{Q_2}}}
  \tensor
  {\idv{\qu V\setminus Q_1Q_2}}}
\overbrace{\paren{\Uvarnames{Q_1}\psi\otimes\phi'}}^{{}=\phi}
=
{\Uvarnames{Q_1}\psi}\otimes\phi' \neq0.
\end{multline*}
By \autoref{lemma:tensor.match}, this implies that there is an
$\eta\in \qu V\setminus Q_1Q_2$ such that
 $ \phi' = {\Uvarnames{Q_2}\adj{U_2}U_1\psi} \otimes \eta$.
 Thus
 \begin{multline*}
   \phi=\Uvarnames{Q_1}\psi\otimes 
   {\Uvarnames{Q_2}\adj{U_2}U_1\psi} \otimes \eta
   \\
  \in
      \pb\paren{\SPAN\{\Uvarnames{Q_1}\psi\} \otimes \elltwov{\qu V\setminus Q_1}}
    \cap
    \pb\paren{\SPAN\{\Uvarnames{Q_2}\adj{U_2}{U_1}\psi\} \otimes \elltwov{\qu V\setminus Q_2}}
    \\
    =
    \pb\paren{\lift{\SPAN\{\psi\}}{Q_1}}
    \cap
    \pb\paren{\lift{\SPAN\{\adj{U_2}{U_1}\psi\}}{Q_2}}
    .
  \end{multline*}
  Hence the lhs of \eqref{eq:qeq.span1} is contained in the rhs.

  \medskip

  We now show that the rhs of \eqref{eq:qeq.span1} is contained in the
  lhs.  Fix some nonzero
  $\gamma\in
    \pb\paren{\lift{\SPAN\{\psi\}}{Q_1}}
    \cap
    \pb\paren{\lift{\SPAN\{\adj{U_2}{U_1}\psi\}}{Q_2}}
%
%  \pb\paren{\SPAN\{\Uvarnames{Q_1}\psi\} \otimes
%    \elltwov{\qu V\setminus Q_1}} \cap
%  \pb\paren{\SPAN\{\Uvarnames{Q_2}\adj{U_2}{U_1}\psi\} \otimes
%  \elltwov{\qu V\setminus Q_2}}
    $.  Since
    $\gamma\in
    \lift{\SPAN\{\psi\}}{Q_1}
    %\pb\paren{\SPAN\{\Uvarnames{Q_1}\psi\} \otimes
    %  \elltwov{\qu V\setminus Q_1}}
    $, we can write
  $\gamma=\Uvarnames{Q_1}\psi\otimes\psi'$
  for some $\psi'\in\elltwov{\qu V\setminus Q_1}$.
  Since
  $\gamma\in
  \lift{\SPAN\{\adj{U_2}{U_1}\psi\}}{Q_2}
  %\pb\paren{\SPAN\{\Uvarnames{Q_2}\adj{U_2}{U_1}\psi\}
  % \otimes \elltwov{\qu V\setminus Q_2}}
  $, we can write
  $\gamma=\Uvarnames{Q_2}\adj{U_2}{U_1}\psi\otimes\phi'$ for some $\phi'\in\elltwov{\qu V\setminus Q_2}$.
  By \autoref{lemma:tensor.match}, this implies that
  $\psi'= \Uvarnames{Q_2}\adj{U_2}{U_1}\psi\otimes\eta$
  for some $\eta\in\elltwov{\qu V\setminus Q_1Q_2}$. Hence $\gamma=\Uvarnames{Q_1}\psi\otimes\Uvarnames{Q_2}\adj{U_2}{U_1}\psi\otimes\eta$. Thus
  \begin{multline*}
    \pB\paren{{\Uvarnames{Q_2}\adj{U_2}U_1\adj{\Uvarnames{Q_1}}}
      \tensor
      {\Uvarnames{Q_1}\adj{U_1}U_2\adj{\Uvarnames{Q_2}}}
      \tensor
      {\idv{\qu V\setminus Q_1Q_2}}}\gamma\\
    =
    \paren{\Uvarnames{Q_2}\adj{U_2}U_1\adj{\Uvarnames{Q_1}}\cdot
    \Uvarnames{Q_1}\psi}
    \otimes
    \paren{\Uvarnames{Q_1}\adj{U_1}U_2\adj{\Uvarnames{Q_2}}\cdot \Uvarnames{Q_2}\adj{U_2}{U_1}\psi}
    \otimes
    \eta\\
    \starrel=
    \Uvarnames{Q_2}\adj{U_2}U_1\psi
    \otimes
    \Uvarnames{Q_1}\psi
    \otimes
    \eta
    =
    \gamma.
  \end{multline*}
  Here $(*)$ uses the assumption
  $\adj{U_1}U_2\adj{U_2}U_1\psi=\psi$.
  Thus $\gamma\in(U_1Q_1\quanteq U_2Q_2)$.
  Since also
  $\gamma\in
\lift{\SPAN\{\psi\}}{Q_1}
  %{\SPAN\{\Uvarnames{Q_1}\psi\} \otimes
  %  \elltwov{\qu V\setminus Q_1}}
  $, we have that $\gamma$ is in the lhs of \eqref{eq:qeq.span1}.
  Thus the rhs is a subset of the lhs.

  \medskip
  
  Summarizing, the lhs and rhs of \eqref{eq:qeq.span1} are subsets of each other, hence \eqref{eq:qeq.span1} holds.
%
%  \medskip
%
%  \eqref{eq:qeq.span2} follows from \eqref{eq:qeq.span1} by symmetry of $\quanteq$.
\end{proof}

\begin{lemma}\label{lemma:quanteqaddstate}
  Fix disjoint lists $Q_1$ and $Q_2$ and $R$ of distinct quantum variables,
  and operators $U_1\in\bounded{\typel{Q_1},Z}$
  and $U_2\in\bounded{\typel{Q_2},Z}$
  and a unit vector $\psi\in\elltwo{\typel{R}}$.
  Then we have
  \begin{equation}
    (U_1Q_1\quanteq U_2Q_2) \cap
    \pb\paren{\lift{\SPAN\{\psi\}}R}
%    \pb\paren{\SPAN\{\Uvarnames{R}\psi\} \tensor \elltwov{\qu V\setminus R}}
    =
    \pb\paren{(U_1\otimes'\id)(Q_1R) \quanteq (U_2\otimes'\psi)Q_2}.
    \label{eq:cap.span}
  \end{equation}
  In this equation, $\otimes'$ refers to the usual positional (unlabeled) tensor product, not to the labeled tensor product
  defined in \autoref{sec:prelim}. And on the rhs, $\psi$ is interpreted as a bounded operator from $\setC$ to $\elltwo{\typel R}$.
\end{lemma}

This rule allows us to rewrite an equality between $Q_1$
and $Q_2$
into an equality between $Q_1R$
and $Q_2$
if we know the content of $R$.
It is the quantum analogue of the following obvious classical fact:
\[
  f(\xx_1)=g(\xx_2) \land \yy_1=y
  \qquad\iff\qquad
  (f\times\id)(\xx_1,\yy_1)
  =
  (\lambda x.\,(x,y))(\xx_2).
\]

\begin{proof}
  We first show the following fact that we need for both directions of the proof:
  \begin{claim}\label{claim:same.qeq}
    If $\phi= \sum_i\phi_{1i}\otimes\phi_{2i}\psi\otimes\phi_{3i}\otimes\Uvarnames R\psi$ for some
    $\phi_{1i}\in\elltwov{Q_1}$, $\phi_{2i}\in\elltwov{Q_2}$, $\phi_{3i}\in\elltwov{\qu V\setminus Q_1Q_2R}$, then
    \begin{multline*}
      \pB\paren{\Uvarnames{Q_2}\adj{(U_2\otimes'\psi)}(U_1\otimes'\id)\adj{\Uvarnames{Q_1R}}
        \tensor
        \Uvarnames{Q_1R}\adj{(U_1\otimes'\id)}(U_2\otimes'\psi)\adj{\Uvarnames{Q_2}}
        \tensor
        \idv{\qu V\setminus Q_1Q_2R}} \phi
      \\
      =
      \pB\paren{
        \Uvarnames{Q_2}\adj{U_2}U_1\adj{\Uvarnames{Q_1}}
        \tensor
        \Uvarnames{Q_1}\adj{U_1}U_2\adj{\Uvarnames{Q_2}}\phi_{2i}
        \tensor
        \idv{\qu V\setminus Q_1Q_2}}
      \phi
    \end{multline*}
  \end{claim}

  \begin{claimproof}
    We have for all $i$:
    \begin{align}
      &\pB\paren{\Uvarnames{Q_2}\adj{(U_2\otimes'\psi)}(U_1\otimes'\id)\adj{\Uvarnames{Q_1R}}}
        (\phi_{1i}\otimes\Uvarnames R\psi)
        \notag\\
      &=
        \Uvarnames{Q_2}\adj{(U_2\otimes'\psi)}(U_1\otimes'\id)
        \pb\paren{\adj{\Uvarnames{Q_1}}\phi_{1i}\otimes'\adj{\Uvarnames{R}}\Uvarnames R\psi}
        \notag\\
      &=
        \Uvarnames{Q_2}
        \pb\paren{\adj{U_2}U_1\adj{\Uvarnames{Q_1}}\phi_{1i}\otimes'
        \underbrace{\adj\psi\adj{\Uvarnames{R}}\Uvarnames R\psi}_{=\norm\psi^2=1}}
        =
        \Uvarnames{Q_2}
        \adj{U_2}U_1\adj{\Uvarnames{Q_1}}\phi_{1i}.
        \label{eq:part1R}
    \end{align}
    And we have for all $i$:
    \begin{align}
      &\pB\paren{\Uvarnames{Q_1R}\adj{(U_1\otimes'\id)}(U_2\otimes'\psi)\adj{\Uvarnames{Q_2}}}\phi_{2i}
        \notag\\
      &=
        \Uvarnames{Q_1R}\adj{(U_1\otimes'\id)}
        (U_2\otimes'\psi)(\adj{\Uvarnames{Q_2}}\phi_{2i}\otimes' \underbrace1_{\in\setC})
        \notag\\[-7pt]
      &=
        \Uvarnames{Q_1R}
        (\adj{U_1}U_2\adj{\Uvarnames{Q_2}}\phi_{2i}\otimes' \psi)
        \notag\\
      &=
        \Uvarnames{Q_1}\adj{U_1}U_2\adj{\Uvarnames{Q_2}}\phi_{2i}\tensor \Uvarnames R \psi
        \label{eq:part2}
    \end{align}
    Thus we have
    \begin{align*}
      &\pB\paren{\Uvarnames{Q_2}\adj{(U_2\otimes'\psi)}(U_1\otimes'\id)\adj{\Uvarnames{Q_1R}}
        \tensor
        \Uvarnames{Q_1R}\adj{(U_1\otimes'\id)}(U_2\otimes'\psi)\adj{\Uvarnames{Q_2}}
        \tensor
        \idv{\qu V\setminus Q_1Q_2R}} \phi
      \\
      &=\sum_i
        \pB\paren{\Uvarnames{Q_2}\adj{(U_2\otimes'\psi)}(U_1\otimes'\id)\adj{\Uvarnames{Q_1R}}}
        (\phi_{1i}\tensor\Uvarnames R\psi) \\[-3pt]
      &\hskip2in \tensor
        \pB\paren{\Uvarnames{Q_1R}\adj{(U_1\otimes'\id)}(U_2\otimes'\psi)\adj{\Uvarnames{Q_2}}}
        \phi_{2i} \tensor
        \phi_{3i}
      \\
      &\txtrel{\eqref{eq:part1R},\eqref{eq:part2}}=\ \sum_i
        \Uvarnames{Q_2}\adj{U_2}U_1\adj{\Uvarnames{Q_1}}\phi_{1i}
        \tensor
        \Uvarnames{Q_1}\adj{U_1}U_2\adj{\Uvarnames{Q_2}}\phi_{2i}\tensor \Uvarnames R \psi
        \tensor
        \phi_{3i}
      \\
      &=
        \pB\paren{
        \Uvarnames{Q_2}\adj{U_2}U_1\adj{\Uvarnames{Q_1}}
        \tensor
        \Uvarnames{Q_1}\adj{U_1}U_2\adj{\Uvarnames{Q_2}}
        \tensor
        \idv{\qu V\setminus Q_1Q_2}}
        \phi.
        \mathQED
    \end{align*}
  \end{claimproof}
  
  We first show that the lhs of \eqref{eq:cap.span} is contained in
  the rhs.  For this, fix some non-zero $\phi$
  in the lhs.  Since
  $\phi\in
  \lift{\SPAN\{\psi\}}R
  %{\SPAN\{\Uvarnames{R}\psi\} \tensor \elltwov{\qu V\setminus R}}
  $, we can decompose $\phi$
  as $\phi= \Uvarnames{R}\psi \tensor \phi'$ for some $\phi'\in\elltwov{\qu V\setminus R}$.
  % And since $\phi\in     (U_1Q_1\quanteq U_2Q_2)$, we have
  % \[
  %   \phi
  %   =
  %   \pb\paren{{\Uvarnames{Q_2}\adj{U_2}U_1\adj{\Uvarnames{Q_1}}}
  %   \tensor
  %   {\Uvarnames{Q_1}\adj{U_1}U_2\adj{\Uvarnames{Q_2}}}
  %   \tensor
  %   {\idv{\qu V\setminus Q_1Q_2}}}
  %   \phi
  % \]
  % and thus
  % \begin{align*}
  %   \Uvarnames{R}\psi \tensor \phi'
  %   &=
  %     \pb\paren{{\Uvarnames{Q_2}\adj{U_2}U_1\adj{\Uvarnames{Q_1}}}
  %     \tensor
  %     {\Uvarnames{Q_1}\adj{U_1}U_2\adj{\Uvarnames{Q_2}}}
  %     \tensor
  %     {\idv{\qu V\setminus Q_1Q_2}}}
  %     ( \Uvarnames{R}\psi \tensor \phi') \\
  %   &=
  %     \pb\paren{{\Uvarnames{Q_2}\adj{U_2}U_1\adj{\Uvarnames{Q_1}}}
  %     \tensor
  %     {\Uvarnames{Q_1}\adj{U_1}U_2\adj{\Uvarnames{Q_2}}}
  %     \tensor
  %     {\idv{\qu V\setminus Q_1Q_2R}}}
  %     \phi'
  %     \
  %     \tensor
  %     \
  %     \Uvarnames{R}\psi.
  % \end{align*}
  % Hence
  % \[
  %   \phi'
  %   =
  %   \pb\paren{{\Uvarnames{Q_2}\adj{U_2}U_1\adj{\Uvarnames{Q_1}}}
  %     \tensor
  %     {\Uvarnames{Q_1}\adj{U_1}U_2\adj{\Uvarnames{Q_2}}}
  %     \tensor
  %     {\idv{\qu V\setminus Q_1Q_2R}}}
  %   \phi'.
  % \]
  We can decompose $\phi'$ as $\phi'=\sum_i\phi_{1i}\otimes\phi_{2i}\otimes\phi_{3i}$ with
  $\phi_{1i}\in\elltwov{Q_1}$, $\phi_{2i}\in\elltwov{Q_2}$, $\phi_{3i}\in\elltwov{\qu V\setminus Q_1Q_2R}$.
Then we have
  \begin{multline*}
    \pB\paren{\Uvarnames{Q_2}\adj{(U_2\otimes'\psi)}(U_1\otimes'\id)\adj{\Uvarnames{Q_1R}}
      \tensor
      \Uvarnames{Q_1R}\adj{(U_1\otimes'\id)}(U_2\otimes'\psi)\adj{\Uvarnames{Q_2}}
      \tensor
      \idv{\qu V\setminus Q_1Q_2R}} \phi
      \\
    \txtrel{\autoref{claim:same.qeq}}=\quad
      \pB\paren{
      \Uvarnames{Q_2}\adj{U_2}U_1\adj{\Uvarnames{Q_1}}
      \tensor
      \Uvarnames{Q_1}\adj{U_1}U_2\adj{\Uvarnames{Q_2}}\phi_{2i}
      \tensor
      \idv{\qu V\setminus Q_1Q_2}}
      \phi
      \ \starrel=\
      \phi
  \end{multline*}
  where $(*)$
  follows from $\phi\in (U_1Q_1\quanteq U_2Q_2)$
  (since $\phi$
  is in the lhs of \eqref{eq:cap.span}).  Thus
  $\phi\in \pb\paren{(U_1\otimes'\id)(Q_1R) \quanteq
    (U_2\otimes'\psi)Q_2}$ which is the rhs of \eqref{eq:cap.span}.
  Thus the lhs of \eqref{eq:cap.span} is contained in the rhs.

  \medskip

  We now prove that the rhs of \eqref{eq:cap.span} is contained in the
  lhs.  Fix some $\phi$
  contained in the rhs.  We can decompose $\phi$
  as $\phi=\sum_i\phi_{1i}\otimes\phi_{2i}\otimes\phi_{3i}$
  with $\phi_{1i}\in\elltwov{Q_1R}$,
  $\phi_{2i}\in\elltwov{Q_2}$,
  $\phi_{3i}\in\elltwov{\qu V\setminus Q_1Q_2R}$.

  Then
  \begin{align}
    \phi
    &\starrel=
      \pB\paren{\Uvarnames{Q_2}\adj{(U_2\otimes'\psi)}(U_1\otimes'\id)\adj{\Uvarnames{Q_1R}}
      \tensor
      \Uvarnames{Q_1R}\adj{(U_1\otimes'\id)}(U_2\otimes'\psi)\adj{\Uvarnames{Q_2}}
      \tensor
      \idv{\qu V\setminus Q_1Q_2R}} \phi
    \notag\\
    &= \sum_i
      \underbrace{\Uvarnames{Q_2}\adj{(U_2\otimes'\psi)}(U_1\otimes'\id)\adj{\Uvarnames{Q_1R}}\phi_{1i}}
         _{=:\hat\phi_{2i}\in\elltwov{Q_2}}
      {}\tensor
      \Uvarnames{Q_1R}\adj{(U_1\otimes'\id)}(U_2\otimes'\psi)\adj{\Uvarnames{Q_2}} \phi_{2i}
      \tensor
      \phi_{3i}
      \notag\\
    &=
      \sum_i
      \hat\phi_{2i} \tensor
      \Uvarnames{Q_1R}\adj{(U_1\otimes'\id)}(U_2\otimes'\psi)
      \pb\paren{\adj{\Uvarnames{Q_2}} \phi_{2i}\otimes'\underbrace1_{\in\setC}}
      \tensor \phi_{3i}
    \notag\\
    &=
      \sum_i
      \hat\phi_{2i} \tensor
      \Uvarnames{Q_1R}
      \pb\paren{\adj{U_1}U_2\adj{\Uvarnames{Q_2}} \phi_{2i}\otimes'\psi}
      \tensor \phi_{3i}
    \notag\\&
      =
      \sum_i
      \hat\phi_{2i} \tensor
      \underbrace{\Uvarnames{Q_1}\adj{U_1}U_2\adj{\Uvarnames{Q_2}} \phi_{2i}}_{=:\hat\phi_{1i}\in\elltwov{Q_1}}
      {} \tensor\Uvarnames R\psi
      \tensor \phi_{3i}
    \notag\\&
        =
        \pB\paren{\sum_i \hat\phi_{1i}\tensor\hat\phi_{2i}\tensor\phi_{3i}}
        {} \tensor \Uvarnames R\psi
        \in
              \pb\paren{\SPAN\{\Uvarnames{R}\psi\} \tensor \elltwov{\qu V\setminus R}}
              = \lift{\SPAN\{\psi\}}R
              .
              \label{eq:phi.in.span}
  \end{align}
  Here $(*)$ holds because $\phi$ is in the rhs of \eqref{eq:cap.span}.

Moreover,  we have
  \begin{align*}
    \phi &\starrel=
    \pB\paren{\Uvarnames{Q_2}\adj{(U_2\otimes'\psi)}(U_1\otimes'\id)\adj{\Uvarnames{Q_1R}}
    \tensor
    \Uvarnames{Q_1R}\adj{(U_1\otimes'\id)}(U_2\otimes'\psi)\adj{\Uvarnames{Q_2}}
    \tensor
    \idv{\qu V\setminus Q_1Q_2R}} \phi
    \\
    &\txtrel{\autoref{claim:same.qeq}}=\ \
    \pB\paren{
    \Uvarnames{Q_2}\adj{U_2}U_1\adj{\Uvarnames{Q_1}}
    \tensor
    \Uvarnames{Q_1}\adj{U_1}U_2\adj{\Uvarnames{Q_2}}\phi_{2i}
    \tensor
    \idv{\qu V\setminus Q_1Q_2}}
    \phi.
  \end{align*}
  Here $(*)$
  follows since $\phi$
  is in the rhs of \eqref{eq:cap.span}.
  And 
  \autoref{claim:same.qeq}
  applies since
  $\phi= \sum_i \hat\phi_{1i}\tensor\hat\phi_{2i}\tensor\phi_{3i}
   \tensor \Uvarnames R\psi$ by \eqref{eq:phi.in.span}.
  Thus
  $\phi\in(U_1Q_1\quanteq U_2Q_2)$.
  Since also
  $\phi\in
  \lift{\SPAN\{\psi\}}R
  %{\SPAN\{\Uvarnames{R}\psi\} \tensor \elltwov{\qu V\setminus R}}
  $ by \eqref{eq:phi.in.span},
  we have that $\phi$ is in the lhs of \eqref{eq:cap.span}.
  Thus the rhs of \eqref{eq:cap.span} is contained in the lhs.

\medskip

  We have shown that the lhs is included in the rhs and that the rhs
  is included in the rhs. Hence \eqref{eq:cap.span} holds.
\end{proof}

\section{Quantum Relational Hoare Logic (qRHL)}
\label{sec:qrhl}

We now formally introduce our quantum Relational Hoare Logic qRHL. See
the introduction (especially the discussion before
\autoref{def:qrhl.informal}) for a motivation of the definition.

\paragraph{Tagged variables.}
In the following, there will always be a fixed set $V$
of program variables, and programs mentioned in qRHL judgments will
always have free variables in $V$.
In addition, we will use two disjoint copies of the set $V$,
called $V_1$
and $V_2$,
consisting of the variables from $V$
tagged with~$1$
or~$2$,
respectively. These tagged variables will be necessary to distinguish
between the variables of the left and right program in a qRHL judgment
(e.g., in the pre- and postcondition of a qRHL judgment).

Formally, $V_1,V_2$
are disjoint sets with variable renamings
\symbolindexmark{\idx}$\idx i:V\to V_i$
for $i=1,2$.
(See the preliminaries for the formal definition of variable
renamings.)
We will write $\xx_i$
for $\idx i(\xx)$.

For an expression $e$
with $\fv(e)\subseteq V$,
we write $\idx i e$
for the expression $e$
with every \emph{classical} variable $\xx$
replaced by $\xx_i$.
Formally, $\denotee{\idx ie}m:=\denotee{e}{m\circ(\idx i\!)^{-1}}$.

For a program $\bc$
with $\fv(\bc)\subseteq V$,
we write $\idx i \bc$
for the program $\bc$
with every classical or quantum variable $\xx$
in statements replaced by $\xx_i$, and every expression $e$ replaced by $\idx ie$.

\paragraph{The definition.}
We now have all the necessary language to make \autoref{def:qrhl.informal} formal:

\begin{definition}[Quantum Relational Hoare judgments]\label{def:rhl}%
  \symbolindexmark{\rhl}%
  Let $\bc_1,\bc_2$ be programs with $\fv(\bc_1),\fv(\bc_2)\subseteq V$.
  Let  $A,B$ be predicates over $V_1V_2$.
  
  Then 
  $\rhl A{\bc_1}{\bc_2}B$ holds iff
  for all $(V_1,V_2)$-separable $\rho\in\traceposcq{V_1V_2}$
  that satisfy $A$,
  we have that there exists a $(V_1,V_2)$-separable
  $\rho'\in\traceposcq{V_1V_2}$ that satisfies $B$ such that
  $\partr{V_1}{V_2}\rho' = \denotc{\idx 1\bc_1}\pb\paren{\partr{V_1}{V_2}\rho}$ and
  $\partr{V_2}{V_1}\rho' = \denotc{\idx 2\bc_2}\pb\paren{\partr{V_2}{V_1}\rho}$.
\end{definition}
(Recall in particular \autoref{def:pred} for the definition of
``predicates'', and \autoref{def:satisfy} for the definition of
``satisfy'', \autoref{sec:qsemantics} for the semantics $\denotc\cdot$,
and the preliminaries for the remaining notation.)

The following is a technical lemma that gives an alternative
characterization of qRHL judgments in terms of deterministic
initial values for the classical variables, and pure initial states
for the quantum variables. This definition will simplify many proofs.
\begin{lemma}[qRHL for pure states]\label{lemma:pure}
  $\rhl A\bc\bd B$ holds iff:

  For all $m_1\in\types{\cl{V_1}}$,
  $m_2\in\types{\cl{V_2}}$ and all normalized $\psi_1\in\elltwov{\qu{V_1}}$,
  $\psi_2\in\elltwov{\qu{V_2}}$ such that $\psi_1\tensor \psi_2\in \denotee A{\memuni{m_1m_2}}$, there exists a $(V_1,V_2)$-separable
  $\rho'\in\traceposcq{V_1V_2}$ such that:

  $\rho'$ satisfies $B$, and
  $\partr{V_1}{V_2}\rho'=\denotc{\idx 1\bc}\pb\paren{\pointstate{\cl{V_1}}{m_1}\tensor\proj{\psi_1}}$, and
  $\partr{V_2}{V_1}\rho'=\denotc{\idx 2\bd}\pb\paren{\pointstate{\cl{V_2}}{m_2}\tensor\proj{\psi_2}}$.
\end{lemma}

\begin{proof}
  First, we show the $\Rightarrow$-direction.
  Assume $\rhl A\bc\bd B$
  and fix $m_1,m_2,\psi_1,\psi_2$
  as in the statement of the lemma.  Since
  $\psi_1\tensor\psi_2\in\denotee A{\memuni{m_1m_2}}$,
  we have that $\rho:=\pb\proj{\basis{\cl{V_1}}{m_1}\tensor \basis{\cl{V_2}}{m_2} \tensor \psi_1\tensor\psi_2}$ satisfies $A$.
  Note that $\rho\in\traceposcq{V_1V_2}$ and $\rho$ is separable.
  Since  $\rhl A\bc\bd B$, there is a separable $\rho'$ that satisfies $B$ and such that
  \begin{align*}
    \partr{V_1}{V_2}\rho' &= \denotc{\idx1\bc}\pb\paren{\partr{V_1}{V_2}\rho} = \denotc{\idx1\bc}\pb\paren{\proj{\basis{\cl{V_1}}{m_1}}\tensor\proj{\psi_1}},
   \\
    \partr{V_2}{V_1}\rho' &= \denotc{\idx2\bd}\pb\paren{\partr{V_2}{V_1}\rho} = \denotc{\idx2\bd}\pb\paren{\proj{\basis{\cl{V_2}}{m_2}}\tensor\proj{\psi_2}}.
  \end{align*}
  This shows the  $\Rightarrow$-direction.

  \medskip
  
  Now we show the $\Leftarrow$-direction. Assume that for all $m_1,m_2,\psi_1,\psi_2$
  as in the statement of the lemma, there exists a separable state
  $\rho'$
  such that: $\rho'$
  satisfies $B$,
  and
  $\partr{V_1}{V_2}\rho'=\denotc{\idx1\bc}\pb\paren{\pointstate{\cl{V_1}}{m_1}\tensor\proj{\psi_1}}$,
  and
  $\partr{V_2}{V_1}\rho'=\denotc{\idx2\bd}\pb\paren{\pointstate{\cl{V_2}}{m_2}\tensor\proj{\psi_2}}$.
  We need to show that $\rhl A\bc\bd B$.
  
  Fix a state $\rho\in\traceposcq{V_1V_2}$
  that satisfies $A$.
  We can decompose $\rho$
  as
  \[
    \rho=\sum_{m_1\in\types{\cl{V_1}},m_2\in\types{\cl{V_2}}}
\pb  \pointstate{\cl{V_1}\cl{V_2}}{\memuni{m_1m_2}}\tensor
  \rho_{\memuni{m_1m_2}}.
  \]
  Since $\rho$ satisfies $A$, we have $\suppo\rho_{\memuni{m_1m_2}}\subseteq \denotee A{\memuni{m_1m_2}}$.
  Furthermore, since $\rho$ is separable, the states $\rho_{\memuni{m_1m_2}}$
  can be decomposed as
  $\rho_{\memuni{m_1m_2}}=\sum_i \lambda_{\memuni{m_1m_2}i} \proj{\psi^{(1)}_{\memuni{m_1m_2}i}\tensor \psi^{(2)}_{\memuni{m_1m_2}i}}
  $
  for some normalized $\psi^{(1)}_{\memuni{m_1m_2}i}\in\elltwov{\qu{V_1}}$,  $\psi^{(2)}_{\memuni{m_1m_2}i}\in\elltwov{\qu{V_2}}$ and $\lambda_{\memuni{m_1m_2}i}\geq0$.
  Since  $\suppo\rho_{\memuni{m_1m_2}}\subseteq \denotee A{\memuni{m_1m_2}}$, 
  $\psi^{(1)}_{\memuni{m_1m_2}i}\tensor \psi^{(2)}_{\memuni{m_1m_2}i}\subseteq \denotee A{\memuni{m_1m_2}}$.

  By assumption, this implies the existence of separable  states
  $\rho'_{\memuni{m_1m_2}i}$ that satisfy $B$, and such that:
  \begin{align}
    \partr{V_1}{V_2}\rho'_{\memuni{m_1m_2}i} &= \denotc{\idx1\bc}\pb\paren{\pointstate{\cl{V_1}}{m_1}\tensor\proj{\psi^{(1)}_{\memuni{m_1m_2}i}}}, \label{eq:x1}\\
    \partr{V_2}{V_1}\rho'_{\memuni{m_1m_2}i} &= \denotc{\idx2\bd}\pb\paren{\pointstate{\cl{V_2}}{m_2}\tensor\proj{\psi^{(2)}_{\memuni{m_1m_2}i}}}. \label{eq:x2}
  \end{align}
  Let $\rho':=\sum_{\memuni{m_1m_2}i}\lambda_{\memuni{m_1m_2}i}\,\rho'_{\memuni{m_1m_2}i}$.
  
  From \eqref{eq:x1}, it follows that
  \[
    \tr\rho'_{\memuni{m_1m_2}i} = \tr\,
  \denotc{\idx1\bc}\pb\paren{\pointstate{\cl{V_1}}{m_1}\tensor\proj{\psi^{(1)}_{\memuni{m_1m_2}i}}}
  \starrel\leq \tr
  \pb\paren{\pointstate{\cl{V_1}}{m_1}\tensor\proj{\psi^{(1)}_{\memuni{m_1m_2}i}}} =
  1.
\]
(Here $(*)$
  holds because $\denotc{\idx1\bc}$
  is a cq-superoperator and thus trace-decreasing.) Thus
  $\sum_{\memuni{m_1m_2}i}\tr\lambda_{\memuni{m_1m_2}i}\,\rho'_{\memuni{m_1m_2}i}\leq
  \sum_{\memuni{m_1m_2}i}\lambda_{\memuni{m_1m_2}i} = \tr\rho \leq \infty$. Thus
  $\rho'=\sum_{\memuni{m_1m_2}i}\lambda_{\memuni{m_1m_2}i}\,\rho'_{\memuni{m_1m_2}i}\in\traceposcq{V_1V_2}$
  exists.  Since all $\rho'_{\memuni{m_1m_2}i}$
  are separable, $\rho'$ is separable. Since all $\rho'_{\memuni{m_1m_2}i}$
  satisfy $B$, $\rho'$ satisfies $B$.

  Furthermore, 
  \begin{align}
    \partr{V_1}{V_2}\rho'
    &= \sum_{\memuni{m_1m_2}i} \lambda_{\memuni{m_1m_2}i}\,\partr{V_1}{V_2}
    \rho'_{\memuni{m_1m_2}i}
    \eqrefrel{eq:x1}=
    \sum_{\memuni{m_1m_2}i} \lambda_{\memuni{m_1m_2}i}\,
    \denotc{\idx1\bc}\pb\paren{\pointstate{\cl{V_1}}{m_1}\tensor\proj{\psi^{(1)}_{\memuni{m_1m_2}i}}}
    \notag\\
 &
   =
   \denotc{\idx1\bc} \pB\paren{ \sum_{m_1m_2i} \lambda_{\scriptscriptstyle  m_1m_2i}\,
   \pb\pointstate{\cl{V_1}}{m_1}\tensor\pb\proj{\psi^{(1)}_{\memuni{m_1m_2}i}}}
    \notag\\
 &
   =
   \denotc{\idx1\bc} \pB\paren{ \partr{V_1}{V_2} \sum_{m_1m_2} \pB\paren{
   \pb\pointstate{\cl{V_1}\cl{V_2}}{\memuni{m_1m_2}}\tensor 
   \underbrace{\sum_i\lambda_{\scriptscriptstyle  m_1m_2i}\,\pb\proj{\psi^{(1)}_{\memuni{m_1m_2}i}\tensor \psi^{(2)}_{\memuni{m_1m_2}i}}}_{{}=\rho_{m_1m_2}}
   }}
    \notag\\
 &
   =
   \denotc{\idx1\bc} \pb\paren{ \partr{V_1}{V_2} \rho}.
   \label{eq:y1}
  \end{align}
And analogously, we get 
  \begin{equation}
    \partr{V_2}{V_1}\rho'
    = 
    \denotc{\idx2\bd} \pb\paren{ \partr{V_2}{V_1} \rho}.
   \label{eq:y2}
  \end{equation}
  
  Summarizing, for any separable $\rho$
  that satisfies $A$,
  there is a separable $\rho'$
  that satisfies $B$
  and such that \eqref{eq:y1} and \eqref{eq:y2} hold. Thus
  $\rhl A\bc\bd B$.  This shows the $\Leftarrow$-direction.
\end{proof}

\subsection{Sound rules for qRHL.}
From the definition of qRHL judgments, we can derive a number of
reasoning rules for qRHL. Our reasoning rules are sound (each rule
follows from the definition of qRHL judgments), but we make no claim
that they are complete. (This situation is similar to that of pRHL,
where no complete set of rules has been proposed either. Instead, the
set of rules is chosen to cover all use cases that occur in common
cryptographic proofs.)

We group the rules in three categories: In
\autoref{fig:rules.general}, we list general reasoning rules.  These
rules are not specific to individual statements (such as assignment or
measurement), but instead describe structural properties of qRHL (such
as the frame rule, or case distinction, etc.) In
\autoref{fig:rules.stmts}, we list the rules specific to the classical
statements in our language (e.g., assignment, sampling, if-statement).
And in \autoref{fig:rules.quantum}, we list rules specific to the
quantum statements (e.g., measurement, quantum initialization).  We
will shortly discuss those rules.
The proofs of the rules are are
given in \autoref{sec:rule-proofs}.

  \begin{theorem}
    The rules in Figures~\ref{fig:rules.general},
    \ref{fig:rules.stmts}, and \ref{fig:rules.quantum} are sound.
  \end{theorem}

(We prove this for each rule individually, see the lemmas
    referenced in the figures.)

\begin{figure}[t]
  \begin{ruleblock}
    \RULE{Sym}{
      \pb\rhl{\Urename{\restrict\sigma{\qu{V_1}\qu{V_2}}}\,(\subst A\sigma)}\bd\bc
      {\Urename{\restrict\sigma{\qu{V_1}\qu{V_2}}}\,(\subst B\sigma)}\\
      \sigma:=\idx 1\circ\idx 2^{-1} \cup \idx 2\circ\idx 1^{-1}
    }{
      \rhl A\bc\bd B
    }
    \RULE{Conseq}{A\subseteq A'\\B'\subseteq B
      \\\rhl{A'}\bc\bd {B'}}{\rhl{A}\bc\bd B}
    \RULE{Seq}{
      \rhl{A}{\bc_1}{\bc_2}{B}
      \\
      \rhl{B}{\bd_1}{\bd_2}{C}
    }{
      \rhl{A}{\seq{\bc_1}{\bd_1}}{\seq{\bc_2}{\bd_2}}{C}
    }
    \RULE{Case}{
      \forall z\in\typee e.\
      \rhl{\CL{e=z}\cap A}\bc\bd{B}
    }{
      \rhl{A}\bc\bd{B}
    }
    \RULE{Equal}{
      \text{$X$ is a list of classical variables} \\
      \text{$Y$ is a list of distinct quantum variables} \\
      \bc\text{ is $XY$-local} \\
      X_i:=\idx iX\\
      Y_i:=\idx iY
    }{
      \pb\rhl
      {\CL{X_1=X_2}\cap (Y_1\quanteq Y_2)}
      \bc\bc
      {\CL{X_1=X_2}\cap (Y_1\quanteq Y_2)}
    }
    \RULE{Frame}{
      X_i':=\idx iX_i \\
      Y_i':=\idx iY_i \\
      \bc\text{ is $X_1$-local}\\
      \bd\text{ is $X_2$-local}\\
      R\text{ is $Y_1'Y_2'$-local}\\
      \text{$X_1\cap Y_1$ and $X_2\cap Y_2$ are classical}\\
      \text{$\bc$ is $(X_1\cap Y_1)$-readonly}
      \\
      \text{$\bd$ is $(X_2\cap Y_2)$-readonly}\\
      \rhl{A}\bc\bd B }{ \rhl{A\cap R}\bc\bd{B\cap R} }
    \RULE{QrhlElim}{
      \rho\in\traceposcq{V_1V_2}\text{ is separable}\\
      \rho\text{ satisfies }A \\
      \Erename{\idx1}(\rho_1)=\partr{V_1}{V_2}\rho \\
      \Erename{\idx2}(\rho_2)=\partr{V_2}{V_1}\rho \\
      \pb\rhl A\bc\bd {\CL{\idx1e \Rightarrow \idx2f}}
    }{
      \pb\prafter e\bc{\rho_1}
      \leq
      \pb\prafter f\bd{\rho_2}
      \rlap{\hskip1.5cm\parbox{4cm}{\footnotesize (also holds for
          $=,\Leftrightarrow$
          and\\\strut\,
          $\geq,\Leftarrow$
          instead of $\leq,\Rightarrow$)}}
    }
    \RULE{QrhlElimEq}{
      \rho\in\traceposcq{V}\\
      \rho\text{ satisfies }A\\
      X\text{ is a list of classical variables}\\
      Y\text{ is a list of distinct quantum variables}\\
      \text{$\bc,\bd$ are $XY$-local} \\
      A\text{ is $\cl VY$-local}\\
     X_i:=\idx iX\\
      Y_i:=\idx iY\\
      A_i:=\Urename{\restrictno{\idx i\!}{\qu{V}}}(\idx i A) \tensor \elltwov{\qu{V_{3-i}}}\\
      \pb\rhl {\CL{X_1=X_2}\cap (Y_1\quanteq Y_2)\cap A_1\cap A_2}{\bc}{\bd} {\CL{\idx1e \Rightarrow \idx2f}}
    }{
      \pb\prafter e\bc{\rho}
      \leq
      \prafter f\bd{\rho}
      \newcommand\alsoholds{\footnotesize (also holds for
      $=,\Leftrightarrow$
      and\\\strut\,
      $\geq,\Leftarrow$
      instead of $\leq,\Rightarrow$)}
    \rlap{\hskip23mm\parbox{4cm}{\alsoholds}}%
  }
  \RULE{TransSimple}{
    X_{pi}:=\idx i\cl{\fv(p)},\ Q_{pi}:=\idx i\qu{\fv(p)}\text{ for }p=\bc,\bd,\be,\ i=1,2
    \\
    \rhl{\CL{X_{c1}=X_{d2}}\cap (Q_{c1}\quanteq Q_{d2})}\bc\bd
    {\CL{X_{c1}=X_{d2}}\cap (Q_{c1}\quanteq Q_{d2})}
    \\
    \rhl{\CL{X_{d1}=X_{e2}}\cap (Q_{d1}\quanteq Q_{e2})}\bd\be
    {\CL{X_{d1}=X_{e2}}\cap (Q_{d1}\quanteq Q_{e2})}
  }{
    \rhl{\CL{X_{c1}=X_{e2}}\cap (Q_{c1}\quanteq Q_{e2})}\bc\be
    {\CL{X_{c1}=X_{e2}}\cap (Q_{c1}\quanteq Q_{e2})}
  }
  \RULE{Trans}{
    Q_{\bc},Q_{\bd},R_{\bd},Q_{\be}\text{ are lists of quantum variables}
    \\
    \qu{\fv(\bc)} \subseteq Q_{\bc}
    \\
    \qu{\fv(\be)} \subseteq Q_{\be}
    \\
    \typee{a_{\bc\bd}},\typee{a_{\bd\be}},\typee{b_{\bc\bd}},\typee{b_{\bd\be}}\subseteq\bool
    \\
    \typee{u_p}\subseteq\iso{\typel{Q_p},Z'}\text{ for }p=\bc,\be
    \\
    \typee{u_{\bd}}\subseteq\uni{\typel{Q_{\bd}},Z'}
    \\
    \typee{v_p}\subseteq\iso{\typel{Q_p},Z'}\text{ for }p=\bc,\be
    \\
    \typee{v_{\bd}}\subseteq\boundedleq{\typel{R_{\bd}},Z'}
    \\
    Q_{pi}:=\idx i Q_p, u_{pi}:=\idx i u_p, v_{pi}:=\idx i v_p\text{ for }i=1,2,\ p=\bc,\bd,\be
    \\
    R_{\bd i}:=\idx i R_{\bd}\text{ for }i=1,2
    \\
    (b_{\bc\bd} \implies (\xx^{(1)}_1,\dots,\xx^{(n)}_1)=\idx2 e_{\bc})
    \text{ holds for some expression }e_{\bc}\text{ and }\{\xx^{(1)},\dots,\xx^{(n)}\}:=\cl{\fv(\bc)}
    \\
    (b_{\bd\be} \implies (\yy^{(1)}_2,\dots,\yy^{(m)}_2)=\idx1 e_{\be})
    \text{ holds for some expression }e_{\be}\text{ and }\{\yy^{(1)},\dots,\yy^{(m)}\}:=\cl{\fv(\bd)}
    \\
    \pB\rhl{\CL{a_{\bc\bd}} \cap \paren{u_{\bc1}Q_{\bc1} \quanteq u_{\bd2}Q_{\bd2}}}\bc\bd{
      \CL{b_{\bc\bd}} \cap \paren{v_{\bc1}Q_{\bc1} \quanteq v_{\bd2}R_{\bd2}}} 
    \\
    \pB\rhl{\CL{a_{\bd\be}} \cap \paren{u_{\bd1}Q_{\bd1} \quanteq u_{\be2}Q_{\be2}}}\bd\be{
      \CL{b_{\bd\be}} \cap \paren{v_{\bd1}R_{\bd1} \quanteq v_{\be2}Q_{\be2}}} 
  }{
    \pB\rhl{\CL{a_{\bc\bd}\circexp a_{\bd\be}} \cap \paren{u_{\bc1}Q_{\bc1} \quanteq u_{\be2}Q_{\be2}}}\bc\be{
      \CL{b_{\bc\bd}\circexp b_{\bd\be}} \cap \paren{v_{\bc1}Q_{\bc1} \quanteq v_{\be2}Q_{\be2}}} 
  }
\end{ruleblock}
\caption{Rules for qRHL (general rules).  Proofs are in
  \autoref{sec:proofs-general}.
    % Rules \rulerefx{QrhlElim} and
    % \rulerefx{QrhlElimEq} also hold with $=,\Leftrightarrow$
    % and $\geq,\Leftarrow$
    % instead of $\leq,\Rightarrow$.
  }
  \label{fig:rules.general}
\end{figure}

\newenvironment{rulesdescitem}{\begin{compactitem}}{\end{compactitem}}

\paragraph{General rules.}
\begin{rulesdescitem}
\item \Ruleref{Sym} states that qRHL is symmetric, i.e., we can exchange the
  left and the right program. However, in that case, the indices of the
  variables in the pre- and postconditions need to be adjusted (e.g.,
  $\xx_1$
  has to become $\xx_2$,
  and $\qq_2$
  has to become $\qq_1$).
  The variable renaming $\sigma$
  from the rule maps $\xx_1$
  to $\xx_2$
  and vice-versa, for classical and quantum variables. Then $A\sigma$
  is the predicate $A$
  with its \emph{classical} variables renamed. ($A\sigma$
  does not rename the quantum variables, because a predicate $A$
  is by definition an expression with free \emph{classical} variables, see
  \autoref{def:pred}.) To rename the quantum variables as well
  (i.e., to apply the renaming ${\restrict\sigma{\qu{V_1}\qu{V_2}}}$
  to the quantum variables), we need to apply the unitary
  $\Urename{\restrict\sigma{\qu{V_1}\qu{V_2}}}$
  that swaps $\elltwov{\qu{V_1}}$
  and $\elltwov{\qu{V_2}}$.
  Thus, to rename both quantum and classical variables in $A$
  we write $\Urename{\restrict\sigma{\qu{V_1}\qu{V_2}}}(A\sigma)$.

\item 
  \Ruleref{Conseq} allows us to weaken the precondition and strengthen
  the postcondition of a qRHL judgment (when doing backwards reasoning).
\item 
  \Ruleref{Seq} allows us to reason about the sequential composition of
  programs by introducing a predicate that has to hold in the middle of
  the execution of the two programs. 
\item
  \Ruleref{Case} allows us to strengthen a precondition by assuming that
  a certain classical variable or expression has a fixed but arbitrary
  value. For example, we can transform a qRHL judgment
  $\rhl{\CL{\xx_1=\xx_2}}\bc\bd{\CL{\xx_1=\xx_2}}$
  into $\rhl{\CL{\xx_1=\xx_2=z}}\bc\bd{\CL{\xx_1=\xx_2}}$.
  Then we can, e.g., perform case distinctions over $z$
  in the ambient logic (and apply different rules and proof steps
  depending on the value of $z$) when proving the latter judgment.

  This rule is similar to the classical proof technique of coin fixing:
  In coin fixing, we replace an adversary that starts out with a random
  randomness tape $\xx$
  by an adversary which has a randomness tape initialized to a fixed
  (but arbitrary) value $z$.

  \Ruleref{Case} is also used in the derivation of other rules. The proofs of \ruleref{If1}
  and \ruleref{JointIf} use \ruleref{Case} to make a case distinction
  between taking the then- and the else-branch of a conditional.
\item \Ruleref{Equal} states that, if all classical and quantum
  variables of $\bc$
  are the same on the left and right before execution, they are the
  same after execution.
  \par
  This rule allows us to reason about unknown code $\bc$ (adversary invocations).
\item 
\Ruleref{Frame} formalizes the fact that if some predicate holds
before execution of the programs, and the predicate refers only to
variables not written in the programs, then it also holds after the
execution. For example, if we know $\rhl A\bc\bd B$,
then we can derive $\rhl{A\cap R}\bc\bd{B\cap R}$,
assuming that $R$
does not refer to any variables occurring in $\bc,\bd$.
(Actually, it is sufficient to assume that $R$
only refers to classical variables that \emph{are not written} by $\bc,\bd$,
and only to quantum variables that \emph{do not occur}\footnote{While
  we allow read-only classical variables of $\bc,\bd$ to occur in $R$,
  we cannot make a comparable exception for quantum variables.
  Measuring a quantum variable modifies it, so there is no obvious concept of
  a read-only quantum variable.}
  in $\bc,\bd$.
See the formal statement of the rule.) 

This rule is useful for compositional analysis.  If $\bc,\bd$
are subprograms of larger programs, we can first analyze $\bc,\bd$
(and derive some judgment $\rhl A\bc\bd B$)
and then use \ruleref{Frame} to show that additional invariants $R$
occurring in the analysis of the larger programs are preserved.

An important special case is where $\bc$
is an unspecified program (modeling an adversary). Say we know that
$\bc$
only accesses variables $\xx$
and $\qq$.
Then by combining \ruleref{Frame} with \ruleref{Equal}, we get
$\rhl{\CL{\xx_1=\xx_2}\cap(\qq_1\quanteq\qq_2)\cap
  R}\bc\bc{\CL{\xx_1=\xx_2}\cap(\qq_1\quanteq\qq_2)\cap R}$
for $R$
that does not refer to $\xx,\qq$.
(Formally, $R$
is $Y$-local
with $\xx,\qq\notin Y$.)
This means that when analyzing a pair $\bd_1,\bd_2$
of programs that both invoke $\bc$,
any invariant $R$
between $\bd_1$
and $\bd_2$
is preserved by the call to the adversary $\bc$
(assuming that $\CL{\xx_1=\xx_2}\cap(\qq_1\quanteq\qq_2)$ holds
before the call to $\bc$, i.e., that the adversary's variables before
invocation have the same state on both sides).

Since typically, we do not know anything about an adversary except
which variables it may or may not access, this is the only way to
reason about adversaries.

\item \Ruleref{QrhlElim} allows us to deduce statements about
  probabilities from qRHL judgments. As the final goal of a proof in
  qRHL is to make statements about, e.g., attack probabilities, it is
  necessary to translate the judgments into statements about
  probabilities of certain events in programs. The \ruleref{QrhlElim}
  allows us to show statements of the form
  $\prafter e\bc{\rho_1}\leq\prafter f\bd{\rho_2}$
  (or analogous equalities) by proving qRHL judgments of the form
  $\rhl\dots\bc\bd{\CL{\idx1e\Rightarrow \idx2f}}$.
  That is, to show that $e$
  is less likely to occur in $\bc$
  (with initial state $\rho_1$)
  than $f$
  is likely to occur in $\bd$
  (with initial state $\rho_2$),
  we need to show that $\bc,\bd$
  can be related in such a way that $e$ implies $f$ after execution.

  Obviously, we cannot expect this to hold for any $\rho_1,\rho_2$,
  so we use the precondition $A$
  to describe the relationship between $\rho_1$
  and $\rho_2$.
  Due to its general treatment of the relationship between $\rho_1$
  and $\rho_2$,
  the rule is probably not easy to apply. For the most important
  special case where $\rho_1=\rho_2$, it is better to use the next rule.

\item
\Ruleref{QrhlElimEq}
is a special case of \ruleref{QrhlElim}
for comparing two programs with the same initial state $\rho$.
To show that $\prafter e\bc{\rho}\leq\prafter f\bd{\rho}$,
we need to show
  $\rhl{\mathit{Eqs}}\bc\bd{\CL{\idx1e\Rightarrow\idx2f}}$
  where $\mathit{Eqs}$
  is the predicate that states equality between all variables of $\bc$
  and $\bd$.
  (And analogously for $\prafter e\bc{\rho}=\prafter f\bd{\rho}$.)

  If additionally, we know that the initial predicate $\rho$
  satisfies a certain predicate $A$,
  we can strengthen the precondition accordingly, see the formal
  statement of the rule. 
  
  An application of \ruleref{QrhlElimEq} is probably the last step (or
  first when performing backwards reasoning) in most qRHL-based proofs,
  since in a game-based proof, we usually compare the probabilities of
  events that hold in different games (programs) when running them
  with the same initial state. (And in most cases, we will have
  $A=\CL{\true}$
  because all relevant variables would be initialized by the
  programs.)
  \item \Ruleref{Trans}: If $\bc$
  stands in a certain relationship to $\bd$,
  and $\bd$
  stands in a certain relationship to $\be$,
  then we expect to be able to conclude what the relationship between
  $\bd$
  and $\be$
  is (transitivity).  \Ruleref{TransSimple} is a simple version of the
  transitivity, where the relationship between $\bc$
  and $\bd$
  and $\be$
  is that all classical and quantum variables are equal in pre- and
  postcondition. In many cases, equality is probably too restrictive,
  but we included this simplified rule for illustrative purposes.
  A more powerful rule is \ruleref{Trans}
  which allows us to use more complex relationships between $\bc$,
  $\bd$
  namely predicates of the form
  $\CL{a_{cd}}\cap(u_cQ_c\quanteq u_dQ_d)$
  (and analogous for $\bd,\be$),
  subject to a number of technical conditions detailed in rule.  That
  is, the relationship between the classical variables can be
  expressed by an expression $a$,
  and the quantum variables can be related by unitaries $u_1,u_2$
  (which in turn may depend on classical variables).  Then the pre-
  and postcondition relating $\bc$ and $\be$ (in the conclusion of the rule) are what one would expect: the
  classical condition becomes $a_{cd}\circexp a_{de}$,
  which is the composition of $a_{cd}, a_{de}$,
  interpreted as relations on classical memories (see below for a
  definition). And the condition on the quantum variables becomes
  $u_cQ_c\quanteq u_dQ_d$
  (that is, quantum equality is transitive in a sense).

  We do not know whether the present formulation of transitivity is
  the most general one. Finding a more general one (or less technical) is an
  interesting open question.

  \medskip
  
  Formally, \symbolindexmark\circexp\pagelabel{page:def:circexp}$e \circexp f$
  is defined as the expression such that for two memories
  $m_1,m_3\in\elltwov{\cl V}$,
  $\denotee{e\circexp
    f}{m_1\circ\idx1^{-1}\,m_3\circ\idx2^{-1}}=\true$ iff there exists
  $m_2\in\elltwov{\cl{V}}$
  such that
  $\denotee{e}{m_1\circ\idx1^{-1}\,m_2\circ\idx2^{-1}}=\true$
  and $\denotee{f}{m_2\circ\idx1^{-1}\,m_3\circ\idx2^{-1}}=\true$.
\end{rulesdescitem}

\begin{figure*}[t]
  \begin{ruleblock}
    \RULE{Skip}{ }{\rhl{A}\Skip\Skip{A}}
    \RULE{Assign1}{ }{
      \pb\rhl{
        \substi B{\idx 1e/\xx_1}
      }{\assign{\xx} e}{\Skip}{B}
    }
    \RULE{Sample1}{
      A:=\pB\paren{\CL{e'\text{ is total}} \cap 
        \bigcap\nolimits_{z\in\suppd e'}\substi B{z/\xx_1}}\\
      e' := \idx1e
    }{
      \rhl{
        A
      }{
        \sample{\xx} e
      }{\Skip} B
    }
    \RULE{JointSample}{
      A
      :=
      \pB\paren{\pb\CL{\marginal1f = \idx1e_1
          \land
          \marginal2f = \idx2e_2}
        \cap
        \bigcap_{(z_1,z_2)\in\suppd f}
        \substi B{z_1/\xx_1,z_2/\yy_2}} \\
      \typee f\subseteq\distr{\typev\xx\times\typev\yy}
    }{
      \rhl{A}
      {\sample{\xx}{e_1}}
      {\sample{\yy}{e_2}}
      {B}
    }
    \RULE{If1}{
      \pb\rhl{\CL{\idx1e}\cap A}\bc\Skip B
      \\
      \pb\rhl{\CL{\lnot\idx1e}\cap A}\bd\Skip B
    }{
      \rhl{A}{\langif e\bc\bd}{\Skip}B
    }
    \RULE{JointIf}{
      A \subseteq \CL{\idx1e_1=\idx2e_2}\\
      \pb\rhl{\CL{\idx1e_1\land\idx2e_2}\cap A}{\bc_1}{\bc_2}B
      \\
      \pb\rhl{\CL{\lnot\idx1e_1\land\lnot\idx2e_2}\cap A}{\bd_1}{\bd_2}B\\
    }{
      \rhl{A}{\langif{e_1}{\bc_1}{\bd_1}}{\langif{e_2}{\bc_2}{\bd_2}}B
    }
    \RULE{While1}{
      \pb\rhl{\CL{\idx1e}\cap A}\bc\Skip A
      \\
      A\subseteq A_1\otimes\elltwov{\qu{V_2}}
      \\
      (\while{\idx1 e}{\idx1\bc})\text{ is total on }A_1
      \quad\text{\small(see Def.~\ref{def:prog.total})}
    }{
      \pb \rhl{A}{\while e\bc}\Skip {\CL{\lnot \idx1e}\cap A}
    }
    % `
    \RULE{JointWhile}{
      A\subseteq \CL{\idx1e_1=\idx2e_2}\\
      \rhl{\CL{\idx1e_1\land\idx2e_2}\cap A}\bc\bd A
    }{
      \rhl{A}{\while{e_1}\bc}{\while{e_2}{\bd}}{\CL{\lnot \idx1e_1\land\lnot \idx2e_2}\cap A}
    }
  \end{ruleblock}
  \caption{Rules for qRHL (related to individual classical statements).  Proofs are in \autoref{sec:proofs-classical}.
    For the rules \rulerefx{Assign1}, \rulerefx{Sample1},
    \rulerefx{If1}, and \rulerefx{While1}, there is also an analogous
    symmetric rule that we do not list explicitly.}
  \label{fig:rules.stmts}
\end{figure*}

\paragraph{Rules for classical statements.}
For each classical
statement of the language (skip, assignment, sampling, if, and while)
we provide one or several rules (\autoref{fig:rules.stmts}). Rules ending with \textsf1 operate
only on the left program and assume that the right program is $\Skip$.
The rules can be applied to a larger program by combining them with
\ruleref{Seq} (as is automatically done by the tactics in our
tool).  The symmetric variants
of those rules (with names ending in \textsf2) are omitted for brevity
but can be derived using \ruleref{Sym}. For some statements, we have
joint rules. E.g., joint sampling for showing qRHL judgments
with a sampling on both sides. These rules are used to prove
qRHL judgments where the two programs are ``in sync'' (e.g., we can
show that two programs $\sample{\xx} e$
and $\sample{\xx}e$
will satisfy postcondition $\xx_1=\xx_2$,
i.e., the two samplings, though random, can be assumed to pick the
same value). The rules concerning classical statements are very
similar to those in pRHL \cite{certicrypt}, except for
\ruleref{JointSample} which we have somewhat generalized:
\begin{itemize}
\item \Ruleref{JointSample} proves qRHL judgments of the form
  $\rhl A{\sample \xx{e_1}}{\sample\yy{e_2}}B$.
  In the simplest case, $e_1,e_2$
  represent the same distribution, and we show the postcondition
  $B:=\CL{\xx_1=\yy_2}$.
  That is, sampling $\xx$
  and $\yy$
  in the two programs proceeds ``in sync''. However, in more general
  cases, we might have different distributions $e_1,e_2$,
  and we might want to show a different relationship between $\xx_1$
  and $\yy_2$.
  In this case, we need to provide a joint distribution $f$
  as a witness, such that $e_1,e_2$
  are the marginals of $f$,
  and $B$
  holds for any $(\xx_1,\yy_2)$
  chosen according to $f$.
  ($f$
  is an expression and thus may itself depend on classical program
  variables, if needed.)
  
  For example, say $e_1$
  is the uniform distribution on $\{0,\dots,n-1\}$
  and $e_2$
  is the uniform distribution on $\{0,\dots,2n-1\}$,
  and $B=\CL{\xx_1\leq\yy_2}$.
  Then we can chose $f:=\pb\paren{\lambda y.\,(\floor{y/2},y)}(e_2)$
  (where we use the notation $g(e)$
  to denote the distribution resulting from applying $g$
  to the values sampled by $e$).
  This $f$
  has the marginals $e_1,e_2$,
  and $(\xx_1,\yy_2)$
  in the support of $f$ always satisfy $\xx_1\leq \yy_2$.

  In contrast, the joint sampling rule \textsc{R-Rand} from \cite{certicrypt} required
  use to provide a bijection $h$
  between the supports of $e_1$
  and $e_2$.
  This is a weaker rule. It cannot be applied to preceding example,
  for example (since there is no bijection between $\{0,\dots,n-1\}$
  and $\{0,\dots,2n-1\}$).
  One can recover the behavior of \textsc{R-Rand} from \cite{certicrypt} as a
  special case of our \ruleref{JointSample} by setting
  $f:=\pb\paren{\lambda x. (x,h(x))}(e_1)$ in our rule.
\end{itemize}

\begin{figure*}[t]
  \begin{ruleblock}
    \RULE{QInit1}{
      Q' := \idx1Q\\
      e' := \Uvarnames{Q'}\idx1e
    }{
      \pb\rhl{(\spaceat A{e'})\otimes\elltwov{Q'}}
      {\Qinit{Q}{e}}\Skip A
    }
    \RULE{QApply1}{
      %Q':=\idx1Q\\
      e':=\lift{\idx1e}{\idx1Q}\\ %\Uvarnames{ Q'}(\idx1 e)\adj{\Uvarnames{ Q'}} \tensor  \idv{\qu{V_1}\qu{V_1}\setminus Q'}\\
    }{
      \rhl{\adj{e'}\cdot(B\cap \im e')}{\Qapply{e}{Q}}\Skip{B}
    }
    \RULE{Measure1}{
      A:=\Bigl(\CL{\idx1e\text{ is a total measurement}}\cap\bigcap\nolimits_{z\in\typev{\xx_1}} \bigl((\substi B{z/\xx_1}\cap \im e'_z) + \orth{(\im e'_z)} \bigr)\Bigr)
      \\
      %Q' :=\idx1 Q\\
      % e'_z := \Uvarnames{Q'}\pb\paren{\idx1e(z)}\adj{\Uvarnames{Q'}} \tensor \idv{\qu{V_1}\qu{V_2}\setminus Q'}
      e'_z := \lift{\idx1e(z)}{\idx1Q}
    }{
      \rhl{A}{\Qmeasure{\xx}{Q}{e}}\Skip B
    }
    \RULE{JointMeasureSimple}{
      \typev\xx=\typev\yy\\
      \typel{Q_1}=\typel{Q_1}\\
      Q_1' := \idx1 Q_1\\
      Q_2' := \idx2 Q_2\\
      e'_{1z} := \lift{\idx1e_1(z)}{Q_1'}\\ %\Uvarnames{Q_1'}\pb\paren{\idx1e_1(z)}\adj{\Uvarnames{Q_1'}} \tensor \idv{\qu{V_1}\qu{V_2}\setminus Q_1'}\\
      e'_{2z} := \lift{\idx2e_2(z)}{Q_2'}\\ %\Uvarnames{Q_2'}\pb\paren{\idx2e_2(z)}\adj{\Uvarnames{Q_2'}} \tensor \idv{\qu{V_1}\qu{V_2}\setminus Q_2'}\\
      \\
      A := \CL{\idx1e_1=\idx2e_2}\cap(Q_1'\quanteq Q_2')\cap \bigcap_{z\in\typev\xx}
      (\substi B{z/\xx_1,z/\yy_2}
      \cap \im e'_{1z} \cap \im e'_{2z})
      + \orth{(\im e'_{1z})} + \orth{(\im e'_{2z})}
    }{
      \rhl{A}
      {\Qmeasure{\xx}{Q_1}{e_1}}{\Qmeasure{\yy}{Q_2}{e_2}}{B}
    }
    \RULE{JointMeasure}{
      \typee f\subseteq\powerset{\typev\xx\times\typev\yy}\\
      \typee{u_1}\subseteq\iso{\typel{Q_1},Z}\\
      \typee{u_2}\subseteq\iso{\typel{Q_2},Z}\\
      Q_1' := \idx1Q_1\\
      Q_2' := \idx2Q_2\\
      e'_{1x} := \lift{\idx1e_1(x)}{Q_1'}\\ %      \Uvarnames{Q_1'}\pb\paren{\idx1e_1(x)}\adj{\Uvarnames{Q_1'}} \tensor \idv{\qu{V_1}\qu{V_2}\setminus Q_1'}\\
      e'_{2y} := \lift{\idx2e_2(y)}{Q_2'}\\ % \Uvarnames{Q_2'}\pb\paren{\idx2e_2(y)}\adj{\Uvarnames{Q_2'}} \tensor \idv{\qu{V_1}\qu{V_2}\setminus Q_2'}\\
      % u'_1 := u_1\adj{\Uvarnames{Q_1'}}\\
      % u'_2 := u_2\adj{\Uvarnames{Q_2'}}\\
      C_f := \pB\CL{
        \pb\paren{\forall x.\idx1e_1(x)\neq0\implies \pb\abs{\{y:(x,y)\in f\}}=1}
        \land
        \pb\paren{\forall y.\idx2e_2(y)\neq0\implies \pb\abs{\{x:(x,y)\in f\}}=1}
      }
      \\
      C_e := \pB\CL{\forall(x,y)\in f.\ u_1\pb\paren{\idx1e_1(x)}\adj{u_1}=u_2\pb\paren{\idx2e_2(y)}\adj{u_2}}
      \\
      A := \bigcap_{(x,y)\in f}
      (\substi B{x/\xx_1,y/\yy_2}
      \cap \im e'_{1x} \cap \im e'_{2y})
      + \orth{(\im e'_{1x})} + \orth{(\im e'_{2y})}
    }{
      \pb\rhl{C_f\cap C_e\cap A\cap(u_1Q_1'\quanteq u_2Q_2')}
      {\Qmeasure{\xx}{Q_1}{e_1}}{\Qmeasure{\yy}{Q_2}{e_2}}{B}
    }
  \end{ruleblock}
  \caption{Rules for qRHL (related to individual quantum statements).
    Proofs are in \autoref{sec:proofs-quantum}.
    For the rules \rulerefx{Measure1}, \rulerefx{QApply1},
    and \rulerefx{QInit1}, there is also an analogous
    symmetric rule that we do not list explicitly.}
  \label{fig:rules.quantum}
\end{figure*}

\paragraph{Rules for quantum statements.}
Like with the classical
statements, we provide one or several rules for each quantum statement
(one-sided ones, and joint ones, see \autoref{fig:rules.quantum}). Again, for the one-sided rules we
only present the ones for the left side (ending with \texttt{1}), the
others are analogous and can be derived using \ruleref{Sym}.
\begin{rulesdescitem}
\item \Ruleref{QInit1} proves qRHL judgments that have a
  quantum initialization of the form $\Qinit Q e$
  on the left side. We assume just a single $\Skip$
  statement on the right side, more complex programs can be first
  decomposed using \ruleref{Seq}.\footnote{%
    Here we implicitly use the (trivial to prove) associativity of $\seq{}{}$ and
    the fact that that $\seq\Skip\bc=\bc=\seq\bc\Skip$ (w.r.t.~to the denotation of programs).
    } The precondition
  $(\spaceat A{e'})\otimes \elltwov{Q'}$
  (with $\spaceat{}{}$
  as in \autoref{def:spaceat}) consists of those states that, when we
  replace the content of variables $Q'$
  by the state described by the expression $e'$,
  we have a state in $A$.
  
  Here $Q'$
  and $e'$
  are derived from $Q$
  and $e$
  by certain natural isomorphisms: $Q'$
  is $\idx1Q$,
  i.e., all variables are indexed with $1$.
  (To distinguish left and right variables in the pre- and
  postcondition.) And $e'$
  is $\Uvarnames{Q'}\idx 1e$,
  i.e., we index the classical variables in $e$
  with $1$,
  and then map the state returned by $e$
  from $\elltwo{\typee e}$
  (not labeled with quantum variables) into $\elltwov{Q'}$
  (labeled with quantum variables).  Similar natural conversions occur
  in the following rules as well, we will ignore them in our informal
  discussions.

  Note that there is no joint rule for quantum initializations. That
  is because initialization does not make any probabilistic choices
  (like measurements do), hence a joint rule would not be any
  different from simply applying rules \rulerefx{QInit1}
  and~\textsc{QInit2} consecutively. (This is analogous to the classical situation
  where \ruleref{Assign1} has no joint variant either.)
\item \Ruleref{QApply1} proves qRHL judgments that have a unitary
  quantum operation of the form $\Qapply eQ$
  on the left side. (As with \ruleref{QInit1}, we assume $\Skip$
  on the right side, and there is an analogous omitted rule
  \textsc{QApply2} for the right side.)

  If $e$
  (and thus $e'$ which is lifted using \autoref{def:lift})
  is unitary, then the precondition can be $\adj{e'}\cdot B$,
  because after applying $e$
  on $Q$ in a state in $\adj{e'}\cdot B$,
  we get a state in $e'\cdot\adj{e'}\cdot B=B$.
  But since we also allow isometries $e$
  in quantum applications, we need to restrict $B$
  to those states that are in the image $\im e'$ of $e'$, leading to the precondition
  $\adj{e'}\cdot (B\cap \im e')$.

  As with quantum initialization, we have no joint rule for quantum
  application since we can simply apply rules \rulerefx{QApply1} and
  \textsc{QApply2} and get the same result.
\item \Ruleref{Measure1} proves qRHL judgments that have a measurement
  of the form $\Qmeasure \xx Qe$
  on the left side. (As with \ruleref{QInit1} and \ruleref{QApply1},
  we assume $\Skip$
  on the right side, and there is an analogous omitted rule
  \textsc{QMeasure2} for the right side.)

  To satisfy the postcondition $B$, we need the following things to hold in the precondition:
  \begin{compactitem}
  \item $e$
    needs to be a total measurement. Otherwise, the probability that
    the left program terminates might be $<1$,
    and the right program terminates with probability $=1$,
    and then $\rhl A {\Qmeasure \xx Qe}\Skip B$
    cannot hold, no matter what the postcondition $B$ is.
  \item For any possible outcome $z$,
    we need that the post-measurement state after outcome $z$
    is in $B\{z/\xx_1\}$ (i.e., $B$ with $\xx_1$ set to be the outcome $z$).
    This is the case if the initial state lies in the complement of
    the image $\im e_z'$
    (then the measurement will not pass), or if it lies in
    $B\cap\im e_z'$
    (then it will pass and stay in $B$),
    or if it is a sum of states satisfying those two conditions.
    Here $e_z'$ is the projector corresponding to outcome $z$.
    Thus,
    for every outcome $z$,
    we have the term
    $(\substi B{z/\xx_1}\cap \im e'_z) + \orth{(\im e'_z)}$
    in the precondition $A$.
  \end{compactitem}
\item \Ruleref{JointMeasureSimple} allows us to analyze two
  measurements that are performed ``in sync''.

  For example, consider a judgment of the following form:
  \[
    \rhl{A'}{\Qmeasure\xx\qq e}{\Qmeasure\yy\qq e}{\CL{\xx_1=\yy_2}}
    \]
  where $e:=M$
  is a measurement in the diagonal basis, and $\qq$
  is a qubit variable, and
  $A':=\SPAN\{\basis{\qq_1\qq_2}{00}\}\otimes\elltwov{V_1V_2\setminus
    \qq_1\qq_2}$ is a precondition that ensures that $\qq_1$
  and $\qq_2$
  are both initialized as $\basis{}0$.
  Then the measurement on each side would yield a uniformly random bit
  $\xx_1$
  and $\yy_2$,
  respectively, and -- analogous to the case where we just sample a
  uniformly random bit on both sides -- we would expect to be able to
  ``match'' the equally-distributed random choices and show the
  postcondition $\CL{\xx_1=\yy_2}$.

  However, by applying rules \rulerefx{Measure1}
  and~\textsc{Measure2}, we will not be able to prove this since they
  treat the measurement outcome as non-deterministic (i.e., they
  ignore the distribution of the outcomes).

Thus we need \ruleref{JointMeasureSimple} for this case. This rule
  requires the following things to hold in the precondition:
  \begin{compactitem}
  \item For any possible outcome $z$,
    we need that the post-measurement state after both measurements
    produce outcome $z$
    is in $B\{z/\xx_1,z/\yy_2\}$
    (i.e., $B$ with $\xx_1,\yy_2$ set to the same outcome $z$).
    Analogous to \ruleref{Measure1} (but generalized to the application of two measurements),
    this means we need to satisfy in the precondition:
    \[
        (\substi B{z/\xx_1,z/\yy_2}
          \cap \im e'_{1z} \cap \im e'_{2z})
          + \orth{(\im e'_{1z})} + \orth{(\im e'_{2z})}
    \]
    for every $z$.
    Here $e_{1z}'$ is the projector for outcome $z$ on variables $Q_1'$ and $e_{2z}'$ 
    the projector for outcome $z$ on $Q_2'$.
    
    In our example, this would become (we omit the
    ${}\otimes\elltwov{\dots}$ subterms for readability):
    \begin{align*}
      &
        \pb\paren{\CL{z=z} \cap \SPAN\{H\basis{\qq_1}{z}\} \cap \SPAN\{H\basis{\qq_2}{z}\}}
        + \orth{\SPAN\{H\basis{\qq_1}{z}\}} + \orth{\SPAN\{H\basis{\qq_2}{z}\}} \\
      &=
        \SPAN\{
        H\basis{\qq_1}{z}
        \otimes
        H\basis{\qq_2}{z}
        \}
        + \orth{\SPAN\{H\basis{\qq_1}{z}\}} + \orth{\SPAN\{H\basis{\qq_2}{z}\}}
        = \elltwov{\qu{V_1}\qu{V_2}}
    \end{align*}%
    where $H$ is the Hadamard matrix $\frac1{\sqrt2}\footnotesize
    \begin{pmatrix}
      1&1\\1&-1
    \end{pmatrix}$. (That this simplifies to
    $\elltwov{\qu{V_1}\qu{V_2}}$ is not unexpected: since in this example, the postcondition 
    $B$
    does not refer to $\qq_1,\qq_2$, the post-measurement state will always be in
    $B$, no matter what the initial state is.)
  \item The measurement on the left and right side need to be the
    same. In the general case, this leads to the condition
    $\CL{\idx1e_1=\idx2e_2}$.
    In our example, this trivializes to $\CL{M=M}=\elltwov{V_1V_2}$.
  \item Finally, in order for the two measurements to give the same
    outcome with the same probability, we need the measured state to
    be the same on both sides. This is captured by the quantum
    equality $Q_1'\quanteq Q_2'$. In our example: $\qq_1\quanteq \qq_2$.
  \end{compactitem}
  So, in our example, the precondition $A$
  given by \ruleref{JointMeasureSimple} is simply
  $A=(\qq_1\quanteq\qq_2)$, and we have
  \[
    \rhl{\qq_1\quanteq\qq_2}{\Qmeasure\xx\qq e\penalty0}{\Qmeasure\yy\qq e}{\CL{\xx_1=\yy_2}}.
  \]
  Since $A'\subseteq\paren{\qq_1\quanteq\qq_2}$
  (because $\basis{\qq_1\qq_2}{00}$
  is invariant under swapping $\qq_1$
  and $\qq_2$),
  we also have
  $\rhl{A'}{\Qmeasure\xx\qq e}{\Qmeasure\yy\qq e}{\CL{\xx_1=\yy_2}}$
  as desired.

  \medskip

  Note that in contrast to \ruleref{Measure1}, we do not require the
  measurements to be total.
\item \Ruleref{JointMeasure} is the general form of
  \ruleref{JointMeasureSimple} which does not require the measurements
  on the left and the right to be the same measurement.

  Instead, we provide a relation $f$
  that pairs up outcomes on the left and right hand side, and provide
  basis transforms that transform $e_1$
  into $e_2$
  and vice versa. More precisely, we provide isometries $u_1$
  and $u_2$,
  such that for all outcomes $(x,y)\in f$,
  $u_1e_1(x)\adj{u_1}=u_2e_2(y)\adj{u_2}$ where $e_1(x),e_2(y)$ are the projectors corresponding
  to outcomes $x,y$ on the left/right side, respectively.
  That is, applying the basis transform $u_1$
  to the first measurement and $u_2$
  to the second gives the same measurement, assuming we pair up the
  projectors $e_1(x),e_2(y)$
  according to $f$.
  ($f,u_1,u_2$
  are all expressions and may depend on classical program variables.)

  Notice that the precondition will contain the quantum equality
  $u_1Q_1\quanteq u_2Q_2$
  instead of simply $Q_1\quanteq Q_2$
  because the measured states need to be equal \emph{up to the basis
    transforms $u_1,u_2$}.
\end{rulesdescitem}

\paragraph{Adversary rule.}
In cryptographic proofs, we will often need to compare two programs
that contain unknown pieces of code, and need to show that an invariant is preserved.
(Usually, this code will be the code of the adversary.)
The simplest case would be if both sides of a qRHL judgment are the same,
i.e., the judgment is of the form $\rhl A\bc\bc A$.
For such cases, the \ruleref{Equal} rule (possibly combined with \rulerefx{Frame}) can be used.
However, often the adversary might call back into other code, e.g.,
when calling an encryption algorithm/oracle.
In that case, the left and right program might not be exactly the same
(e.g., because we are changing the encryption algorithm).
Thus we need a rule for judgments of the form $\rhl{A}\bc{\bc'}A$
where $\bc$ and $\bc'$ are almost the same, except for some known differences.
More formally, let $C$ be a context with multiple holes.\footnote{\label{footnote:context}%
  Formally, a \emph{multi-hole context}%
  \index{multi-hole context}%
  \index{context!multi-hole}
  $C$ matches the grammar $C::=\Box_i \mid \seq CC \mid \langif e C C \mid \while e C \mid \bc$
  where $i=1,\dots,n$, and $\bc$ is an arbitrary program, and $e$ an arbitrary expression.
  Then $C[\bc_1,\dots,\bc_n]$ is the program resulting from replacing $\Box_i$ by $\bc_i$.}
Then we want to prove judgments of the form  $\rhl{A}{C[\bc_1,\dots,\bc_n]}{C[\bc'_1,\dots,\bc'_n]}A$,
without knowing anything about $C$ (except which free variables it has).
We can, however, assume that we know $\bc_i,\bc'_i$ since those programs represent known algorithms.
Thus, we want to reduce  $\rhl{A}{C[\bc_1,\dots,\bc_n]}{C[\bc'_1,\dots,\bc'_n]}A$
to qRHL judgments about $\bc_i,\bc'_i$.
This can be done by induction over the structure of $C$.
The following general rule summarizes the result of such an induction:
\[
\RULE{Adversary}{
  \text{$X,\Tilde X$ are lists of classical variables} \\
  \text{$Y,Z,W$ are lists of distinct quantum variables} \\\\
  X_i:=\idx i X \\
  \Tilde X_i:=\idx i\Tilde X \\
  Y_i:=\idx i Y \\
  W_i:=\idx i W \\
  Y_1Y_2\cap Z=\varnothing \\ Y\cap W=\varnothing \\\\
  \text{$C$ is a context with multiple holes (cf.~\autoref{footnote:context})} \\\\
  \text{C is $XY$-local} \\
  \text{$R$ is $\Tilde X_1\Tilde X_2Z$-local} \\
  \text{$C$ is $(X\cap\Tilde X)$-readonly} \\\\
  A := \CL{X_1=X_2} \cap \paren{Y_1W_1\quanteq Y_2W_2} \cap R \\
  \forall i.\ \rhl{A}{\bc_i}{\bc'_i}{A}
}{
  \rhl{A}{C[\bc_1,\dots,\bc_n]}{C[\bc'_1,\dots,\bc'_n]}{A}
}
\]
In this rule, we used two new definitions: ``$C$ is $V$-readonly''
means that all programs $\bc$ in $C$
(according to the grammar in \autoref{footnote:context}) are $V$-readonly. 
And ``$C$ is $V$-local'' means that all programs $\bc$ in $C$ are $V$-local,
and all expressions $e$ in $C$ satisfy $\fv(e)\subseteq V$.
\Ruleref{Adversary} is a derived rule.
That is, it can be proven by using only rules \rulerefx{Equal}, \rulerefx{Seq},
\rulerefx{Conseq}, \rulerefx{JointIf}, \rulerefx{JointWhile}, \rulerefx{Frame},
without recourse to the semantics of qRHL judgments or the language.

\subsection{Relationship to classical pRHL}
\label{sec:rel.classical}

In this section, we show that the existing classical logic pRHL
\cite{certicrypt} is a special case of our new logic qRHL. Namely, if
no quantum variables are used, qRHL and pRHL coincide. In fact, all
three variants of qRHL (\autoref{def:rhl}, \autoref{def:qrhl.first},
\autoref{def:qrhl.uniform}) are equivalent in that case.

We recap the definition of pRHL, formulated with out notation and with
respect to our semantics (see \autoref{sec:classical.sem}). An
informal definition was given in \autoref{def:prhl.informal}.
\begin{definition}[Probabilistic Relational Hoare Logic (pRHL)]\label{def:prhl}
  Let $\bc_1,\bc_2$
  be classical programs (in the sense of \autoref{sec:classical.sem}).
  Assume that $\qu V=\varnothing$
  (and thus also $\qu{V_1},\qu{V_2}=\varnothing$).
  Let $a,b$ be expressions of type $\bool$.

  Then \symbolindexmark\prhl$\prhl a{\bc_1}{\bc_2}b$
  holds iff for all $m_1\in\types{V_1}$,
  $m_2\in\types{V_2}$
  with $\denotee a{m_1m_2}=\true$,
  there exists a distribution $\mu'\in\distrv{V_1V_2}$
  such that $\denotcl{\bc_1}(\pointdistr{m_1})=\marginal1{\mu'}$
  and $\denotcl{\bc_2}(\pointdistr{m_2})=\marginal2{\mu'}$
  and $\denotee{b}{m_1'm_2'}=\true$ for all $m_1'm_2'\in\suppd{\mu'}$.
\end{definition}
(Recall that $\denotcl{\bc_1}(\pointdistr{m_1})$
is the distribution of the memory of $\bc_1$
after execution if the initial memory was $m_1$.)

The relationship between the four relational Hoare logics in the
classical case are summarized by the following theorem:
\begin{theorem}
  Let $\bc_1,\bc_2$
  be classical programs (in the sense of \autoref{sec:classical.sem}).
  Assume that $\qu V=\varnothing$
  (and thus also $\qu{V_1},\qu{V_2}=\varnothing$).
  Let $a,b$ be expressions of type $\bool$.
  Then the following four statements are equivalent:
  \begin{compactenum}[(i)]
  \item\label{item:prhl} $\prhl a{\bc_1}{\bc_2}b$
    \tabto{2in}{(\autoref{def:prhl}, informally \autoref{def:prhl.informal})}
  \item\label{item:qrhl} $\rhl {\CL a}{\bc_1}{\bc_2}{\CL b}$
    \tabto{2in}{(\autoref{def:rhl}, informally \autoref{def:qrhl.informal})}
  \item\label{item:first} $\rhlfirst {\CL a}{\bc_1}{\bc_2}{\CL b}$
    \tabto{2in}{(\autoref{def:qrhl.first}, informally \autoref{def:qrhl.first.informal})}
  \item\label{item:uniform} $\rhlunif {\CL a}{\bc_1}{\bc_2}{\CL b}$
    \tabto{2in}{(\autoref{def:qrhl.uniform})}
  \end{compactenum}
\end{theorem}

\begin{proof}
  First, note that if $\qu{V_1}\qu{V_2}=\varnothing$,
  then the states $\rho\in\traceposcq{V_1V_2}$
  are just the diagonal operators
  $\rho=\sum_{m_1m_2\in\types{V_1V_2}}\pointstate{V_1V_2}{m_1m_2}\cdot
  p_{m_1m_2}$ with $p_{m_1m_2}\geq0$
  and $\sum_{m_1m_2}p_{m_1m_2}<\infty$.
  (This is because a state
  $\rho_{m_1m_2}\in\traceposv{\qu{V_1}\qu{V_2}}=\traceposv{\varnothing}=\setR$
  is just a scalar $p_{m_1m_2}$.)

  In particular, any $\rho\in\traceposcq{V_1V_2}$
  is separable. Thus \autoref{def:rhl} and \autoref{def:qrhl.first}
  are trivially equivalent in this case, and we have
  $\eqref{item:qrhl}\iff\eqref{item:first}$.

  Furthermore, \autoref{def:qrhl.uniform} immediately implies
  \autoref{def:qrhl.first} since in
  \autoref{def:qrhl.first} we can choose $\rho':=\calE(\rho)$
  where $\calE$
  is the witness from \autoref{def:qrhl.uniform}. Thus
  $\eqref{item:uniform}\implies\eqref{item:first}$.

  \medskip

  We now show $\eqref{item:qrhl}\iff\eqref{item:prhl}$.
  By \autoref{lemma:pure}, $\rhl {\CL a}{\bc_1}{\bc_2}{\CL b}$ is equivalent to:

  \begin{quote}
    \noindent\llap{$(*)$\ }%
    For all $m_1$,
    $m_2$
    and all normalized $\psi_1\in\elltwov{\qu{V_1}}$,
    $\psi_2\in\elltwov{\qu{V_2}}$
    such that $\psi_1\tensor \psi_2\in \denotee{\CL a}{\memuni{m_1m_2}}$,
    there exists a $(V_1,V_2)$-separable
    $\rho'\in\traceposcq{V_1V_2}$ such that:
    
    $\rho'$
    satisfies $\CL{b}$,
    and
    $\partr{V_1}{V_2}\rho'=\denotc{\idx
      1{\bc_1}}\pb\paren{\pointstate{\cl{V_1}}{m_1}\tensor\proj{\psi_1}}$,
    and
    $\partr{V_2}{V_1}\rho'=\denotc{\idx
      2{\bc_2}}\pb\paren{\pointstate{\cl{V_2}}{m_2}\tensor\proj{\psi_2}}$.
  \end{quote}

  For $m_1,m_2$
  with $\denotee{a}{m_1m_2}=\false$,
  $\denotee{\CL a}{\memuni{m_1m_2}}=0$
  and thus there are no normalized
  $\psi_1\tensor \psi_2\in \denotee{\CL a}{\memuni{m_1m_2}}$.
  And for $m_1,m_2$
  with $\denotee{a}{m_1m_2}=\true$,
  normalized
  $\psi_1\tensor \psi_2\in \denotee{\CL
    a}{\memuni{m_1m_2}}=\elltwov{\qu{V_1}\qu{V_2}}=\elltwov\varnothing=\setC$
  are simply scalars with $\abs{\psi_1}=\abs{\psi_2}=1$ and hence
  $\proj{\psi_1}=\proj{\psi_2}=1\in\setC$.
  Thus $(*)$ is equivalent to:
\begin{quote}
    \noindent\llap{$(**)$\ }%
    For all $m_1$,
    $m_2$ with  $\denotee{a}{m_1m_2}=\true$,
    there exists a $(V_1,V_2)$-separable
    $\rho'\in\traceposcq{V_1V_2}$ such that:
    
    $\rho'$
    satisfies $\CL{b}$,
    and
    $\partr{V_1}{V_2}\rho'=\denotc{\idx
      1{\bc_1}}\pb\paren{\pointstate{\cl{V_1}}{m_1}}$,
    and
    $\partr{V_2}{V_1}\rho'=\denotc{\idx
      2{\bc_2}}\pb\paren{\pointstate{\cl{V_2}}{m_2}}$.
  \end{quote}

  As mentioned above, $\rho'\in\traceposcq{V_1V_2}$
  is always separable. And by definition, $\rho'$
  satisfies $\CL b$
  iff $\rho'$
  is of the form
  $\rho'=\sum_{\substack{m_1m_2\ \text{s.t.}\\\denotee
      b{m_1m_2}=\true}}\pointstate{V_1V_2}{m_1m_2}\cdot \mu'(m_1m_2)
  =\qlift{\mu'}$
  with $\mu'\in\ellonev{V_1V_2}$ and $\denotee b{m_1m_2}=\true$ for all $m_1m_2\in\suppd\mu'$.
  ($\qlift{\cdot}$ was defined just before \autoref{lemma:classical.sem}.)
  Furthermore, since $\denotc{\idx i\bc_i}$
  are trace-reducing, we have that $\rho'$
  can only have $\tr\leq 1$, hence $\sum \mu'(m_1m_2)\leq 1$, i.e., $\mu'\in \distrv{V_1V_2}$.

  Thus $(**)$ is equivalent to:
\begin{quote}
    \noindent\llap{$(*\mathord**)$\ }%
    For all $m_1$,
    $m_2$ with  $\denotee{a}{m_1m_2}=\true$,
    there exists a $\mu'\in\distrv{V_1V_2}$ such that:
    
    $\denotee b{m_1'm_2'}=\true$ for all $m_1'm_2'\in\suppd\mu'$
    and
    $\partr{V_1}{V_2}\qlift{\mu'}=\denotc{\idx
      1{\bc_1}}\pb\paren{\pointstate{\cl{V_1}}{m_1}}$,
    and
    $\partr{V_2}{V_1}\qlift{\mu'}=\denotc{\idx
      2{\bc_2}}\pb\paren{\pointstate{\cl{V_2}}{m_2}}$.
  \end{quote}
  
  We have $\partr{V_1}{V_2}\qlift{\mu'}
  =
  \qlift{\marginal1{\mu'}}$, and
  by \autoref{lemma:classical.sem},  we have
  \[
    \denotc{\idx1{\bc_1}}\pb\paren{\pointstate{\cl{V_1}}{m_1}}
    =
    \denotc{\idx1{\bc_1}}\pb\paren{\qlift{\pointdistr{m_1}}}
    =
    \pb\qlift{\denotcl{\idx1{\bc_1}}\paren{\pointdistr{m_1}}}.
  \]
  Thus
  \begin{align*}
    \partr{V_1}{V_2}\qlift{\mu'}=\denotc{\idx
    1{\bc_1}}\pb\paren{\pointstate{\cl{V_1}}{m_1}}
    &\iff
    \qlift{\marginal1{\mu'}}
    =
    \pb\qlift{\denotcl{\idx1{\bc_1}}\paren{\pointdistr{m_1}}} \\
    &\iff
      \marginal1{\mu'}
      =
      \denotcl{\idx1{\bc_1}}\paren{\pointdistr{m_1}}.
  \end{align*}
  and analogously
  \[
    \partr{V_2}{V_1}\qlift{\mu'}=\denotc{\idx
    2{\bc_2}}\pb\paren{\pointstate{\cl{V_2}}{m_2}}
  \iff
      \marginal2{\mu'}
      =
      \denotcl{\idx2{\bc_2}}\paren{\pointdistr{m_2}}.
  \]
  Hence $(*\mathord**)$ is equivalent to:
\begin{quote}
    For all $m_1$,
    $m_2$ with  $\denotee{a}{m_1m_2}=\true$,
    there exists a $\mu'\in\distrv{V_1V_2}$ such that:
    
    $\denotee b{m_1'm_2'}=\true$ for all $m_1'm_2'\in\suppd\mu'$
    and
    $\marginal1{\mu'}
      =
      \denotcl{\idx1{\bc_1}}\paren{\pointdistr{m_1}}$
and
    $\marginal2{\mu'}
      =
      \denotcl{\idx2{\bc_2}}\paren{\pointdistr{m_2}}$.
  \end{quote}
  And this is $\prhl{a}{\bc_1}{\bc_2}b$ by definition.

  Thus $\rhl {\CL a}{\bc_1}{\bc_2}{\CL b}$ is equivalent to 
  $\prhl{a}{\bc_1}{\bc_2}b$, hence we have
  shown  $\eqref{item:qrhl}\iff\eqref{item:prhl}$.

  \medskip

  We now show $\eqref{item:prhl}\implies\eqref{item:uniform}$.
  Thus, assume that $\prhl a{\bc_1}{\bc_2}b$
  holds. We want to show $\rhlunif {\CL a}{\bc_1}{\bc_2}{\CL b}$.
  For this, we first construct a witness $\calE$
  for $\rhlunif {\CL a}{\bc_1}{\bc_2}{\CL b}$.
  Since $\prhl a{\bc_1}{\bc_2}b$,
  for every $m_1\in\types{V_1}$,
  $m_2\in\types{V_2}$
  with $\denotee a{m_1m_2}=\true$,
  there is a $\mu'_{m_1m_2}\in\distrv{V_1V_2}$
  such that
  \begin{equation}\label{eq:class.marg}
    \denotcl{\idx1\bc_1}(\pointdistr{m_1})=\marginal1{\mu'_{m_1m_2}}
    \qquad\text{and}\qquad
    \denotcl{\idx2\bc_2}(\pointdistr{m_2})=\marginal2{\mu'_{m_1m_2}}
  \end{equation}
  and
  \begin{equation}\label{eq:support}
    \denotee{b}{m_1'm_2'}=\true
    \quad\text{for all}
    \quad m_1'm_2'\in\suppd{\mu'_{m_1m_2}}.
  \end{equation}
  For $m_1,m_2$ with $\denotee a{m_1m_2}=\false$, we choose $\mu'_{m_1m_2}:=0$.
  Then \eqref{eq:support} holds for all $m_1,m_2$. (Not just for those with 
  $\denotee a{m_1m_2}=\true$.)
  
  Let $\calE$ be the cq-superoperator on $V_1V_2$ with
  \begin{equation}\label{eq:E.def1}
    \calE(\pointstate{V_1V_2}{m_1m_2})
    :=
    \qlift{\mu'_{m_1m_2}}
  \end{equation}
  (This extends linearly to all of $\traceposcq{V_1V_2}$
  which is the set of diagonal positive operators and thus uniquely
  defines a cq-superoperator.)

  By \eqref{eq:support} and by definition of $\qlift\cdot$,
  for all $m_1m_2$,
  $\qlift{\mu'_{m_1m_2}}$
  is a linear combination of states of the form
  $\pointstate{V_1V_2}{m_1'm_2'}$
  with $\denotee{b}{m_1'm_2'}=\true$.
  That implies that $\qlift{\mu'_{m_1m_2}}$
  satisfies $\CL b$.
  Since for all $\rho$,
  $\calE(\rho)$
  is a linear combination of states $\qlift{\mu'_{m_1m_2}}$,
  it follows that $\calE(\rho)$ satisfies $\CL b$ as well.

  Fix some $m_1,m_2$ with $\denotee a{m_1m_2}=\true$. Then
  \begin{multline*}
    \partr{V_1}{V_2}\,\calE(\pointstate{V_1V_2}{m_1m_2})
    \eqrefrel{eq:E.def1}=
      \partr{V_1}{V_2}\,\qlift{\mu'_{m_1m_2}} 
    = \pb\qlift{\marginal1{\mu'_{m_1m_2}}} 
    \eqrefrel{eq:class.marg}=
      \pb\qlift{\denotcl{\idx1\bc_1}(\pointdistr{m_1})} \\
    \txtrel{\def\lemmaautorefname{Lem.\!}\autoref{lemma:classical.sem}}=
      \quad
      \denotc{\idx1\bc_1}\pb\paren{\qlift{\pointdistr{m_1}}} 
    =
      \denotc{\idx1\bc_1}\pb\paren{\pointstate{V_1}{m_1}} 
    =
      \denotc{\idx1\bc_1}\pb\paren{\partr{V_1}{V_2}\pointstate{V_1V_2}{m_1m_2}}.
    \end{multline*}
  Thus
  \begin{equation}\label{eq:rho.bc1}
    \partr{V_1}{V_2}\,\calE(\rho)
    = \denotc{\idx1\bc_1}\pb\paren{\partr{V_1}{V_2}\rho}
  \end{equation}
  holds for any $\rho$
  of the form $\pointstate{V_1V_2}{m_1m_2}$ with $\denotee a{m_1m_2}=\true$,
  and thus, by linearity, for every $\rho\in\traceposcq{V_1V_2}$ that satisfies $\CL a$.
Analogously,
\begin{equation}\label{eq:rho.bc2}
    \partr{V_2}{V_1}\,\calE(\rho)
    = \denotc{\idx2\bc_2}\pb\paren{\partr{V_2}{V_1}\rho}
  \end{equation}
  for every $\rho\in\traceposcq{V_1V_2}$ that satisfies $\CL a$.

  From the fact that $\calE(\rho)$
  satisfies $\CL b$
  for all $\rho$,
  and \eqref{eq:rho.bc1} and \eqref{eq:rho.bc2}, it follows that
  $\calE$
  is a witness for $\rhlunif {\CL a}{\bc_1}{\bc_2}{\CL b}$.
  This implies that $\rhlunif {\CL a}{\bc_1}{\bc_2}{\CL b}$
  holds. Hence we have shown
  $\eqref{item:prhl}\implies\eqref{item:uniform}$.

  \medskip

  Altogether, we showed
  $\eqref{item:uniform}\implies\eqref{item:first}\iff\eqref{item:qrhl}\iff\eqref{item:prhl}\implies
  \eqref{item:uniform}$.
  Hence the four statements from the lemma are equivalent.
\end{proof}

\subsection{Alternative definitions of qRHL}
\label{sec:alt.def}

In this section, we discuss some alternative attempts at defining
qRHL. These attempts have not been successful but we believe that
including them in this paper may help other researchers to avoid
repeating those attempts.

\paragraph{Non-separable definition.}
The first definition we studied was the following:
\begin{definition}[qRHL, non-separable definition]\label{def:qrhl.first}%
  \symbolindexmark{\rhlfirst}%
  Let $\bc_1,\bc_2$ be programs with $\fv(\bc_1),\fv(\bc_2)\subseteq V$.
  Let  $A,B$ be predicates over $V_1V_2$.
  
  Then 
  $\rhlfirst A{\bc_1}{\bc_2}B$ holds iff
  for all $\rho\in\traceposcq{V_1V_2}$
  that satisfy $A$,
  we have that there exists a
  $\rho'\in\traceposcq{V_1V_2}$ that satisfies $B$ such that
  $\partr{V_1}{V_2}\rho' = \denotc{\idx 1\bc_1}\pb\paren{\partr{V_1}{V_2}\rho}$ and
  $\partr{V_2}{V_1}\rho' = \denotc{\idx 2\bc_2}\pb\paren{\partr{V_2}{V_1}\rho}$.
\end{definition}
This definition arose by direct analogy to the definition of pRHL.
The only difference between this definition and \autoref{def:rhl}  (the
definition of qRHL that we actually use) is that \autoref{def:rhl}
restricts $\rho$
and $\rho'$
to separable states. We already described this definition informally in the
introduction as \autoref{def:qrhl.first.informal}.

When analyzing this definition, we got stuck trying to prove a frame rule such as the following (other rules posed no major problems):
\[
  \inferrule{
  A,B,\bc,\bd\text{ are $X$-local}\\
  R\text{ is $Y$-local}\\
  X\cap Y=\varnothing\\
  \rhlfirst A\bc\bd B
}{
  \rhlfirst{A\cap R}\bc\bd {B\cap R}
}
\]
(For to keep notation simple, we often omit the indices $1$
and $2$
from variable names and applications of the functions $\idx1,\idx2$
during this discussion.)  The reason why it is difficult to prove this
rule is the following property of $\rhlfirst A\bc\bd B$:
When proving this judgment, we start with some initial joint state
$\rho$
that satisfies $A$,
compute the marginals $\rho_1:=\partr{V_1}{V_2}\rho$
and $\rho_2:=\partr{V_2}{V_1}\rho$,
compute the final states $\rho_1':=\denotc\bc(\rho_1)$,
$\rho_2':=\denotc\bd(\rho_2)$,
and finally we recombine $\rho_1'$
and $\rho_2'$
into a state $\rho$
that satisfies $B$.
This recombination leaves us a lot of freedom (all that is required is
that $\rho_1',\rho'_2$
are the marginals of $\rho$).
In particular, $\rho'$
may decompose into $\rho_1',\rho_2'$
in a very different way than $\rho$
decomposes into $\rho_1,\rho_2$.

The following example illustrates this: Let $A:=\CL{\xx_1=\xx_2}$
and $B:=\CL{\xx_1\neq \xx_2}$
and $\bc,\bd:=\Skip$.
Consider the initial state
$\rho:=\frac12\proj{\basis{\xx_1\xx_2}{00}}+\frac12\proj{\basis{\xx_1\xx_2}{11}}$,
i.e., a state where $\xx_1$
and $\xx_2$
are uniformly random but equal.  $\rho$
satisfies $A$.
Then
$\rho_1=\rho_2=\frac12\proj{\basis{\xx}0}+\frac12\proj{\basis{\xx}1}$,
i.e., a uniformly random bit.  Since $\bc,\bd=\Skip$,
$\rho_1'=\rho_2'=\rho_1=\rho_2$.
It would seem natural to then pick $\rho':=\rho$,
but we have another choice:
$\rho':=
\frac12\proj{\basis{\xx_1\xx_2}{01}}+\frac12\proj{\basis{\xx_1\xx_2}{10}}$. This
$\rho'$
also has marginals $\rho_1',\rho_2'$,
and $\rho'$
satisfies $B$.
Of course, this does not show that $\rhlfirst A\bc\bd B$
holds, since there other initial $\rho$
that satisfy $A$
for which this reasoning does not apply.  But it illustrates that
there could be many different ways how the final state $\rho'$
is constructed.  And, what is worse, it may be that for each initial
$\rho$
satisfying $A$,
the way how the corresponding $\rho'$
is constructed is very different. (In other words, the mapping
$\rho\mapsto\rho'$ is not, e.g., a continuous function.)

What does this have to do with the frame rule?

If we want to show $\rhlfirst{A\cap R}\bc\bd {B\cap R}$,
we start with a state $\rho$
that satisfies $A\cap R$,
and that has marginals $\rho_1$
and $\rho_2$,
which are transformed by $\bc,\bd$
into $\rho'_1,\rho'_2$,
respectively. Now, since $\rho$
satisfies $A$,
and $\rhlfirst A\bc\bd B$
holds by assumption, we know that $\rho_1',\rho_2'$
can be recombined into some $\rho'$
that satisfies $B$.
We can also show that $\rho_1',\rho_2'$
can be recombined into some $\rho''$
that satisfies $R$.\footnote{We
  skip the proof here, but at least for total $\bc,\bd$,
  it is easy to see that $\rho'':=(\denotc\bc\otimes\denotc\bd)(\rho)$
  satisfies $R$
  if $\rho$
  satisfies $R$.}
However, $\rho'$
and $\rho''$
are not the same state. It might be that any way to recombine
$\rho_1',\rho_2'$
into some state that satisfies $B$
will necessarily have to combine $\rho_1',\rho_2'$
in a way that violates the predicate $R$.
While we were not able to find a concrete example where the frame rule
breaks down for \autoref{def:qrhl.first}, we were also not able to
come up with a proof that solves the problem of recombining
$\rho_1',\rho_2'$ in a way that both $B$ and $R$ are satisfied.

\paragraph{Uniform definition.}
The difficulties with the frame rule seem to hint that the problem
with \autoref{def:qrhl.first} is the fact that $\rho'$
can be constructed from $\rho_1',\rho_2'$
in very arbitrary ways, leading to a possibly not even continuous map
$\rho\mapsto\rho'$ (and such a map is unlikely to have a natural operational meaning).

Our second attempt was therefore to formulate a definition that forces
$\rho'$
to depend in a more uniform way on $\rho$.
More specifically, we require $\rho\mapsto\rho'$
to be a cq-superoperator (which we will call the ``witness'').

We get the following definition:
\begin{definition}[qRHL, uniform definition]\label{def:qrhl.uniform}%
  \symbolindexmark{\rhlunif}%
  $\rhlunif A{\bc_1}{\bc_2}B$
  holds iff there exists a cq-superoperator $\calE$
  on $V_1V_2$
  (the witness\index{witness}) such that for all
  $\rho\in\traceposcq{V_1V_2}$ that satisfy $A$, we have:
  \begin{compactitem}
  \item $\rho':=\calE(\rho)$ satisfies $B$,
  \item $\partr{V_1}{V_2}\rho' = \denotc{\idx 1\bc_1}\pb\paren{\partr{V_1}{V_2}\rho}$,
  \item $\partr{V_2}{V_1}\rho' = \denotc{\idx 2\bc_2}\pb\paren{\partr{V_2}{V_1}\rho}$.
  \end{compactitem}
\end{definition}
With this definition, the frame rule holds: One first proves that a
witness for $\rhlunif A\bc\bd B$
(where $A,B,\bc,\bd$
are $X$-local)
can without loss of generality be chosen to be
$X$-local.\footnote{For a witness $\calE$ for $\rhlunif A\bc\bd B$, let
  $\calE'(\rho):=
  \partr{XZ}{Z'}
  \Erename\sigma \circ (\calE\otimes\id_{Z'}) \circ \Erename\sigma
  (\rho\otimes\sigma_0) $ where $Z:=V_1V_2\setminus X$,
  $Z'$
  is a fresh copy of the variables $Z$,
  $\sigma$
  is the variable renaming that exchanges $Z$
  and $Z'$,
  and $\sigma_0\in\traceposcq{Z'}$
  is an arbitrary operator with $\tr\sigma_0=1$.

  Then $\calE'$
  is a witness for $\rhlunif A\bc\bd B$
  as well, and $\calE'$
  is $X$-local.
  (Assuming that $A,\bc,\bd,B$
  are $X$-local.)}
Since $\calE$
is a witness for $\rhlunif A\bc\bd B$,
for any $\rho$
that satisfies $A\cap B$
(and thus $A$),
we have that $\rho':=\calE(\rho)$
satisfies $B$.
And
$\partr{V_1}{V_2}\rho' = \denotc{\idx
  1\bc}\pb\paren{\partr{V_1}{V_2}\rho}$ and
$\partr{V_2}{V_1}\rho' = \denotc{\idx
  2\bd}\pb\paren{\partr{V_2}{V_1}\rho}$.  Furthermore, since $\calE$
is $X$-local
and $R$
is $Y$-local
with $X\cap Y=\varnothing$,
and since $\rho$
satisfies $R$,
we have that $\rho'=\calE(\rho)$
satisfies $R$
as well. Thus $\rho'$
satisfies $B\cap R$.
This implies that $\calE$
is in fact already a witness for $\rhlunif{A\cap R}\bc\bd{B\cap R}$.
The frame rule follows. (We omit a detailed proof because we do not
use this version of qRHL in the end, anyway.)

So, what is the problem with \autoref{def:qrhl.uniform} then?  The
problem is to find a notion of quantum equality. As discussed in the
introduction, we need some notion of equality of quantum
variables. Besides other rules, such an equality notion should satisfy
the following rule:
\[
  \inferrule{
    \bc\text{ is $Q$-local} \\
  }{
    \rhlunif
    {Q_1\quanteq Q_2}
    \bc\bc
    {Q_1\quanteq Q_2}
  }
\]
In other words, if the left and right variables are equal before
execution, and we execute the same program on both sides, then the
variables should be equal afterwards.

Unfortunately, for \autoref{def:qrhl.uniform}, no such equality notion
exists (at least if we want the equality notion to be a quantum
predicate that can be represented as a subspace)!

The problem is the following: say we are given an initial state
$\rho=\rho_1\otimes\rho_2$
satisfying $\qq_1\quanteq\qq_2$.
Then $\rho_1=\rho_2$.
Say we execute on both side a program $\bc$
that effectively applies some projector $P$
to the variable $\qq_1$
(e.g., because $\bc$
is a non-total projective measurement). Then the final state will be
$\rho_1':=P\rho_1 P$
and $\rho_2':=P\rho_2 P$
on the left and right side, respectively. We need to find a state
$\rho'$
such that $\rho_1'$
and $\rho_2'$
are the marginals of $\rho'$.
At first glance, it seems the best choice is
$\rho':=\rho'_1\otimes\rho'_2$.
But then $\partr{V_1}{V_2}\rho'=\rho'_1\cdot \tr\rho'_2$
which does in general not equal $\rho_1'$
because $\rho_2'$
may not have trace $1$
(namely, if $P$
corresponds to a measurement that fails with non-zero probability).
In order to get the correct marginals, we need to define
$\rho':=\frac{\rho_1'\otimes\rho_2'}{\tr\rho_1'}=\frac{\rho_1'\otimes\rho_2'}{\tr\rho_2'}$.
(In the formal impossibility proof, we show that at least in some cases, this is
indeed the only possible choice of $\rho'$.)
But now the mapping $\rho\mapsto\rho'$
cannot be realized by a cq-superoperator any more (or any linear
function), hence there is no witness for
$\rhlunif{\qq_1\quanteq\qq_2}\bc\bc{\qq_2\quanteq\qq_2}$.
(This holds not only for the definition of $\quanteq$
given in \autoref{def:quanteq.simple}, but for any equality definition that can
be represented as a subspace.)

The formal impossibility result is given in \autoref{sec:qeq.uniform}.

One may wonder whether this impossibility result hinges on the fact
that we used a program $\bc$
that was not total (because it maps $\rho_1$
to $P\rho_1P$).
However, similar problems occur also with total programs, e.g., if
$\bc$
measures with a projective measurement $M=\{P,1-P\}$
and stores the result in a classical variable.

\medskip

Because of the problem with equality, we abandon the uniform
definition of qRHL (\autoref{def:qrhl.uniform}) in favor of
\autoref{def:rhl} which is closer in spirit to
\autoref{def:qrhl.first} (except that it requires separable states
$\rho,\rho'$).
In some sense, this is unfortunate because
\autoref{def:qrhl.uniform} feels somewhat more natural: it does not
allow us to pick the state $\rho'$
in arbitrary dependency of $\rho$,
but instead requires that dependency to be natural in a physical
sense.  It is an interesting open question whether we can find some
definition in the spirit of \autoref{def:qrhl.uniform} but which has a
quantum equality.

As a side effect, note the problem with \autoref{def:qrhl.uniform} also shows that the
approach of \cite{barthe17coupling} is unlikely to work in the quantum
setting (at least not when translated straightforwardly to the quantum
setting). In \cite{barthe17coupling}, a logic $\times$qRHL
is developed in which an explicit product program is used as a witness
for a pRHL judgment. Those explicit product programs from $\times$qRHL
are exactly the classical special cases of the witness in
\autoref{def:qrhl.uniform}.

\section{Examples}
\label{sec:example}

\subsection{EPR pairs}

We present a small example of a derivation in our logic.  What we show
is that, if we initialize two variables with an EPR pair, it does not
matter whether we apply a Hadamard transform to the first or to the second variable.

The two programs we analyze are thus:
\begin{equation*}
  \bc\ :=\ \seq{\Qinit{\qq\rr}{\mathrm{EPR}}}{\Qapply H{\qq}},
  \qquad
  \bd\ :=\ \seq{\Qinit{\qq\rr}{\mathrm{EPR}}}{\Qapply H{\rr}}.
\end{equation*}
Here $\mathrm{EPR}:=1/\sqrt2 \pb\basis{}{(0,0)} + 1/\sqrt2 \pb\basis{}{(1,1)}$ is an EPR pair, and $H:=
\frac1{\sqrt2}\scriptscriptstyle
\begin{pmatrix}
  1&1\\1&-1
\end{pmatrix}$ is the Hadamard transform.

We want to show that both programs return the same quantum state, i.e., we
want to show
\begin{equation}
  \label{eq:example}
  \rhl{\CL{\true}}\bc\bd{\qq_1\rr_1\quanteq \qq_2\rr_2}.
\end{equation}

We refer to the first/second statement of $\bc$ as $\bc_1,\bc_2$. Analogously for $\bd$.
Let $I_0:=\paren{\qq_1\rr_1\quanteq \qq_2\rr_2}$.

By \ruleref{QApply1}, we have
$\rhl{I_1}{\bc_2}{\Skip}{I_0}$ with
\[
  I_1 := \adj{\paren{\lift H{\qq_1}}} \cdot \pb\paren{I_0 \cap \im (\lift H{\qq_1})}
  \starrel= \pb\paren{ (H\otimes\id) \qq_1\rr_1\quanteq \qq_2\rr_2 }.
\]
Here $(*)$
follows by the fact that $\im (\lift H{\qq_1})$
is the full space (since $H$
is unitary),
and \autoref{lemma:qeq.inside}.

By rule \textsc{QApply2}, we have
$\rhl{I_2}{\Skip}{\bd_2}{I_1}$ with
\[
  I_2 := \adj{\paren{\lift H{\rr_2}}} \cdot \pb\paren{I_1 \cap \im (\lift H{\rr_2})}
  \starrel= \pb\paren{ (H\otimes\id) \qq_1\rr_1\quanteq (\id\otimes H) \qq_2\rr_2}.
\]
Here $(*)$ is simplified analogously as the previous step.

By \ruleref{QInit1}, we have $\rhl{I_3}{\bc_1}\Skip{I_2}$ with
\[
  I_3 := \spaceat{I_2}{\Uvarnames{\qq_1\rr_1}\mathrm{EPR}}\otimes\elltwov{\qq_1\rr_1}.
\]

By rule \textsc{QInit2}, we have $\rhl{I_4}{\Skip}{\bd_1}{I_3}$ with
\[
  I_4 := \spaceat{I_3}{\Uvarnames{\qq_2\rr_2}\mathrm{EPR}}\otimes\elltwov{\qq_2\rr_2}.
\]

Using three applications of \ruleref{Seq},\footnote{%
  Here we implicitly use the (trivial to prove) fact that that
  $\seq\Skip\bc=\bc=\seq\bc\Skip$
  (w.r.t.~to the denotation of programs).} those three qRHL judgments
imply $\rhl{I_4}\bc\bd{I_0}$.

In order to show $\rhl{\CL{\true}}\bc\bd{I_0}$
using \rulerefx{Conseq}, we need to show that $\CL{\true}\subseteq I_4$.

We have
\begin{align*}
  & \CL{\true}\subseteq I_4
  \\
  \txtrel{\autoref{lemma:spacediv.leq}}\iff
  &\CL{\true} \cap \SPAN\{\Uvarnames{\qq_2\rr_2}\mathrm{EPR}\}\otimes\elltwov{\qu{V_1}\qu{V_2}\setminus\qq_2\rr_2}
    \subseteq
    I_3
  \\
  \iff
  &
    \SPAN\{\lift{\mathrm{EPR}}{\qq_2\rr_2}\}
    \subseteq
    I_3
  \\
  \txtrel{\autoref{lemma:spacediv.leq}}\iff
  &
    \SPAN\{\lift{\mathrm{EPR}}{\qq_2\rr_2}\}
    \cap \SPAN\{\Uvarnames{\qq_1\rr_1}\mathrm{EPR}\}\otimes\elltwov{\qu{V_1}\qu{V_2}\setminus\qq_1\rr_1}
    \subseteq I_2
  \\
  \iff
  &
    \SPAN\{\lift{\mathrm{EPR}}{\qq_2\rr_2}\}
    \cap \SPAN\{\lift{\mathrm{EPR}}{\qq_1\rr_1}\}
    \subseteq I_2
  \\
    \iff
  &
    \forall x.\ 
    \Uvarnames{\qq_2\rr_2}{\mathrm{EPR}}
    \otimes
    \Uvarnames{\qq_1\rr_1}{\mathrm{EPR}}
    \otimes
    \basis{\qu{V_1}\qu{V_2}\setminus\qq_1\rr_1\qq_2\rr_2}{x}
    \in
    I_2 = \pb\paren{(H\otimes\id) \qq_1\rr_1\quanteq (\id\otimes H) \qq_2\rr_2}
  \\
  \txtrel{\autoref{lemma:quanteq}}\iff
  &
    (H\otimes\id) \mathrm{EPR}
    =
    (\id\otimes H) \mathrm{EPR}.
\end{align*}
The last line can be verified by elementary computation. Hence
$\CL{\true}\subseteq I_4$
is true.  (In fact, we can also prove the statement
$\CL{\true}\subseteq I_4$
directly by explicit computation, but the involved matrices are too
unwieldy for a pen-and-paper proof. However, in our tool such a
computation can be performed automatically.)

From $\CL{\true}\subseteq I_4$
and $\rhl{I_4}\bc\bd{I_0}$,
by \ruleref{Conseq}, we get \eqref{eq:example} as desired.

We have formalized this proof in our tool in the contributed file \texttt{epr.qrhl}.

\subsection{Measurements}

We extend the previous example to include a measurement of the two
quantum registers in the computational basis. The fact that the quantum
registers between the two programs are equal (in the sense of $\quanteq$)
before the measurements will imply that the outcomes are the same. That
is, with $\bc,\bd$ as in the previous section, we define
\begin{equation*}
  \bc'\ :=\ \seq\bc{\underbrace{\Qmeasure\xx{\qq\rr}M_\mathit{class}}_{=:\be}},
  \qquad
  \bd'\ :=\ \seq\bd{\Qmeasure\xx{\qq\rr}M_\mathit{class}}.
\end{equation*}
Here $\xx$
has type $\typev\xx=\bool\times\bool$,
i.e., two bits, and $M_\mathit{class}$
is a measurement in the computational basis on two qubits. (Formally, 
$M_\mathit{class}\pb\paren{(a,b)}:=\pb\proj{\basis{}{(a,b)}}$.)

And we want to show:
\begin{equation}
  \label{eq:example.measure}
  \rhl{\CL{\true}}{\bc'}{\bd'}{\CL{\xx_1=\xx_2}}.
\end{equation}

By \ruleref{JointMeasureSimple} (with $Q_1:=Q_2:=\qq\rr$, $e_1:=e_2:=M_\mathit{class}$), we have
\[
  \rhl
  {I_1}
  \be\be
  {\CL{\xx_1=\xx_2}}
\]
with
\[
  I_1
  :=
  {\CL{M_\mathit{class}=M_\mathit{class}}
    \cap \paren{\qq_1\rr_1\quanteq\qq_2\rr_2}
    \cap
    \bigcap_{a,b\in\bit}
    \pb \paren{\CL{z=z}\cap A_{ab}^1 \cap A_{ab}^2}
    + \orth{(A_{ab}^1)}+ \orth{(A_{ab}^2)}}.
\]
and $A_{ab}^i:=\im\proj{\basis{\qq_i\rr_i}{(a,b)}}$.

Since $M_\mathit{class}=M_\mathit{class}$ and $z=z$ are trivially true, the above simplifies to 
\[
  I_1
  =
   \paren{\qq_1\rr_1\quanteq\qq_2\rr_2}
    \cap
    \bigcap_{a,b\in\bit}
    \paren{ A_{ab}^1 \cap A_{ab}^2}
    + \orth{(A_{ab}^1)}+ \orth{(A_{ab}^2)}.
\]
Furthermore,
$A_{ab}^1\cap A_{ab}^2 + \orth{(A_{ab}^1)}+
\orth{(A_{ab}^2)}=\elltwov{V_1V_2}$. This is shown by noting that
every vector $\psi$
is a linear combination of vectors in $A_{ab}^1\cap A_{ab}^2$,
$\orth{(A_{ab}^1)}$,
$\orth{(A_{ab}^2)}$
and then verifying the equation for each of those three cases. Thus
the definition of $I_1$
simplifies to $I_1=\paren{\qq_1\rr_1\quanteq\qq_2\rr_2}$.
Hence $\rhl{\qq_1\rr_1\quanteq\qq_2\rr_2}\be\be{\CL{\xx_1=\xx_2}}$.
With \eqref{eq:example} and \ruleref{Seq}, we get
$\rhl{\CL{\true}}{\seq\bc\be}{\seq\bd\be}{\CL{\xx_1=\xx_2}}$.
Using the definitions of $\bc',\bd',\be$,
we obtain our goal \eqref{eq:example.measure}.

We have formalized this proof in our tool in the contributed file \texttt{epr-measure.qrhl}.

\subsection{Post-quantum cryptography}
\label{sec:enc.example}

\newcommand\kk{\mathbf k}
\newcommand\cc{\mathbf c}
\newcommand\mm{\mathbf m}
\newcommand\sss{\mathbf s}
\newcommand\bb{\mathbf b}
\newcommand\calU{\mathcal U}
\newcommand\enc{\mathsf{enc}}
\newcommand\same[1]{\underline{#1}}
\newcommand\qsame{\underline Q}

Our second example proof is the ROR-OT-CPA security of a simple
one-time encryption scheme. This example illustrates how post-quantum
cryptographic proofs may end up being very similar to classical pRHL proofs,
except for a bit of additional syntactic clutter.

\paragraph{The setting.} The encryption scheme is defined by
\begin{align*}
  &\mathrm{enc}:K\times M\to M,\qquad
    \mathrm{enc}(k,m) := G(k)\oplus m \\
  &\mathrm{dec}:K\times M\to M,\qquad
    \mathrm{dec}(k,c) := G(k)\oplus c
\end{align*}
where $G:K\to M$
is a pseudorandom generator, $k$
is the key, and $m$
is the message (plaintext). 

The \index{ROR-OT-CPA}ROR-OT-CPA security notion says,
informally: The adversary cannot distinguish between an encryption
of~$m$
and an encryption of a random message, even if the adversary
itself chooses~$m$. More formally:
\begin{definition}[ROR-OT-CPA advantage]\label{def:roradv}
  For a stateful adversary $A_1,A_2$, let \symbolindexmark\AdvROR
  \begin{align*}
    \AdvROR^{A_1A_2} :=
    \Bigl\lvert
    &\Pr\bigl[b=1:
      \sample k\calU,\, m\leftarrow A_1(),\, c\leftarrow\mathrm{enc}(k,m),\, b\leftarrow A_2(c)\bigr]
    \\    -
    &\Pr\bigl[b=1:
      \sample k\calU,\, m\leftarrow A_1(),\, \sample r\calU,\, c\leftarrow\mathrm{enc}(k,r),\, b\leftarrow A_2(c)\bigr]
      \Bigr\rvert      
  \end{align*}
  where $\calU$
  is the uniform distribution (on the domain of the respective
  variable), and the notation $\Pr[e:G]$
  denotes the probability that $e$
  holds after executing the instructions in $G$.

  We call $\AdvROR^{A_1A_2}$ the \emph{ROR-OT-CPA advantage}%
  \index{advantage!ROR-OT-CPA}%
  \index{ROR-OT-CPA!advantage}
  of $A_1,A_2$.
\end{definition}

Analogously, we define pseudorandomness of $G:K\to M$ by defining the \emph{PRG advantage}%
\index{PRG advantage}%
\index{advantage!PRG} of $G$:
\begin{definition}[PRG advantage]\label{def:prg.adv}
  For an adversary $B$, let \symbolindexmark\AdvPRG
  \begin{equation*}
    \AdvPRG^{B} :=
    \Bigl\lvert
    \Pr\bigl[b=1:
    \sample s\calU,\, r \leftarrow G(s),\, b\leftarrow B(r)\bigr]
    -
    \Pr\bigl[b=1:
    \sample r\calU,\, b\leftarrow B(r)\bigr]
    \Bigr\rvert      .
  \end{equation*}
\end{definition}

What we want to show is the following well-known fact: ``If $G$
is pseudorandom, then $\mathrm{enc}$
is ROR-OT-CPA.''  In other words, if $\AdvROR^A$
is big for some efficient adversary $A$,
then $\AdvPRG^B$
is big for some related efficient adversary $B$.
In the present case, we can even show the stronger
result $\AdvROR^A=\AdvPRG^B$ for some $B$.
This is shown by proving that for a suitable $B$,
the lhs of $\AdvROR^A$
equals the lhs of $\AdvPRG^B$,
and the rhs of $\AdvROR^A$
equals the rhs of $\AdvPRG^B$.
Since our formalism does not include procedures with parameters and
return values, we first need to rewrite the games (program fragments)
in the definitions of $\AdvROR^A$
and $\AdvPRG^B$.
This is done by defining $A_1$
to be an arbitrary adversary that can access the output variable
$\mm$,
and $A_2$ to be one
that can access the input/output variables $\cc,\bb$.
In addition, they can keep an internal state, thus $A_1$
and $A_2$
have access to some additional classical variables $C$
and (very important for the post-quantum setting) quantum variables
$Q$.
That is, $A_1$
has free variables $\mm,C,Q$,
and $A_2$
has free variables $\cc,\bb,C,Q$.
We can then write the lhs of $\AdvROR$
and $\AdvPRG$ as the following programs:
\begin{align*}
  G_1 &\quad:=\quad \sample \kk \calU;\ A_1;\ \assign\cc{\enc(\kk,\mm)};\ A_2 \\
  G_2 &\quad:=\quad \sample \sss \calU;\ \assign\rr{G(\sss)};\ B
\end{align*}
where $\kk,\rr,\cc,\sss,\mm,\bb$ are classical variables.
(We omit the proof of the equality of the right hand sides from this example.)
Furthermore, we need to chose an efficient adversary $B$:
\begin{equation*}
  B \quad:=\quad A_1;\assign\cc{\rr\oplus \mm}; A_2
\end{equation*}
We now need to show
$ \prafter {\bb=1}{G_1}\rho = \prafter {\bb=1}{G_2}\rho$
for arbitrary initial state $\rho$.
(See \autoref{def:prafter} for the notation
$\prafter {\mathord\cdot}{\mathord\cdot}\rho$.)
For this, we first show $\rhl{I}{G_1}{G_2}{\CL{\bb_1=\bb_2}}$
for a suitable precondition $A$.
To make the following calculation more compact, let
$\same{\xx\yy\dots}$
mean $\xx_1=\xx_2\land\yy_1=\yy_2\land\dots$
for classical $\xx,\yy,\dots$,
and let $\qsame$
mean $Q_1\quanteq Q_2$.
We perform a backward reasoning, starting with the desired
post-condition, and working out way backward through the statements in
$G_1,G_2$
(sometimes individually, sometimes pairwise). Note that the
precondition of each line matches the postcondition of the next one.
\begin{align*}
  & \rhl
    {\CL{\same{\bb\cc C}}\cap \qsame}
    {A_2}
    {A_2}
    {\CL{\same\bb}}
  && \rulerefx{Equal},\rulerefx{Conseq}
  \\
  & \rhl
    {\CL{\enc(\kk_1,\mm_1)=\cc_2\land\same{\bb C}}\cap \qsame}
    {\assign\cc{\enc(\kk,\mm)}}
    {\Skip}
    {\CL{\same{\bb\cc C}}\cap \qsame}
  && \rulerefx{Assign1}
  \\
  & \rhl
    {\CL{\enc(\kk_1,\!\mm_1)\!=\!\rr_2\!\oplus\!\mm_2\land\same{\bb C}}\cap \qsame}
    \Skip
    {\assign{\cc\!}{\!\rr\!\oplus\! \mm}}
    {\CL{\enc(\kk_1,\!\mm_1)\!=\!\cc_2\land\same{\bb C}}\cap \qsame}
    \hskip-1cm
  && \hskip1cm \textsc{Assign2}
  \\
  &
    \braces{\CL{G(\kk_1)=\rr_2\land\same{\mm\bb C}}\cap \qsame}
    \ \subseteq \ 
    \braces{\CL{\enc(\kk_1,\mm_1)=\rr_2\!\oplus\!\mm_2\land\same{\bb C}}\cap \qsame}
  && \text{(def.~of $\enc$)}
  \\
  & \rhl
    {\CL{G(\kk_1)=\rr_2\land\same{\mm\bb C}}\cap \qsame}
    {A_1}
    {A_1}
    {\CL{G(\kk_1)=\rr_2\land\same{\mm\bb C}}\cap \qsame}
  && \rulerefx{Equal},\rulerefx{Frame}
  \\
  & \rhl
    {\CL{G(\kk_1)=G(\sss_2)\land\same{\mm\bb C}}\cap \qsame}
    \Skip
    {\assign\rr{G(\sss)}}
    {\CL{G(\kk_1)=\rr_2\land\same{\mm\bb C}}\cap \qsame}
  && \textsc{Assign2}
  \\
  & \rhl
    {A}
    {\sample \kk \calU}
    {\sample \sss \calU}
    {\CL{G(\kk_1)=G(\sss_2)\land\same{\mm\bb C}}\cap \qsame}
  && \rulerefx{JointSample}
\end{align*}
In the last step (\rulerefx{JointSample}), we define $\mu$
to be the uniform distribution on pairs $(z,z)$. Using $f:=\mu$ in \ruleref{JointSample}, we get the precondition
\begin{align*}
  A &:=
      \CL{\marginal1\mu=\calU \land \marginal2\mu=\calU}
      \cap
      \pB\bigcap_{(z_1,z_2)\in\suppd \mu}
      \paren{
      \CL{G(z_1)=G(z_2)\land\same{\mm\bb C}}\cap \qsame
      } \\
    &=
      \CL{
      \forall(z_1,z_2)\in\suppd \mu.\
      G(z_1)=G(z_2)\land\same{\mm\bb C}}\cap \qsame
      =
      \CL{\same{\mm\bb C}}\cap \qsame
      \supseteq
      \CL{\same{\kk\rr\cc\sss\mm\bb C}}\cap \qsame
\end{align*}
Here the first expression after $:=$
is what \ruleref{JointSample} gives us, the next equality holds since
both marginals of of $\mu$
are the uniform distribution, and by \autoref{lemma:cl.simps}. The
last equality is elementary first-order logic (using the fact that
$\suppd\mu$ contains only pairs $(z,z)$).

Combining all those judgments using rules \rulerefx{Seq} and~\rulerefx{Conseq}, we get
\begin{equation*}
  \rhl
  {\CL{\same{\kk\rr\cc\sss\mm\bb C}}\cap \qsame}
  {G_1}{G_2}
  {\CL{\same\bb}}.
\end{equation*}
By \ruleref{QrhlElimEq}, this implies
\begin{equation*}
  \prafter {\bb=1}{G_1}\rho
  =
  \prafter {\bb=1}{G_2}\rho
\end{equation*}
as desired.

Notice how each predicate in the above calculation is of the form
$\CL{e}\cap(Q_1\quanteq Q_2)$
for some classical expression $e$,
and that the $(Q_1\quanteq Q_2)$-part
of the predicates is not touched by the rules.\footnote{\Ruleref{JointSample}
    gave us a more complex precondition but we can easily see that the precondition
    of \rulerefx{JointSample} can always be simplified to something of
    the form $\CL{\dots}\cap \qsame$ with steps analogous to what we
    did here, provided the postcondition is of the form
    $\CL{\dots}\cap \qsame$. In fact, we could formulate a
    special case of \rulerefx{JointSample} that directly
    has pre- and postcondition of the form $\CL{\dots}\cap \qsame$.}
This means that, as
long as we treat the quantum adversary as a black box, and no other
quantum operations are performed in our games, the reasoning in qRHL
will be almost identical to that in pRHL. This is good news because it
means that we need to deal with the extra complexity of quantum
mechanics only in those cases where we would have to deal with them in
a pen-and-paper proof as well.

We have formalized this proof in our tool in the contributed file \texttt{prg-enc-rorcpa.qrhl}.

\paragraph{Acknowledgments.}
We thank Gilles Barthe, Tore Vincent Carstens, François Dupressoir,
Benjamin Gregoire, Yangjia Li,
Pierre-Yves Strub for valuable discussions.  This
work was supported by institutional research funding IUT2-1 of the
Estonian Ministry of Education and Research, by the Estonian Centre of
Exellence in IT (EXCITE) funded by ERDF, and by the Air Force Office
of Scientific Research through the project ``Verification of quantum
cryptography'' (AOARD Grant FA2386-17-1-4022).  We also used Sage
\cite{sage} for calculations and experiments, and the Sage Cluster
funded by National Science Foundation Grant No.~DMS-0821725.

\appendix

\section{Proofs of rules}
\label{sec:rule-proofs}

\subsection{General rules}
\label{sec:proofs-general}

\begin{lemma}[Symmetry]\label{rule-lemma:Sym}
  \Ruleref{Sym} holds.
\end{lemma}

\begin{proof} We first show an auxiliary fact:
  \begin{claim}\label{claim:rename.ass}
    If $\rho\in\traceposcq{V_1V_2}$
    satisfies a predicate $A$,
    and $\sigma:V_1V_2\to V_1V_2$
    is a variable renaming, then $\Erename\sigma(\rho)$
    satisfies $A^*:=\Urename{\restrict\sigma{\qu{V_1}\qu{V_2}}}\ \subst (A\sigma)$.
  \end{claim}

  \begin{claimproof}
    Since $\rho$
    is a cq-operator, we can write $\rho$
    as
    $\rho=\sum_{m\in\types{\cl{V_1}\cl{V_2}}}
    \pointstate{\cl{V_1}\cl{V_2}}{m} \otimes \rho_m$.  Since $\rho$
    satisfies $A$,
    it follows that $\suppo\rho_m\subseteq \denotee{A}{m}$
    for all $m$.
    Let $\rho'_m:=\Erename{\restrict\sigma{\qu{V_1}\qu{V_2}}}(\rho_m)$.
    Then
    \begin{align*}
      \suppo\rho'_m
%      &= \suppo\Erename{\restrict\sigma{\qu{V_1}\qu{V_2}}}(\rho_m)\\
         &=      \suppo\Urename{\restrict\sigma{\qu{V_1}\qu{V_2}}}\rho_m
      \adj{\Urename{\restrict\sigma{\qu{V_1}\qu{V_2}}}}\\
      &=\Urename{\restrict\sigma{\qu{V_1}\qu{V_2}}}\suppo\rho_m
        \subseteq \Urename{\restrict\sigma{\qu{V_1}\qu{V_2}}}\denotee Am \\
      &=  \Urename{\restrict\sigma{\qu{V_1}\qu{V_2}}}\denotee{\subst A\sigma}{m\circ
        (\restrict{\sigma}{\cl{V_1}\cl{V_2}})^{-1}} \\
      &=
        \denotee{A^*}{m\circ(\restrict{\sigma}{\cl{V_1}\cl{V_2}})^{-1}}
        .
    \end{align*}
    Furthermore,
    \begin{align*}
      \Erename\sigma(\rho)
      &=
        \sum_m \Erename{\restrict\sigma{\cl{V_1}\cl{V_2}}}\pb\paren{\pointstate{{\cl{V_1}\cl{V_2}}}m}
        \tensor
        \Erename{\restrict\sigma{\qu{V_1}\qu{V_2}}}(\rho_m)\\
      &=
        \sum_m \pB\proj{\Urename{\restrict\sigma{\cl{V_1}\cl{V_2}}}\basis{{\cl{V_1}\cl{V_2}}}m}
        \tensor
        \rho'_m\\
      &=
        \sum_m \pB\proj{\pb\basis{{\cl{V_1}\cl{V_2}}}
        {m\circ(\restrict{\sigma}{\cl{V_1}\cl{V_2}})^{-1}}}
        \tensor
        \rho'_m.
    \end{align*}
    Since
    $\suppo\rho'_m\in
    \denotee{A^*}{m\circ(\restrict{\sigma}{\cl{V_1}\cl{V_2}})^{-1}}$,
    this implies that
    $\rho''_m := \pb\proj{\basis{{\cl{V_1}\cl{V_2}}}
      {m\circ(\restrict{\sigma}{\cl{V_1}\cl{V_2}})^{-1}}} \tensor
    \rho'_m$ satisfies $A^*$,
    and hence $\Erename\sigma(\rho)=\sum_m\rho''_m$ satisfies $A^*$.
  \end{claimproof}

  Let
  $A^*:= \Urename{\restrict\sigma{\qu{V_1}\qu{V_2}}}\ \subst (A\sigma)$
  and
  $B^*:= \Urename{\restrict\sigma{\qu{V_1}\qu{V_2}}}\ \subst (B\sigma)$.
  
  Fix a separable $\rho\in\traceposcq{V_1V_2}$
  that satisfies $A$.
  Then $\hat\rho:=\Erename\sigma(\rho)$
  satisfies
  $A^*$
  by \autoref{claim:rename.ass}. Since $\rho$
  is separable, so is $\hat\rho$.
  We have $\rhl{A^*}\bd\bc{B^*}$
  from the premises of \ruleref{Sym}.  Thus there
  exists a separable $\hat\rho'$ that satisfies $B^*$ and such that
  \begin{equation}
    \denotc{\idx1\bd}\pb\paren{\partr{V_1}{V_2}\hat\rho} =
    \partr{V_1}{V_2}\hat\rho'
    \qquad\text{and}\qquad
    \denotc{\idx2\bc}\pb\paren{\partr{V_2}{V_1}\hat\rho} =
    \partr{V_2}{V_1}\hat\rho'.
    \label{eq:sym.c.d}
  \end{equation}
  Let $\rho':=\Erename\sigma(\hat\rho')$. Then
  \begin{align*}
    \denotc{\idx1\bc}\pb\paren{\partr{V_1}{V_2}\rho}
    &=
    \Erename{\restrict\sigma{V_2}}\circ\denotc{\idx2\bc}\circ\Erename{\restrict\sigma{V_1}}\pb\paren{\partr{V_1}{V_2}\rho}\\
    &=
    \Erename{\restrict\sigma{V_2}}\circ\denotc{\idx2\bc}\pb\paren{\partr{V_2}{V_1}\Erename{\sigma}(\rho)}\\
    &=
      \Erename{\restrict\sigma{V_2}}\circ\denotc{\idx2\bc}\pb\paren{\partr{V_2}{V_1}\hat\rho}\\
    &
      \eqrefrel{eq:sym.c.d}=
      \Erename{\restrict\sigma{V_2}}\pb\paren{\partr{V_2}{V_1}\hat\rho'}
      =
      \partr{V_1}{V_2}\,\Erename{\sigma}(\hat\rho')
      =
      \partr{V_1}{V_2}\rho'.
  \end{align*}
  Analogously, we get
  \begin{equation*}
    \denotc{\idx2\bd}\pb\paren{\partr{V_2}{V_1}\rho}
    =
    \partr{V_2}{V_1}\rho'.
  \end{equation*}    
  Furthermore, since $\hat\rho'$
  is separable, $\rho'$
  is separable. And since $\hat\rho'$
  satisfies $B^*$,
  by \autoref{claim:rename.ass}, $\rho'$
  satisfies
  $ \Urename{\restrict\sigma{\qu{V_1}\qu{V_2}}}\,(\subst{B^*}\sigma) =
  \Urename{\restrict\sigma{\qu{V_1}\qu{V_2}}}
  \Urename{\restrict\sigma{\qu{V_1}\qu{V_2}}}\
  \pb\paren{\subst{B}{(\sigma\circ\sigma)}} = B.  $
  (We use that $\sigma\circ\sigma=\id$.)

  Since this holds for any separable $\rho$ that satisfies $A$, we have $\rhl A\bc\bd B$.
\end{proof}

\begin{lemma}[Frame rule]\label{rule-lemma:Frame}
  \Ruleref{Frame} holds.
\end{lemma}

\begin{proof}\stepcounter{claimstep}
  Without loss of generality, we assume that $\qu{X_1}\qu{Y_1}=\qu{V}$
  and $\qu{X_2}\qu{Y_2}=\qu{V}$.
  (Otherwise, we replace $Y_i$ by $Y_i\cup\qu{V}\setminus\qu{X_i}$ without making any 
  of the premises false.
  Note that $\qu{X_1},\qu{Y_1}$
  as well as $\qu{X_2},\qu{Y_2}$
  are disjoint by assumption.)
  Thus $\qu{{X_1'}}\qu{{Y_1'}}=\qu{V_1}$ and 
  $\qu{{X_2'}}\qu{{Y_2'}}=\qu{V_2}$.
  
  Fix $m_1\in\types{\cl{V_1}}$
  and $m_2\in\types{\cl{V_2}}$
  and normalized $\psi_1\in\elltwov{\qu{V_1}}$
  and $\psi_2\in\elltwov{\qu{V_2}}$
  with $\psi_1\tensor \psi_2\in\denotee{A\cap R}{\memuni{m_1m_2}}$.

  Let
  $\rho:=\pb\proj{\basis{\cl{V_1}\cl{V_2}}{\memuni{m_1m_2}}\tensor \psi_1\tensor \psi_2}$.
  Since $\psi_1\tensor\psi_2\in\denotee{A\cap R}{\memuni{m_1m_2}}\subseteq\denotee A{\memuni{m_1m_2}}$,
  we have that $\rho$ satisfies $A$.

  Since $\rhl A\bc\bd B$, this implies that there exists a separable
  $\rho'\in\traceposcq{V_1V_2}$ such that
  \begin{gather}
    \rho'\text{ satisfies }B, \label{eq:rho'.sat.B} \\
    \partr{V_1}{V_2}\rho'=\denotc{\idx1\bc}\pb\paren{\pointstate{\cl{V_1}}{m_1}\tensor\proj{\psi_1}}, \label{eq:rho'.a1.x1} \\
    \partr{V_2}{V_1}\rho'=\denotc{\idx2\bd}\pb\paren{\pointstate{\cl{V_2}}{m_2}\tensor\proj{\psi_2}}. \label{eq:rho'.a2.x2}
  \end{gather}
  For applying \autoref{lemma:pure}, it remains to show that $\rho'$
  satisfies $B\cap R$.

  By \autoref{lemma:schmidt}, we can decompose
  $\psi_1=\sum_i \lambda_{1i} \psi_{1i}^X\tensor \psi_{1i}^Y$
  for some $\lambda_{1i}>0$,
  and orthonormal $\psi_{1i}^X\in\elltwov{\qu{{X_1'}}}$,
  $\psi_{1i}^Y\in\elltwov{\qu{{Y_1'}}}$.
  And 
  $\psi_2=\sum_i \lambda_{2i} \psi_{2i}^X\tensor \psi_{2i}^Y$
  for some $\lambda_{2i}>0$,
  and orthonormal $\psi_{2i}^X\in\elltwov{\qu{{X'_2}}}$,
  $\psi_{2i}^Y\in\elltwov{\qu{{Y_2'}}}$.

  Let $R_1:=\sum_i \SPANO{\psi_{1i}^Y}\subseteq\elltwov{\qu{{Y_1'}}}$
  and $R_2:=\sum_i \SPANO{\psi_{2i}^Y}\subseteq\elltwov{\qu{{Y_2'}}}$.
  Both $R_1$ and $R_2$ are subspaces.

  Let $R'\subseteq\elltwov{\qu{{Y_1'}}\qu{{Y_2'}}}$
  be the subspace such that
  $\denotee R{\memuni{m_1m_2}}=R'\tensor \elltwov{\qu{{X_1'}}\qu{{X_2'}}}$
  ($R'$ exists since $R$ is is $Y_1'Y_2'$-local and $\qu{{X_1'}}\qu{{X_2'}}\cap \qu{{Y_1'}}\qu{{Y_2'}}=\varnothing$).

  \begin{claim}\label{claim:R1R2.R'}
    $R_1\tensor R_2\subseteq R'$.  
  \end{claim}

  \begin{claimproof}
  Since $\psi_1\tensor\psi_2\in\denotee{A\cap R}{\memuni{m_1m_2}}\subseteq \denotee R{\memuni{m_1m_2}}=R'\tensor \elltwov{\qu{{X_1'}}\qu{{X_2'}}}$, we have 
  $\suppo\proj{\psi_1\otimes\psi_2}\subseteq R'\otimes\elltwov{\qu{{X_1'}}\qu{{X_2'}}}$
  and thus for all $i,j$:
  \begin{align*}
    \psi_{1i}^Y\tensor\psi_{2j}^Y 
    &
      \in
      \suppo\sum_{i,j}\lambda_{1i}\lambda_{2j}\proj{\psi_{1i}^Y\otimes\psi_{2j}^Y}
      =
      \suppo\partr{\qu{{Y_1'}}\qu{{Y_2'}}}{\qu{{X_1'}}\qu{{X_2'}}}
      \proj{\psi_1\otimes\psi_2}
      \subseteq R'.
  \end{align*}

  Since $R_1\tensor R_2$ is the span of 
  $\{\psi_{1i}^Y\tensor \psi_{2j}^Y\}_{ij}$, it follows that $R_1\tensor R_2\subseteq R'$.  
\end{claimproof}

Since $\rho'\in\traceposcq{V_1V_2}$, we can decompose it as $\rho'=\sum_{\bar m_1\bar m_2}\pb\proj{\basis{\cl{V_1}\cl{V_2}}{\bar m_1\bar m_2}}\tensor\rho'_{\bar m_1\bar m_2}$ with $\rho'_{\bar m_1\bar m_2}\in\traceposcq{\qu{V_1}\qu{V_2}}$.

\begin{claim}\label{claim:rhobarm.R1R2}
  For all $\bar m_1,\bar m_2$,
  $\suppo\rho'_{\bar m_1\bar m_2}\subseteq
  R_1\tensor R_2
  \tensor\elltwov{\qu{{X_1'}}\qu{{X_2'}}}$.
\end{claim}

\begin{claimproof}    
  Since $\idx1\bc$ is $X_1'$-local, and $\qu{{X_1'}}\cap\qu{{Y_1'}}=\varnothing$, we can
  write $\denotc{\idx1\bc}=\calE_{\bc}\tensor\idv{\qu{{Y_1'}}}$ for some cq-superoperator
  $\calE_{\bc}$ on $V_1\setminus\qu{{Y_1'}}$. Then
  \begin{align*}
    \partr{\qu{{Y_1'}}}{V_1V_2\setminus\qu{{Y_1'}}}\,\rho'
    &=
      \partr{\qu{{Y_1'}}}{V_1\setminus\qu{{Y_1'}}}\, \partr{V_1}{V_2}\rho'
      \eqrefrel{eq:rho'.a1.x1}=
      \partr{\qu{{Y_1'}}}{V_1\setminus\qu{{Y_1'}}}
      \denotc{\idx1\bc}\pb\paren{\pointstate{\cl{V_1}}{m_1}\tensor\proj{\psi_1}}
    \\&
      =
      \partr{\qu{{Y_1'}}}{V_1\setminus\qu{{Y_1'}}}
        (\calE_{\bc}\tensor\idv{\qu{{Y_1'}}})\pb\paren{\pointstate{\cl{V_1}}{m_1}\tensor\proj{\psi_1}}
      =
      \partr{\qu{{Y_1'}}}{V_1\setminus\qu{{Y_1'}}}
        \pb\paren{\pointstate{\cl{V_1}}{m_1}\tensor\proj{\psi_1}}
    \\
    &=
      \partr{\qu{{Y_1'}}}{\qu{V_1}\setminus \qu{{Y_1'}}}\, \proj{\psi_1}
      = \sum_i\lambda_{1i}^2\, \pb\proj{\psi_{1i}^Y} =: \rho''.
  \end{align*}
  Since $\psi_{1i}^R\in R_1$
  for all $i$
  by definition of $R_1$,
  we have that $\suppo\rho''\subseteq R_1$.
  Thus 
  $\suppo\partr{\qu{{Y_1'}}}{V_1V_2\setminus\qu{{Y_1'}}}\rho'\subseteq R_1$ as well.
  Hence
  $\suppo\rho'\subseteq R_1\tensor\elltwov{{V_1}{V_2}\setminus\qu{{Y_1'}}}$.
  Hence
  $\suppo(\sum_{\bar m_1,\bar m_2}\rho'_{\bar m_1,\bar m_2})=\suppo\partr{\qu{V_1}\qu{V_2}}{\cl{V_1}\cl{V_2}}\rho'\subseteq
  R_1\tensor\elltwov{\qu{V_1}\qu{V_2}\setminus\qu{{Y_1'}}}$.
  Hence $\suppo\rho'_{\bar m_1\bar m_2}\subseteq R_1\tensor \elltwov{\qu{V_1}\qu{V_2}\setminus\qu{{Y_1'}}}$.

  Analogously, we show
  $\suppo\rho'_{\bar m_1\bar m_2}\subseteq R_2\tensor
  \elltwov{\qu{V_1}\qu{V_2}\setminus\qu{{Y_2'}}}$.

  Since
  \begin{multline*}
    (R_1\tensor \elltwov{\qu{V_1}\qu{V_2}\setminus\qu{{Y_1'}}}) \cap
    (R_2\tensor \elltwov{\qu{V_2}\qu{V_1}\setminus\qu{{Y_2'}}})
\\
  = R_1\tensor
  R_2\tensor\elltwov{\qu{V_1}\qu{V_2}\setminus\qu{{Y_1'}}\qu{{Y_2'}}}
  = R_1\tensor
  R_2\tensor\elltwov{\qu{{X_1'}}\qu{{X_2'}}}
,
\end{multline*}
 the claim
  follows.
\end{claimproof}

\begin{claim}\label{claim:R'indep}
  For all $\bar m_1,\bar m_2$
  with $\rho_{\bar m_1\bar m_2}\neq 0$,
  we have $\denotee R{\bar m_1\bar m_2}=R'\tensor\elltwov{\qu{{X_1'}}\qu{{X_2'}}}$.
\end{claim}

\begin{claimproof}
  Fix $\bar m_1,\bar m_2$ with $\rho_{\bar m_1\bar m_2}\neq 0$.

  We first show that $\bar m_1=\bar m_2$ on $\cl{{Y_1'}}$.
  Fix $\yy_1\in \cl{{Y_1'}}$.
  Let $z:=m_1(\yy_1)$.
  Then we have that $\pb\pointstate{\cl{V_1}}{m_1}\tensor\proj{\psi_1}$
  satisfies $\CL{\yy_1=z}$.
  If $\yy_1\in \cl{{X_1'}}\cap \cl{{Y_1'}}$,
  the fact that $\idx1\bc$
  is $(X_1'\cap Y_1')$-readonly (from the premises of \ruleref{Frame})
  implies that
  $\partr{V_1}{V_2}\rho'\eqrefrel{eq:rho'.a1.x1}=\denotc{\idx1\bc}\pb\paren{\pointstate{\cl{V_1}}{m_1}\tensor\proj{\psi_1}}$
  satisfies $\CL{\yy_1=z}$.
  If $\yy_1\in \cl{{Y_1'}}\setminus \cl{{X_1'}}$,
  the fact that $\idx1\bc$
  is $\cl{{X_1'}}$-local (from the premises)
  implies that
  $\partr{V_1}{V_2}\rho'\eqrefrel{eq:rho'.a1.x1}=\denotc{\idx1\bc}\pb\paren{\pointstate{\cl{V_1}}{m_1}\tensor\proj{\psi_1}}$
  satisfies $\CL{\yy_1=z}$.
  So $\partr{V_1}{V_2}\rho'$
  satisfies $\CL{\yy_1=z}$ in both cases.
  Hence
  $\rho'$
  satisfies $\CL{\yy_1=z}$.
  Since  $\rho'=\sum_{\bar m_1\bar m_2}\pb\proj{\basis{\cl{V_1}\cl{V_2}}{\bar m_1\bar m_2}}\tensor\rho'_{\bar m_1\bar m_2}$,
  this implies $\suppo\rho_{\bar m_1\bar m_2}\subseteq\denotee{\CL{\yy_1=z}}{\memuni{m_1m_2}}$.
  Since $\rho_{\bar m_1\bar m_2}\neq0$,
  this implies
  $\denotee{\CL{\yy_1=x}}{\memuni{m_1m_2}}\neq 0$, hence
  $(\bar m_1(\yy_1)=z)=\denotee{\yy_1=z}{\bar m_1\bar m_2}=\true$,
  thus $\bar m_1(\yy_1)=z=m_1(\yy_1)$.
  Since this holds for every $\yy_1\in \cl{{Y_1'}}$,
  we have that $\bar m_1=m_1$ on $\cl{{Y_1'}}$.

  Analogously we show that $\bar m_2=m_2$ on $\cl{{Y_2'}}$.

  Since $R$ is $Y_1'Y_2'$-local,
  we have $\fv(R)\subseteq\cl{{Y_1'}}\cl{{Y_2'}}$.
 And since $\bar m_1\bar m_2=\memuni{m_1m_2}$ on $\cl{{Y_1'}}\cl{{Y_2'}}$, we have
  $\denotee R{\bar m_1\bar m_2}=\denotee R{\memuni{m_1m_2}}$.
  And $\denotee R{\memuni{m_1m_2}}=R'\tensor\elltwov{\qu{{{{X_1'}}}}\qu{{X_2'}}}$ by definition of $R'$.
Thus $\denotee R{\bar m_1\bar m_2}=R'\tensor\elltwov{\qu{{{{X_1'}}}}\qu{{X_2'}}}$.
\end{claimproof}

\begin{claim}
  $\rho$ satisfies $B\cap R$.
\end{claim}

\begin{claimproof}
  By \autoref{claim:rhobarm.R1R2},
  $\suppo\rho'_{\bar m_1\bar m_2}\subseteq R_1\tensor R_2
  \tensor\elltwov{\qu{{{{X_1'}}}}\qu{{X_2'}}}$ for all $\bar m_1,\bar m_2$.  By
  \autoref{claim:R1R2.R'}, it follows that
  $\suppo\rho'_{\bar m_1\bar m_2}\subseteq R'
  \tensor\elltwov{\qu{{{{X_1'}}}}\qu{{X_2'}}}$.  If
  $\rho'_{\bar m_1\bar m_2}\neq 0$,
  this implies
  $\suppo\rho'_{\bar m_1\bar m_2}\subseteq \denotee{R}{\bar m_1\bar
    m_2}$ by \autoref{claim:R'indep}.  If
  $\rho'_{\bar m_1\bar m_2}= 0$, we have
  $\suppo\rho'_{\bar m_1\bar m_2}=0\subseteq \denotee{R}{\bar m_1\bar
    m_2}$.
  Thus $\rho'=\sum_{\bar m_1\bar m_2}\pb\proj{\basis{\cl{V_1}\cl{V_2}}{\bar m_1\bar m_2}}\tensor \rho'_{\bar m_1\bar m_2}$ satisfies $R$.

  By \eqref{eq:rho'.sat.B}, $\rho'$ satisfies $B$. Thus $\rho'$ satisfies $B\cap R$.
\end{claimproof}

Summarizing, for all $m_1\in\types{\cl{V_1}}$
and $m_2\in\types{\cl{V_2}}$
and normalized $\psi_1\in\elltwov{\qu{V_1}}$
and $\psi_2\in\elltwov{\qu{V_2}}$
with $\psi_1\tensor \psi_2\in\denotee{A\cap R}{\memuni{m_1m_2}}$,
there is a separable $\rho'\in\traceposcq{V_1V_2}$
that satisfies $B\cap R$,
and such that \eqref{eq:rho'.a1.x1} and \eqref{eq:rho'.a2.x2} hold. By
\autoref{lemma:pure}, this implies $\rhl{A\cap R}\bc\bd{B\cap R}$.
\end{proof}

\begin{lemma}[Weakening]\label{rule-lemma:Conseq}
  \Ruleref{Conseq} holds.
\end{lemma}

\begin{proof}
  Fix a separable $\rho\in\traceposcq{V_1V_2}$
  that satisfies $A$.
  By \autoref{lemma:leq.sat}, $\rho$
  satisfies $A'$.
  Then $\rhl{A'}\bc\bd {B'}$
  implies that there is a separable $\rho'\in\traceclcq{V_1V_2}$
  that satisfies $B'$
  and such that
  $\partr{V_1}{V_2}\rho' =
  \denotc{\idx1\bc_1}\pb\paren{\partr{V_1}{V_2}\rho}$ and
  $\partr{V_2}{V_1}\rho' =
  \denotc{\idx2\bc_2}\pb\paren{\partr{V_2}{V_1}\rho}$. By
  \autoref{lemma:leq.sat}, $\rho'$
  satisfies $B$. Thus $\rhl{A}\bc\bd B$ holds.
\end{proof}

\begin{lemma}[Program composition]\label{rule-lemma:Seq}
  \Ruleref{Seq} holds.
\end{lemma}

\begin{proof}
  Fix a separable $\rho\in\traceposcq{V_1V_2}$
  that satisfies $A$.
  Then by $\rhl{A}{\bc_1}{\bc_2}{B}$,
  there is a separable $\rho'\in\traceposcq{V_1V_2}$
  that satisfies $B$ such that
  \begin{equation}\label{eq:rho.rho'}
    \partr{V_1}{V_2}\rho' = \denotc{\idx1\bc_1}\pb\paren{\partr{V_1}{V_2}\rho}
    \quad\text{and}\quad
    \partr{V_2}{V_1}\rho' = \denotc{\idx2\bc_2}\pb\paren{\partr{V_2}{V_1}\rho}.
  \end{equation}
  Then by $\rhl{B}{\bd_1}{\bd_2}{C}$,
  there is a separable $\rho''\in\traceposcq{V_1V_2}$ that satisfies $C$ such that
  \begin{equation}\label{eq:rho'.rho''}
    \partr{V_1}{V_2}\rho'' = \denotc{\idx1\bd_1}\pb\paren{\partr{V_1}{V_2}\rho'}
    \quad\text{and}\quad
    \partr{V_2}{V_1}\rho'' = \denotc{\idx2\bd_2}\pb\paren{\partr{V_2}{V_1}\rho'}.
  \end{equation}
  By definition, $\denotc{\idx1(\seq{\bc_1}{\bd_1})}=\denotc{\seq{\idx1\bc_1}{\idx1\bd_1}} = \denotc{\idx1\bd_1}\circ\denotc{\idx1\bc_1}$ and thus
  \begin{equation*}
    \partr{V_1}{V_2}\rho''
    \eqrefrel{eq:rho'.rho''}= \denotc{\idx1\bd_1}\pb\paren{\partr{V_1}{V_2}\rho'}
    \eqrefrel{eq:rho.rho'}= \denotc{\idx1\bd_1}\pb\paren{\denotc{\idx1\bc_1}(\partr{V_1}{V_2}\rho)}
    =
    \pb\denotc{\idx1(\seq{\bc_1}{\bd_1})}\pb\paren{\partr{V_1}{V_2}\rho}.
  \end{equation*}
  Similarly,
  \begin{equation*}
    \partr{V_2}{V_1}\rho''
    =
    \pb\denotc{\idx2(\seq{\bc_2}{\bd_2})}\pb\paren{\partr{V_2}{V_1}\rho}.
  \end{equation*}
  Thus $\rhl{A}{\seq{\bc_1}{\bd_1}}{\seq{\bc_2}{\bd_2}}{C}$.
\end{proof}

\begin{lemma}[Case distinction]\label{rule-lemma:Case}
  \Ruleref{Case} holds.
\end{lemma}

\begin{proof}
  Fix $m_1\in\types{\cl{V_1}}$,
  $m_2\in\types{\cl{V_2}}$, and normalized
  $\psi_1\in\elltwov{\qu{V_1}}$,
  $\psi_2\in\elltwov{\qu{V_2}}$
  with $\psi_1\tensor \psi_2\in\denotee{A}{\memuni{m_1m_2}}$.

  Let
  $\rho:=\pointstate{\cl{V_1}\cl{V_2}}{\memuni{m_1m_2}}\tensor \proj{\psi_1\tensor
  \psi_2}\in\traceposcq{V_1V_2}$.
  The operator $\rho$ is separable.
Since $\psi_1\tensor \psi_2\in \denotee{A}{\memuni{m_1m_2}}$,
  we have that $\rho$
  satisfies $ A$.
  Let $z:=\denotee e{\memuni{m_1m_2}}$.
  Then $\psi_1\tensor \psi_2\in\elltwov{\qu{V_1}\qu{V_2}}=\denotee{\CL{e=z}}{\memuni{m_1m_2}}$.
  Thus $\rho$ satisfies $\CL{e=z}$.
  Since $\rho$ satisfies both $\CL{e=z}$ and $A$, $\rho$
  satisfies $\CL{e=z}\cap A$.
  By assumption, we have $\rhl{\CL{e=z}\cap A}\bc\bd B$.
  Thus there exists a separable $\rho'\in\traceposcq{V_1V_2}$
  that satisfies~$B$
  and such that
  $\partr{V_1}{V_2}\rho' =
  \denotc{\idx1\bc}(\partr{V_1}{V_2}\rho)$ and
  $\partr{V_2}{V_1}\rho' =
  \denotc{\idx2\bd}(\partr{V_2}{V_1}\rho)$.  Since
  $\partr{V_1}{V_2}\rho=
  \pointstate{\cl{V_1}}{m_1}\tensor\proj{\psi_1}$ and
  $\partr{V_2}{V_1}\rho=
  \pointstate{\cl{V_2}}{m_2}\tensor\proj{\psi_2}$, it follows that 
  $\partr{V_1}{V_2}\rho'=\denotc{\idx1\bc}\pb\paren{\pointstate{\cl{V_1}}{m_1}\tensor\proj{\psi_1}}$,
  and
  $\partr{V_2}{V_1}\rho'=\denotc{\idx2\bd}\pb\paren{\pointstate{\cl{V_2}}{m_2}\tensor\proj{\psi_2}}$.

  By \autoref{lemma:pure}, this implies $\rhl A\bc\bd B$.
\end{proof}

\begin{lemma}[Equality]\label{rule-lemma:Equal}
  \Ruleref{Equal} holds.
\end{lemma}

\begin{proof}\stepcounter{claimstep}
  Let $A:={\CL{X_1=X_2}\cap (Y_1\quanteq Y_2)}$. We need to show $\rhl A\bc\bc A$.
  Let $Z_1:={V_1}\setminus X_1Y_1$ and
  $Z_2:={V_2}\setminus X_2 Y_2$.

  Fix   $m_1\in\types{\cl{V_1}}$,
  $m_2\in\types{\cl{V_2}}$ and normalized $\psi_1\in\elltwov{\qu{V_1}}$,
  $\psi_2\in\elltwov{\qu{V_2}}$ such that $\psi_1\tensor \psi_2\in\denotee A{\memuni{m_1m_2}}$.

  Let $\rho_1':=\denotc{\idx1\bc}\pb\paren{\pointstate{\cl{V_1}}{m_1}\tensor\proj{\psi_1}}$
  and $\rho_2':=\denotc{\idx2\bc}\pb\paren{\pointstate{\cl{V_2}}{m_2}\tensor\proj{\psi_2}}$.

  Our goal is to show that there exists a separable
  $\rho'\in\traceclcq{V_1V_2}$
  such that: $\rho'$
  satisfies $A$,
  and $\partr{V_1}{V_2}\rho'=\rho'_1$,
  and $\partr{V_1}{V_2}\rho'=\rho'_2$.
  This will imply $\rhl A\bc\bc A$ by \autoref{lemma:pure}.

  \begin{claim}\label{claim:m1m2.eq}
    $m_1(\xx_1)=m_2(\xx_2)$ for all $\xx\in X$.
  \end{claim}
  
  \begin{claimproof}
    If
    $m_1(\xx_1)\neq m_2(\xx_2)$
    for some $\xx\in X$, we have $\denotee{X_1}{\memuni{m_1m_2}}\neq\denotee{X_2}{\memuni{m_1m_2}}$.
    Thus $\denotee{X_1=X_2}{\memuni{m_1m_2}}=\false$, and thus $\denotee{\CL{X_1=X_2}}{\memuni{m_1m_2}}=0$.
    Since $A\subseteq \CL{X_1=X_2}$
    this implies $\denotee{A}{\memuni{m_1m_2}}=0$
    in contradiction to $\psi_1\tensor\psi_2\in\denotee{A}{\memuni{m_1m_2}}$.
  \end{claimproof}

  Let $\sigma:V_1\to V_2$
  denote the variable renaming with $\sigma(\xx_1)=\xx_2$
  for all $\xx\in V$.

  \begin{claim}\label{claim:rho'12.decomp}
    We can write $\rho_1'$
    and $\rho_2'$
    as $\rho_1'=\rho'_{1,XY}\tensor\rho'_{1,Z}$
    and $\rho_2'=\rho'_{2,XY}\tensor\rho'_{2,Z}$
    for some $\rho'_{1,XY}\in\traceposcq{X_1Y_1}$,
    $\rho'_{2,XY}\in\traceposcq{X_2Y_2}$,
    $\rho'_{1,Z}\in\traceposcq{Z_1}$,
    and $\rho'_{2,Z}\in\traceposcq{Z_2}$
    such that $\tr\rho'_{1,Z}=\tr\rho'_{2,Z}=1$
    and $\Erename{\restrict\sigma{X_1Y_1}}(\rho'_{1,XY}) =\rho'_{2,XY}$.
  \end{claim}

  \begin{claimproof}
    Since $\idx1\bc$
    is $X_1Y_1$-local,
    we can write $\denotc{\idx1\bc}$
    as $\denotc{\idx1\bc}=\calE_{\bc}\tensor\idv{Z_1}$
    for some cq-superoperator $\calE_{\bc}$
    on $X_1Y_1$.
    Furthermore,
    $\denotc{\idx2\bc}=\Erename{\sigma}\circ\denotc{\idx1\bc}\circ\Erename{\sigma^{-1}}$.
    Hence $\denotc{\idx2\bc}=\calE_{\bc}'\otimes\idv{Z_2}$ for
    $\calE_{\bc}':=\Erename{\restrict\sigma{X_1Y_1}}\circ\calE_{\bc}\circ\Erename{(\restrict\sigma{X_1Y_1})^{-1}}$.
    % And hence
    % $\calE_{\bc}=\Erename{\restrict\sigma{X_2Y_2}}\circ\calE_{\bc}'\circ\Erename{\restrict\sigma{X_1Y_1}}$.

    Note that $\restrict\sigma{Y_1}$
    is the variable renaming with $\sigma(\qq_i)=\qq_i'$
    when $\qq_i$
    and $\qq_i'$
    are the $i$-th
    variable in $Y_1$
    and $Y_2$,
    respectively. And $\psi_1\otimes\psi_2\in\denotee A{\memuni{m_1m_2}}\subseteq(Y_1\quanteq Y_2)$. Thus by \autoref{coro:quanteq}, we can write
    $\psi_1,\psi_2$
    as $\psi_1=\psi^Y_1\tensor\psi^Z_1$
    and $\psi_2=\psi^Y_2\tensor\psi^Z_2$
    for some normalized $\psi_1^Y\in\elltwov{Y_1}$,
    $\psi_1^Z\in\elltwov{\qu{Z_1}}$,
    $\psi_2^X\in\elltwov{Y_2}$,
    $\psi_2^Y\in\elltwov{\qu{Z_2}}$
    with $\psi^Y_1=\Urename{(\restrict\sigma{Y_1})^{-1}}\psi^Y_2$.
    (Note that in \autoref{coro:quanteq}, the variable renaming $\sigma:V_2\to V_1$ is the inverse 
    of the variable renaming $\sigma$ defined here.)

    We then have
    \begin{align*}
      \rho_1'
      &=
        \denotc{\idx1\bc}\pb\paren{\pointstate{\cl{V_1}}{m_1}\tensor\proj{\psi_1}}
        =
        \underbrace{\calE_{\bc}\plr\paren{\pb\proj{\pb\basis{X_1}{\restrict{m_1}{X_1}}\tensor\psi_1^Y}}}_{{}=:\rho'_{1,XY}}
        \tensor\underbrace{\pb\proj{\pb\basis{\cl{Z_1}}{\restrict{m_1}{\cl{Z_1}}}\otimes\psi_1^Z}}_{{}=:\rho'_{1,Z}}
   \end{align*}
    and
    \begin{align*}
      \rho_2'
      &=
        \denotc{\idx2\bc}\pb\paren{\pointstate{\cl{V_2}}{m_2}\tensor\proj{\psi_2}}
        =
        \calE_{\bc}'\plr\paren{\pb\proj{\pb\basis{X_2}{\restrict{m_2}{X_2}}\tensor\psi_2^Y}}\tensor\pb\proj{\pb\basis{\cl{Z_2}}{\restrict{m_2}{\cl{Z_2}}}\tensor\psi_2^Z}
      \\&=
          \Erename{\restrict\sigma{X_1Y_1}}\circ \calE_{\bc}\circ\Erename{(\restrict\sigma{X_1Y_1})^{-1}}\pB\paren{\pb\proj{\pb\basis{X_2}{\restrict{m_2}{X_2}}\tensor\psi_2^Y}}
          \tensor\pb\proj{\pb\basis{\cl{Z_2}}{\restrict{m_2}{\cl{Z_2}}}\tensor\psi_2^Z}
      \\&=
          \Erename{\restrict\sigma{X_1Y_1}}\circ\calE_{\bc}\pB\paren{\pb\proj{\pb\basis{X_1}{(\restrict{m_2}{X_2}\circ\restrict\sigma{X_1})}}
          \tensor\pb\proj{\Urename{(\restrict\sigma{Y_1})^{-1}}\psi_2^Y}}
          \tensor\pb\proj{\pb\basis{\cl{Z_2}}{\restrict{m_2}{\cl{Z_2}}}\tensor\psi_2^Z}
      \\&\starrel=
          \Erename{\restrict\sigma{X_1Y_1}}\circ\calE_{\bc}\pB\paren{\pb\proj{\pb\basis{X_1}{\restrict{m_1}{X_1}}\tensor\pb\psi_1^Y}}
          \tensor\pb\proj{\pb\basis{\cl{Z_2}}{\restrict{m_2}{\cl{Z_2}}}\tensor\psi_2^Z}
      \\&=
          \underbrace{\Erename{\restrict\sigma{X_1Y_1}}\pb\paren{\rho'_{1,XY}}}_{{}=:\rho'_{2,XY}}\tensor
          \underbrace{\pb\proj{\pb\basis{\cl{Z_2}}{\restrict{m_2}{\cl{Z_2}}}\tensor\psi_2^Z}}_{{}=:\rho'_{2,Z}}.
    \end{align*}
    Here $(*)$ uses that
    $\restrict{m_2}{X_2}\circ\restrict\sigma{X_1}=\restrict{m_1}{X_1}$
    by \autoref{claim:m1m2.eq} and
    $\Urename{(\restrict\sigma{Y_2})^{-1}}\psi_2^Y=\psi_1^Y$ by definition of~$\psi_1^Y$ and~$\psi_2^Y$.

    And $\rho'_{1,Z},\rho'_{2,Z}$ have trace $1$ since $\psi_1^Z$ and $\psi_2^Z$ are normalized.
  \end{claimproof}

  Since $\rho'_{1,XY}$
  is a cq-density operator, we can decompose $\rho'_{1,XY}$
  as
  $\rho'_{1,XY}=\sum_{m_1^X,i} \lambda_{m_1^X,i} \pb\pointstate{X_1}{m_1^X}
  \tensor \proj{\phi_{m_1^Xi}}$ for some normalized
  $\phi_{m_1^Xi}\in\elltwov{Y_1}$ and some $\lambda_{m_1^X,i}>0$. Here $m_1^X$ ranges over $\types{\qu{X_1}}$. Note that $\sum_{m_1^X,i}\lambda_{m_1^Xi}=\tr\rho'_{1,XY}<\infty$.

  Let
  \begin{equation}\label{eq:rho'.def}
    \rho':= \sum_{m_1^X,i} \lambda_{m_1^X,i}\,
    \pB\proj{\pb\basis{X_1X_2}{(m_1^X\ m_1^X\circ(\restrict\sigma{X_2}))}}
    \tensor \proj{\phi_{m_1^Xi}}
    \tensor \pb\proj{\Urename{\restrict\sigma{Y_1}}\phi_{m_1^Xi}}
    \tensor \rho'_{1,Z}
    \tensor \rho'_{2,Z}.
  \end{equation}

  \begin{claim}\label{claim:rho'.sep}
    $\rho'\in\traceposcq{V_1V_2}$
    exists (i.e., the sum defining $\rho'$
    converges) and $\rho'$
    is separable. 
  \end{claim}

  \begin{claimproof}
    Since $\lambda_{m_1^X,i}>0$
    and $\sum_{m_1^X,i}\lambda_{m_1^Xi}=\tr\rho'_{1,XY}<\infty$,
and the tensor product in \eqref{eq:rho'.def} has trace $1$,
    we have $\rho'\in\traceposcq{V_1V_2}$.
    And by construction, $\rho'$ is separable.
  \end{claimproof}
  
  \begin{claim}\label{claim:rho'.proj}
    $\partr{V_1}{V_2}\rho'=\rho'_1$
    and $\partr{V_2}{V_1}\rho'=\rho'_2$.
  \end{claim}

  \begin{claimproof}
    We compute $\partr{V_1}{V_2}\rho'$:
    \begin{align*}
      \partr{V_1}{V_2}\rho'
      &=
        \sum_{m_1^X,i} \lambda_{m_1^X,i}\,
        \pointstate{X_1}{m_1^X}
        \tensor \proj{\phi_{m_1^Xi}}
        \tensor \rho'_{1,Z}
        \cdot
        \underbrace{\tr\proj{\Urename{\restrict\sigma{Y_1}}\phi_{m_1^Xi}}}_{{}=\norm{\phi_{m_1^Xi}}^2=1}
        \cdot
        \underbrace{\tr \rho'_{2,Z}}_{{}=1}
      \\[-7pt]& =
          \rho'_{1,XY}\tensor \rho'_{1,Z}
          =
          \rho'_1.
    \end{align*}
    We now compute $\partr{V_2}{V_1}\rho'$:
    \begin{align*}
      \partr{V_2}{V_1}\rho'
      &=
        \sum_{m_1^X,i} \lambda_{m_1^X,i}\,
        \pB\proj{\pb\basis{X_2}{(m_1^X\circ\restrict\sigma{X_2})}}
        \tensor \pb\proj{\Urename{\restrict\sigma{Y_1}}\phi_{m_1^Xi}}
        \tensor \rho'_{2,Z}
        \cdot
        \underbrace{\tr \pb\proj{\phi_{m_1^Xi}}}_{{}=\norm{ \phi_{m_1^Xi}}^2=1}
        \cdot
        \underbrace{\tr  \rho'_{1,Z}}_{{}=1} \\
      &=
        \sum_{m_1^X,i} \lambda_{m_1^X,i}\,
        \Erename{\restrict\sigma{X_1Y_1}}\pB\paren{
        \pb\proj{\basis{X_1}{m_1^X}}
        \tensor \pb\proj{\phi_{m_1^Xi}}}
        \tensor \rho'_{2,Z}
      \\&=
          \Erename{\restrict\sigma{X_1Y_1}}(\rho'_{1,XY})
          \tensor \rho'_{2,Z}
          =
          \rho'_{2,XY} \tensor \rho'_{2,Z}
          =
          \rho'_2.
          \mathQED
    \end{align*}
  \end{claimproof}

  \begin{claim}\label{claim:rho'.A}
    $\rho'$ satisfies $A$.
  \end{claim}

  \begin{claimproof}
    For any $m_1,m_2$ such that $\restrict{m_1}{X_1}=m_1^X$ and
    $\restrict{m_2}{X_2}=m_1^X\circ(\restrict\sigma{X_2})$ for some
    $m_1^X\in\types{X_1}$, we have
    $\denotee{X_1=X_2}{\memuni{m_1m_2}}=\true$, thus
    $\denotee{\CL{X_1=X_2}}{\memuni{m_1m_2}}=\elltwov{\cl{V_1}\cl{V_2}}$.
    Hence any cq-operator of the form
    $ \sum_{m_1^X,i} \lambda_{m_1^X,i}\,
    \pB\proj{\pb\basis{X_1X_2}{(m_1^X\
        m_1^X\circ(\restrict\sigma{X_2}))}} \tensor \dots $
    satisfies $\CL{X_1=X_2}$. Thus $\rho'$ as defined in \eqref{eq:rho'.def} satisfies  $\CL{X_1=X_2}$.

    From \eqref{eq:rho'.def} it follows that $\rho'$
    can be decomposed as
    $\rho'=\sum_j\lambda_j\pb\proj{\phi_j\tensor\Urename{\restrict\sigma{Y_1}}{\phi_j}\tensor
      {\phi'_{1j}}\tensor\phi'_{2j}}$ for some normalized $\phi_j\in\elltwov{Y_2}$,
    $\phi'_{1j}\in\elltwov{X_1\qu{Z_1}}$,
     $\phi'_{2j}\in\elltwov{X_2\qu{Z_2}}$,
    and $\lambda_i>0$.
    (To arrive at this decomposition, we start with
    \eqref{eq:rho'.def}, and decompose $\rho'_{1,Z}$
    and $\rho'_{2,Z}$ as mixtures of pure states.)

    Fix some $j$.
    By~\autoref{coro:quanteq} (with $\psi_1:=\phi_j\tensor\phi'_{1j}$
    and $\psi_2:=\Urename{\restrict\sigma{Y_1}}\phi_{j}\tensor\phi'_{2j}$
    and $X_1:=Y_1$
    and $X_2:=Y_2$
    and $V_1:={V_1}$
    and $V_2:={V_2}$
    and $\sigma:=(\restrict\sigma{Y_1})^{-1}$),
    we have that
    $(\phi_j\tensor\phi'_{1j})\tensor
    (\Urename{\restrict\sigma{Y_1}}\phi_j\tensor\phi'_{2j})\in(Y_1\quanteq
    Y_2)$, since
    $\phi_j =
    \Urename{(\restrict\sigma{Y_1})^{-1}}\Urename{\restrict\sigma{Y_1}}\phi_j$.

    Since for all $j$,
    $(\phi_j\tensor\phi'_{1j})\tensor
    (\Urename{\restrict\sigma{Y_1}}\phi_j\tensor\phi'_{2j})\in(Y_1\quanteq
    Y_2)$, it follows that $\rho'$ satisfies $(Y_1\quanteq Y_2)$.

    Thus $\rho$'
    satisfies both $\CL{X_1=X_2}$
    and $Y_1\quanteq Y_2$,
    hence $\rho'$
    satisfies $A = \CL{X_1=X_2} \cap (Y_1\quanteq Y_2)$.
  \end{claimproof}

  Summarizing, for any $\psi_1,\psi_2,m_1,m_2$
  as in \autoref{lemma:pure}  with $\psi_1\otimes\psi_2\in\denotee{A}{\memuni{m_1}{m_2}}$, there exists a separable
  $\rho'\in\traceclcq{V_1V_2}$
  (\autoref{claim:rho'.sep}) such that
  $\partr{V_1}{V_2}\rho'=\rho'_1=\denotc{\idx1\bc}(\pointstate{\cl{V_1}}{m_1}\tensor\proj{\psi_1})$
  and
  $\partr{V_2}{V_1}\rho'=\rho'_2=\denotc{\idx2\bc}(\pointstate{\cl{V_2}}{m_2}\tensor\proj{\psi_2})$
  (\autoref{claim:rho'.proj}) and such that $\rho'$
  satisfies $A$ (\autoref{claim:rho'.A}).

  By \autoref{lemma:pure}, this implies $\rhl A\bc\bc A$.
  Since $A$
  was defined as $A:=\CL{X_1=X_2}\cap(Y_1\quanteq Y_2)$,
  the lemma follows.
\end{proof}

\begin{lemma}[Elimination rule for qRHL]\label{rule-lemma:QrhlElim}
  \Ruleref{QrhlElim} holds.
\end{lemma}

\begin{proof}
  Let $\rho_1':=\Erename{\idx1}(\rho_1)$ and  $\rho_2':=\Erename{\idx2}(\rho_2)$.
  
  Since $\rho$
  is separable, and $\rho$
  satisfies $A$,
  and since $\rhl A\bc\bd {\CL{\idx1 e \Rightarrow \idx2f}}$,
  there exists a state $\rho'$ such that $\rho'$ satisfies $\CL{\idx1e \Rightarrow \idx2f}$, and
  \begin{align*}
    \partr{V_1}{V_2}\rho' &= \denotc{\idx1\bc}\pb\paren{\partr{V_1}{V_2} \rho} = \denotc{\idx1\bc}(\rho_1'), \\
    \partr{V_2}{V_1}\rho' &= \denotc{\idx2\bd}\pb\paren{\partr{V_2}{V_1} \rho} = \denotc{\idx2\bd}(\rho_2').
  \end{align*}
  Since $\rho'$ is a separable and $\rho'\in\traceposcq{V_1V_2}$, we can write $\rho'$ as
  \begin{equation*}
    \rho' = \sum_{\memuni{m_1m_2}}\pb\proj{\basis{\cl{V_1}\cl{V_2}}{\memuni{m_1m_2}}}\tensor\rho'_{\memuni{m_1m_2}}
  \end{equation*}
  for some $\rho'_{\memuni{m_1m_2}}\in\traceposv{\qu{V_1}\qu{V_2}}$.
  Then
  \begin{align*}
    \denotc{\idx1\bc}(\rho_1') &=
                                 \partr{V_1}{V_2}\rho'=
                                 \sum_{m_1} \pb\proj{\basis{\cl{V_1}}{m_1}} \tensor
    \underbrace{\sum_{m_2}\partr{\qu{V_1}}{\qu{V_2}}\rho'_{\memuni{m_1m_2}}}_{{}=:\rho''_{m_1}}, \\
    \denotc{\idx2\bd}(\rho_2') &= 
                                 \partr{V_2}{V_1}\rho'=
                                 \sum_{m_2} \pb\proj{\basis{\cl{V_2}}{m_2}} \tensor
                         \underbrace{\sum_{m_1}\partr{\qu{V_2}}{\qu{V_1}}\rho'_{\memuni{m_1m_2}}}_{{}=:\rho''_{m_2}}.
  \end{align*}
  Let $p_{\memuni{m_1m_2}}:=\tr\rho'_{\memuni{m_1m_2}}$.
  Let $E:=\{m_1:\denotee{\idx1e}{m_1}=\true\}$
  and $F:=\{m_2:\denotee{\idx2f}{m_2}=\true\}$.
  By definition of 
  $ \prafter e\bc{\rho_1}$, we have
  \begin{align*}
    \pb\prafter e\bc{\rho_1}
    =
    \pb\prafter{\idx1e}{(\idx1\bc)}{\rho_1'}
    = \sum_{m_1\in E} \tr \rho''_{m_1} =
    \sum_{\substack{m_1,m_2\\m_1\in E}} \tr \partr{\qu{V_1}}{\qu{V_2}}\rho'_{\memuni{m_1m_2}}
    =
    \sum_{\substack{m_1,m_2\\m_1\in E}} p_{\memuni{m_1m_2}}.
  \end{align*}
  And analogously
  \begin{equation*}
    \pb\prafter f\bd{\rho_2} =
    \sum_{\substack{m_1,m_2\\m_2\in F}} p_{\memuni{m_1m_2}}.
  \end{equation*}
  
  Since $\rho'$
  satisfies $\CL{\idx1e\Rightarrow \idx2f}$,
  we have that $\rho'_{\memuni{m_1m_2}}= 0$
  whenever $\denotee{\idx1e\Rightarrow \idx2f}{\memuni{m_1m_2}}=\false$.
  In other words, $p_{\memuni{m_1m_2}}=0$ whenever $m_1\in E\land m_2\notin F$.

  Thus
  \vspace{-11pt}
  \begin{align*}
    \pb\prafter f\bd{\rho_2} &=
    \sum_{\substack{m_1,m_2\\m_2\in F}} p_{\memuni{m_1m_2}}
    =
    \sum_{\substack{m_1,m_2\\m_2\in F}} p_{\memuni{m_1m_2}} + 
    \overbrace{\sum_{\substack{m_1,m_2\\m_1\in E \land m_2\notin F}} p_{\memuni{m_1m_2}}}^{{}=0}
    =
    \sum_{\substack{m_1,m_2\\m_1\in E\vee m_2\in F}} p_{\memuni{m_1m_2}} \\ 
    &\geq
    \sum_{\substack{m_1,m_2\\m_1\in E}} p_{\memuni{m_1m_2}} 
    =
    \pb\prafter e\bc{\rho_1}.
    %\mathQED
  \end{align*}
  This shows  \ruleref{QrhlElim}.

  \medskip

  The variants of \ruleref{QrhlElim} with $=,\Leftrightarrow$
  and $\geq,\Leftarrow$
  instead of $\leq,\Rightarrow$
  follow from the $\leq,\Rightarrow$-case using \ruleref{Sym} and \ruleref{Conseq}.
\end{proof}

\begin{lemma}[Elimination rule for qRHL, for quantum equality]\label{rule-lemma:QrhlElimEq}
\Ruleref{QrhlElimEq} holds.
\end{lemma}

\begin{proof}\stepcounter{claimstep}
  Fix some arbitrary normalized $\psi_0\in\elltwov{\qu{V}\setminus Y}$.
  Let
  $\calE(\hat\rho):=\partr{\cl{V}Y}{\qu{V}\setminus
    Y}(\hat\rho)\tensor\proj {\psi_0}$
  for all $\hat\rho\in\traceclcq V$.  Then $\calE$
  is an $(\qu{V}\setminus Y)$-local cq-superoperator on $V$.

  Let $\rho^*:=\calE(\rho)$.
  Since $\calE$
  is $(\qu{V}\setminus Y)$-local,
  and $\bc$
  is $XY$-local,
  and $(\qu{V}\setminus Y)\cap XY=\varnothing$
  (since $X_1$
  are classical), we have that $\denotc\bc$
  and $\calE$
  commute.  Furthermore, since $\calE$
  is $(\qu{V}\setminus Y)$-local,
  it operates only on quantum variables.
  Thus $\calE=\idv V\otimes\calE'$ for some trace-preserving $\calE'$ on $\qu V$.
  Let $\denotc\bc(\rho)=:\sum_m\pb\proj{\basis{\cl V}{m}}\otimes\rho_m$ for some $\rho_m\in\traceclcq{\qu V}$.
  Then 
  \[
  \denotc\bc\paren{\rho^*}
  =
  \denotc\bc\pb\paren{\calE(\rho)}
  =
  \sum_m\calE\pb\paren{\proj{\basis{\cl V}{m}}\otimes\rho_m}
  =
  \sum_m{\proj{\basis{\cl V}{m}}\otimes\calE'(\rho_m)}.
  \]
  Hence
  \begin{equation}\label{eq:ebcrho1}
    \pb\prafter e\bc{\rho}
    =\!\!\!\sum_{m:\denotee em=\true}\!\!\!\!\tr\rho_m
    =\!\!\!\sum_{m:\denotee em=\true}\!\!\!\!\tr\calE'(\rho_m)
    =\pb\prafter e\bc{\rho^*}.
\end{equation}
  Analogously, we get
  \begin{equation}
   \pb\prafter f\bd{\rho}=\pb\prafter f\bd{\rho^*}.\label{eq:ebdrho1}
 \end{equation}
   Since $\calE$
   is a cq-superoperator, we have that
   $\rho^*=\calE(\rho)\in\traceposcq{V}$.
   Thus we can decompose $\rho^*$
   as
   $\rho^*=\sum_{m,i} \alpha_{m,i}\,\pb\proj{\basis{\cl{V}}m\tensor
     \psi_{m,i}\tensor \psi_0}$ for some $\alpha_{m,i}>0$
   and normalized $\psi_{m,i}\in\elltwov{Y}$.  Let
  \begin{align*}
    \hat\rho &:= \sum_{m,i}\alpha_{m,i}\pB\proj{
               \,
               \underbrace{
               \vphantom{\Bigl(\Bigr)}
               \Urename{\idx1}\pb\paren{\basis{\cl{V}}m \tensor \psi_{m,i}\tensor \psi_0}
               \tensor
               \Urename{\idx2}\pb\paren{\basis{\cl{V}}m \tensor \psi_{m,i}\tensor \psi_0}
               }_{{}=:\phi_{m,i}}
               \,
}.
  \end{align*}
  
  \begin{claim}\label{claim:hatrho.sep}
    $\hat \rho$ is separable.
  \end{claim}

  \begin{claimproof}
    Immediate from the definition of $\hat\rho$.
  \end{claimproof}

  \begin{claim}
    $\hat\rho$ satisfies $\CL{X_1=X_2}$.
  \end{claim}

  \begin{claimproof}
  We have
  $\denotee {X_1}{m\,\circ\,{\idx1^{-1}}}=\denotee{X}m=\denotee
  {X_2}{m\,\circ\,{\idx2^{-1}}}$ for all $m\in\types{\cl V}$. Thus
  $\denotee {X_1=X_2}{m\,\circ\,{\idx1^{-1}}\, m\,\circ\,{\idx2^{-1}}}=\true$.
  Furthermore
  \begin{align*}
    \phi_{m,i}&=\Urename{\restrictno{\idx1\!}{\cl V}}\basis{\cl V}m \otimes
    \Urename{\restrictno{\idx2\!}{\cl V}}\basis{\cl V}m \otimes \cdots
    =\pb\basis{\cl{V_1}}{m\circ\idx1^{-1}} \otimes
                \pb\basis{\cl{V_2}}{m\circ\idx2^{-1}} \otimes\cdots
            \\&
    = \pb\basis{\cl{V_1}\cl{V_2}}{m\circ\idx1^{-1}\, m\circ\idx2^{-1}} \otimes \cdots
  \end{align*}
  where $\dots$ stands for the other terms from $\phi_{m,i}$ which are irrelevant here.
  And $\hat\rho=\sum_{m,i}\alpha_{m,i}\,\phi_{m,i}$. 
  Thus $\proj{\phi_{m,i}}$ satisfies $\CL{X_1=X_2}$.
  Thus $\hat\rho$ satisfies $\CL{X_1=X_2}$.
\end{claimproof}

\begin{claim}
  $\hat\rho$ satisfies $Y_1\quanteq Y_2$.
\end{claim}

\begin{claimproof}
  Let $\sigma:Y_2\to Y_1$ be defined by $\sigma(\qq_2):=\qq_1$ for all $\qq\in Y$.
  Note that
  \[
    \Urename{\restrictno{\idx1\!}{Y}}\psi_{m,i}
    =
    \Urename{\idx1\!\circ{\idx2\!}^{-1}}\, \Urename{\restrictno{\idx2\!}{Y}}\psi_{m,i}
    =
    \Urename\sigma\, \Urename{\restrictno{\idx2\!}{Y}}\psi_{m,i}.
  \]
  Thus
  by \autoref{coro:quanteq} (with $X_1:=Y_1$,
  $X_2:=Y_2$,
  $V_1:={V_1}$,
  $V_2:={V_2}$,
  $\psi_1^X:=\Urename{\restrictno{\idx1\!}{Y}}\psi_{m,i}$,
  $\psi_2^X:=\Urename{\restrictno{\idx2\!}{Y}}\psi_{m,i}$,
  $\psi_1^Y:=\Urename{\restrictno{\idx1\!}{\qu{V}\setminus Y}}\psi_0$,
  $\psi_2^Y:= \Urename{\restrictno{\idx2\!}{\qu{V}\setminus Y}}\psi_0$,
  $\sigma:=\sigma$)
  we have
  \begin{multline*}
    \phi'_{m,i}:=\Urename{\restrictno{\idx1\!}{Y}}\psi_{m,i}\tensor
    \Urename{\restrictno{\idx1\!}{\qu{V}\setminus Y}}\psi_0\tensor
    \Urename{\restrictno{\idx2\!}{Y}}\psi_{m,i}\tensor
    \Urename{\restrictno{\idx2\!}{\qu{V}\setminus Y}}\psi_0
    \\
    \in
    (Y_1\quanteq Y_2).
  \end{multline*}
  Since
  $\phi_{m,i}=\pb\basis{\cl{V_1}\cl{V_2}}{m\circ(\idx1)^{-1}\ 
    m\circ(\idx2)^{-1}} \otimes\phi'_{m,i}$, it
  follows that $\proj{\phi_{m,i}}$
  satisfies $Y_1\quanteq Y_2$,
  and thus $\hat\rho=\sum_{m,i}\alpha_{m,i}\,\proj{\phi_{m,i}}$
  satisfies $Y_1\quanteq Y_2$.
\end{claimproof}

\begin{claim}\label{claim:A1}
  $\hat\rho$ satisfies $A_1$.
\end{claim}

\begin{claimproof}
  Since $\rho$
  satisfies $A$ by assumption,
  and $\calE$
  is $(\qu V\setminus Y)$-local,
  and $A$
  is $\cl VY$-local by assumption,
  we have that $\rho^*=\calE(\rho)$ satisfies $A$.  Since 
   $\rho^*=\sum_{m,i} \alpha_{m,i}\,\pb\proj{\basis{\cl{V}}m}\tensor\proj{
     \psi_{m,i}\tensor \psi_0}$,
   it follows that $     \psi_{m,i}\tensor \psi_0\in\denotee A{m}=
\denotee{\idx1A}{m\circ(\idx1)^{-1}\,
    m\circ(\idx2)^{-1}}$.
  Hence
$\Urename{\restrictno{\idx1}{\qu V}}
  (\psi_{m,i}\tensor \psi_0)\in 
\pb\denotee{\Urename{\restrictno{\idx1}{\qu V}}\idx1A}{m\circ(\idx1)^{-1}\,
    m\circ(\idx2)^{-1}}$.
Thus
\begin{multline*}
  \phi'_{m,i} = 
\Urename{\restrictno{\idx1}{\qu V}}
  (\psi_{m,i}\tensor \psi_0)
\tensor
\Urename{\restrictno{\idx2}{\qu V}}
  (\psi_{m,i}\tensor \psi_0) \\
\in
\pb\denotee{\Urename{\restrictno{\idx1}{\qu V}}\idx1A}{m\circ(\idx1)^{-1}\,
    m\circ(\idx2)^{-1}}
\otimes\elltwov{\qu{V_2}}
=
\denotee{A_1}{m\circ(\idx1)^{-1}\,
    m\circ(\idx2)^{-1}}.
\end{multline*}
Thus
\begin{equation*}
  \proj{\phi_{m,i}} =
\pB\proj{\pb\basis{\cl{V_1}\cl{V_2}}{m\circ(\idx1)^{-1}\,
    m\circ(\idx2)^{-1}} \otimes\phi'_{m,i}}
\text{ satisfies }A_1.
  \end{equation*}
  Thus  $\hat\rho=\sum_{m,i}\alpha_{m,i}\,\proj{\phi_{m,i}}$ satisfies $A_1$.
\end{claimproof}

\begin{claim}\label{claim:A2}
  $\hat\rho$ satisfies $A_2$
\end{claim}

\begin{claimproof}
  Analogous to \autoref{claim:A1}.
\end{claimproof}

\begin{claim}\label{claim:ef}
  $\pb\prafter {e}{\bc}{\rho^*}\leq\pb\prafter {f}{\bd}{\rho^*}$.
\end{claim}
  
\begin{claimproof}
  Since $\hat\rho$ satisfies $\CL{X_1=X_2}$ and $Y_1\quanteq Y_1$ and $A_1$ and $A_2$, we have
  that $\hat\rho$ satisfies
  $\Hat A:={\CL{X_1=X_2}\cap (Y_1\quanteq Y_1)}\cap A_1\cap A_2$.

  We have
  \begin{align*}
 \partr{V_1}{V_2}\hat\rho
    &= 
      \sum_{m,i}\alpha_{m,i}\,\pB\proj{
      \Urename{\idx1}\pb\paren{
      \basis{\cl{V}}m\tensor
      \psi_{m,i}\tensor\psi_0
      }}
             =
                 \Erename{\idx1}(\rho^*),
    \\
 \partr{V_2}{V_1}\hat\rho
    &= 
      \sum_{m,i}\alpha_{m,i}\,\pB\proj{
      \Urename{\idx2}\pb\paren{
      \basis{\cl{V}}m\tensor
      \psi_{m,i}\tensor\psi_0
      }}
             =
                 \Erename{\idx2}(\rho^*).
  \end{align*}

  Then we can apply \ruleref{QrhlElim}, with $e:=e$,
  $f:=f$,
  $\bc:=\bc$,
  $\bd:=\bd$.
  $\rho:=\hat\rho$,
  $\rho_1:=\rho^*$,
  $\rho_2:=\rho^*$,
  $A:=\Hat A$.
  (The premises of \ruleref{QrhlElim} are satisfied by Claims \ref{claim:hatrho.sep}--\ref{claim:A2}.)
  Thus we get
  $\pb\prafter {e}{\bc}{\rho^*}\leq\pb\prafter {f}{\bd}{\rho^*}$.
\end{claimproof}

  We conclude
  \begin{align*}
    \pb\prafter e\bc{\rho}
    &  \eqrefrel{eq:ebcrho1}=
      \pb\prafter e\bc{\rho^*}
%    =
%      \prafter{\idx1e}{(\idx1\bc)}{\Erename{\idx1}(\rho^*)}
%    \eqrefrel{eq:hatrho1.rho*}=
%      \prafter {\idx1e}{\idx1\bc}{\hat\rho_1} \\
%    &\eqrefrel{eq:prafter.ineq}\leq
%      \prafter {\idx2f}{\idx2\bd}{\hat\rho_2} 
%    \eqrefrel{eq:hatrho2.rename}=
%      \prafter {\idx2f}{(\idx2\bd)}{\Erename{\idx2}(\rho^*)} \\
    \quad
    \txtrel{\autoref{claim:ef}}\leq
    \quad
      \pb\prafter {f}{\bd}{\rho^*}
    \eqrefrel{eq:ebdrho1}=
      \pb\prafter {f}{\bd}{\rho_1}.
  \end{align*}
  This shows \ruleref{QrhlElimEq}.

  \medskip

  The variants of \ruleref{QrhlElimEq} with $=,\Leftrightarrow$
  and $\geq,\Leftarrow$
  instead of $\leq,\Rightarrow$
  are shown analogously using the corresponding variants of
  \ruleref{QrhlElim} in \autoref{claim:ef}.
\end{proof}

\begin{lemma}[Transitivity]\label{rule-lemma:Trans}
  \Ruleref{Trans} holds.
\end{lemma}

% The following counterexample is not correct. There is no reason to assume that
% $\rhl\dots\bc\bd{\qq_1\quanteq\qq_2}$. In fact,
% ~/svn/research-notes/shared/qrhl/basic-qrhl/qotp.pdf would contradict this, I think.
% 
% {Counterexample for unrestricted relations $b_{\bc\bd}$,
%   $b_{\bd\be}$:
%   $\bd$
%   outputs fully mixed state on $\qq$.
%   $\bc$
%   picks $\xx$
%   uniform, sets $\qq\leftarrow\basis{}\xx$.
%   $\bd$
%   picks $\xx$
%   uniform, sets $\qq\leftarrow H\basis{}\xx$.
%   Then $\rhl\dots\bc\bd{\qq_1\quanteq\qq_2}$
%   and $\rhl\dots\bd\be{\qq_1\quanteq\qq_2}$,
%   but not $\rhl\dots\bc\be{\qq_1\quanteq\qq_2}$.
%   % 
%   Need additional condition, e.g., $m_{\bc}$ restricted to $\fv(\bc)$ is uniquely determined by $m_{\bd}$ (assuming relation $b_{\bc\bd}$ holds). And same for $m_{\be}$.
%   % 
%   Or formally: $\exists f$ s.t. $b_{\bc\bd}\implies X_1=f(X_2)$ holds where $X_1:=\idx1\fv(\bc)$ and $X_2:=\fv(b_{\bc\be})\cap V_2$.
% }

\begin{proof}\stepcounter{claimstep}
  Let
  \begin{align*}
    A_{\bc\bd} &:= {\CL{a_{\bc\bd}} \cap \paren{u_{\bc1}Q_{\bc1} \quanteq u_{\bd2}Q_{\bd2}}} \\
    A_{\bd\be} &:= {\CL{a_{\bd\be}} \cap \paren{u_{\bd1}Q_{\bd1} \quanteq u_{\be2}Q_{\be2}}} \\
    A_{\bc\be} &:= {\CL{a_{\bc\bd}\circexp a_{\bd\be}} \cap \paren{u_{\bc1}Q_{\bc1} \quanteq u_{\be2}Q_{\be2}}} \\
    B_{\bc\bd} &:= {\CL{b_{\bc\bd}} \cap \paren{v_{\bc1}Q_{\bc1} \quanteq v_{\bd2}R_{\bd2}}} \\
    B_{\bd\be} &:= {\CL{b_{\bd\be}} \cap \paren{v_{\bd1}R_{\bd1} \quanteq v_{\be2}Q_{\be2}}} \\
    B_{\bc\be} &:= \CL{b_{\bc\bd}\circexp b_{\bd\be}} \cap \paren{v_{\bc1}Q_{\bc1} \quanteq v_{\be2}Q_{\be2}}
  \end{align*}
  With this notation, the qRHL judgments in the premise of
  \ruleref{Trans} are $\rhl{A_{\bc\bd}}\bc\bd{B_{\bc\bd}}$
  and $\rhl{A_{\bd\be}}\bd\be{B_{\bd\be}}$, and we need to show the conclusion
  $\rhl{A_{\bc\be}}\bc\be{B_{\bc\be}}$.

  \medskip
  
  Furthermore, for $i=1,2$ we abbreviate
  $U_{\idx i}:= \Urename{\restrictno{\idx i}{\qu V}}$,
  i.e., $U_{\idx i}$
  is the canonical isomorphism between $\elltwov{\qu V}$
  and $\elltwov{\qu{V_i}}$ (where $V$ is clear from the context).
  Let $\calE_{\idx i}:=\Erename{\idx i}$,
  i.e., the cq-superoperator $\calE_{\idx i}$
  is the canonical isomorphism between $\traceposcq{V}$
  and $\traceposcq{V_i}$.

\medskip

Furthermore, let 
\begin{align*}
  B_{\bc\bd}^Q &:= {\CL{b_{\bc\bd}} \cap \paren{v_{\bc1}Q_{\bc1} \quanteq v_{\bd2}R_{\bd2}}} \subseteq \elltwov{\cl{V_1}\cl{V_2}Q_{\bc1}R_{\bd2}}  \\
  B_{\bd\be}^Q &:= {\CL{b_{\bd\be}} \cap \paren{v_{\bd1}R_{\bd1} \quanteq v_{\be2}Q_{\be2}}} \subseteq \elltwov{\cl{V_1}\cl{V_2}R_{\bd1}Q_{\be2}}
\end{align*}
That is,
$B_{\bc\bd}=B_{\bc\bd}^Q\otimes\elltwov{\qu{V_1}\qu{V_2}\setminus Q_{\bc1}R_{\bd2}}$
and $B_{\bd\be}=B_{\bd\be}^Q\otimes\elltwov{\qu{V_1}\qu{V_2}\setminus R_{\bd1}Q_{\be2}}$
(In other words, $B_{\bc\bd}$
and $B_{\bc\bd}^Q$
are basically the same predicate, but the latter is defined on a
subset of variables.)

Define $R_{\bc\bd}^A, R_{\bd\be}^A, R_{\bc\bd}^B,  R_{\bd\be}^B\subseteq \types{\cl V}\times\types{\cl V}$ to be the
relations on memories encoded by the expressions $a_{\bc\bd},a_{\bd\be},b_{\bc\bd},b_{\bd\be}$. Formally:
\begin{align*}
  (m_{\bc},m_{\bd})\in R^A_{\bc\bd}
  &:\!\!\iff
    \denotee{a_{\bc\bd}}{m_{\bc}\circ\idx1^{-1}\,m_{\bd}\circ\idx2^{-1}}=\true
  \\
  (m_{\bd},m_{\be})\in R^A_{\bd\be}
  &:\!\!\iff
  \denotee{a_{\bd\be}}{m_{\bd}\circ\idx1^{-1}\,m_{\be}\circ\idx2^{-1}}=\true
  \\
  (m_{\bc},m_{\bd})\in R^B_{\bc\bd}
  &:\!\!\iff
    \denotee{b_{\bc\bd}}{m_{\bc}\circ\idx1^{-1}\,m_{\bd}\circ\idx2^{-1}}=\true
  \\
  (m_{\bd},m_{\be})\in R^B_{\bd\be}
  &:\!\!\iff
  \denotee{b_{\bd\be}}{m_{\bd}\circ\idx1^{-1}\,m_{\be}\circ\idx2^{-1}}=\true
\end{align*}

% Let $R^B_{\bc\bd}\subseteq\types{\cl V}\times\types{\cl V}$
% be the set of all $(m_{\bc},m_{\bd})$
% such that
% $\denotee{b_{\bc\bd}}{m_{\bc}\circ\idx1^{-1}\,m_{\bd}\circ\idx2^{-1}}=\true$.

% And let $R^B_{\bd\be}\subseteq\types{\cl V}\times\types{\cl V}$
% be the set of all $(m_{\bd},m_{\be})$
% such that
% $\denotee{b_{\bd\be}}{m_{\bd}\circ\idx1^{-1}\,m_{\be}\circ\idx2^{-1}}=\true$.

Then in particular, $R^A_{\bc\bd}\circrel R^A_{\bd\be}$
and $R^B_{\bc\bd}\circrel R^B_{\bd\be}$
are the relations encoded by the expressions $a_{\bc\bd}\circexp a_{\bd\be}$
and $b_{\bc\bd}\circexp b_{\bd\be}$. Formally:
\begin{align}
  (m_{\bc},m_{\be}) \in  R^A_{\bc\bd}\circrel R^A_{\bd\be}
  &\iff
    \denotee{a_{\bc\bd}\circexp
    a_{\bd\be}}{m_{\bc}\circ\idx1^{-1}\,m_{\be}\circ\idx2^{-1}}=\true
    \label{eq:RAce}
  \\
  (m_{\bc},m_{\be}) \in  R^B_{\bc\bd}\circrel R^B_{\bd\be}
  &\iff
    \denotee{b_{\bc\bd}\circexp
    b_{\bd\be}}{m_{\bc}\circ\idx1^{-1}\,m_{\be}\circ\idx2^{-1}}=\true
    \notag
    \label{eq:RBce}
\end{align}
(This is immediate from the definition of $\circexp$, \autopageref{page:def:circexp}.)

% In particular, $R^B_{\bc\bd}\circrel R^B_{\bd\be}$
% is the set of all $(m_{\bc},m_{\be})$
% such that
% $\denotee{b_{\bc\bd}\circexp
%   b_{\bd\be}}{m_{\bc}\circ\idx1^{-1}\,m_{\be}\circ\idx2^{-1}}=\true$
% by definition of $\circexp$.

% Let $R^A_{\bc\bd}\subseteq\types{\cl V}\times\types{\cl V}$
% be the set of all $(m_{\bc},m_{\bd})$
% such that
% $\denotee{a_{\bc\bd}}{m_{\bc}\circ\idx1^{-1}\,m_{\bd}\circ\idx2^{-1}}=\true$.

% And let $R^A_{\bd\be}\subseteq\types{\cl V}\times\types{\cl V}$
% be the set of all $(m_{\bd},m_{\be})$
% such that
% $\denotee{a_{\bd\be}}{m_{\bd}\circ\idx1^{-1}\,m_{\be}\circ\idx2^{-1}}=\true$.

% In particular, $R^A_{\bc\bd}\circrel R^A_{\bd\be}$
% is the set of all $(m_{\bc},m_{\be})$
% such that
% $\denotee{a_{\bc\bd}\circexp
%   a_{\bd\be}}{m_{\bc}\circ\idx1^{-1}\,m_{\be}\circ\idx2^{-1}}=\true$
% by definition of $\circexp$.

\medskip

   \begin{claim}\label{claim:pure.index}
     To show $\rhl{A_{\bc\be}}\bc\bd{B_{\bc\be}}$,
     it is sufficient to show that for all $m_{\bc}^0\in\types{\cl{V}}$,
     $m_{\be}^0\in\types{\cl{V}}$
     and all normalized $\psi_{\bc},\psi_{\be}\in\elltwov{\qu{V}}$,
     such that
     $U_{\idx1}\psi_{\bc}\tensor U_{\idx2}\psi_{\be}\in \denotee
     {A_{\bc\be}}{\memuni{m_{\bc}^0\circ\idx1^{-1}\, m_{\be}^0\circ\idx2^{-1}}}$,
     there is a
     separable
     $\rho'_{\bc\be}\in\traceposcq{V_1V_2}$ such that:
     \begin{gather}
       \rho'_{\bc\be}\text{   satisfies }B_{\bc\be}, \notag\\
       \partr{V_1}{V_2}\rho'_{\bc\be}=\calE_{\idx 1}(\rho'_{\bc})\text{ for }\rho'_{\bc}
       :=\denotc{\bc}\pb\paren{\pointstate{\cl{V}}{m_{\bc}^0}\tensor\proj{\psi_{\bc}}}, \label{def:rho'c} \\
       \partr{V_2}{V_1}\rho'_{\bc\be}=\calE_{\idx 2}(\rho'_{\be})\text{ for }\rho'_{\be}:=\denotc{\be}\pb\paren{\pointstate{\cl{V}}{m_{\be}^0}\tensor\proj{\psi_{\be}}}. \label{def:rho'e}
     \end{gather}
   \end{claim}

   \begin{claimproof}
     This is a simple corollary of \autoref{lemma:pure}, with $A:=A_{\bc\be}$ and $B:=B_{\bc\be}$:

     In
     \autoref{lemma:pure} we quantify over all
     $m_1\in\types{\cl{V_1}}$,
     $m_2\in\types{\cl{V_2}}$,
     here we quantify over all $m_{\bc}^0,m_{\be}^0\in\types{\cl{V}}$
     instead and use the fact that
     $m_{\bc}^0\mapsto m_{\bc}^0\circ\idx1^{-1}$
     and $m_{\bd}^0\mapsto m_{\bd}^0\circ\idx2^{-1}$
     are bijections from $\types{\cl V}$
     to $\types{\cl{V_1}}$
     and $\types{\cl{V_2}}$, respectively.
     We thus replace occurences of $m_1,m_2$ in \autoref{lemma:pure}
     by $ m_{\bc}^0\circ\idx1^{-1}$ and $ m_{\bd}^0\circ\idx2^{-1}$.

     Similarly, we quantify over $\psi_{\bc},\psi_{\bd}\in\elltwo{\qu V}$
     instead of $\psi_1\in\elltwov{\qu{V_1}}$,
     $\psi_2\in\elltwov{\qu{V_2}}$,
     and replace $\psi_1,\psi_2$
     by $U_{\idx1}\psi_{\bc},U_{\idx2}\psi_{\bd}$
     (because $U_{\idx1},U_{\idx2}$
     is unitary and thus an isomorphism).

     And we rename $\rho'$ as $\rho'_{\bc\bd}$.
     
     These substitution in \autoref{lemma:pure} already almost give us
     the present claim, except that the last two lines become:
     \begin{align}
       &\partr{V_1}{V_2}\rho'_{\bc\be}=\denotc{\idx 1\bc}\pb\paren{\pointstate{\cl{V_1}}{m_{\bc}^0\circ\idx1^{-1}}\tensor\proj{U_{\idx1}\psi_{\bc}}}, \label{eq:rho'ce.1}
         \\
       &\partr{V_2}{V_1}\rho'_{\bc\be}=\denotc{\idx 2\be}\pb\paren{\pointstate{\cl{V_2}}{m_{\be}^0\circ\idx2^{-1}}\tensor\proj{U_{\idx2}\psi_{\be}}}. \label{eq:rho'ce.2}
     \end{align}
     Since $\basis{\cl{V_1}}{m_{\bc}^0\circ\idx1^{-1}}=U_{\idx1}\basis{\cl V}{m_{\bc}^0}$, we have
     \begin{align*}
       &\denotc{\idx 1\bc}\pb\paren{\pointstate{\cl{V_1}}{m_{\bc}^0\circ\idx1^{-1}}\tensor\proj{U_{\idx1}\psi_{\bc}}}
         \\
       &=
         \denotc{\idx 1\bc}\pb\paren{\proj{U_{\idx1}\basis{\cl{V}}{m_{\bc}^0}}\tensor\proj{U_{\idx1}\psi_{\bc}}}
         \\
       &=
         \denotc{\idx 1\bc}\pb\paren{\calE_{\idx1}\paren{\proj{\basis{\cl{V}}{m_{\bc}^0}}\tensor \psi_{\bc}}}
         \\
       &=
       \calE_{\idx1}\paren{\denotc{\bc}\pb\paren{\proj{\basis{\cl{V}}{m_{\bc}^0}}\tensor \psi_{\bc}}}.
     \end{align*}
     Thus \eqref{eq:rho'ce.1} becomes \eqref{def:rho'c}. Analogously,
     \eqref{eq:rho'ce.2} becomes \eqref{def:rho'e}.
   \end{claimproof}

  Thus, fix some $m_{\bc}^0,m_{\be}^0,\psi_{\bc},\psi_{\be}$
  with the properties from \autoref{claim:pure.index}.  We need to
  show the existence of $\rho'_{\bc\be}$
  with the properties from \autoref{claim:pure.index}.

  Let $U_{\bc}:=\denotee{u_{\bc}}{m_{\bc}^0}=\denotee{u_{\bc1}}{m_{\bc}^0\circ\idx1^{-1}}\in\iso{\typel{Q_{\bc}},Z}$
  and $U_{\be}:=\denotee{u_{\be}}{m_{\be}^0}=\denotee{u_{\be2}}{m_{\be}^0\circ\idx2^{-1}}\in\iso{\typel{Q_{\be}},Z}$.

  \begin{claim}\label{claim:psi.Q}
    $(m_{\bc}^0,m_{\be}^0)\in R_{\bc\bd}^A\circrel R_{\bd\be}^A$.
    And $\psi_{\bc}$
    and $\psi_{\be}$
    are of the form $\psi_{\bc}=\psi_{\bc Q}\otimes\psi_{\bc\bullet}$
    and $\psi_{\be}=\psi_{\be Q}\otimes\psi_{\be\bullet}$
    for some normalized $\psi_{\bc Q}\in\elltwov{Q_{\bc}}$,
    $\psi_{\be Q}\in\elltwov{Q_{\be}}$,
    $\psi_{\bc\bullet}\in\elltwov{\qu V\setminus Q_{\bc}}$,
    $\psi_{\be\bullet}\in\elltwov{\qu V\setminus Q_{\be}}$.
    And
    $ U_{\bc}\adj{\Uvarnames{Q_{\bc}}}\psi_{\bc Q} =
    U_{\be}\adj{\Uvarnames{Q_{\be}}}\psi_{\be Q} $.
    \end{claim}

  \begin{claimproof}
    By choice of $m_{\bc}^0,m_{\be}^0,\psi_{\bc},\psi_{\be}$
    (after \autoref{claim:pure.index}) we have that
    \begin{equation}
      \label{eq:psic.psie.Ace}
      U_{\idx1}\psi_{\bc}\tensor U_{\idx2}\psi_{\be}\in \denotee
      {A_{\bc\be}}{\memuni{m_{\bc}^0\circ\idx1^{-1}\, m_{\be}^0\circ\idx2^{-1}}},
    \end{equation}
    and $\psi_{\bc},\psi_{\be}$ are normalized. This implies 
    $\denotee{A_{\bc\be}}{\memuni{m_{\bc}^0\circ\idx1^{-1}\, m_{\be}^0\circ\idx2^{-1}}}\neq0$, and
    by definition of $A_{\bc\be}$,
    $\denotee{\CL{a_{\bc\bd}\circexp a_{\bd\be}}}{\memuni{m_{\bc}^0\circ\idx1^{-1}\, m_{\be}^0\circ\idx2^{-1}}}=\true$.
    Hence $(m_{\bc}^0,m_{\be}^0)\in R^A_{\bc\bd}\circrel R^A_{\bd\be}$ by \eqref{eq:RAce}.

    By definition of $A_{\bc\be}$, \eqref{eq:psic.psie.Ace} also implies
    \begin{equation*}
      U_{\idx1}\psi_{\bc}\tensor U_{\idx2}\psi_{\be}
      \in \denotee
      {u_{\bc1}Q_{\bc1}\quanteq u_{\be2}Q_{\be2}}
      {\memuni{m_{\bc}^0\circ\idx1^{-1}\, m_{\be}^0\circ\idx2^{-1}}}
      =
      (U_{\bc}Q_{\bc1}\quanteq U_{\be}Q_{\be2}).
    \end{equation*}
    By \autoref{lemma:quanteq} this implies that there exist
    normalized $\psi_1^Q\in\elltwov{Q_{\bc1}}$,
    $\psi_1^Y\in\elltwov{\qu{V_1}\setminus Q_{\bc1}}$,
    $\psi_2^Q\in\elltwov{Q_{\be2}}$,
    $\psi_2^Y\in\elltwov{\qu{V_2}\setminus Q_{\be2}}$
    such that
    $U_{\bc}\adj{\Uvarnames{Q_{\bc1}}}\psi^Q_1=
    U_{\be}\adj{\Uvarnames{Q_{\be2}}}\psi^Q_2$ and
    $U_{\idx1}\psi_{\bc}=\psi^Q_1\tensor\psi^Y_1$
    and $U_{\idx2}\psi_{\be}=\psi^Q_2\tensor\psi^Y_2$.
    
    Let $\psi_{\bc Q}:=U_{\idx1}^{-1}\psi_1^Q$,
    $\psi_{\bc\bullet}:=U_{\idx1}^{-1}\psi_1^Y$,
    $\psi_{\be Q}:=U_{\idx2}^{-1}\psi_2^Q$,
    $\psi_{\be\bullet}:=U_{\idx2}^{-1}\psi_2^Y$.
    Then
    $\psi_{\bc}=U_{\idx1}^{-1}\paren{\psi^Q_1\tensor\psi^Y_1}=\psi_{\bc Q}
    \otimes\psi_{\bc\bullet}$. And analogously
    $\psi_{\be}= \psi_{\be Q} \otimes\psi_{\be\bullet}$.

    And finally,
    \begin{multline*}
      U_{\bc}\adj{\Uvarnames{Q_{\bc}}}\psi_{\bc Q}
      =
      U_{\bc}\adj{\Uvarnames{Q_{\bc}}}U_{\idx1}^{-1}\psi_1^Q
      =
      U_{\bc}\adj{\Uvarnames{Q_{\bc1}}}\psi_1^Q
      \\
      =
      U_{\be}\adj{\Uvarnames{Q_{\be2}}}\psi_2^Q
      =
      U_{\be}\adj{\Uvarnames{Q_{\be}}}U_{\idx2}^{-1}\psi_2^Q
      =
      U_{\be}\adj{\Uvarnames{Q_{\be}}}\psi_{\be Q}
      \mathQED
    \end{multline*}
  \end{claimproof}

  Let
  \begin{gather}
    \label{eq:def:rho''c}
    \rho''_{\bc}:=\denotc\bc\pb\paren{\proj{\basis{\cl V}{m_{\bc}^0}}\otimes\proj{\psi_{\bc Q}}}
    \in\traceposcq{\cl VQ_{\bc}} \\
    \llap{\text{and}\quad}
    \rho''_{\be}:=\denotc\be\pb\paren{\proj{\basis{\cl V}{m_{\be}^0}}\otimes\proj{\psi_{\be Q}}}
    \in\traceposcq{\cl VQ_{\be}}
    \notag
  \end{gather}
  (It is well-defined to apply $\bc$ and $\be$ to
  the smaller sets of quantum variables $Q_{\bc}$ and $Q_{\be}$ since
  $\qu{\fv(\bc)}\subseteq Q_{\bc}$ and
  $\qu{\fv(\be)}\subseteq Q_{\be}$ by assumption.)
  
  Then
  \begin{equation}
    \rho'_{\bc} = \rho''_{\bc}\otimes\proj{\psi_{\bc\bullet}}
    \quad\text{and}\quad
    \rho'_{\be} = \rho''_{\be}\otimes\proj{\psi_{\be\bullet}}
    \label{eq:rho'''}
  \end{equation}
  by \autoref{claim:psi.Q} and \eqref{def:rho'c},\eqref{def:rho'e}.

  (In the following, operators $\rho''_{\dots}$
  will always be variants of operators $\rho'_{\dots}$
  that live on a smaller set of quantum variables.)
  
  \begin{claim}\label{claim:ex.rho.cd.de}
    There are
    \begin{align*}
      &\text{separable operators }
      \rho_{\bc\bd}''\in\traceposcq{\cl{V_1}\cl{V_2}Q_{\bc1}R_{\bd2}}
      \text{ and }
      \rho_{\bd\be}''\in\traceposcq{\cl{V_1}\cl{V_2}R_{\bd1}Q_{\be2}}, \\
      &\text{and an operator }\rho_{\bd}''\in\traceposcq{\cl VR_{\bd}},
%      &\text{vectors }\psi''_{\bc\bd}\in\elltwov{\qu{V_1}\qu{V_2}V_1V_2\setminus Q_{\bc1}R_{\bd2}}
%      \text{ and }
%      \psi''_{\bd\be}\in\elltwov{\qu{V_1}\qu{V_2}\setminus Q_{\bd1}Q_{\be2}}
    \end{align*}
    such that 
    \begin{align*}
      &\rho''_{\bc\bd} \text{ satisfies }B^Q_{\bc\bd}
      &&\partr{\cl{V_1}Q_{\bc1}}{\cl{V_2}R_{\bd2}}\rho''_{\bc\bd} = \calE_{\idx1}(\rho''_{\bc}) 
      &&\partr{\cl{V_1}R_{\bd1}}{\cl{V_2}Q_{\be2}}\rho''_{\bd\be} = \calE_{\idx1}(\rho''_{\bd}) 
      \\
      &\rho''_{\bd\be} \text{ satisfies }B^Q_{\bd\be} 
      &&\partr{\cl{V_2}R_{\bd2}}{\cl{V_1}Q_{\bc1}}\rho''_{\bc\bd} = \calE_{\idx2}(\rho''_{\bd}) 
      &&\partr{\cl{V_2}Q_{\be2}}{\cl{V_1}R_{\bd1}}\rho''_{\bd\be} = \calE_{\idx2}(\rho''_{\be})
    \end{align*}
  \end{claim}

  \begin{claimproof}
    Since $(m_{\bc}^0,m_{\be}^0)\in R_{\bc\bd}^A\circrel R_{\bd\be}^A$
    by \autoref{claim:psi.Q}, we can choose an
    $m_{\bd}^0\in\types{\cl{V}}$
    such that $(m_{\bc}^0,m_{\bd}^0)\in R^A_{\bc\bd}$
    and $(m_{\bd}^0,m_{\be}^0)\in R^A_{\bd\be}$.

    Also by \autoref{claim:psi.Q}, we have
    \begin{equation*}
      \psi := U_{\bc}\adj{\Uvarnames{Q_{\bc}}}\psi_{\bc Q} = U_{\be}\adj{\Uvarnames{Q_{\be}}}\psi_{\be Q}
      \in \elltwo Z.
    \end{equation*}

    Let  $U_{\bd}:=\denotee{u_{\bd}}{m_{\bd}^0}=\denotee{u_{\bd1}}{m_{\bd}^0\circ\idx1^{-1}}=\denotee{u_{\bd2}}{m_{\bd}^0\circ\idx2^{-1}}\in\uni{\typel{Q_{\bd}},Z}$.
    Fix an arbitrary normalized $\psi_{\bd\bullet}\in\elltwov{\qu V\setminus Q_{\bd}}$.
    Let $\psi_{\bd Q}:=\Uvarnames{Q_{\bd}}\adj{U_{\bd}}\psi\in\elltwov{Q_{\bd}}$.
    Let $\psi_{\bd} := \psi_{\bd Q}\otimes\psi_{\bd\bullet}\in\elltwov{\qu V}$.
    Note that $\psi_{\bd Q}$ is normalized since $\psi$ is and $U_{\bd}$
    is unitary.
    Then
    \begin{equation*}
      U_{\bc}\adj{\Uvarnames{Q_{\bc1}}}U_{\idx1}\psi_{\bc Q}
      =
      U_{\bc}\adj{\Uvarnames{Q_{\bc}}}\psi_{\bc Q}
      =
      \psi
      =
      U_{\bd}\adj{\Uvarnames{Q_{\bd}}}\psi_{\bd Q}
      =
      U_{\bd}\adj{\Uvarnames{Q_{\bd2}}}U_{\idx2}\psi_{\bd Q}
    \end{equation*}
    and hence by \autoref{lemma:quanteq} (converse direction),
    \begin{multline*}
      U_{\idx1}\psi_{\bc}\otimes U_{\idx2}\psi_{\bd}
      =
      U_{\idx1}\psi_{\bc Q}\otimes U_{\idx2}\psi_{\bd Q}
      \otimes U_{\idx1}\psi_{\bc\bullet}\otimes U_{\idx2}\psi_{\bd\bullet}
      \\
      \in (U_{\bc}Q_{\bc1}\quanteq U_{\bd}Q_{\bd2})
      =
      \denotee{u_{\bc1}Q_{\bc1}\quanteq u_{\bd2}Q_{\bd2}}
      {\memuni{m_{\bc}^0\circ\idx1^{-1}\, m_{\bd}^0\circ\idx2^{-1}}}.
%      \label{eq:psic.psid}
    \end{multline*}
    Analogously,
    \begin{equation*}
      U_{\idx1}\psi_{\bd}\otimes U_{\idx2}\psi_{\be}
      =
      \denotee{u_{\bd1}Q_{\bd1}\quanteq u_{\be2}Q_{\be2}}
      {\memuni{m_{\bd}^0\circ\idx1^{-1}\, m_{\be}^0\circ\idx2^{-1}}}.
    \end{equation*}
    This implies
    \begin{align}
      &\rho_{\bc\bd}^0 :=
      \proj{U_{\idx1}\basis{\cl V}{m_{\bc}^0}}\otimes
      \proj{U_{\idx2}\basis{\cl V}{m_{\bd}^0}}\otimes
      U_{\idx1}\psi_{\bc}\otimes
      U_{\idx2}\psi_{\bd}
      \notag\\
      &\hskip2.7in\text{satisfies }
      (u_{\bc1}Q_{\bc1}\quanteq u_{\bd2}Q_{\bd2}) \label{def:rho.cd0}\\
      &\rho_{\bd\be}^0 :=
      \proj{U_{\idx1}\basis{\cl V}{m_{\bd}^0}}\otimes
      \proj{U_{\idx2}\basis{\cl V}{m_{\be}^0}}\otimes
      U_{\idx1}\psi_{\bd}\otimes
      U_{\idx2}\psi_{\be} \notag\\
      &\hskip2.7in\text{satisfies }
        \paren{u_{\bd1}Q_{\bd1}\quanteq u_{\be2}Q_{\be2}}
        \notag
    \end{align}

    Since $(m^0_{\bc},m^0_{\bd})\in R_{\bc\bd}^A$, we also have that $\rho_{\bc\bd}^0$ satisfies $\CL{a_{\bc\bd}}$. Thus $\rho_{\bc\bd}^0$ satisfies $A_{\bc\bd}$.
    Analogously $\rho_{\bd\be}^0$ satisfies $A_{\bd\be}$.

    Let $\rho_{\bd}':=\denotc\bd\pb\paren{\proj{\basis{\cl V}{m_{\bd}^0}\otimes\psi_{\bd}}}$.

    Since $\rho_{\bd\be}^0$
    satisfies $A_{\bc\bd}$,
    and since $\rhl{A_{\bc\bd}}\bc\bd{B_{\bc\bd}}$,
    there exists a separable $\rho_{\bc\bd}'\in\traceposcq{V_1V_2}$
    satisfying $B_{\bc\bd}$
    such that
    $\partr{V_1}{V_2}\rho'_{\bc\bd}=\denotc{\idx1\bc}\paren{\partr{V_1}{V_2}\rho_{\bc\bd}^0}$
    and
    $\partr{V_2}{V_1}\rho'_{\bc\bd}=\denotc{\idx2\bd}\paren{\partr{V_2}{V_1}\rho_{\bc\bd}^0}$.
    Furthermore, we have
    \begin{equation*}
      \calE_{\idx1}(\rho'_{\bc})
      \eqrefrel{def:rho'c}=
      \denotc{\idx1\bc}\paren{
        \proj{U_{\idx1}\basis{\cl V}{m_{\bc}^0}\otimes U_{\idx1}\psi_{\bc}}
      }
      \eqrefrel{def:rho.cd0}=
      \denotc{\idx1\bc}\paren{
        \partr{V_1}{V_2}\rho_{\bc\bd}^0
      }
      =
      \partr{V_1}{V_2}\rho'_{\bc\bd}
    \end{equation*}
    and
    \begin{equation*}
      \calE_{\idx2}(\rho'_{\bd})
      =
      \denotc{\idx2\bd}\paren{
        \proj{U_{\idx2}\basis{\cl V}{m_{\bd}^0}\otimes U_{\idx2}\psi_{\bd}}
      }
      \eqrefrel{def:rho.cd0}=
      \denotc{\idx2\bd}\paren{
        \partr{V_2}{V_1}\rho_{\bc\bd}^0
      }
      =
      \partr{V_2}{V_1}\rho'_{\bc\bd}.
    \end{equation*}
    That is,
    \begin{equation}\label{tr.rho'cd}
      \partr{V_1}{V_2}\rho'_{\bc\bd} = \calE_{\idx1}(\rho'_{\bc})
      \quad\text{and}\quad
      \partr{V_2}{V_1}\rho'_{\bc\bd} = \calE_{\idx2}(\rho'_{\bd}).
    \end{equation}
    Analogously,
    \begin{equation*}
      \partr{V_1}{V_2}\rho'_{\bd\be} = \calE_{\idx1}(\rho'_{\bd})
      \quad\text{and}\quad
      \partr{V_2}{V_1}\rho'_{\bd\be} = \calE_{\idx2}(\rho'_{\be}).
    \end{equation*}
    % Since $Q_{\bc}\supseteq\qu{\fv(\bc)}$, $\bc$ is $\cl VQ_{\bc}$-local. Since $\psi_{\bc}=\psi_{\bc Q}\otimes\psi_{\bc\bullet}$ with $\psi_{\bc\bullet}\in$

    Let
    $\rho''_{\bc\bd}:=\partr{\cl{V_1}\cl{V_2}Q_{\bc1}R_{\bd2}}{\qu{V_1}\qu{V_2}\setminus Q_{\bc1}R_{\bd2}}\rho'_{\bc\bd}$
    and
    $\rho''_{\bd\be}:=\partr{\cl{V_1}\cl{V_2}R_{\bd1}Q_{\be2}}{\qu{V_1}\qu{V_2}\setminus R_{\bd1}Q_{\be2}}\rho'_{\bd\be}$.
    These are separable because $\rho'_{cd}$ and $\rho'_{de}$ are separable.
    Since
    $B_{\bc\bd}=B_{\bc\bd}^Q\otimes\elltwov{\qu{V_1}\qu{V_2}\setminus
      Q_{\bc1}R_{\bd2}}$ and $\rho'_{\bc\bd}$
    satisfies $B_{\bc\bd}$,
    we have that $\rho''_{\bc\bd}$
    satisfies $B_{\bc\bd}^Q$.
    Analogously, $\rho''_{\bd\be}$ satisfies $B_{\bd\be}^Q$.
    Furthermore, 
    \begin{multline*}
      \partr{\cl{V_1}Q_{\bc1}}{\cl{V_2}R_{\bd2}}\rho''_{\bc\bd} =
      \partr{\cl{V_1}Q_{\bc1}}{\cl{V_2}R_{\bd2}} \partr{\cl{V_1}\cl{V_2}Q_{\bc1}R_{\bd2}}{\qu{V_1}\qu{V_2}\setminus Q_{\bc1}R_{\bd2}}\rho'_{\bc\bd}
      =
      \partr{\cl{V_1}Q_{\bc1}}{\qu{V_1}\setminus Q_{\bc1}}
      \partr{V_1}{V_2}\rho'_{\bc\bd}
      \\
      \eqrefrel{tr.rho'cd}=
      \partr{\cl{V_1}Q_{\bc1}}{\qu{V_1}\setminus Q_{\bc1}}
      \calE_{\idx1}(\rho'_{\bc})
      =
      \calE_{\idx1}\pb\paren{\partr{\cl{V}Q_{\bc}}{\qu{V}\setminus Q_{\bc}} \rho'_{\bc}}
      \eqrefrel{eq:rho'''}=
      \calE_{\idx1}\pb\paren{\rho''_{\bc}}.
    \end{multline*}
    and analogously,
    $\partr{\cl{V_2}Q_{\be2}}{\cl{V_1}R_{\bd1}}\rho''_{\bd\be} =
    \calE_{\idx2}(\rho''_{\be})$.  And if we define
    $\rho''_{\bd}:=\partr{\cl{V}R_{\bd}}{\qu{V}\setminus R_{\bd}} \rho'_{\bd}$,
    we also get analogously
    $\partr{\cl{V_2}R_{\bd2}}{\cl{V_1}Q_{\bc1}}\rho''_{\bc\bd} =
    \calE_{\idx2}(\rho''_{\bd}) $ and
    $\partr{\cl{V_1}R_{\bd1}}{\cl{V_2}Q_{\be2}}\rho''_{\bd\be} =
    \calE_{\idx1}(\rho''_{\bd})$.
  \end{claimproof}

  Let $\rho''_{\bc\bd}$
  and $\rho''_{\bd\be}$ be as in \autoref{claim:ex.rho.cd.de}. Then we can write
  \begin{equation}
    \rho_{\bc\bd}''=\sum_{m_{\bc}m_{\bd}}\mu_{\bc\bd}(m_{\bc},m_{\bd})\,
    \proj{U_{\idx 1}\basis{\cl{V}}{m_{\bc}}\otimes U_{\idx 2}\basis{\cl V}{m_{\bd}}}
    \otimes\hat\rho_{\bc\bd,m_{\bc}m_{\bd}}
    \label{eq:rho.cd.sum}
  \end{equation}  
  for some $\mu_{\bc\bd}\in\ellone{\typel{\cl{V}} \times \typel{\cl{V}}}$
  and some separable operators $\hat\rho_{\bc\bd,m_{\bc}m_{\bd}}\in\traceposcq{Q_{\bc1}R_{\bd2}}$
  with $\tr\hat\rho_{\bc\bd,m_{\bc}m_{\bd}}=1$.
  (Here $m_{\bc},m_{\bd}\in \types{\cl{V}}$.)
  And since $\rho_{\bc\bd}''$
  satisfies $B_{\bc\bd}^Q\subseteq (v_{\bc1}Q_{\bc1}\quanteq v_{\bd2}R_{\bd2}) $,
  it follows that
  \begin{equation}
    \suppo\hat\rho_{\bc\bd,m_{\bc}m_{\bd}} \subseteq(V_{\bc,m_{\bc}}Q_{\bc1}\quanteq
    V_{\bd,m_{\bd}}R_{\bd2})
    \label{eq:supp.hat.rho.cd}
  \end{equation}
  for all $(m_{\bc},m_{\bd})\in\suppd\mu_{\bc\bd}$
  where $V_{\bc,m_{\bc}}:=\denotee{v_{\bc}}{m_{\bc}}$
  and $V_{\bd,m_{\bd}}:=\denotee{v_{\bd}}{m_{\bd}}$.\pagelabel{page:vd.def}
  (Recall that $v_{\bc1}=\idx1 v_{\bc}$
  and $v_{\bd2}=\idx2v_{\bd}$.)
  And since  $\rho_{\bc\bd}''$
  satisfies $B_{\bc\bd}^Q\subseteq\CL{b_{\bc\bd}} $,
  we have
  $\mu_{\bc\bd}(m_{\bc},m_{\bd})\neq0$ only if
  $\denotee{b_{\bc\bd}}{m_{\bc}\circ\idx1^{-1}\,m_{\bd}\circ\idx2^{-1}}=\true$,
  i.e., we have $\suppd\mu_{\bc\bd}\subseteq  R^B_{\bc\bd}$.
  
  Analogously, we can write 
  \begin{equation}\label{eq:rho.de.sum}
    \rho_{\bd\be}''=\sum_{m_{\bd}m_{\be}}\mu_{\bd\be}(m_{\bd},m_{\be})
    \proj{U_{\idx 1}\,\
      basis{\cl{V}}{m_{\bd}}\otimes U_{\idx 2}\basis{\cl V}{m_{\be}}}\otimes\hat\rho_{\bd\be,m_{\bd}m_{\be}}
  \end{equation}
  for some $\mu_{\bd\be}\in\ellone{\typel{\cl{V}} \times \typel{\cl{V}}}$
  and some separable operators $\hat\rho_{\bd\be,m_{\bd}m_{\be}}\in\traceposcq{R_{\bd1}Q_{\be2}}$
  with $\tr\hat\rho_{\bd\be,m_{\bd}m_{\be}}=1$
  and
  \begin{equation}
    \suppo\hat\rho_{\bd\be,m_{\bd}m_{\be}} \subseteq(V_{\bd,m_{\bd}}R_{\bd1}\quanteq
    V_{\be,m_{\be}}Q_{\be2})
    \label{eq:supp.hat.rho.de}
  \end{equation}
  for all $(m_{\bd},m_{\be})\in\suppd\mu_{\bd\be}$.
  Here  $V_{\be,m_{\be}}:=\denotee{v_{\be}}{m_{\be}}$ and $V_{\bd,m_{\be}}$ as above.
  And $\suppd\mu_{\bd\be}\subseteq  R^B_{\bd\be}$.

  \begin{claim}\label{claim:mu.cde}
    There is a distribution $\mu_{\bc\bd\be}\in\ellone{\typel{\cl{V}}\times\typel{\cl{V}}\times\typel{\cl{V}}}$ such that
    $\marginal{1,2}{\mu_{\bc\bd\be}}=\mu_{\bc\bd}$
    and $\marginal{2,3}{\mu_{\bc\bd\be}}=\mu_{\bd\be}$.

    Here $\marginal{1,2}{\mu_{\bc\bd\be}}$
    is the marginal distribution that projects onto the first two
    components, i.e., the result of summing $\mu_{\bc\bd\be}$ over $m_{\be}$.
    $\marginal{2,3}{\mu_{\bc\bd\be}}$ analogously.
  \end{claim}

  \begin{claimproof}
    We first show
    \begin{equation}
      \label{eq:same.marg}
      \marginal2{\mu_{\bc\bd}} = \marginal1{\mu_{\bd\be}} =: \mu_{\bd}.
    \end{equation}
    To show this, we calculate:
    \begin{align*}
      \hskip2cm & \hskip-2cm
                  \calE_{\idx2}\pB\paren{\pb\qlift{\marginal2{\mu_{\bc\bd}}}} \\
      &\starrel=
        \calE_{\idx2}\pB\paren{
        \sum_{m_{\bc},m_{\bd}}
        \mu_{\bc\bd}(m_{\bc},m_{\bd})\,
        \proj{\basis{\cl V}{m_{\bd}}}
        } \\
      &=
        \sum_{m_{\bc},m_{\bd}}
        \mu_{\bc\bd}(m_{\bc},m_{\bd})\,
        \pb\proj{U_{\idx2}\basis{\cl V}{m_{\bd}}} \\
      &\starstarrel=
        \partr{\cl{V_2}}{\cl{V_1}Q_{\bc1}R_{\bd2}}
        \sum_{m_{\bc},m_{\bd}}
        \mu_{\bc\bd}(m_{\bc},m_{\bd})\,
        \pb\proj{U_{\idx1}\basis{\cl V}{m_{\bc}}\otimes U_{\idx2}\basis{\cl V}{m_{\bd}}}
        \otimes
        \hat\rho_{\bc\bd,m_{\bc}m_{\bd}} \\
      &\eqrefrel{eq:rho.cd.sum}=
        \partr{\cl{V_2}}{\cl{V_1}Q_{\bc1}R_{\bd2}} \
        \rho_{\bc\bd}''
      =
        \partr{\cl{V_2}}{R_{\bd2}} \
        \partr{\cl{V_2}R_{\bd2}}{\cl{V_1}Q_{\bc1}} \
        \rho_{\bc\bd}'' 
      \tristarrel=
        \partr{\cl{V_2}}{R_{\bd2}} \
        \calE_{\idx2} ( \rho''_{\bd} )
      =
      \calE_{\idx2} \pb\paren{
        \partr{\cl{V}}{R_{\bd}}\,
      \rho''_{\bd} }.
    \end{align*}
    Here $(*)$
    is by definition of $\marginal1{\cdot}$
    and $\qlift{\cdot}$.
    (Recall: $\qlift{\cdot}$
    maps a distribution to a diagonal operator and was defined on
    \autopageref{page:def:qlift}.)  And $(**)$
    follows since $\tr\hat\rho_{\bc\bd,m_{\bc}m_{\bd}}=1$.
    And $(*\mathord**)$
    follows from the properties of $\rho''_{\bc\bd}$
    stated in \autoref{claim:ex.rho.cd.de}.

    Since $\calE_{\idx2}$ is injective, this implies
    $\pb\qlift{\marginal2{\mu_{\bc\bd}}} =
    \partr{\cl{V}}{R_{\bd}}\, \rho''_{\bd}$.
    Analogously, 
    we show
    $\pb\qlift{\marginal1{\mu_{\bd\be}}} =
    \partr{\cl{V}}{R_{\bd}}\, \rho''_{\bd}$.
    Thus $\pb\qlift{\marginal2{\mu_{\bc\bd}}} =
    \pb\qlift{\marginal1{\mu_{\bd\be}}}$. Since $\qlift{\cdot}$ is injective,
    \eqref{eq:same.marg} follows.

    \medskip

    Let
    \begin{equation*}
      \mu_{\bc\bd\be}(m_{\bc},m_{\bd},m_{\be}) :=\mu_{\bc\bd}(m_{\bc},m_{\bd})\mu_{\bd\be}(m_{\bd},m_{\be})/\mu_{\bd}(m_{\bd})
%      \mu_{\bc\be} &:= \marginal{1,3}{\mu_{\bc\bd\be}}.
    \end{equation*}
    
    It is then easy to see (using $\eqref{eq:same.marg}$)
    that $\marginal{1,2}{\mu_{\bc\bd\be}}=\mu_{\bc\bd}$
    and $\marginal{2,3}{\mu_{\bc\bd\be}}=\mu_{\bd\be}$.
  \end{claimproof}

  \begin{claim}\label{claim:fc.fe}
    There are functions $f_{\bc},f_{\be}$ such that $(m_{\bc},m_{\bd})\in\suppd\mu_{\bc\bd}\implies m_{\bc}=f_{\bc}(m_{\bd})$
    and  $(m_{\bd},m_{\be})\in\suppd\mu_{\bd\be}\implies m_{\be}=f_{\be}(m_{\bd})$.
  \end{claim}

  \begin{claimproof}
    By assumption (of \ruleref{Trans}), we have
    that $b_{\bc\bd} \implies (\xx^{(1)}_1,\dots,\xx^{(n)}_1)=\idx2 e_{\bc}$ holds.
    for some expression $e_{\bc}$ and $\{\xx^{(1)},\dots,\xx^{(n)}\}=\cl{\fv(\bc)}$.
    For any $(m_{\bc},m_{\bd})\in R_{\bc\bd}^B$ we have
    $\denotee{b_{\bc\bd}}{m_{\bc}\circ\idx1^{-1}\,m_{\bd}\circ\idx2^{-1}}=\true$ (by definition of $R_{\bc\bd}^B$).
    Thus for  $(m_{\bc},m_{\bd})\in R_{\bc\bd}^B$, we have
    \begin{multline*}
      \pb\paren{m_{\bc}(\xx^{(1)}),\dots,m_{\bc}(\xx^{(1)})}
      =
      \denotee{(\xx^{(1)}_1,\dots,\xx^{(n)}_1)}{m_{\bc}\circ\idx1^{-1}\,m_{\bd}\circ\idx2^{-1}}
      \\
      \starrel=
      \denotee{\idx2e_{\bc}}{m_{\bc}\circ\idx1^{-1}\,m_{\bd}\circ\idx2^{-1}}
      =
      \denotee{e_{\bc}}{m_{\bd}} =: f_{\bc}'(m_{\bd})
    \end{multline*}
    Here $(*)$
    follows since
    $b_{\bc\bd} \implies (\xx^{(1)}_1,\dots,\xx^{(n)}_1)=\idx2 e_{\bc}$
    holds.  Since $\{\xx^{(1)},\dots,\xx^{(n)}\}=\cl{\fv(\bc)}$,
    there is a function $f_{\bc}''$
    such that
    $f_{\bc}''\pb\paren{m_{\bc}(\xx^{(1)}),\dots,m_{\bc}(\xx^{(1)})}=\restrict{m_{\bc}}{\cl{\fv(\bc)}}$.
    Then $f_{\bc}''\circ f_{\bc}'(m_{\bd})=\restrict{m_{\bc}}{\cl{\fv(\bc)}}$ for any $(m_{\bc},m_{\bd})\in R^B_{\bc\bd}$.
    Since $\suppd\mu_{\bc\bd}\subseteq R_{\bc\bd}^B$ (discussion after \autoref{claim:ex.rho.cd.de}),
    we get:
    \begin{equation}
      \label{eq:fc''fc'}
      (m_{\bc},m_{\bd})\in \suppd\mu_{\bc\bd}
      \quad\text{implies}\quad
      f_{\bc}''\circ f_{\bc}'(m_{\bd})=\restrict{m_{\bc}}{\cl{\fv(\bc)}}.
    \end{equation}

    Let $m':=\restrict{m_{\bc}^0}{\cl V\setminus \cl{\fv(\bc)}}$.
    By definition of $\rho''_{\bc}$
    in \eqref{eq:def:rho''c} and the fact that $\bc$
    is $\cl{\fv(\bc)}$-local,
    we have that $\rho''_{\bc}$
    is of the form
    $\sum_{m_{\bc}}\lambda_{m_{\bc}}\proj{\basis{V}{m_{\bc}}}\cdot\rho_{m_{\bc}}$
    for some $\lambda_{m_{\bc}}\geq0$ and positive operators
    $\rho_{m_{\bc}}$
    where the sum ranges only over $m_{\bc}$
    with $\restrict{m_{\bc}}{\cl V\setminus \cl{\fv(\bc)}}=m'$.
    Since
    $\partr{\cl{V_1}Q_{\bc1}}{\cl{V_2}R_{\bd2}}\rho''_{\bc\bd} =
    \calE_{\idx1}(\rho''_{\bc})$ (\autoref{claim:ex.rho.cd.de}), this
    implies that the sum in \eqref{eq:rho.cd.sum} ranges only over
    $m_{\bc}$ with
    $\restrict{m_{\bc}}{\cl V\setminus \cl{\fv(\bc)}}=m'$. In other words,
    $\mu_{\bc\bd}(m_{\bc},m_{\bd})=0$ unless 
    $\restrict{m_{\bc}}{\cl V\setminus \cl{\fv(\bc)}}=m'$.
    Thus
    \begin{equation}
      \label{eq:m'}
      (m_{\bc},m_{\bd})\in\suppd\mu_{\bc\bd}
      \quad\text{implies}\quad
      \restrict{m_{\bc}}{\cl V\setminus \cl{\fv(\bc)}}=m'.
    \end{equation}
    
    Hence $(m_{\bc},m_{\bd})\in\suppd\mu_{\bc\bd}$ implies
    \begin{equation*}
      f_{\bc}(m_{\bd}) := m' \cup \pb\paren{f''_{\bc}\circ f'_{\bc}(m_{\bd})}
      \ \ \txtrel{\eqref{eq:fc''fc'},\eqref{eq:m'}}=\ \ 
      \restrict{m_{\bc}}{\cl{\fv(\bc)}} \cup
      \restrict{m_{\bc}}{\cl V\setminus \cl{\fv(\bc)}}
      =
      m_{\bc}.
    \end{equation*}
    This shows the existence of $f_{\bc}$ as required by the current claim.

    The existence of $f_{\be}$ is shown analogously.
  \end{claimproof}
  
  \begin{claim}\label{claim:hat.rho.d}
    Let $\mu_{\bc\bd\be}$
    be as in \autoref{claim:mu.cde}.  For every
    $(m_{\bc},m_{\bd},m_{\be})\in\suppd\mu_{\bc\bd\be}$,
    there is a $\hat\rho_{\bd,m_{\bc}m_{\bd}m_{\be}}\in\traceposcq{R_{\bd}}$
    such that
    $\partr{R_{\bd2}}{Q_{\bc1}}\hat\rho_{\bc\bd,m_{\bc}m_{\bd}}=\calE_{\idx2}(\hat\rho_{\bd,m_{\bc}m_{\bd}m_{\be}})$
    and
    $\partr{R_{\bd1}}{Q_{\be2}}\hat\rho_{\bd\be,m_{\bd}m_{\be}}=\calE_{\idx1}(\hat\rho_{\bd,m_{\bc}m_{\bd}m_{\be}})$.
  \end{claim}

  \begin{claimproof}
    Let $f_{\bc},f_{\be}$ be as in \autoref{claim:fc.fe}.
    Then
    \begin{multline*}
      \calE_{\idx2}(\rho''_{\bd})
      \ \ \txtrel{\autoref{claim:ex.rho.cd.de}}=\ \
      \partr{\cl{V_2}R_{\bd2}}{\cl{V_1}Q_{\bc1}} \rho_{\bc\bd}''
      \eqrefrel{eq:rho.cd.sum}=
        \sum_{m_{\bc}m_{\bd}}\mu_{\bc\bd}(m_{\bc},m_{\bd})\,
        \proj{U_{\idx 2}\basis{\cl V}{m_{\bd}}}
        \otimes \partr{R_{\bd2}}{Q_{\bc1}}\hat\rho_{\bc\bd,m_{\bc}m_{\bd}}
        \\
      \txtrel{\autoref{claim:fc.fe}}=\ \ 
        \sum_{m_{\bd}}\mu_{\bc\bd}(f_{\bc}(m_{\bd}),m_{\bd})\,
        \proj{U_{\idx 2}\basis{\cl V}{m_{\bd}}}
        \otimes \partr{R_{\bd2}}{Q_{\bc1}}\hat\rho_{\bc\bd,f_{\bc}(m_{\bd})m_{\bd}}
    \end{multline*}
    and analogously (using \eqref{eq:rho.de.sum} instead of \eqref{eq:rho.cd.sum}),
    \begin{align*}
      \calE_{\idx1}(\rho''_{\bd})
      =
      \sum_{m_{\bd}}\mu_{\bd\be}(m_{\bd},f_{\be}(m_{\bd}))\,
      \proj{U_{\idx 1}\basis{\cl V}{m_{\bd}}}
      \otimes \partr{R_{\bd1}}{Q_{\be2}}\hat\rho_{\bd\be,m_{\bd}f_{\be}(m_{\bd})}.
    \end{align*}
    Hence
    \begin{multline*}
      \sum_{m_{\bd}}\mu_{\bc\bd}(f_{\bc}(m_{\bd}),m_{\bd})\,
      \proj{\basis{\cl V}{m_{\bd}}}
      \otimes \calE_{\idx2}^{-1}(\partr{R_{\bd2}}{Q_{\bc1}}\hat\rho_{\bc\bd,f_{\bc}(m_{\bd})m_{\bd}})
      =
      \rho_{\bd}'' \\
      =
      \sum_{m_{\bd}}\mu_{\bd\be}(m_{\bd},f_{\be}(m_{\bd}))\,
      \proj{\basis{\cl V}{m_{\bd}}}
      \otimes \calE_{\idx1}^{-1}(\partr{R_{\bd1}}{Q_{\be2}}\hat\rho_{\bd\be,m_{\bd}f_{\be}(m_{\bd})}).
    \end{multline*}
    It follows that, whenever $\mu_{\bc\bd}(f_{\bc}(m_{\bd}),m_{\bd})\neq0$, we have 
    \begin{equation}\label{eq:rhocd=rhode}
      \hat\rho_{\bd,m_{\bc}m_{\bd}m_{\be}}
      :=
      \calE_{\idx2}^{-1}(\partr{R_{\bd2}}{Q_{\bc1}}\hat\rho_{\bc\bd,f_{\bc}(m_{\bd})m_{\bd}})
      =
      \calE_{\idx1}^{-1}(\partr{R_{\bd1}}{Q_{\be2}}\hat\rho_{\bd\be,m_{\bd}f_{\be}(m_{\bd})}).
    \end{equation}
    If $(m_{\bc},m_{\bd},m_{\be})\in\suppd\mu_{\bc\bd\be}$,
    by \autoref{claim:mu.cde}, $(m_{\bc},m_{\bd})\in\suppd\mu_{\bc\bd}$,
    and thus, by \autoref{claim:fc.fe},
    $(f_{\bc}(m_{\bd}),m_{\bd})\in\suppd\mu_{\bc\bd}$.
    Thus \eqref{eq:rhocd=rhode} holds whenever
    $(m_{\bc},m_{\bd},m_{\be})\in\suppd\mu_{\bc\bd\be}$. 

    Hence for every $(m_{\bc},m_{\bd},m_{\be})\in\suppd\mu_{\bc\bd\be}$,
    by \eqref{eq:rhocd=rhode}, we have that
    \begin{equation}
      \partr{R_{\bd2}}{Q_{\bc1}}\hat\rho_{\bc\bd,f_{\bc}(m_{\bd})m_{\bd}} = \calE_{\idx2}(\hat\rho_{\bd,m_{\bc}m_{\bd}m_{\be}})
      \quad\text{and}\quad
      \partr{R_{\bd1}}{Q_{\be2}}\hat\rho_{\bd\be,m_{\bd}f_{\be}(m_{\bd})} = \calE_{\idx1}(\hat\rho_{\bd,m_{\bc}m_{\bd}m_{\be}})
    \end{equation}
    and $m_{\bc}=f_{\bc}(m_{\bd})$,
    $m_{\be}=f_{\be}(m_{\be})$,
    which implies the properties required in the current claim.
  \end{claimproof}

  \begin{claim}\label{claim:join.rho.quant}
    Fix $V_{\bc}\in\iso{\typel{Q_{\bc}},Z}$, $V_{\bd}\in\boundedleq{\typel{R_{\bd}},Z}$,
    $V_{\be}\in\iso{\typel{Q_{\be}},Z}$.
    Fix a separable $\hat\rho_{\bc\bd}\in\traceposcq{Q_{\bc1}R_{\bd2}}$ and a separable
    $\hat\rho_{\bd\be}\in\traceposcq{R_{\bd1}Q_{\be2}}$ and $\hat\rho_{\bd}\in\traceposcq{R_{\bd}}$ such that
    $\suppo\hat\rho_{\bc\bd}\subseteq (V_{\bc}Q_{\bc1}\quanteq V_{\bd}R_{\bd2})$ and
    $\suppo\hat\rho_{\bd\be}\subseteq (V_{\bd}R_{\bd1}\quanteq V_{\be}Q_{\be2})$ and
    $\partr{R_{\bd2}}{Q_{\bc1}}\hat\rho_{\bc\bd}=
    \calE_{\idx2}(\hat\rho_{\bd})$ and
    $\partr{R_{\bd1}}{Q_{\be2}}\hat\rho_{\bd\be}=
    \calE_{\idx1}(\hat\rho_{\bd})$.
    
    Then there exists a separable $\hat\rho_{\bc\be}\in\traceposcq{Q_{\bc1}Q_{\be2}}$ such that
    \begin{equation*}
      \suppo\hat\rho_{\bc\be}\subseteq (V_{\bc}Q_{\bc1}\quanteq V_{\be}Q_{\be2}), \quad
      \partr{Q_{\bc1}}{Q_{\be2}}\hat\rho_{\bc\be}=\partr{Q_{\bc1}}{R_{\bd2}}\hat\rho_{\bc\bd}, \quad
      \partr{Q_{\be2}}{Q_{\bc1}}\hat\rho_{\bc\be}=\partr{Q_{\be2}}{R_{\bd1}}\hat\rho_{\bd\be}.
    \end{equation*}
  \end{claim}

  \begin{claimproof}
    Since $\hat\rho_{\bc\bd}$ is separable, we can write it as:
    \begin{equation}\label{eq:hat.rho.ce}
      \hat\rho_{\bc\bd} = \sum_i \lambda_i\, \proj{\psi_{\bc i}\otimes\psi_{\bd i}}
    \end{equation}
    with $\lambda_i\geq0$,
    $\sum\lambda_i<\infty$,
    and unit vectors $\psi_{\bc i}\in \elltwov{Q_{\bc1}}$
    and $\psi_{\bd i}\in\elltwov{R_{\bd2}}$.
    Since
    $\suppo\hat\rho_{\bc\bd}\subseteq (V_{\bc}Q_{\bc1}\quanteq V_{\bd}R_{\bd2})$,
    it follows that
    $\psi_{\bc i}\otimes\psi_{\bd i}\in (V_{\bc}Q_{\bc1}\quanteq V_{\bd}R_{\bd2})$.
    By \autoref{lemma:quanteq} (in the special case
    $\qu{V_1},Q_1:=Q_{\bc1}$
    and $\qu{V_2},Q_2:=R_{\bd2}$),
    this implies that
    \begin{equation}
      V_{\bc}\adj{\Uvarnames{Q_{\bc1}}}\psi_{\bc i} = V_{\bd}\adj{\Uvarnames{R_{\bd2}}}\psi_{\bd i}.
      \label{eq:vc.vd}
    \end{equation}
    Then
    \begin{multline}
      V_{\bc}\adj{\Uvarnames{Q_{\bc1}}}\,
      (\partr{Q_{\bc1}}{R_{\bd2}}\hat\rho_{\bc\bd})\,
      {\Uvarnames{Q_{\bc1}}}\adj{V_{\bc}}
      \eqrefrel{eq:hat.rho.ce}=
      \sum_i\lambda_i\,\proj{V_{\bc}\adj{\Uvarnames{Q_{\bc1}}}\psi_{\bc i}}
      \eqrefrel{eq:vc.vd}=
      \sum_i\lambda_i\,\proj{V_{\bd}\adj{\Uvarnames{R_{\bd2}}}\psi_{\bd i}}
      \\
      \eqrefrel{eq:hat.rho.ce}=
      V_{\bd}\adj{\Uvarnames{R_{\bd2}}}\,
      (\partr{R_{\bd2}}{Q_{\bc1}}\hat\rho_{\bc\bd})\,
      {\Uvarnames{R_{\bd2}}}\adj{V_{\bd}}
      =
      V_{\bd}\adj{\Uvarnames{R_{\bd2}}}\,
      \calE_{\idx2}(\hat\rho_{\bd})\,
      {\Uvarnames{R_{\bd2}}}\adj{V_{\bd}}.
      \label{eq:VcVd}
    \end{multline}
    Analogously,
    \begin{equation}
      V_{\be}\adj{\Uvarnames{Q_{\be2}}}\,
      (\partr{Q_{\be2}}{R_{\bd1}}\hat\rho_{\bd\be})\,
      {\Uvarnames{Q_{\be2}}}\adj{V_{\be}}
      =
      V_{\bd}\adj{\Uvarnames{R_{\bd1}}}\,
      \calE_{\idx1}(\hat\rho_{\bd})\,
      {\Uvarnames{R_{\bd1}}}\adj{V_{\bd}}.
      \label{eq:VeVd}
    \end{equation}
    % By definition of
    % $\calE_{\idx1},\calE_{\idx2},\Uvarnames{\dots},R_{\bd1},R_{\bd2}$,
    We have that
    \begin{equation*}
      \adj{\Uvarnames{R_{\bd2}}}\, \calE_{\idx2}(\hat\rho_{\bd})\,
      {\Uvarnames{R_{\bd2}}}
      =
       \adj{\Uvarnames{R_{\bd}}}\, \hat\rho_{\bd}\,
      {\Uvarnames{R_{\bd}}}
      =
      \adj{\Uvarnames{R_{\bd1}}}\,
      \calE_{\idx1}(\hat\rho_{\bd})\, {\Uvarnames{R_{\bd1}}}.
    \end{equation*}
  Thus the rhs of
    \eqref{eq:VcVd} and \eqref{eq:VeVd} are equal, hence
    \begin{equation}
      \label{eq:sigma}
      V_{\bc}\adj{\Uvarnames{Q_{\bc1}}}\,
      (\partr{Q_{\bc1}}{R_{\bd2}}\hat\rho_{\bc\bd})\,
      {\Uvarnames{Q_{\bc1}}}\adj{V_{\bc}}
      =
      V_{\be}\adj{\Uvarnames{Q_{\be2}}}\,
      (\partr{Q_{\be2}}{R_{\bd1}}\hat\rho_{\bd\be})\,
      {\Uvarnames{Q_{\be2}}}\adj{V_{\be}}
      =:
      \sigma \in \tracepos{Z}.
    \end{equation}
    Since $\sigma\in\tracepos Z$,
    we can write $\sigma=\sum_{i\in I}\gamma(i)\,\proj{\psi_i}$
    for some set $I$,
    some $\gamma\in\ellone I$
    and normalized $\psi_i\in\elltwo Z$.
    From \eqref{eq:sigma}, we get that $\suppo\sigma\subseteq\im V_{\bc}$,
    hence $\psi_i\in\im V_{\bc}$. Analogously $\psi_i\in\im V_{\be}$.

    Let
    \begin{equation}
      \hat\rho_{\bc\be} := \sum_i\gamma(i)\,\proj{
        \underbrace{
          \Uvarnames{Q_{\bc1}}\adj{V_{\bc}}\psi_i\otimes
          \Uvarnames{Q_{\be2}}\adj{V_{\be}}\psi_i
        }_{{}=:\ \phi_i \ 
        \in \ \elltwov{Q_{\bc1}Q_{\be2}}}
      }.
      \label{eq:def:rho.ce}
    \end{equation}
    Since $\gamma\in\ellone I$, and
    $\psi_i$ are
    normalized, and $V_{\bc},V_{\be}$ are isometries,
    this sum converges and $\hat\rho_{\bc\be}$
    exists. $\hat\rho_{\bc\be}$ is separable by construction.

    We now show that $\phi_i\in(V_{\bc}Q_{\bc1}\quanteq V_{\be}Q_{\be2})$.
    By \autoref{def:quanteq}, this means that $\phi_i$
    is fixed by
    $U:=\Uvarnames{Q_{\be2}}\adj{V_{\be}}V_{\bc}\adj{\Uvarnames{Q_{\bc1}}}\otimes
    \Uvarnames{Q_{\bc1}}\adj{V_{\bc}}V_{\be}\adj{\Uvarnames{Q_{\be2}}}$.
    (The third factor from \autoref{def:quanteq} vanishes because
    $\phi_i\in\elltwov{Q_{\bc1}Q_{\be2}}$,
    and hence the third factor is the identity on
    $\elltwov{Q_{\bc1}Q_{\be2}\setminus
      Q_{\bc1}Q_{\be2}}=\elltwov\varnothing=\setC$.) We calculate:
    \begin{align*}
      U\phi_i
      &=
        \Uvarnames{Q_{\be2}}\adj{V_{\be}}V_{\bc}\adj{\Uvarnames{Q_{\bc1}}}
        \Uvarnames{Q_{\bc1}}\adj{V_{\bc}}\psi_i
        \otimes
        \Uvarnames{Q_{\bc1}}\adj{V_{\bc}}V_{\be}\adj{\Uvarnames{Q_{\be2}}}
        \Uvarnames{Q_{\be2}}\adj{V_{\be}}\psi_i
      \\
      &=
        \Uvarnames{Q_{\be2}}\adj{V_{\be}}V_{\bc}\adj{V_{\bc}}\psi_i
        \otimes
        \Uvarnames{Q_{\bc1}}\adj{V_{\bc}}V_{\be}\adj{V_{\be}}\psi_i
      \\
      &\starrel=
        \Uvarnames{Q_{\be2}}\adj{V_{\be}}\psi_i
        \otimes
        \Uvarnames{Q_{\bc1}}\adj{V_{\bc}}\psi_i = \phi_i.
    \end{align*}
    Here $(*)$
    uses that $\psi_i\in\im V_{\bc}$
    and thus $V_{\bc}\adj{V_{\bc}}\psi_i=\psi_i$
    (since $V_{\bc}$
    is an isometry), and analogously $V_{\be}\adj{V_{\be}}\psi_i=\psi_i$.

    Thus $\phi_i$
    is invariant under $U$,
    thus $\phi_i\in(V_{\bc}Q_{\bc1}\quanteq V_{\be}Q_{\be2})$,
    and thus
    $\suppo\hat\rho_{\bc\be}\subseteq(V_{\bc}Q_{\bc1}\quanteq V_{\be}Q_{\be2})$.
    This is the first property that we needed to prove about
    $\hat\rho_{\bc\be}$.

    The second required property ($\partr{Q_{\bc1}}{Q_{\be2}}\hat\rho_{\bc\be}=\partr{Q_{\bc1}}{R_{\bd2}}\hat\rho_{\bc\bd}$) is shown as follows:
    \begin{align*}
      \partr{Q_{\bc1}}{Q_{\be2}}\hat\rho_{\bc\be}
      &\starrel=
        \sum\nolimits_i\gamma(i)\,\proj{
        \Uvarnames{Q_{\bc1}}\adj{V_{\bc}}\psi_i
        }
        = \Uvarnames{Q_{\bc1}}\adj{V_{\bc}}\sigma V_{\bc}\adj{\Uvarnames{Q_{\bc1}}}
      \\
      &\eqrefrel{eq:sigma}=
        \Uvarnames{Q_{\bc1}}\adj{V_{\bc}}V_{\bc}\adj{\Uvarnames{Q_{\bc1}}}\,
        (\partr{Q_{\bc1}}{R_{\bd2}}\hat\rho_{\bc\bd})\,
        {\Uvarnames{Q_{\bc1}}}\adj{V_{\bc}}V_{\bc}\adj{\Uvarnames{Q_{\bc1}}}
      \starstarrel=
        \partr{Q_{\bc1}}{R_{\bd2}}\hat\rho_{\bc\bd}
    \end{align*}
    Here $(*)$ uses  the definition \eqref{eq:def:rho.ce} of $\hat\rho_{\bc\be}$, and the fact that 
    $\norm{\Uvarnames{Q_{\be2}}\adj{V_{\be}}\psi_i}=1$ since $V_{\be}$ is an isometry and $\psi_i\in\im V_{\be}$.
    And $(**)$ uses that $V_{\bc}$ is an isometry and thus $\adj{V_{\bc}}V_{\bc}=\id$.

    The last required property, $\partr{Q_{\be2}}{Q_{\bc1}}\hat\rho_{\bc\be}=\partr{Q_{\be2}}{R_{\bd1}}\hat\rho_{\bd\be}$, is shown analogously.
  \end{claimproof}

  \begin{claim}\label{claim:rho.ce''}
    There is a separable $\rho_{\bc\be}''\in \traceposcq{\cl{V_1}\cl{V_2}Q_{\bc1}Q_{\be2}}$
    such that $\rho''_{\bc\be}$ satisfies $B_{\bc\be}^Q$ and
    $\partr{\cl{V_1}{Q_{\bc1}}}{\cl{V_2}{Q_{\be2}}}\rho_{\bc\be}''=
    \partr{\cl{V_1}{Q_{\bc1}}}{\cl{V_2}{R_{\bd2}}}\rho_{\bc\bd}''$
    and
    $\partr{\cl{V_2}{Q_{\be2}}}{\cl{V_1}{Q_{\bc1}}}\rho_{\bc\be}''=
    \partr{\cl{V_2}{Q_{\be2}}}{\cl{V_1}{R_{\bd1}}}\rho_{\bd\be}''$.
  \end{claim}

\begin{claimproof}
    % From \autoref{claim:ex.rho.cd.de}, we get
    % $\rho''_{\bc\bd}\in\traceposcq{\cl{V_1}\cl{V_2}Q_{\bc1}R_{\bd2}}$
    % and $\rho''_{\bd\be}\in\traceposcq{\cl{V_1}\cl{V_2}R_{\bd1}Q_{\be2}}$
    % with the properties listed in \autoref{claim:ex.rho.cd.de}.
%
    % Let $T:=\tr\rho_{\bc\bd}''=\tr\rho_{\bd\be}''$.
    % (The traces are equal since
    % $\tr\rho_{\bc\bd}''=\tr\partr{V_2}{V_1}\rho_{\bc\bd}''$
    % $\tr\rho_{\bd\be}''=\tr\partr{V_1}{V_2}\rho_{\bd\be}''$
    % are equal from the properties listed in
    % \autoref{claim:ex.rho.cd.de}.)
%
    %
    % And furthermore, if $(m_{\bc},m_{\bd},m_{\be})\in\suppd\mu_{\bc\bd\be}$,
    % then $(m_{\bc},m_{\bd})\in\suppd\mu_{\bc\bd}\subseteq R_{\bc\bd}^B$
    % and $(m_{\bd},m_{\be})\in\suppd\mu_{\bd\be}\subseteq R_{\bd\be}^B$,
    % hence $(m_{\bc},m_{\bd})\in R_{\bc\bd}^B\circrel R_{\bd\be}^B$.
    % Thus $\suppd\mu_{\bc\be}\subseteq R_{\bc\bd}^B\circrel R_{\bd\be}^B$.
%
 %   
  %  
    For each $m_{\bc},m_{\bd},m_{\be}$,
    we apply \autoref{claim:join.rho.quant} to $\hat\rho_{\bc\bd,m_{\bc}m_{\bd}}$
    and $\hat\rho_{\bd\be,m_{\bd}m_{\be}}$.
    That is, we invoke \autoref{claim:join.rho.quant} with $V_{\bc}:=V_{\bc,m_{\bc}}$,
    $V_{\bd}:=V_{\bd,m_{\bd}}$,
    $V_{\be}:=V_{\be,m_{\be}}$
    (as defined on \autopageref{page:vd.def}),
    $\hat\rho_{\bc\bd}:=\hat\rho_{\bc\bd,m_{\bc}m_{\bd}}$,
    $\hat\rho_{\bd\be}:=\hat\rho_{\bd\be,m_{\bd}m_{\be}}$
    (defined in \eqref{eq:rho.cd.sum},\eqref{eq:rho.de.sum}),
    $\hat\rho_{\bd}:=\hat\rho_{\bd,m_{\bc}m_{\bd}m_{\be}}$ (from \autoref{claim:hat.rho.d}).
    The assumptions of  \autoref{claim:join.rho.quant}
    where shown in \eqref{eq:supp.hat.rho.cd}, \eqref{eq:supp.hat.rho.de}, and \autoref{claim:hat.rho.d}.
    Then by  \autoref{claim:join.rho.quant}, there exists a separable
    $\hat\rho_{\bc\bd,m_{\bc}m_{\bd}m_{\be}}\in\traceposcq{Q_{\bc1}Q_{\be2}}$ such that
    \begin{gather}
      \suppo\hat\rho_{\bc\be,m_{\bc}m_{\bd}m_{\be}}\subseteq (V_{\bc,m_{\bc}}Q_{\bc1}\quanteq V_{\be,m_{\be}}Q_{\be2}), \label{eq:rho.ce.qeq} \\
      \partr{Q_{\bc1}}{Q_{\be2}}\suppo\hat\rho_{\bc\be,m_{\bc}m_{\bd}m_{\be}}=\partr{Q_{\bc1}}{R_{\bd2}}\suppo\hat\rho_{\bc\bd,m_{\bc}m_{\bd}},
      \label{eq:rho.ce.c}
      \\
      \partr{Q_{\be2}}{Q_{\bc1}}\suppo\hat\rho_{\bc\be,m_{\bc}m_{\bd}m_{\be}}=\partr{Q_{\be2}}{R_{\bd1}}\suppo\hat\rho_{\bd\be,m_{\bd}m_{\be}}.
      \notag % \label{eq:rho.ce.e}
    \end{gather}
    (For all $m_{\bc},m_{\bd},m_{\be}$.)

    Since $\hat\rho_{\bc\bd,m_{\bc}m_{\bd}}$ has trace $1$, it follows that
    $\tr\hat\rho_{\bc\be,m_{\bc}m_{\bd}m_{\be}}=1$ for all $m_{\bc},m_{\bd},m_{\be}$.

    Let $\mu_{\bc\bd\be}$ be as in \autoref{claim:mu.cde}.
    Since  $\sum\mu_{\bc\bd\be}(m_{\bc},m_{\bd},m_{\be})$ is finite,
    \begin{equation}\label{eq:rho.ce''}
      \rho_{\bc\be}'' := \sum_{m_{\bc},m_{\bd},m_{\be}} \mu_{\bc\bd\be}(m_{\bc},m_{\bd},m_{\be})\,
      \proj{U_{\idx1}\basis {\cl V}{m_{\bc}}\otimes U_{\idx2}\basis{\cl V}{m_{\be}}}
      \otimes
      \hat\rho_{\bc\be,m_{\bc}m_{\bd}m_{\be}}
    \end{equation}
    exists. $\rho''_{\bc\be}$ is separable since $\hat\rho_{\bc\be,m_{\bc}m_{\bd}m_{\be}}$ are.

    If $(m_{\bc},m_{\bd},m_{\be})\in\suppd\mu_{\bc\bd\be}$, then
    $(m_{\bc},m_{\bd})\in\suppd\mu_{\bc\bd}\subseteq R_{\bc\bd}^B$
    and $(m_{\bd},m_{\be})\in\suppd\mu_{\bd\be}\subseteq R_{\bd\be}^B$,
    thus
    $(m_{\bc},m_{\be})\in R_{\bc\bd}^B\circrel R_{\bd\be}^B$, hence
    $\denotee{b_{\bc\bd}\circexp b_{\bd\be}}{m_{\bc}\circ\idx1^{-1}\, m_{\be}\circ\idx2^{-1}}=\true$
    by \eqref{eq:RBce}.
    This implies that $\rho''_{\bc\be}$ satisfies $\CL{b_{\bc\bd}\circexp b_{\bd\be}}$.

    By \eqref{eq:rho.ce.qeq}, $\suppd\rho_{\bc\be,m_{\bc}m_{\bd}m_{\be}}\subseteq
    (V_{\bc,m_{\bc}}Q_{\bc1}\quanteq V_{\be,m_{\be}}Q_{\be2})
    =
    \denotee{
      v_{\bc}Q_{\bc1}\quanteq v_{\be}Q_{\be2}
    }{m_{\bc}\circ\idx1^{-1}\, m_{\be}\circ\idx2^{-1}}$.
    Thus $\rho_{\bc\be}''$ satisfies $v_{\bc}Q_{\bc1}\quanteq v_{\be}Q_{\be2}$.
    Together, we have that $\rho_{\bc\be}''$ satisfies $B_{\bc\bd}^Q=
    \CL{b_{\bc\bd}\circexp b_{\bd\be}}\cap
    (v_{\bc}Q_{\bc1}\quanteq v_{\be}Q_{\be2})$.
    We have shown the first required property of $\rho''_{\bc\be}$.

    Furthermore, we have
    \begin{align*}
      \partr{\cl{V_1}{Q_{\bc1}}}{\cl{V_2}{Q_{\be2}}}\rho_{\bc\be}'' 
      &\eqrefrel{eq:rho.ce''}=
        \sum_{m_{\bc},m_{\bd},m_{\be}} \mu_{\bc\bd\be}(m_{\bc},m_{\bd},m_{\be})\,
        \proj{U_{\idx1}\basis{\cl V}{m_{\bc}}}
        \otimes
        \partr{Q_{\bc1}}{Q_{\be2}}
        \hat\rho_{\bc\be,m_{\bc}m_{\bd}m_{\be}} \\
      &\eqrefrel{eq:rho.ce.c}=
        \sum_{m_{\bc},m_{\bd},m_{\be}} \mu_{\bc\bd\be}(m_{\bc},m_{\bd},m_{\be})\,
        \proj{U_{\idx1}\basis{\cl V}{m_{\bc}}}
        \otimes
        \partr{Q_{\bc1}}{R_{\bd2}}
        \hat\rho_{\bc\bd,m_{\bc}m_{\bd}} \\
      &\starrel=
        \sum_{m_{\bc},m_{\bd}} \mu_{\bc\bd}(m_{\bc},m_{\bd})\,
        \proj{U_{\idx1}\basis{\cl V}{m_{\bc}}}
        \otimes
        \partr{Q_{\bc1}}{R_{\bd2}}
        \hat\rho_{\bc\bd,m_{\bc}m_{\bd}} \\
      &\eqrefrel{eq:rho.cd.sum}=
        \partr{\cl{V_1}Q_{\bc1}}{\cl{V_2}R_{\bd2}}
        \rho_{\bc\bd}''.
    \end{align*}
    Here $(*)$ uses that $\mu_{\bc\bd}=\marginal{1,2}{\mu_{\bc\bd\be}}$ (\autoref{claim:mu.cde}).

    Thus $\partr{\cl{V_1}{Q_{\bc1}}}{\cl{V_2}{Q_{\be2}}}\rho_{\bc\be}''  = 
    \partr{\cl{V_1}Q_{\bc1}}{\cl{V_2}R_{\bd2}}
    \rho_{\bc\bd}''$, this shows the second required property of $\rho_{\bc\be}''$.

    Analogously, we get 
    $\partr{\cl{V_2}{Q_{\be2}}}{\cl{V_1}{Q_{\bc1}}}\rho_{\bc\be}''=
    \partr{\cl{V_2}{Q_{\be2}}}{\cl{V_1}{R_{\bd1}}}\rho_{\bd\be}''$ which is the third (and last) required property of $\rho''_{\bc\be}$.
  \end{claimproof}

  By \autoref{claim:pure.index}, all we need to do is to find a
  $\rho'_{\bc\be}$
  satisfying the conditions given in \autoref{claim:pure.index}.
  Let
  \[
    \rho_{\bc\be}' := \rho_{\bc\be}'' \otimes \proj{U_{\idx1}\psi_{\bc\bullet}\otimes U_{\idx2}\psi_{\be\bullet}}
  \]
  where $ \rho_{\bc\be}''$
  is the operator from \autoref{claim:rho.ce''}.  Since $\rho''_{\bc\be}$
  satisfies $B_{\bc\be}^Q$,
  $\rho'_{\bc\be}$ satisfies $B_{\bc\be}=B_{\bc\be}^Q\otimes\elltwov{\qu{V_1}\qu{V_2}\setminus Q_{\bc1}R_{\bd2}}$.
  Furthermore,
  \begin{multline*}
    \partr{V_1}{V_2}\rho'_{\bc\be}
    =
    \pb  \paren{ \partr{\cl{V_1}Q_{\bc1}}{\cl{V_2}Q_{\be2}}\rho''_{\bc\be} }
    \otimes
    U_{\idx1}\psi_{\bc\bullet}
    \ \txtrel{\autoref{claim:rho.ce''}}= \ \
    \pb \paren{ \partr{\cl{V_1}{Q_{\bc1}}}{\cl{V_2}{R_{\bd2}}}\rho_{\bc\bd}'' }
    \otimes
    U_{\idx1}\psi_{\bc\bullet}
    \\
    \txtrel{\autoref{claim:ex.rho.cd.de}}= \ \
    \calE_{\idx1}(\rho''_{\bc})
    \otimes
    U_{\idx1}\psi_{\bc\bullet}
    =
    \calE_{\idx1}\pb\paren{\rho''_{\bc}\otimes\proj{\psi_{\bc\bullet}}}
    \eqrefrel{eq:rho'''}=
    \calE_{\idx1}\pb\paren{\rho'_{\bc}}.
  \end{multline*}
  Analogously, 
  $\partr{V_2}{V_1}\rho'_{\bc\be} = 
  \calE_{\idx2}\pb\paren{\rho'_{\be}}$.
  Thus the conditions required by \autoref{claim:pure.index} are satisfied, and 
  $\rhl{A_{\bc\be}}\bc\bd{B_{\bc\be}}$ follows.
\end{proof}

\begin{lemma}[Transitivity, simple]\label{rule-lemma:TransSimple}
  \Ruleref{TransSimple} holds.
\end{lemma}

\begin{proof}
  Follows as a special case from \ruleref{Trans}, taking
  $Q_{\bc}:=\qu{\fv(\bc)}$,
  $Q_{\bd}:=R_{\bd}:=\qu{\fv(\bd)}$,
  $Q_{\be}:=\qu{\fv(\be)}$,
  $a_{\bc\bd}:=b_{\bc\bd}:=(X_{\bc1}=X_{\bd2})$,
  $a_{\bd\be}:=b_{\bd\be}:=(X_{\bd1}=X_{\be2})$,
  $u_p,v_p:=\id$
  for $p=\bc,\bd,\be$, $e_{\bc}:=\cl{\fv(\bc)}$, $e_{\be}:=\cl{\fv(\be)}$.
  Note that then $a_{\bc\bd}\circexp a_{\bd\be} = b_{\bc\bd} \circexp b_{\bd\be} = (X_{\bc1}=X_{\be2})$.
\end{proof}

\subsection{Rules for classical statements}
\label{sec:proofs-classical}

\begin{lemma}[Empty program]\label{rule-lemma:Skip}
  \Ruleref{Skip} holds.
\end{lemma}

\begin{proof}
  Immediate from \autoref{def:rhl} and the semantics of $\Skip$.
\end{proof}

\begin{lemma}[Assignment]\label{rule-lemma:Assign1}
  \Ruleref{Assign1} holds.
\end{lemma}

\begin{proof}
Fix $m_1\in\types{\cl{V_1}}$,
$m_2\in\types{\cl{V_2}}$
and normalized $\psi_1\in\elltwov{\qu{V_1}}$,
$\psi_2\in\elltwov{\qu{V_2}}$
such that
$\psi_1\tensor \psi_2\in\pb\denotee{\substi B{\idx1e/\xx_1}}{\memuni{m_1m_2}}$.

  Let
  \begin{equation*}
    \rho' := \pB\proj{\pb\basis{\cl{V_1}\cl{V_2}}{\upd{m_1}{\xx_1}{\denotee{\idx1 e}{m_1}}\,m_2}
    \tensor\psi_1\tensor\psi_2}.
\end{equation*}
Then
  \begin{align}
    &\pb\denotc{\idx1(\assign{\xx} e)}\pB\paren{\pb\pointstate{\cl{V_1}}{m_1}\tensor\proj{\psi_1}}
    =\denotc{\assign{\xx_1}{\idx1e}}\pB\paren{\pb\proj{\basis{\cl{V_1}}{m_1}\tensor\psi_1}}
      \notag\\&
                \hskip1in\starrel=
             \pB\proj{\pb\basis{\cl{V_1}}{\upd{m_1}{\xx_1}{\denotee{\idx1e}{m_1}}} \tensor\psi_1}
             =
             \partr{V_1}{V_2}\rho', \label{eq:assign.rho'} \\
    & \denotc{\idx1\Skip}\pB\paren{\pb\pointstate{\cl{V_2}}{m_2}\tensor\proj{\psi_2}}
    = \denotc{\Skip}\pB\paren{\pb\proj{\basis{\cl{V_2}}{m_2}\tensor\psi_2}}
              \starrel=
                  \pb\proj{\basis{\cl{V_2}}{m_2}\tensor\psi_2} = \partr{V_2}{V_1}\rho'.
      \label{eq:skip.rho'}
  \end{align}
  Here $(*)$ follows from the semantics of assignments and $\Skip$.

  We have
  \begin{equation*}
    \psi_1\tensor\psi_2\in
    \pb\denotee{\substi B{\idx1e/\xx_1}}{\memuni{m_1m_2}}
%    = \denotee B{\upd{(\memuni{m_1m_2})}{\xx_1}{\denotee{\idx1e}{\memuni{m_1m_2}}}}
    = \denotee B{\upd{m_1}{\xx_1}{\denotee{\idx1e}{m_1}}\,m_2}.
  \end{equation*}
  Hence
  $\rho' = \pB\proj{\pb\basis{\cl{V_1}\cl{V_2}}{\upd{m_1}{\xx_1}{\denotee{\idx1e}{m_1}}\,m_2}
    \tensor\psi_1\tensor\psi_2}$ satisfies $B$.
  
  Furthermore, $\rho'$ is separable.

  Thus for all $m_1,m_2,\psi_1,\psi_2$
  with $\psi_1\tensor \psi_2\in\denotee{\substi B{\idx1e/\xx_1}}{\memuni{m_1m_2}}$,
  there is a separable $\rho'$
  that satisfies~$B$
  and such that \eqref{eq:assign.rho'} and \eqref{eq:skip.rho'} holds.
  By \autoref{lemma:pure}, this implies
  $\pb\rhl{\substi B{\idx1e/\xx_1}}{\assign {\xx}e}\Skip B$.
\end{proof}

\begin{lemma}[Sampling (one-sided)]\label{rule-lemma:Sample1}
  \Ruleref{Sample1} holds.
\end{lemma}

\begin{proof}
  Fix $m_1\in\types{\cl{V_1}}$,
  $m_2\in\types{\cl{V_2}}$ and normalized $\psi_1\in\elltwov{\qu{V_1}}$,
  $\psi_2\in\elltwov{\qu{V_2}}$ such that $\psi_1\tensor \psi_2\in \denotee A{\memuni{m_1m_2}}$.

  Let $\mu:=\denotee{e'}{\memuni{m_1m_2}}=\denotee{e'}{m_1}$.
  Since $\psi_1\tensor \psi_2\in \denotee A{\memuni{m_1m_2}}
  \subseteq\pb\denotee{\CL{e'\text{ is total}}}{\memuni{m_1m_2}}$
  we have
  $\denotee{e'\text{ is total}}{\memuni{m_1m_2}}=\true$. Thus $\mu$ is total.

  For $z\in\suppd\mu$,
  let
  $\phi'_z:= \pb\basis{\cl{V_1}\cl{V_2}}{\upd{m_1}{\xx_1}z\,m_2}
  \tensor\psi_1\tensor\psi_2$
  and let
  \begin{equation*}
    \rho' := \sum_z \mu(z)\, \proj{\phi_z'}.
  \end{equation*}

  Then
  \begin{align}
    &\pb\denotc{\idx1(\sample\xx e)}\pB\paren{\pb\pointstate{X_1}{m_1}\tensor\proj{\psi_1}}
      =
      \pb\denotc{\sample{\xx_1} e'}\pB\paren{\pb\pointstate{X_1}{m_1}\tensor\proj{\psi_1}}
    \notag\\&\hskip.6in
              \starrel= \sum_z \denotee{e'}{m_1}(z)\,
                        \pB\proj{\pb\basis{\cl{V_1}}{\upd{m_1}{\xx_1}z}} \tensor\proj{\psi_1}
              = \sum_z\mu(z)\, \partr{V_1}{V_2}\proj{\phi'_z}
             =
             \partr{V_1}{V_2}\rho', \label{eq:sample.rho'} \\
    &
      \pb\denotc{\idx2\Skip}\pB\paren{\pb\pointstate{X_2}{m_2}\tensor\proj{\psi_2}}
      =
      \denotc\Skip\pB\paren{\pb\proj{\basis{X_2}{m_2}\tensor\psi_2}}
              \starrel=
      \pb\proj{\basis{\cl{V_2}}{m_2}\tensor\psi_2}
      \notag\\&
                \hskip.6in
      \starstarrel= \sum_z \mu(z)\,   \pb\proj{\basis{\cl{V_2}}{m_2}\tensor\psi_2}
      = \sum_z \mu(z)\,  \partr{V_2}{V_1}\proj{\phi'_i}
       = \partr{V_2}{V_1}\rho'
      .
      \label{eq:2.skip.rho'}
  \end{align}
  Here $(*)$ follows from the semantics of the sampling and the $\Skip$ statement.
  And $(**)$ follows since $\mu$ is total.

  For any $z\in\suppd\mu=\denotee{\suppd e'}{\memuni{m_1m_2}}$, we have
  \begin{align*}
    \psi_1\tensor\psi_2 &\in
                          \denotee A{\memuni{m_1m_2}}
                          \subseteq \pB\denotee{\bigcap\nolimits_{z\in\suppd e'}
                                    \substi B{z/\xx_1}}{\memuni{m_1m_2}}
    \subseteq
                          \pb\denotee{\substi B{z/\xx_1}}{\memuni{m_1m_2}}
    = \denotee B{\upd{m_1}{\xx_1}z\,m_2}.
  \end{align*}
  Hence $\proj{\phi'_z}$
  satifies $B$
  for $z\in\suppd\mu$.
  Hence $\rho'=\sum_z\mu(z)\,\proj{\phi'_z}$ satisfies $B$.

  Furthermore, $\rho'$ is separable.

  Thus for all $m_1,m_2,\psi_1,\psi_2$
  with $\psi_1\tensor \psi_2\in\denotee{A}{\memuni{m_1m_2}}$,
  there is a separable $\rho'$
  that satisfies~$B$
  and such that \eqref{eq:sample.rho'} and \eqref{eq:2.skip.rho'} holds.
  By \autoref{lemma:pure}, this implies
  $\rhl A{\sample {\xx}e}\Skip B$.
\end{proof}

\begin{lemma}[Sampling (joint)]\label{rule-lemma:JointSample}
  \Ruleref{JointSample} holds.
\end{lemma}

\begin{proof}
  Fix   $m_1\in\types{\cl{V_1}}$,
  $m_2\in\types{\cl{V_2}}$ and normalized $\psi_1\in\elltwov{\qu{V_1}}$,
  $\psi_2\in\elltwov{\qu{V_2}}$ such that $\psi_1\tensor \psi_2\in \denotee A{\memuni{m_1m_2}}$.
  
  Let $\mu:=\denotee f{\memuni{m_1m_2}}$.
%  Since $\psi_1\tensor \psi_2\in \denotee A{\memuni{m_1m_2}}$,
%  we have
%  $\basis{\cl{V_1}\cl{V_2}}{\memuni{m_1m_2}}\tensor\psi_1\tensor\psi_2\in
%  A\subseteq \CL{e\text{ is total}}$. Thus
%  $\denotee{e\text{ is total}}{m_1}=\denotee{e\text{ is total}}{\memuni{m_1m_2}}=\true$. Thus $\mu$ is total.
%  
  For $(z_1,z_2)\in\suppd\mu$,
  let
  $\phi'_{z_1z_2} := \pb\basis{\cl{V_1}\cl{V_2}}{\upd{m_1}{\xx_1}{z_1}\,\upd{m_2}{\xx_2}{z_2}}
  \tensor\psi_1\tensor\psi_2$
  and let
  \begin{equation*}
    \rho' := \sum_{z_1,z_2} \mu(z_1,z_2)\, \proj{\phi'_{z_1z_2}}.
  \end{equation*}

  Since $\psi_1\tensor\psi_2\in\denotee A{\memuni{m_1m_2}}\subseteq \pb\denotee{\pb\CL{\marginal1f=\idx1e_1}}{\memuni{m_1m_2}}$,
  we have $\pb\denotee{\marginal1f=\idx1e_1}{\memuni{m_1m_2}}=\true$,
  thus
  \begin{equation}\label{eq:e1-marginal}
    \marginal1\mu=\pb\marginal1{\denotee f{\memuni{m_1m_2}}}
    =\denotee{\idx1e_1}{\memuni{m_1m_2}}=\denotee{\idx1e_1}{m_1}.
  \end{equation}
  
  Analogously, we show $\denotee{\idx2e_2}{m_2}=\marginal2\mu$.

  Then
  \begin{align}
    &\pb\denotc{\idx1(\sample\xx{e_1})}\pB\paren{\pb\pointstate{\cl{V_1}}{m_1}\tensor\proj{\psi_1}}
      =
      \pb\denotc{\sample{\xx_1}{\idx1e_1}}\pB\paren{\pb\pointstate{\cl{V_1}}{m_1}\tensor\proj{\psi_1}}
      \notag\\&\qquad
      \starrel=
      \sum_{z_1} \denotee {\idx1e_1}{m_1}(z_1)\,
      \pb\proj{\basis{\cl{V_1}}{\upd{m_1}{\xx_1}{z_1}} \tensor\psi_1}
      \notag\\&\qquad
                \eqrefrel{eq:e1-marginal}=\sum_{z_1} \pb\paren{\marginal1\mu}(z_1)\,
                \pB\proj{\pb\basis{\cl{V_1}}{\upd{m_1}{\xx_1}{z_1}} \tensor\psi_1}
                \notag\\&\qquad
                =\sum_{z_1,z_2} \mu(z_1,z_2)\,
                \pB\proj{\pb\basis{\cl{V_1}}{\upd{m_1}{\xx_1}{z_1}} \tensor\psi_1}
                \notag\\&\qquad
                =\sum_{z_1,z_2} \mu(z_1,z_2)\,
                          \partr{V_1}{V_2} \proj{\phi'_{z_1,z_2}}
                          =\partr{V_1}{V_2}\rho'. \label{eq:sample.rho'-1}
  \end{align}
  Here $(*)$ follows from the semantics of the sampling statement.

  Analogously, we get
  \begin{equation}
    \label{eq:sample.rho'-2}
    \pb\denotc{\idx2(\sample{\yy}{e_2})}\pB\paren{\pb\pointstate{\cl{V_1}}{m_2}\tensor\proj{\psi_2}}
    =
    \partr{V_2}{V_1}\rho'.
  \end{equation}

  For any $(z_1,z_2)\in\suppd\mu$, we have that $\denotee{(z_1,z_2)\in\suppd f}{\memuni{m_1m_2}}=\true$. Thus
  \begin{align*}
    \psi_1\tensor\psi_2
    &\in
      \denotee A{\memuni{m_1m_2}}
      \subseteq \pB\denotee{\bigcap\nolimits_{(z_1,z_2)\in\suppd f} \substi B{z_1/\xx_1,z_2/\xx_2}}{\memuni{m_1m_2}}
      \subseteq \pb\denotee{\substi B{z_1/\xx_1,z_2/\xx_2}}{\memuni{m_1m_2}} \\
    &
    = \denotee B{\upd{m_1}{\xx_1}{z_1}\,\upd{m_2}{\xx_2}{z_2}}.
  \end{align*}
  Hence $\proj{\phi'_{z_1z_2}}$ satisfies $B$ for all
   $(z_1,z_2)\in\suppd\mu$.
  Since $\rho'=\sum_{z_1,z_2}\mu(z_1,z_2)\,\proj{\phi'_{z_1,z_2}}$, $\rho'$ satisfies~$B$.
  
  Furthermore, $\rho'$ is separable.

  Thus for all $m_1,m_2,\psi_1,\psi_2$
  with $\psi_1\tensor \psi_2\in\denotee{A}{\memuni{m_1m_2}}$,
  there is a separable $\rho'$
  that satisfies~$B$
  and such that \eqref{eq:sample.rho'-1} and \eqref{eq:sample.rho'-2} hold.
  By \autoref{lemma:pure}, this implies
  $\rhl A{\sample {\xx}{e_2}}{\sample {\yy}{e_2}} B$.
\end{proof}

\begin{lemma}[Conditionals (one-sided)]\label{rule-lemma:If1}
  \Ruleref{If1} holds.
\end{lemma}

\begin{proof}\stepcounter{claimstep}
  Let $e':=\idx1e$ and $\bc':=\idx1\bc$ and $\bd':=\idx1\bd$.
  Let $A_b:=\CL{e'=b}\cap A$ for $b=\true,\false$.
  \begin{claim}\label{claim:atrue}
    $\rhl{A_{\true}}{\langif e\bc\bd}\Skip B$.
  \end{claim}

  \begin{claimproof}
    Fix a separable $\rho$
    that satisfies $A_{\true}$.
    Since $\rhl{A_{\true}}\bc\Skip B$
    by assumption, there exists a separable $\rho'$
    that satisfies $B$ and such that
    \begin{align}
      \denotc{\bc'}\pb\paren{\partr{V_1}{V_2}\rho} = \partr{V_1}{V_2}\rho',
      \label{eq:bc'}
      \\
      \denotc{\idx1\Skip}\pb\paren{\partr{V_1}{V_2}\rho} = \partr{V_1}{V_2}\rho'.
      \label{eq:skip'}
    \end{align}
    Since $\rho$ satisfies $A_{\true}\subseteq\CL{e'=\true}$, and $\fv(e')\subseteq V_1$, we have that
    $\partr{V_1}{V_2}\rho$ satisfies $\CL{e'=\true}$, 
    and thus
    $\restricte {e'}(\partr{V_1}{V_2}\rho)=\partr{V_1}{V_2}\rho$ and $\restricte {\lnot e'}(\partr{V_1}{V_2}\rho)=0$.
    Thus
    \begin{equation}
      \pb\denotc{\langif{e'}{\bc'}{\bd'}}\pb\paren{\partr{V_1}{V_2}\rho} \starrel=
      \denotc{\bc'}\pb\paren{\restricte {e'}(\partr{V_1}{V_2}\rho)} + 
      \denotc{\bd'}\pb\paren{\restricte {\lnot e'}(\partr{V_1}{V_2}\rho)}
      =
      \denotc{\bc'}\pb\paren{\partr{V_1}{V_2}\rho}.
      \label{eq:ifsem.restr}
    \end{equation}
    Here $(*)$ follows from the semantics of the if statement.

    Thus
    \begin{equation}
      \pb\denotc{\idx1(\langif{e}{\bc}{\bd})}\pb\paren{\partr{V_1}{V_2}\rho}
      =
      \pb\denotc{\langif{e'}{\bc'}{\bd'}}\pb\paren{\partr{V_1}{V_2}\rho}
      \eqrefrel{eq:ifsem.restr}=
      \denotc{\bc'}\pb\paren{\partr{V_1}{V_2}\rho}
      \eqrefrel{eq:bc'}=
      \partr{V_1}{V_2}\rho'.
      \label{eq:idx1langif}
    \end{equation}
    
    Thus $\rho'$
    is separable and satisfies $B$,
    and \eqref{eq:idx1langif} and \eqref{eq:skip'} hold. Thus
    $\rhl{A_{\true}}{\langif e\bc\bd}\Skip B$ holds.
  \end{claimproof}

  \begin{claim}\label{claim:afalse}
    $\rhl{A_{\false}}{\langif e\bc\bd}\Skip B$.
  \end{claim}

  \begin{claimproof}
    Analogous to \autoref{claim:atrue}.
  \end{claimproof}

  Note that $\typee{e'}=\bool$.
  From \autoref{claim:atrue} and \autoref{claim:afalse} we thus have
  \begin{equation*}
    \forall z\in\typee{e'}.\quad
    \pb\rhl{\CL{e'=z}\cap A}{\langif e\bc\bd}\Skip B.
  \end{equation*}
  By \ruleref{Case}, this implies
  $\rhl{A}{\langif e\bc\bd}\Skip B$.
\end{proof}

\begin{lemma}[Conditionals (joint)]\label{rule-lemma:JointIf}
  \Ruleref{JointIf} holds.
\end{lemma}

\begin{proof}\stepcounter{claimstep}
  Let $e_i':=\idx ie_i$ and $\bc'_i:=\idx i\bc_i$ and $\bd'_i:=\idx i\bd_i$.
  Let $A_b:=\CL{e'_1=b\land e'_2=b}\cap A$ for $b=\true,\false$.

  \begin{claim}\label{claim:ab.e1}
    $A_b=\CL{e_1'=b}\cap A$.
  \end{claim}

  \begin{claimproof} We have
    \begin{align*}
      A_b&=\pb\CL{e'_1=b\land e'_2=b}\cap A \starrel\subseteq
      \pb\CL{e'_1=b\land e'_2=b}\cap \pb\CL{e'_1=e'_2}\cap A
      \\&
      =
      \pb\CL{e'_1=b\land e'_2=b\land e'_1=e'_2}\cap A
          =
      \pb\CL{e'_1=b\land e'_1=e'_2}\cap A
      \subseteq
          \pb\CL{e'_1=b}\cap A
    \end{align*}
    and
    \begin{align*}
      \pb\CL{e'_1=b}\cap A
      &\starrel\subseteq
        \pb\CL{e'_1=b}\cap \pb\CL{e'_1=e'_2}\cap A
        =
        \pb\CL{e'_1=b\land e'_1=e'_2}\cap A
      \\&
          =
        \pb\CL{e'_1=b\land e_2'=b\land e'_1=e'_2}\cap A
        \subseteq
        \pb\CL{e'_1=b\land e_2'=b}\cap A
        =
        A_b.
    \end{align*}
    In both calculations, $(*)$
    holds since we have $A\subseteq\CL{e'_1=e'_2}$
    by assumption. 
  \end{claimproof}

  \begin{claim}\label{claim:atrue.joint}
    $\pb\rhl{\CL{e_1'}\cap A}{\langif{e_1}{\bc_1}{\bd_1}}{\langif{e_2}{\bc_2}{\bd_2}} B$.
  \end{claim}

  \begin{claimproof}
    Fix a separable $\rho$
    that satisfies ${\CL{e_1'}\cap A}$.
    By \autoref{claim:ab.e1}, $\rho$ satiesfies $A_{\true}$.
    Since $\rhl{A_{\true}}{\bc_1}{\bc_2} B$
    by assumption, there exists a separable $\rho'$
    that satisfies $B$ and such that
    \begin{align}
      \denotc{\bc_1'}\pb\paren{\partr{V_1}{V_2}\rho} = \partr{V_1}{V_2}\rho',
      \label{eq:bc'.joint}
      \\
      \denotc{\bc_2'}\pb\paren{\partr{V_2}{V_1}\rho} = \partr{V_2}{V_1}\rho',
      \label{eq:bc'2.joint}
    \end{align}
    Since $\rho$ satisfies $A_{\true}\subseteq\CL{e'_1}$, and $\fv(e_1')\subseteq V_1$, we have that
    $\partr{V_1}{V_2}\rho$ satisfies $\CL{e'_1}$, 
    and thus
    $\restricte {e'_1}(\partr{V_1}{V_2}\rho)=\partr{V_1}{V_2}\rho$ and $\restricte {\lnot e'_1}(\partr{V_1}{V_2}\rho)=0$.
    Thus
    \begin{equation}
      \pb\denotc{\langif{e'_1}{\bc'_1}{\bd'_1}}\pb\paren{\partr{V_1}{V_2}\rho} \starrel=
      \pb\denotc{\bc'_1}\pb\paren{\restricte {e'_1}(\partr{V_1}{V_2}\rho)} + 
      \denotc{\bd'_1}\pb\paren{\restricte {\lnot e'_1}(\partr{V_1}{V_2}\rho)}
      =
      \denotc{\bc'_1}(\partr{V_1}{V_2}\rho).
      \label{eq:ifsem.restr1}
    \end{equation}
    Here $(*)$ follows from the semantics of the if statement.
    Analogously,
    \begin{equation}\label{eq:ifsem.restr2}
      \pb\denotc{\langif{e'_2}{\bc'_2}{\bd'_2}}\pb\paren{\partr{V_2}{V_1}\rho}
      =
      \denotc{\bc'_2}\pb\paren{\partr{V_2}{V_1}\rho}.
    \end{equation}
    Thus
    \begin{gather}
      \pb\denotc{\idx1(\langif{e_1}{\bc_1}{\bd_1})}\pb\paren{\partr{V_1}{V_2}\rho}
      =
      \pb\denotc{\langif{e'_1}{\bc'_1}{\bd'_1}}\pb\paren{\partr{V_1}{V_2}\rho}
      \eqrefrel{eq:ifsem.restr1}=
      \denotc{\bc'_1}\pb\paren{\partr{V_1}{V_2}\rho}
      \eqrefrel{eq:bc'.joint}=
      \partr{V_1}{V_2}\rho',
      \label{eq:idx1langif1}
      \\
      \pb\denotc{\idx2(\langif{e_2}{\bc_2}{\bd_2})}\pb\paren{\partr{V_2}{V_1}\rho}
      =
      \pb\denotc{\langif{e'_2}{\bc'_2}{\bd'_2}}\pb\paren{\partr{V_2}{V_1}\rho}
      \eqrefrel{eq:ifsem.restr2}=
      \denotc{\bc'_2}\pb\paren{\partr{V_2}{V_1}\rho}
      \eqrefrel{eq:bc'2.joint}=
      \partr{V_2}{V_1}\rho'.
      \label{eq:idx1langif2}
    \end{gather}
    So $\rho'$
    is separable and satisfies $B$,
    and \eqref{eq:idx1langif1} and \eqref{eq:idx1langif2} hold. Thus
    \[
      \pb\rhl{A_{\true}}{\langif{e_1}{\bc_1}{\bd_1}}{\langif{e_2}{\bc_2}{\bd_2}} B
    \] holds.
  \end{claimproof}

  \begin{claim}\label{claim:afalse.joint}
    $\pb\rhl{\CL{\lnot e_1'}\cap A}{\langif{e_1}{\bc_1}{\bd_1}}{\langif{e_2}{\bc_2}{\bd_2}} B$.
  \end{claim}

  \begin{claimproof}
    Analogous to \autoref{claim:atrue.joint}.
  \end{claimproof}

  Note that $\typee{e_1'}=\bool$.
  From \autoref{claim:atrue.joint} and \autoref{claim:afalse.joint} we thus have
  \[
  \forall z\in\typee{e_1'}.\ \pb\rhl{\CL{e_1'=z}\cap A}{\langif{e_1}{\bc_1}{\bd_1}}{\langif{e_2}{\bc_2}{\bd_2}} B.
  \]
  By \ruleref{Case}, this implies
  $\rhl{A}{\langif{e_1}{\bc_1}{\bd_1}}{\langif{e_2}{\bc_2}{\bd_2}}B$.
\end{proof}

%\newpage

\begin{lemma}[While loops (one-sided)]\label{rule-lemma:While1}
  \Ruleref{While1} holds.
\end{lemma}

\begin{proof}\stepcounter{claimstep}
  Let $e':=\idx1e$ and $\bc':=\idx1\bc$.

  Since $\pb\rhl{\CL{e'}\cap A}\bc\Skip A$,
  for any separable $\rho\in\traceposcq{V_1V_2}$
  that satisfies $\CL{e'}\cap A$,
  there is a separable $\rho'$ that satisfies $A$ and such that
  \begin{align*}
    \denotc{\bc'}\pb\paren{\partr{V_1}{V_2}\rho} = \partr{V_1}{V_2}\rho'
                                          \qquad\text{and}\qquad
    \partr{V_2}{V_1}\rho = \pb\denotc{\idx2\Skip}\pb\paren{\partr{V_2}{V_1}\rho} = \partr{V_2}{V_1}\rho'.
  \end{align*}
  Thus there exists a function (not a superoperator!)
  $E:\traceposcq{V_1V_2}\to\traceposcq{V_1V_2}$ such that for all separable $\rho$ that satisfy
  $\CL{e'}\cap A$, we have that $E(\rho)$ is separable and satisfies $A$, and
  \begin{equation}
    \denotc{\bc'}\pb\paren{\partr{V_1}{V_2}\rho} = \partr{V_1}{V_2} E(\rho)
    \qquad\text{and}\qquad
    \partr{V_2}{V_1}\rho = \partr{V_2}{V_1} E(\rho).
    \label{eq:E.def}
  \end{equation}

  Fix some separable $\rho\in\traceposcq{V_1V_2}$
  that satisfies $A$.
  To prove the rule, we need to show that there exists a separable
  $\rho'$ such that $\rho'$ satisfies $\CL{\lnot e'}\cap A$, and
  \begin{align*}
    \pb\denotc{\idx1(\while e\bc)}(\rho_1) = \partr{V_1}{V_2}\rho'
                                          \qquad\text{and}\qquad
     \denotc{\idx2\Skip}(\rho_2) = \partr{V_1}{V_2}\rho'
  \end{align*}
  where
  \[
    \rho_1:=\partr{V_1}{V_2}\rho
    \qquad\text{and}\qquad
    \rho_2:=\partr{V_2}{V_1}\rho.
  \]
  
  Let $\hat\rho_0:=\rho$,
  and $\hat\rho_{i+1}:=E\pb\paren{\restricte{e'}(\hat\rho_i)}$.
  Let $\rho'_i:=\restricte {\lnot e'}(\hat\rho_i)$.
  Let $\rho':=\sum_{i=0}^\infty\rho'_i$. (Existence of $\rho'$ is shown below in \autoref{claim:rho'.exi}.)

  \begin{claim}\label{claim:rhoi.A}
    $\hat\rho_i$ is separable and satisfies $A$.
    $\rho'_i$ is separable and satisfies $\CL{\lnot e'}\cap A$.
  \end{claim}

  \begin{claimproof}
    We prove that     $\hat\rho_i$ is separable and satisfies $A$ by induction over $i$.
    For $i=0$,
    this follows since $\hat\rho_0=\rho$
    is separable and satisfies $A$
    by assumption. For $i>0$,
    by induction hypothesis, $\rho_{i-1}$
    is separable and satisfies $A$.
    Thus $\restricte {e'}(\rho_{i-1})$
    is separable and satisfies $\CL{e'}\cap A$.
    Thus $\rho_i=E\pb\paren{\restricte {e'}(\rho_{i-1})}$
    is separable and satisfies $A$.

    To show that $\rho'_i$
    is separable and satisfies $\CL{\lnot e'}A$,
    note that $\rho'_i=\restricte{\lnot e}'(\hat\rho_i)$.
    Since $\hat\rho_i$
    is separable and satisfies $A$,
    $\rho'_i$ is separable and satisfies $\CL{\lnot e'}\cap A$,
  \end{claimproof}

  \begin{claim}\label{claim:rhoi'.formula}
    $\partr{V_1}{V_2}\rho_i' =
      \restricte {\lnot e'}\bigl((\denotc{\bc'}\circ\restricte {e'})^i(\rho_1)\bigr)$.
  \end{claim}

  \begin{claimproof}
    By \autoref{claim:rhoi.A}, and by definition of $E$, we have
    \[
    \partr{V_1}{V_2}\hat\rho_i=
    \partr{V_1}{V_2}E\pb\paren{\restricte{e'}(\hat\rho_{i-1})}
    =
    \denotc{\bc'}\pb\paren{\partr{V_1}{V_2}\restricte {e'}(\hat\rho_{i-1})}
    =
    \pb\paren{\denotc{\bc'}\circ\restricte {e'}}\pb\paren{\partr{V_1}{V_2}\hat\rho_{i-1}}
    \qquad
    \text{for }i\geq1.
    \]
    Since $\partr{V_1}{V_2}\hat\rho_0=\partr{V_1}{V_2}\rho=
    \pb\paren{\denotc{\bc'}\circ\restricte {e'}}^0(\rho_1)$, it follows by induction that
    $\partr{V_1}{V_2}\hat\rho_i =
    \pb\paren{\denotc{\bc'}\circ\restricte {e'}}^i(\rho_1)$ for $i\geq0$.
    And thus 
    $\partr{V_1}{V_2}\rho_i' =
    \partr{V_1}{V_2}\,\restricte {\lnot e'}(\hat\rho_i)=
    \restricte {\lnot e'}\pb\paren{\partr{V_1}{V_2}\hat\rho_i}=
    \restricte {\lnot e'}\bigl((\denotc{\bc'}\circ\restricte {e'})^i(\rho_1)\bigr)$.
  \end{claimproof}

  \begin{claim}\label{claim:rho'.exi}
    $\rho'\in\traceposcq{V_1V_2}$ exists, and $\tr\rho' = \tr\rho$.
  \end{claim}

  \begin{claimproof}
    Since $\rho$
    satisfies $A\subseteq A_1\otimes\elltwov{\qu{V_2}}$,
    we have that $\rho_1=\partr{V_1}{V_2}\rho$
    satisfies $A_1$.
    Since $\while{e'}{\bc'}$
    is total on $A_1$
    by assumption, we have
    $\tr\denotc{\while{e'}{\bc'}}(\rho_1)=\tr\rho_1$ .

    We have
    \begin{align*}
      \sum_{i=0}^\infty \tr \rho'_i
      &=
        \sum_{i=0}^\infty \tr \partr{V_1}{V_2} \rho'_i
        \starrel=
        \sum_{i=0}^\infty \tr
        \restricte {\lnot e'}\bigl((\denotc{\bc'}\circ\restricte {e'})^i(\rho_1)\bigr)
        =
        \tr      \sum_{i=0}^\infty 
        \restricte {\lnot e'}\bigl((\denotc{\bc'}\circ\restricte {e'})^i(\rho_1)\bigr)
      \\&
          \starstarrel=
          \tr \, \pb\denotc{\while{e'}{\bc'}}(\rho_1)
          =
          \tr \rho_1
          = \tr \rho.
    \end{align*}
    Here $(*)$
    follows from \autoref{claim:rhoi'.formula}.
    And $(**)$ follows from the semantics of while statements.

    Since $ \sum_{i=0}^\infty \tr \rho'_i = \tr\rho <\infty$,
    we have that $\rho'=\sum_{i=0}^\infty\rho'_i$ exists and $\tr\rho'=\tr\rho$.
  \end{claimproof}
  
  \begin{claim}\label{claim:rho'.CLnot.A}
    $\rho'$ is separable and satisfies $\CL{\lnot e'}\cap A$.
  \end{claim}

  \begin{claimproof}
    Since $\rho'_i$
    is separable and satisfies $\CL{\lnot e'}\cap A$
    by \autoref{claim:rhoi.A}, we immediately have that
    $\rho'=\sum_i\rho'_i$
    is separable and satisfies $\CL{\lnot e'}\cap A$.
  \end{claimproof}
  
  \begin{claim}\label{claim:rho'.while}
    $\partr{V_1}{V_2}\rho' = \pb\denotc{\idx1(\while{e}{\bc})}(\rho_1)$.
  \end{claim}

  \begin{claimproof}
    We have
    \begin{align*}
      \partr{V_1}{V_2}\rho' &=
        \sum_{i=0}^\infty \partr{V_1}{V_2} \rho'_i
        \starrel=
        \sum_{i=0}^\infty
        \restricte {\lnot e'}\bigl((\denotc{\bc'}\circ\restricte {e'})^i(\rho_1)\bigr)
                              \starstarrel=
         \pb\denotc{\while{e'}{\bc'}}(\rho_1)
                              =
                              \pb\denotc{\idx1(\while{e}{\bc})}(\rho_1).
    \end{align*}
    Here $(*)$ follows from \autoref{claim:rhoi'.formula}.
    And $(**)$ follows from the semantics of while statements.
  \end{claimproof}

  \begin{claim}
    $\partr{V_2}{V_1}\rho' \leq \tr\rho_2$.
  \end{claim}

  \begin{claimproof}
    We first show for all $n\geq 0$,
    \begin{equation}
      \partr{V_2}{V_1}\pB\paren{\hat\rho_n+\sum_{i=0}^{n-1}\rho'_i} = \rho_2.
      \label{eq:rho2-sum}
    \end{equation}
    We show this by induction over $n$.
    For $n=0$, this follows since $\hat\rho_0=\rho$ and $\rho_2=\partr{V_2}{V_1}\rho$ by definition.
    Assume \eqref{eq:rho2-sum} holds for $n$, then
    \begin{align*}
      \partr{V_2}{V_1}\pB\paren{\hat\rho_{n+1}+\sum_{i=0}^{n}\rho'_i}
      &=
      \partr{V_2}{V_1}E\pb\paren{\restricte{e'}(\hat\rho_{n})}+\partr{V_2}{V_1}\pB\paren{\sum_{i=0}^{n}\rho'_i}
      \eqrefrel{eq:E.def}=
      \partr{V_2}{V_1}\,\restricte{e'}(\hat\rho_{n})
        +\partr{V_2}{V_1}\pB\paren{\sum_{i=0}^{n}\rho'_i} \\
      &=
        \partr{V_2}{V_1}\pb\paren{\restricte{e'}(\hat\rho_{n})
        +\underbrace{\rho'_n}_{\hskip-1in{}=\restricte{\lnot e'}(\hat\rho_n)\hskip-1in}\,}
        +\partr{V_2}{V_1}\pB\paren{\sum_{i=0}^{n-1}\rho'_i}
        =
        \partr{V_2}{V_1}\hat\rho_n
        +\partr{V_2}{V_1}\pB\paren{\sum_{i=0}^{n-1}\rho'_i}
        \eqrefrel{eq:rho2-sum}=
        \rho_2.
    \end{align*}
    Here \eqref{eq:E.def} can be applied since $\hat\rho_n$
    is separable and satisfies $A$
    by \autoref{claim:rhoi.A} and thus $\restricte{e'}(\hat\rho_n)$
    is separable and satisfies $\CL{e'}\cap A$.
    Thus we have shown \eqref{eq:rho2-sum} for $n+1$,
    and thus \eqref{eq:rho2-sum} holds for all $n$ by induction.

    Then for $n\geq 0$,
    \begin{equation*}
      \sum_{i=0}^n\rho'_i
      =
      \rho'_n + 
      \sum_{i=0}^{n-1}\rho'_i
      \starrel=
      \restricte{\lnot e'}(\hat\rho_n)
      +
      \sum_{i=0}^{n-1}\restricte{\lnot e'}(\rho'_i)
      =
      \restricte{\lnot e'}\pB\paren{\hat\rho_n
      +
      \sum_{i=0}^{n-1}\rho'_i}
      \starstarrel\leq
      \hat\rho_n
      +
      \sum_{i=0}^{n-1}\rho'_i.
    \end{equation*}
    Here $(*)$
    holds since $\rho'_n=\restricte{\lnot e'}(\hat\rho_n)$
    by definition, and
    $\restricte{\lnot e'}(\rho'_i)=\restricte{\lnot
      e'}\pb\paren{\restricte{\lnot e'}(\hat\rho_i)}= \restricte{\lnot
      e'}(\hat\rho_i)=\rho'_i$.  And $(**)$
    holds since for any $\rho$,
    $\restricte{e'}(\rho)=\rho-\restricte{\lnot e'}(\rho)$
    and $\restricte{\lnot e'}(\rho)\geq 0$,
    and thus $\restricte{e'}(\rho)\leq\rho$.

    Since superoperators are monotonous and $\partr{V_2}{V_1}$
    is a superoperator, it follows
    \begin{equation}\label{eq:rho2bound}
      \sum_{i=0}^n       \partr{V_2}{V_1} \rho'_i
      =
      \partr{V_2}{V_1}\pB\paren{\sum_{i=0}^n\rho'_i}
      \leq \partr{V_2}{V_1}\pB\paren{
            \hat\rho_n
      +
      \sum_{i=0}^{n-1}\rho'_i.
      }
      \eqrefrel{eq:rho2-sum}=
      \rho_2.
    \end{equation}
    Thus also
    \begin{align*}
      \partr{V_2}{V_1}\rho'
      =
      \sum_{i=0}^\infty \partr{V_2}{V_1}\rho'_i
      =
      \sup_n
      \sum_{i=0}^n \partr{V_2}{V_1}\rho'_i
      \eqrefrel{eq:rho2bound}\leq
      \rho_2.
      \mathQED
    \end{align*}
  \end{claimproof}

  \begin{claim}\label{claim:rho'.skip}
    $\partr{V_2}{V_1}\rho'=\pb\denotc{\idx2\Skip}(\rho_2)$.
  \end{claim}

  \begin{claimproof}
    We have
    \[
      \tr \partr{V_2}{V_1}\rho'
      = \tr\rho'
      \starrel= \tr\rho
      = \tr\rho_2
    \]
    where $(*)$
    is by \autoref{claim:rho'.exi}. And we have
    $\partr{V_2}{V_1}\rho' \leq \rho_2$.
    Together, this implies $\partr{V_2}{V_1}\rho' =
    \rho_2$. (Otherwise $\rho_2-\partr{V_2}{V_1}\rho'$
    would be positive, nonzero, and have zero trace.)
    Hence
    \begin{gather*}
      \partr{V_2}{V_1}\rho'=
      \rho_2
      =
      \denotc{\Skip}(\rho_2)
      =
      \pb\denotc{\idx2\Skip}(\rho_2).
      \mathQED
    \end{gather*}
  \end{claimproof}

  To summarize, for any separable $\rho\in\traceposcq{V_1V_2}$
  that satisfies $A$,
  there is a $\rho'\in\traceposcq{V_1}$
  that is separable and satisfies $\CL{\lnot \idx1e}\cap A$
  (\autoref{claim:rho'.CLnot.A}), and such that
  $\pb\denotc{\idx1(\while
    e\bc)}\pb\paren{\partr{V_1}{V_2}\rho}=\partr{V_1}{V_2}\rho'$
  (\autoref{claim:rho'.while}) and
  $\pb\denotc{\idx2\Skip}\pb\paren{\partr{V_2}{V_1}\rho}=\partr{V_2}{V_1}\rho'$
  (\autoref{claim:rho'.skip}). This implies 
  $\pb\rhl{A}{\while e\bc}\Skip {\CL{\lnot \idx1e}\cap A}$.
\end{proof}

\begin{lemma}[While loops (joint)]\label{rule-lemma:JointWhile}
  \Ruleref{JointWhile} holds.
\end{lemma}

\begin{proof}\stepcounter{claimstep}
  Let $e'_i:=\idx ie_i$ for $i=1,2$ and $\bc':=\idx1\bc$
  and $\bd':=\idx2\bd$.

  Since $\pb\rhl{\CL{e'_1\land e'_2}\cap A}\bc\bd A$ by assumption,
  for any separable $\rho\in\traceposcq{V_1V_2}$
  that satisfies $\CL{e'_1\land e'_2}\cap A$,
  there is a separable $\rho'$ that satisfies $A$ and such that
  \begin{align*}
    \denotc{\bc'}\pb\paren{\partr{V_1}{V_2}\rho} = \partr{V_1}{V_2}\rho'
                                          \qquad\text{and}\qquad
    \denotc{\bd'}\pb\paren{\partr{V_2}{V_1}\rho} = \partr{V_2}{V_1}\rho'.
  \end{align*}
  Thus there exists a function (not a superoperator!)
  $E:\traceposcq{V_1V_2}\to\traceposcq{V_1V_2}$ such for separable $\rho$ that satisfies 
  $\CL{e_1'\cap e_2'}\cap A$, we have that $E(\rho)$ is separable and satisfies $A$, and
  \begin{equation}
    \denotc{\bc'}\pb\paren{\partr{V_1}{V_2}\rho} = \partr{V_1}{V_2} E(\rho)
    \qquad\text{and}\qquad
    \denotc{\bd'}\pb\paren{\partr{V_2}{V_1}\rho} = \partr{V_2}{V_1} E(\rho).
    \label{eq:E.def-joint}
  \end{equation}

  Fix some separable $\rho\in\traceposcq{V_1V_2}$
  that satisfies $A$.
  To prove the lemma, we need to show that there exists a separable
  $\rho'$ such that $\rho'$ satisfies $\CL{\lnot e'_1\land\lnot e'_2}\cap A$, and
  \begin{align*}
    \pb\denotc{\idx1(\while{e_1}\bc)}(\rho_1) = \partr{V_1}{V_2}\rho'
                                          \qquad\text{and}\qquad
    \pb\denotc{\idx2(\while{e_2}\bd)}(\rho_2) = \partr{V_2}{V_1}\rho'
  \end{align*}
  where
  \[
    \rho_1:=\partr{V_1}{V_2}\rho
    \qquad\text{and}\qquad
    \rho_2:=\partr{V_2}{V_1}\rho.
  \]
  
  Let $\hat\rho_0:=\rho$,
  and $\hat\rho_{i+1}:=E\pb\paren{\restricte{e_1'\land e_2'}(\hat\rho_i)}$.
  Let $\rho'_i:=\restricte {\lnot e'_1\land \lnot e'_2}(\hat\rho_i)$.
  Let $\rho':=\sum_{i=0}^\infty\rho'_i$. (Existence of $\rho'$ is shown below in \autoref{claim:rho'.exi-joint}.)

  \begin{claim}\label{claim:rhoi.A-joint}
    $\hat\rho_i$ is separable and satisfies $A$.
    $\rho'_i$ is separable and satisfies $\CL{\lnot e'_1\land\lnot e'_2}\cap A$.
  \end{claim}

  \begin{claimproof}
    We prove that     $\hat\rho_i$ is separable and satisfies $A$ by induction over $i$.
    For $i=0$,
    this follows since $\hat\rho_0=\rho$
    is separable and satisfies $A$
    by assumption. For $i>0$,
    by induction hypothesis, $\hat\rho_{i-1}$
    is separable and satisfies $A$.
    Thus $\restricte {e'_1\land e'_2}(\hat\rho_{i-1})$
    is separable and satisfies $\CL{e'_1\land e'_2}\cap A$.
    Thus $\hat\rho_i=E(\restricte {e'_1\land e'_2}(\hat\rho_{i-1}))$
    is separable and satisfies $A$ by definition of $E$.

    To show that $\rho'_i$
    is separable and satisfies $\CL{\lnot e'_1\land\lnot e'_2}\cap A$,
    note that $\rho'_i=\restricte{\lnot e'_1\land\lnot e'_2}(\hat\rho_i)$.
    Since $\hat\rho_i$
    is separable and satisfies $A$,
    $\rho'_i$ is separable and satisfies $\CL{\lnot e'_1\land\lnot e'_2}\cap A$.
  \end{claimproof}

  \begin{claim}\label{claim:restrict.same}
    For any $\sigma\in\traceposcq{V_1V_2}$
    that satisfies $A$,
    it holds that
    $\restricte{e'_1\land
      e'_2}(\sigma)=\restricte{e'_1}(\sigma)=\restricte{e'_2}(\sigma)$
    and
    $\restricte{\lnot e'_1\land \lnot e'_2}(\sigma)=\restricte{\lnot
      e'_1}(\sigma)=\restricte{\lnot e'_2}(\sigma)$.
  \end{claim}

  \begin{claimproof}
    Since $\sigma\in\traceposcq{V_1V_2}$,
    we can decompose $\sigma$
    as
    $\sigma=\sum_m\pb\proj{\basis{m}{\cl{V_1}\cl{V_2}}}\otimes\sigma_m=:\sum_m\sigma'_m$.
    Let $S:=\{m:\sigma_m\neq 0\}$.
    Let $S_1:=\pb\braces{m\in S:\denotee{e_1'}{m}}$ and
    $S_2:=\pb\braces{m\in S:\denotee{e_2'}{m}}$.
    Then 
    \begin{equation}\label{eq:sigma.decomps}
      \restricte{e'_1\land e'_2}(\sigma)=\sum_{m\in S_1\cap S_2}\sigma'_m,
                                          \qquad\qquad
      \restricte{e'_1}(\sigma)=\sum_{m\in S_1}\sigma'_m,
                                          \qquad\qquad
      \restricte{e'_2}(\sigma)=\sum_{m\in S_2}\sigma'_m.
    \end{equation}
    Since $\sigma$
    satisfies $A\subseteq\CL{e_1'=e_2'}$,
    we have $\sigma_m=0$
    for all $m$
    with $\denotee{e'_1}{m}\neq\denotee{e'_2}{m}$.
    Thus for all $m\in S$,
    $\denotee{e'_1}{m}=\denotee{e'_2}{m}$.
    Hence $S_1=S_2=S_1\cap S_2$.
    With \eqref{eq:sigma.decomps}, this implies
    $\restricte{e'_1\land
      e'_2}(\sigma)=\restricte{e'_1}(\sigma)=\restricte{e'_2}(\sigma)$.

    The fact
    $\restricte{\lnot e'_1\land \lnot e'_2}(\sigma)=\restricte{\lnot
      e'_1}(\sigma)=\restricte{\lnot e'_2}(\sigma)$ is shown
    analogously.
  \end{claimproof}

  \begin{claim}\label{claim:rhoi'.formula-joint}
    $\partr{V_1}{V_2}\rho_i' =
      \restricte {\lnot e_1'}\bigl((\denotc{\bc'}\circ\restricte {e_1'})^i(\rho_1)\bigr)$ for all $i\geq 0$.
  \end{claim}

  \begin{claimproof}
    For $i\geq1$, we have
    \begin{align*}
      \partr{V_1}{V_2}\hat\rho_i
      &=
        \partr{V_1}{V_2}E\pb\paren{\restricte{e_1'\land e_2'}(\hat\rho_{i-1})}
        \eqrefrel{eq:E.def-joint}=
        \denotc{\bc'}\pb\paren{\partr{V_1}{V_2}\restricte{e_1'\land e_2'}(\hat\rho_{i-1})}
        \starrel=
        \denotc{\bc'}\pb\paren{\partr{V_1}{V_2}\restricte{e_1'}(\hat\rho_{i-1})} 
\\&
        \starstarrel=
        \denotc{\bc'}\pb\paren{\restricte{e_1'}(\partr{V_1}{V_2}\hat\rho_{i-1})}
        =\pb\paren{\denotc{\bc'}\circ\restricte{e_1'}}(\partr{V_1}{V_2}\hat\rho_{i-1})
        .
    \end{align*}
    Here \eqref{eq:E.def-joint} can be applied because
    $\hat\rho_{i-1}$
    satisfies $A$ by \autoref{claim:rhoi.A-joint}
    and thus $\restricte{e_1'\land e_2'}(\hat\rho_{i-1})$
    satisfies $\CL{e_1'\land e_2'}\cap A$.
    $(*)$ follows from \autoref{claim:restrict.same} and the fact that
    $\hat\rho_{i-1}$
    satisfies $A$
    by \autoref{claim:rhoi.A-joint}.  And $(**)$
    uses the fact that $\fv(e_1')=\fv(\idx1e_1)\subseteq \cl{V_1}$.

    Since $\partr{V_1}{V_2}\hat\rho_0=\rho_1$
    by definition of $\hat\rho_0$
    and $\rho_1$,
    we get
    $\partr{V_1}{V_2}\hat\rho_i= \pb\paren{\denotc{\bc'}\circ\restricte {e_1'}}^i(\rho_1)$ for all $i\geq0$.

    And thus 
    \begin{align*}
      \partr{V_1}{V_2}\rho_i' =
      \partr{V_1}{V_2}\restricte {\lnot e'_1\land\lnot e'_2}(\hat\rho_i)
      \starrel=
      \partr{V_1}{V_2}\restricte {\lnot e'_1}(\hat\rho_i)=
      \restricte {\lnot e'_1}\pb\paren{\partr{V_1}{V_2}\hat\rho_i}=
      \restricte {\lnot e'_1}\bigl((\denotc{\bc'}\circ\restricte {e'_1})^i(\rho_1)\bigr).
    \end{align*}
    Here $(*)$
    uses \autoref{claim:restrict.same} and the fact that $\hat\rho_i$
    satisfies $A$ by \autoref{claim:rhoi.A-joint}.
  \end{claimproof}

  \begin{claim}\label{claim:rhoi'.formula-joint2}
    $\partr{V_2}{V_1}\rho_i' =
    \restricte {\lnot e_2'}\bigl((\denotc{\bd'}\circ\restricte {e_2'})^i(\rho_2)\bigr)$.
  \end{claim}
  
  \begin{claimproof}
    Analogous to \autoref{claim:rhoi'.formula-joint}.
  \end{claimproof}

  \begin{claim}\label{claim:rho'.exi-joint}
    $\rho'\in\traceposcq{V_1V_2}$ exists.
  \end{claim}

  \begin{claimproof}
    We have
    \begin{align*}
      \sum_{i=0}^\infty \tr \rho'_i
      &=
        \sum_{i=0}^\infty \tr \partr{V_1}{V_2} \rho'_i
        \starrel=
        \sum_{i=0}^\infty \tr
        \restricte {\lnot e_1'}\bigl((\denotc{\bc'}\circ\restricte {e'_1})^i(\rho_1)\bigr)
        =
        \tr      \sum_{i=0}^\infty 
        \restricte {\lnot e'_1}\bigl((\denotc{\bc'}\circ\restricte {e'_1})^i(\rho_1)\bigr)
      \\&
          \starstarrel=
          \tr\,\pb \denotc{\while{e'_1}{\bc'}}(\rho_1)
          \tristarrel\leq
          \tr \rho_1
          = \tr \rho.
    \end{align*}
    Here $(*)$
    follows from \autoref{claim:rhoi'.formula-joint}.  
    And $(**)$ follows from the semantics of while statements.
    And $(*{*}*)$
    follows since $ \denotc{\while{e'_1}{\bc'}}$ is a superoperator and thus trace-decreasing.

    Since $ \sum_{i=0}^\infty \tr \rho'_i \leq \tr\rho <\infty$,
    we have that $\rho'=\sum_{i=0}^\infty\rho_i'$ exists.
  \end{claimproof}

  \begin{claim}\label{claim:rho'.CLnot.A-joint}
    $\rho'$ is separable and satisfies $\CL{\lnot e'_1\land \lnot e'_2}\cap A$.
  \end{claim}

  \begin{claimproof}
    Since $\rho'_i$
    is separable and satisfies $\CL{\lnot e'_1\land \lnot e'_2}\cap A$
    by \autoref{claim:rhoi.A-joint}, we immediately have that
    $\rho'=\sum_i\rho'_i$
    is separable and satisfies $\CL{\lnot e'_1\land\lnot e'_2}\cap A$.
  \end{claimproof}

  \begin{claim}\label{claim:rho'.while-joint}
    $\partr{V_1}{V_2}\rho' = \pb\denotc{\idx1(\while{e_1}{\bc})}(\rho_1)$
    and
    $\partr{V_2}{V_1}\rho' = \pb\denotc{\idx2(\while{e_2}{\bd})}(\rho_2)$
  \end{claim}

  \begin{claimproof}
    We have
    \begin{align*}
      \partr{V_1}{V_2}\rho' &=
        \sum_{i=0}^\infty \partr{V_1}{V_2} \rho'_i
        \starrel=
        \sum_{i=0}^\infty
        \restricte {\lnot e'_1}\bigl((\denotc{\bc'}\circ\restricte {e'_1})^i(\rho_1)\bigr)
                              \starstarrel=
         \pb\denotc{\while{e'_1}{\bc'}}(\rho_1)
                              =
                              \pb\denotc{\idx1(\while{e_1}{\bc})}(\rho_1).
    \end{align*}
    Here $(*)$
    follows from \autoref{claim:rhoi'.formula-joint}.
    And $(**)$ follows from the semantics of while statements.

    The fact $\partr{V_2}{V_1}\rho' = \pb\denotc{\idx2(\while{e_2}{\bd})}(\rho_2)$
    is shown analogously, using \autoref{claim:rhoi'.formula-joint2}.
  \end{claimproof}

  To summarize, for any separable $\rho\in\traceposcq{V_1V_2}$
  that satisfies $A$,
  there is a $\rho'\in\traceposcq{V_1}$
  that is separable and satisfies $\CL{\lnot e'_1\land\lnot e'_2}\cap A$
  (\autoref{claim:rho'.CLnot.A-joint}), and such that
  $\pb\denotc{\idx1(\while
   {e_1}\bc)}\pb\paren{\partr{V_1}{V_2}\rho}=\partr{V_1}{V_2}\rho'$
  (\autoref{claim:rho'.while-joint}) and
  $\pb\denotc{\idx2(\while
   {e_2}\bd)}\pb\paren{\partr{V_2}{V_1}\rho}=\partr{V_2}{V_1}\rho'$
 (\autoref{claim:rho'.while-joint}). This implies
 \begin{gather*}
   \pb\rhl{A}{\while{e_1}\bc}{\while{e_2}{\bd}}{\CL{\lnot e'_1\land\lnot e'_2}\cap A}.
   \mathQED
 \end{gather*}
\end{proof}

\subsection{Rules for quantum statements}
\label{sec:proofs-quantum}

\begin{lemma}[Measurements (one-sided)]\label{rule-lemma:Measure1}
  \Ruleref{Measure1} holds.
\end{lemma}

\begin{proof}
  Fix   $m_1\in\types{\cl{V_1}}$,
  $m_2\in\types{\cl{V_2}}$ and normalized $\psi_1\in\elltwov{\qu{V_1}}$,
  $\psi_2\in\elltwov{\qu{V_2}}$ such that $\psi_1\tensor \psi_2\in \denotee A{\memuni{m_1m_2}}$.

  Let $M:=\denotee{\idx1e}{\memuni{m_1m_2}}$.
  Let $P_z:=M(z)=\denotee{(\idx1e)(z)}{\memuni{m_1m_2}}$.
  From the well-typedness of $\Qmeasure\xx eQ$,
  it follows that
  $M\in\Meas{\typev\xx}{\typel Q}=\Meas{\typev{\xx_1}}{\typel{\idx1Q}}$.
  Thus  $P_z$
  are mutually orthogonal projectors on $\elltwo{\typel{\idx1\!Q}}$   for $z\in\typev{\xx_1}$.
  Let
  $P_z':=\Uvarnames{Q'}
  P_z\adj{\Uvarnames{Q'}}\otimes\idv{\qu{V_1}\setminus Q'}$.  Then 
  $P_z'$
  are mutually orthogonal projectors on $\elltwov{\qu{V_1}}$
  for $z\in\typev{\xx_1}$,
  and $P_z'\otimes\idv{\qu{V_2}}=\denotee{e_z'}{\memuni{m_1m_2}}$.

  Since $\psi_1\tensor \psi_2\in \denotee A{\memuni{m_1m_2}}\subseteq
  \pb\denotee{\CL{e\text{ is a total measurement}}}{\memuni{m_1m_2}}$,
  we have
  $\denotee{e\text{ is a total measurement}}{\memuni{m_1m_2}}=\true$.
  Thus $M$ is a total measurement, i.e., $\sum_zP_z=\id$. Thus $\sum_z P_z'=\idv{\qu{V_1}}$.
  This implies
  \begin{equation}
    \label{eq:sumpz.psi1}
    \sum_{z\in\typev{\xx_1}} \!\! \tr \proj{P'_z\psi_1}
    = 
    \!\! \sum_{z\in\typev{\xx_1}} \!\! \pb\norm{P'_z\psi_1}^2
    \ \starrel=\ 
    \pB\norm{\, \overbrace{\sum_{z\in\typev{\xx_1}} P'_z}^{{}=\idv{\qu{V_1}}}\psi_1}^2
    =
    \norm{\psi_1}^2 = 1.
  \end{equation}
  Here $(*)$ uses that all $P'_z\psi_1$ are orthogonal (since the $P'_z$ are orthogonal).

  For all $z\in\typev{\xx_1}$,
  let
  $\phi_z:=\pb\basis{\cl{V_1}\cl{V_2}}{\upd{m_1}{\xx_1}{z}\, m_2} \tensor
  P'_z\psi_1 \tensor \psi_2$. Let
  $\rho':=\sum_{z\in\typev{\xx_1}}\proj{\phi_z}$.

  We have
  \begin{align}
    &\pb\denotc{\idx1(\Qmeasure{\xx}Qe)}\pB\paren{\pb\proj{\basis{\cl{V_1}}{m_1}\tensor\psi_1}}
      \notag\\&\hskip.6in
      =\pb\denotc{\Qmeasure{\xx_1}{Q'}{\idx1e}}\pB\paren{\pb\proj{\basis{\cl{V_1}}{m_1}\tensor\psi_1}}
      \notag\\&\hskip.6in
    \starrel=
                \!\!\sum_{z\in\typev{\xx_1}}\!\!\! \pB\proj{\pb\basis{\cl{V_1}}{\upd {m_1}{\xx_1}z} \tensor P'_z\psi_1}
%    \notag\\&
        =
               \partr{V_1}{V_2}\rho',
    \label{eq:rmeas.1}           \\
    &\pb\denotc{\idx2\Skip}\pB\paren{\pb\proj{\basis{\cl{V_2}}{m_2}\tensor\psi_2}}
      =\denotc{\Skip}\pB\paren{\pb\proj{\basis{\cl{V_2}}{m_2}\tensor\psi_2}}
      \starrel=
      \pb\proj{\basis{\cl{V_2}}{m_2}\tensor\psi_2}
      \notag\\&\hskip.6in
      \eqrefrel{eq:sumpz.psi1}=
                \!\!\sum_{z\in\typev{\xx_1}}\!\!\! \pb\proj{\basis{\cl{V_2}}{m_2}\tensor\psi_2} \cdot \tr\proj{P'_z\psi_1}
        =
      \partr{V_2}{V_1}\rho'.
      \label{eq:rmeas.2}
  \end{align}
  Here $(*)$ is by definition of the semantics of the measurement statement and the $\Skip$ statement.

  For all $z\in\typev{\xx_1}$, we have
  \begin{align*}
    \psi_1\tensor\psi_2 &\in\denotee A{\memuni{m_1m_2}}
                          \subseteq
                          \pB\denotee{{\substi B{z/\xx_1}\cap \im e'_z + \orth{(\im e'_z)} }}{\memuni{m_1m_2}} \\
    &=
      \denotee{B}{\upd{m_1}{\xx_1}z\,m_2}
      \cap\im(P'_z\otimes\idv{\qu{V_2}})
      +
      \orth{\pb\paren{\im(P'_z\otimes\idv{\qu{V_2}})}}.
  \end{align*}
  Since the $\orth{\pb\paren{\im(P'_z\otimes\idv{\qu{V_2}})}}$-part of $\psi_1\tensor\psi_2$ vanishes
  under $P'_z\tensor \idv{\qu{V_2}}$, and since the $\pb\paren{\denotee{B}{\upd{m_1}{\xx_1}z\,m_2}
  \cap\im(P'_z\otimes\idv{\qu{V_2}})}$-part is invariant under  $P'_z\tensor \idv{\qu{V_2}}$ and lies in $\denotee{B}{\upd{m_1}{\xx_1}z\,m_2}$, we have
  \begin{equation*}
    P'_z\psi_1\tensor\psi_2
    =
    (P'_z\tensor \idv{\qu{V_2}})(\psi_1\tensor\psi_2)
    \in       \denotee B{\upd{m_1}{\xx_1}z\, m_2}.
  \end{equation*}
  Thus
  \begin{equation*}
    \proj{\phi_z}=\pB\proj{\pb\basis{\cl{V_1}\cl{V_2}}{\upd{m_1}{\xx_1}{z}\, m_2} \tensor
     P'_z\psi_1 \tensor \psi_2}
    \text{ satisfies } B.
  \end{equation*}
  Thus $\rho':=\sum_{z\in\typev{\xx_1}}\proj{\phi_z}$ satisfies $B$.
  Furthermore, $\rho'$ is separable.

  Thus for all $m_1,m_2,\psi_1,\psi_2$
  with $\psi_1\tensor \psi_2\in\denotee{A}{\memuni{m_1m_2}}$,
  there is a separable $\rho'$
  that satisfies~$B$
  and such that \eqref{eq:rmeas.1} and \eqref{eq:rmeas.2} hold.
  By \autoref{lemma:pure}, this implies
  $\pb\rhl A{\Qmeasure{\xx}Qe}\Skip B$.
\end{proof}

\begin{lemma}[Measurements (joint)]\label{rule-lemma:JointMeasure}
  \Ruleref{JointMeasure} holds.
\end{lemma}

\begin{proof}\stepcounter{claimstep}
  Let $A^*:=C_f\cap C_e\cap A\cap(u_1Q_1'\quanteq u_2Q_2')$.

  Fix   $m_1\in\types{\cl{V_1}}$,
  $m_2\in\types{\cl{V_2}}$ and normalized $\psi_1\in\elltwov{\qu{V_1}}$,
  $\psi_2\in\elltwov{\qu{V_2}}$ such that $\psi_1\tensor \psi_2\in \denotee{A^*}{\memuni{m_1m_2}}$.

  Let $M_i:=\denotee{\idx ie_i}{\memuni{m_1m_2}}$.
  Let $P_{iz}:=M_i(z)=\denotee{\idx ie_i(z)}{\memuni{m_1m_2}}$.
  From the well-typedness of $\Qmeasure\xx{Q_1}{e_1}$
  and $\Qmeasure\yy{Q_2}{e_2}$,
  it follows that
  $M_1\in\Meas{\typev\xx}{\typel {Q_1}}=\Meas{\typev{\xx_1}}{\typel{Q_1'}}$
  and  $M_2\in\Meas{\typev\yy}{\typel {Q_2}}=\Meas{\typev{\yy_2}}{\typel{Q_2'}}$.
  Thus $P_{iz}$ 
  for $z\in\typev{\xx_1}$ or $z\in\typev{\yy_2}$, respectively,
  are mutually orthogonal projectors on $\elltwov{Q_i}$.
  Let
  \begin{equation}
    P_{iz}':=\Uvarnames{Q'_i}
    P_{iz}\adj{\Uvarnames{Q'_i}}\otimes\idv{\qu{V_i}\setminus Q'_i}.
    \label{eq:Piz'.def}
  \end{equation}
  Then the
  $P_{iz}'$
  with $z\in\typev{\xx_1}$ or  $z\in\typev{\yy_2}$, respectively,
  are mutually orthogonal projectors on $\elltwov{\qu{V_2}}$,
  and $P_{iz}'\otimes\idv{\qu{V_{3-i}}}=\denotee{e_{iz}'}{\memuni{m_1m_2}}$.

  Let $U_i:=\denotee{u_i}{\memuni{m_1m_2}}\in\iso{\typel{Q_i},Z}$.
  Let
  $U_i':=U_i\adj{\Uvarnames{Q_i'}}
  \in\iso{\types{Q_i'},Z}$.

  Let $R:=\denotee{f}{\memuni{m_1m_2}}\subseteq\typev\xx\times\typev\yy$. Since $\psi_1\tensor\psi_2\in\denotee{A^*}{\memuni{m_1m_2}}\subseteq C_f$, we have
  \begin{align*}
    &\forall x.\
    \pb\denotee{\idx1 e_1(x)}{\memuni{m_1m_2}}\neq0 \implies \pb\abs{\{y:(x,y)\in R\}}=1
    \\\text{and}\qquad
    &\forall y.\
    \pb\denotee{\idx2 e_2(y)}{\memuni{m_1m_2}}\neq0 \implies \pb\abs{\{x:(x,y)\in R\}}=1.
  \end{align*}
  Since $P_{iz}=\denotee{\idx i e_i(z)}{\memuni{m_1m_2}}$, this implies
  \begin{gather}
    \forall x.\
    P_{1x}\neq0 \implies \pb\abs{\{y:(x,y)\in R\}}=1
    \qquad\text{and}\qquad
    \forall y.\
    P_{2y}\neq0 \implies  \pb\abs{\{x:(x,y)\in R\}}=1.
    \label{eq:P.R}
  \end{gather}

  \begin{claim}\label{claim:samelen} For all $(x,y)\in R$,
    $\norm{P'_{1x} \psi_1}=\norm{P'_{2y}\psi_2}$.
  \end{claim}

  \begin{claimproof}
    Since
    $\psi_1\otimes\psi_2\in\denotee{A^*}{\memuni{m_1m_2}}\subseteq\denotee{u_1Q_1'\quanteq
      u_2Q_2'}{\memuni{m_1m_2}}=(U_1Q_1'\quanteq U_2Q_2')$, by
    \autoref{lemma:quanteq}, there are normalized
    $\psi^X_1\in\elltwov{Q_1'}$,
    $\psi^Y_1\in\elltwov{\qu{V_1}\setminus Q_1'}$,
    $\psi_2^X\in\elltwov{Q_2'}$,
    $\psi_2^Y\in\elltwov{\qu{V_2}\setminus Q_2'}$ such that
    \begin{equation}
      U_1'\psi_1^X=U_1\adj{\Uvarnames{Q_1'}}\psi_1^X=U_2\adj{\Uvarnames{Q_2'}}\psi_2^X
      =      U_2'\psi_2^X
      \text{ and }
      \psi_1=\psi_1^X\otimes\psi_1^Y
      \text{ and }
      \psi_2=\psi_2^X\otimes\psi_2^Y.
      \label{eq:psi1psi2}
    \end{equation}
    
    Since
    \begin{equation*}
      \psi_1\otimes\psi_2\in\denotee{A^*}{\memuni{m_1m_2}}\subseteq\denotee{C_e}{\memuni{m_1m_2}}
      = \pb\CL{\forall(x,y)\in R.\ U_1P_{1x}\adj{U_1}=U_2P_{2y}\adj{U_2}},
    \end{equation*}
    it holds that
    \begin{equation}
      U_1P_{1x}\adj{U_1}=U_2P_{2y}\adj{U_2}
      \label{eq:U1P1.U2P2}
    \end{equation}
    (note that $(x,y)\in R$ by assumption of the claim).

    We have\vspace{-12pt}
    \begin{align*}
      \pb\norm{P_{1x}'\psi_1}
      \ \ &\txtrel{\eqref{eq:Piz'.def},\eqref{eq:psi1psi2}}=\ \ 
        \pb\norm{(\Uvarnames{Q_1'}P_{1x}\adj{\Uvarnames{Q_1'}} \otimes \idv{\qu{V_1}\setminus Q_1'})(\psi_1^X\otimes\psi_1^Y)}
      =
      \pb\norm{\Uvarnames{Q_1'}P_{1x}\adj{\Uvarnames{Q_1'}}\psi_1^X} \cdot \overbrace{\pb\norm{ \idv{\qu{V_1}\setminus Q_1'}\psi_1^Y}}^{{}=1} \\
      &\starrel=
        \pb\norm{U_1P_{1x}\adj{\Uvarnames{Q_1'}}\psi_1^X}
        \starstarrel=
        \pb\norm{U_1P_{1x}\adj{U_1}U_1\adj{\Uvarnames{Q_1'}}\psi_1^X}
        \eqrefrel{eq:U1P1.U2P2}=
        \pb\norm{U_2P_{2y}\adj{U_2}U_1\adj{\Uvarnames{Q_1'}}\psi_1^X}
        \\&
            \tristarrel=
            \pb\norm{U_2P_{2y}\adj{U_2}U_1'\psi_1^X}
      \eqrefrel{eq:psi1psi2}=
            \pb\norm{U_2P_{2y}\adj{U_2}U_2'\psi_2^X}
            \tristarrel=
            \pb\norm{U_2P_{2y}\adj{U_2}U_2\adj{\Uvarnames{Q_2'}}\psi_2^X}
            \starstarrel=
            \pb\norm{U_2P_{2y}\adj{\Uvarnames{Q_2'}}\psi_2^X}
      \\
      &\starrel=
        \pb\norm{\Uvarnames{Q_2'}P_{2y}\adj{\Uvarnames{Q_2'}}\psi_2^X}
        =
        \pb\norm{\Uvarnames{Q_2'}P_{2y}\adj{\Uvarnames{Q_2'}}\psi_2^X} \cdot \underbrace{\pb\norm{ \idv{\qu{V_2}\setminus Q_2'}\psi_2^Y}}_{{}=1}
        \\&
        =
        \pb\norm{(\Uvarnames{Q_2'}P_{2y}\adj{\Uvarnames{Q_2'}} \otimes \idv{\qu{V_2}\setminus Q_2'})(\psi_2^X\otimes\psi_2^Y)} 
      \ \ \txtrel{\eqref{eq:Piz'.def},\eqref{eq:psi1psi2}}=\ \ 
      \pb\norm{P_{2y}'\psi_2}.
    \end{align*}
    Here $(*)$ follows since $\Uvarnames{Q_i'}$ and $U_i$ and are isometries, so neither change the norm.
    $(**)$ follows since $U_1$ and $U_2$ are isometries, and thus $\adj{U_1}U_1=\id$ and $\adj{U_2}U_2=\id$.
    And $(*{*}*)$ holds since $U_1'=\adj{U_1\Uvarnames{Q_1'}}$ and $U_2'=\adj{U_2\Uvarnames{Q_2'}}$ .
  \end{claimproof}

  For all $(x,y)\in R$,
  let
  \[
    \phi_{xy} := 
    \frac1{\norm{P'_{2y}\psi_2}}\cdot \pb\basis{\cl{V_1}\cl{V_2}}{\upd{m_1}{\xx_1}{x}\, \upd{m_2}{\yy_2}{y}} \tensor
    P'_{1x}\psi_1 \tensor P'_{2y}\psi_2.
  \]
  (When $P'_{2y}\psi_2=0$, let $\phi_{xy}:=0$.)
  Let
  $\rho':=\sum_{(x,y)\in R}\proj{\phi_{xy}}$.

  We have
  \begin{align}
    \partr{V_1}{V_2} \, \proj{\phi_{xy}} &=
    \frac1{\norm{P'_{2y}\psi_2}^2}
    \cdot
    \pB\proj{\pb\basis{\cl{V_1}}{\upd{m_1}{\xx_1}{x}}
    \tensor
    P'_{1x}\psi_1}\cdot \underbrace{\tr\proj{P'_{2y}\psi_2}}_{{}=\norm{P'_{2y}\psi_2}^2}
    \notag\\
    &=
    \pB\proj{\pb\basis{\cl{V_1}}{\upd{m_1}{\xx_1}{x}}
    \tensor
      P'_{1x}\psi_1}.
      \label{eq:tr1.phixy}
  \end{align}
  And we have
  \begin{align}
    \partr{V_2}{V_1} \, \proj{\phi_{xy}} &=
    \underbrace{ \frac{\norm{P'_{1x}\psi_1}^2}{\norm{P'_{2y}\psi_2}^2}}
      _{{}=1\text{ by \autoref{claim:samelen}}}
    \cdot
    \pB\proj{\pb\basis{\cl{V_2}}{\upd{m_2}{\yy_2}{y}}
    \tensor    P'_{2y}\psi_2}
                                         =
    \pB\proj{\pb\basis{\cl{V_2}}{\upd{m_2}{\yy_2}{y}}
    \tensor    P'_{2y}\psi_2}.
                                         \label{eq:tr2.phixy}
  \end{align}

  We have
  \begin{align}
    &\pb\denotc{\idx1(\Qmeasure{\xx}{Q_1}{e_1})}\pB\paren{\pb\proj{\basis{\cl{V_1}}{m_1}\tensor\psi_1}}
      \notag\\&\qquad
      =\denotc{\Qmeasure{\xx_1}{Q'_1}{\idx1e_1}}\pB\paren{\pb\proj{\basis{\cl{V_1}}{m_1}\tensor\psi_1}}
      \notag\\&\qquad
                \starrel=
                \sum_{x\in\typev{\xx_1}} \pB\proj{\pb\basis{\cl{V_1}}{\upd {m_1}{\xx_1}x} \tensor P'_{1x}\psi_1}
                =
                \sum_{x\in\typev{\xx_1}} \pB\proj{\pb\basis{\cl{V_1}}{\upd {m_1}{\xx_1}x} \tensor P'_{1x}\psi_1} \cdot
                \underbrace{\pb\abs{\{y:(x,y)\in R\}}}_{{}=1\text{ when }P'_{1x}\neq 0\text{ by }\eqref{eq:P.R}}
                \notag
    \\
    &\qquad
      =
    \sum_{(x,y)\in R} \pB\proj{\pb\basis{\cl{V_1}}{\upd {m_1}{\xx_1}x} \tensor P'_{1x}\psi_1}
      \eqrefrel{eq:tr1.phixy}=
      \sum_{(x,y)\in R}  \partr{V_1}{V_2} \proj{\phi_{xy}}
        =
        \partr{V_1}{V_2}\rho'.
    \label{eq:rmeas.1-joint}           
  \end{align}
  Here $(*)$ is by definition of the semantics of the measurement statement.

  And analogously (using \eqref{eq:tr2.phixy} instead of
  \eqref{eq:tr1.phixy}), we get
  \begin{equation}
    \pb\denotc{\idx2(\Qmeasure{\yy}{Q_2}{e_2})}\pB\paren{\pb\proj{\basis{\cl{V_2}}{m_2}\tensor\psi_2}}
    =
    \partr{V_2}{V_1}\rho'.
    \label{eq:rmeas.2-joint}          
  \end{equation}

  For any $(x,y)\in R=\denotee f{\memuni{m_1m_2}}$, we have
  \begin{align*}
    \psi_1\tensor\psi_2 &\in\denotee{A^*}{\memuni{m_1m_2}}
                          \subseteq
                          \denotee{A}{\memuni{m_1m_2}}
                          \subseteq
                          \pB\denotee{\substi B{x/\xx_1,y/\yy_2}\cap \im e'_{1x} \cap \im e'_{2y} + \orth{(\im e'_{1x})}  \orth{(\im e'_{2y})} }{\memuni{m_1m_2}} \\
    &=
      \denotee{B}{\upd{m_1}{\xx_1}x\,\upd{m_2}{\yy_2}y}
      \cap\im(P'_{1x}\otimes\idv{\qu{V_2}})
      \cap\im(P'_{2y}\otimes\idv{\qu{V_1}})
      +
      \orth{\pb\paren{\im(P'_{1x}\otimes\idv{\qu{V_2}})}}
      +
      \orth{\pb\paren{\im(P'_{2y}\otimes\idv{\qu{V_1}})}}.
  \end{align*}
  Since the $\orth{\pb\paren{\im(P'_{1x}\otimes\idv{\qu{V_2}})}}$-part
  and the $\orth{\pb\paren{\im(P'_{2y}\otimes\idv{\qu{V_1}})}}$-part of $\psi_1\tensor\psi_2$ both vanish
  under $P'_{1x}\tensor P'_{2y}$, and since the $\pb\paren{\denotee{B}{\upd{m_1}{\xx_1}x\,\upd{m_2}{\yy_2}y}
      \cap\im(P'_{1x}\otimes\idv{\qu{V_2}})
      \cap\im(P'_{2y}\otimes\idv{\qu{V_1}})}$-part is invariant under  $P'_{1x}\tensor P'_{2y}$ and lies
      in $\denotee{B}{\upd{m_1}{\xx_1}x\,\upd{m_2}{\yy_2}y}$, we have for all $(x,y)\in R$,
  \begin{equation*}
    P'_{1x}\psi_1\tensor P_{2y}\psi_2
    =
    (P'_{1x}\tensor P'_{2y})(\psi_1\tensor\psi_2)
    \in  \denotee{B}{\upd{m_1}{\xx_1}x\,\upd{m_2}{\yy_2}y}.
  \end{equation*}
  Thus for all $(x,y)\in R$, 
  \begin{equation*}
    \proj{\phi_{xy}}=\pB\proj{\pb\basis{\cl{V_1}\cl{V_2}}{\upd{m_1}{\xx_1}{x}\, \upd{m_2}{\yy_2}{y}} \tensor
     P'_{1x}\psi_1 \tensor P'_{2y} \psi_2}
    \text{ satisfies } B.
  \end{equation*}
  Thus $\rho'=\sum_{(x,y)\in R}\proj{\phi_{xy}}$ satisfies $B$.

  Furthermore, $\rho'$ is separable.

  Thus for all $m_1,m_2,\psi_1,\psi_2$
  with $\psi_1\tensor \psi_2\in\denotee{A^*}{\memuni{m_1m_2}}$,
  there is a separable $\rho'$
  that satisfies~$B$
  and such that \eqref{eq:rmeas.1-joint} and \eqref{eq:rmeas.2-joint} hold.
  By \autoref{lemma:pure}, this implies
  $\rhl{ A^*}{\Qmeasure{\xx}{Q_1}{e_1}}{\Qmeasure{\yy}{Q_2}{e_2}}B$.
\end{proof}

\begin{lemma}[Measurement (joint, same measurement)]\label{rule-lemma:JointMeasureSimple}
  \Ruleref{JointMeasureSimple} holds.
\end{lemma}

\begin{proof}
  Define the (constant) expression $f$ as
  $f:=\pb\braces{(x,x):x\in\typev\xx}$.  Since $\typev\xx=\typev\yy$,
  we have that $\typee f\subseteq\powerset{\typev\xx\times\typev\yy}$.
  Let $u_1,u_2:=\id$ be constant expressions where $\id$ is the
  identity on $\elltwo{\typel{Q_1}}$.  Let $Z:=\typel{Q_1}$.  Then
  $\typee{u_i}\subseteq\iso{\typel{Q_i},Z}$. (Since
  $\typel{Q_1}=\typel{Q_2}$.)  Let $Q_1',Q_2',e_{1z}',e_{2z}'$ be
  defined as in \ruleref{JointMeasureSimple}.  (Note that
  $Q_1',Q_2',e_{1z}',e_{2z}'$ are define in
  \ruleref{JointMeasureSimple} in the same way as in
  \ruleref{JointMeasure}, except that the index $z$ is called $x$ or
  $y$ there.)

  Define $C_f$
  and $C_e$
  as in \ruleref{JointMeasure}.  Then all conditions for
  \ruleref{JointMeasure} are satisfied, and we have
  \begin{equation}\label{eq:meas.simp.final}
   \pb \rhl{A^*}{\Qmeasure{\xx}{Q_1}{e_1}}{\Qmeasure{\yy}{Q_2}{e_2}}{B}
  \end{equation}
  with
  \begin{align*}
    A^* &:= C_f\cap C_e\cap B^*\cap (u_1Q_1'\quanteq u_2Q_2'), \\
    B^* &:=  \bigcap_{(x,y)\in f}
        \substi B{x/\xx_1,y/\yy_2}
          \cap \im e'_{1x} \cap \im e'_{2y}
          + \orth{(\im e'_{1x})} + \orth{(\im e'_{2y})}.
  \end{align*}
  (Note, $B^*$
  is called $A$
  in \ruleref{JointMeasure}. We cannot call it $A$
  here because in \ruleref{JointMeasureSimple} there is already a different
  predicate called $A$.)

  Since
  $\pb\abs{\{y:(x,y)\in f\}}=1$
  and $\pb\abs{\{x:(x,y)\in f\}}=1$
  for all $x,y\in\typev\xx=\typev\yy$, we have that $C_f=\CL{\true}$.
  And
  $C_e$ simplifies to:
  \begin{align*}
    C_e &= \pB\CL{\forall(x,y)\in f.\ u_1\pb\paren{\idx1e_1(x)}\adj{u_1}=u_2\pb\paren{\idx2e_2(y)}\adj{u_2}} \\
    &= \pB\CL{\forall x\in\typev\xx.\ \idx1e_1(x)=\idx2e_2(y)}
    = \pb\CL{\idx1e_1=\idx2e_2}.
  \end{align*}
  And
  \begin{equation*}
    \pb\paren{u_1Q_1'\quanteq u_2Q_2'} =
    \pb\paren{\id \, Q_1'\quanteq \id \, Q_2'}
    \starrel=
    \pb\paren{Q_1'\quanteq Q_2'}.
  \end{equation*}
  Here $(*)$ follows from \autoref{def:quanteq.simple}.
And 
\begin{equation*}
\bigcap_{z\in\typev{\xx}}
        \substi B{z/\xx_1,z/\yy_2}
          \cap \im e'_{1z} \cap \im e'_{2z}
          + \orth{(\im e'_{1z})} + \orth{(\im e'_{2z})}.
\end{equation*}

  Thus
  \begin{equation*}
    A^* = \CL{\true}\cap \CL{\idx1e_1=\idx2e_2}\cap (Q_1'\quanteq Q_2')\cap B^*
    = A.
  \end{equation*}
  Thus \eqref{eq:meas.simp.final} is the conclusion of
  \ruleref{JointMeasureSimple}.
\end{proof}

\begin{lemma}[Unitary quantum operation]\label{rule-lemma:QApply1}
  \Ruleref{QApply1} holds.
\end{lemma}

\begin{proof}
  Fix   $m_1\in\types{\cl{V_1}}$,
  $m_2\in\types{\cl{V_2}}$ and normalized $\psi_1\in\elltwov{\qu{V_1}}$,
  $\psi_2\in\elltwov{\qu{V_2}}$ such that $\psi_1\tensor \psi_2\in \denotee {\adj{e'}\cdot(B\cap \im e')}{\memuni{m_1m_2}}$.

  Let $\bc:=(\Qapply eQ)$, $\bd:=\Skip$, 
  
  Let $U:=\denotee{\idx1e}{\memuni{m_1m_2}}$.
  From the well-typedness conditions for programs, it follows that
  $\typee e\subseteq\iso{\typel Q}$,
  and thus $\denotee{e'}{\memuni{m_1m_2}}\in \isov{\qu{V_1}\qu{V_2}}$.
  Let $U':=\denotee{e'}{\memuni{m_1m_2}} = \Uvarnames{Q'} U\adj{\Uvarnames{Q'}} \tensor  \idv{\qu{V_1}\qu{V_1}\setminus Q'}$. (This uses \autoref{def:lift}.)
  % and thus $U':=\Urename{\idx1} e'\adj{\Urename{\idx1}}\in\isov{Q'}$.
  % Thus $U'':=U'\tensor\elltwov{\qu{V_1}\qu{V_2}\setminus Q'}\in\isov{\qu{V_1}\qu{V_2}}$.

  Let
  \begin{align*}
    \phi_1 &:= \pb\paren{\Uvarnames{Q'} U \adj{\Uvarnames{Q'}} \tensor \idv{\qu{V_1}\setminus Q'}}\psi_1,  \qquad
    \rho' := \pb\proj{\basis{\cl{V_1}\cl{V_2}}{\memuni{m_1m_2}}\tensor\phi_1\tensor\psi_2}.
  \end{align*}
  Then $\rho'$ is separable and
  \begin{align}
    \pb\denotc{\idx1\bc}\pB\paren{\pb\proj{\basis{\cl{V_1}}{m_1}\tensor\psi_1}}
    &=
      \pb\denotc{\Qapply{\idx1e}{Q'}}\pB\paren{\pb\proj{\basis{\cl{V_1}}{m_1}\tensor\psi_1}}
      \notag\\&
    \starrel= \pb\proj{\basis{\cl{V_1}}{m_1}\tensor\phi_1} =
      \partr{V_1}{V_2}\rho',
        \label{eq:qapply1.1}
    \\
    \pb\denotc{\idx1\bd}\pB\paren{\pb\proj{\basis{\cl{V_2}}{m_2}\tensor\psi_2}}
    &= \denotc\Skip\pB\paren{\pb\proj{\basis{\cl{V_2}}{m_2}\tensor\psi_2}}
    \starrel= \pb\proj{\basis{\cl{V_2}}{m_2}\tensor\psi_2} =
      \partr{V_2}{V_1}\rho'.
    \label{eq:qapply1.2}
  \end{align}
  Here $(*)$ follows from the semantics of the respective statements.

  Since
  $\psi_1\tensor\psi_2\in \denotee{{\adj{e'}\cdot(B\cap \im e')}}{\memuni{m_1m_2}}= \adj{U'}\cdot\pb\paren{\denotee B{\memuni{m_1m_2}}\cap \im U'} $,
  there exists a $\psi_0\in \denotee B{\memuni{m_1m_2}}\cap \im U'$ with $\psi_1\tensor\psi_2 = \adj {U'}\psi_0$.
  Thus
  \begin{align*}
    \phi_1\tensor\psi_2
    &=
    \pb\paren{\Uvarnames{Q'} U \adj{\Uvarnames{Q'}} \tensor \idv{\qu{V_1}\setminus Q'}}\psi_1 \tensor \psi_2
    =
      U'(\psi_1\tensor\psi_2)
      \\&
    =
    U'\adj{U'}\psi_0 
    \starrel=
    \psi_0
    \in
    \denotee{B}{\memuni{m_1m_2}}\cap \im U'
    \subseteq
    \denotee{B}{\memuni{m_1m_2}}.
  \end{align*}
  Here $(*)$
  follows since $U'\adj{U'}$
  is the projector onto $\im U'$
  (this holds for any isometry $U'$)
  and $\psi_0\in\denotee{ B}{\memuni{m_1m_2}}\cap\im U'\subseteq\im U'$.

  Since $\phi_1\tensor\psi_2\in\denotee{B}{\memuni{m_1m_2}}$,
  it follows that $\rho'$ satisfies $B$ by definition of $\rho'$.

  Thus for all $m_1,m_2,\psi_1,\psi_2$
  with $\psi_1\tensor \psi_2\in\denotee{{\adj{e'}\cdot(B\cap \im e')}}{\memuni{m_1m_2}}$,
  there is a separable $\rho'$
  that satisfies~$B$
  and such that \eqref{eq:qapply1.1} and \eqref{eq:qapply1.2} hold.
  By \autoref{lemma:pure}, this implies
  $\rhl {\adj{e'}\cdot(B\cap \im e')}\bc\bd B$.
  This proves the rule by definition of~$\bc,\bd$.
\end{proof}

\begin{lemma}[Quantum initialization]\label{rule-lemma:QInit1}
  \Ruleref{QInit1} holds.
\end{lemma}

\begin{proof}
  Let $A^*:=(\spaceat A{e'})\otimes\elltwov{Q'}$.
  Fix $m_1\in\types{\cl{V_1}}$,
  $m_2\in\types{\cl{V_2}}$
  and normalized $\psi_1\in\elltwov{\qu{V_1}}$,
  $\psi_2\in\elltwov{\qu{V_2}}$
  such that $\psi:=\psi_1\tensor \psi_2\in \denotee{A^*}{\memuni{m_1m_2}}$.

  Let $\phi:=\denotee{e'}{\memuni{m_1m_2}}\in \elltwov{Q'}$
  and $S:=\denotee{A}{\memuni{m_1m_2}}\subseteq\elltwov{\qu{V_1}\qu{V_2}}$.
  Then
  $\psi\in\denotee{A^*}{\memuni{m_1m_2}}=(\spaceat S\phi)\otimes\elltwov{ Q'}$.
  And from the well-typedness of $\Qinit eQ$,
  it follows that all $\typee e$
  contains only normalized vectors, thus $\typee{e'}=\typee{e}$
  contains only normalized vectors, thus $\norm\phi=1$.

  Let $\psi_1=\sum_i\lambda_i\psi_i^S\otimes\psi_i^Q$
  with $\lambda_i>0$ 
  and orthonormal $\psi_i^S\in\elltwov{\qu{V_1}\setminus Q'}$
  and $\psi_i^Q\in\elltwov{Q'}$
  be the Schmidt decomposition of $\psi_1$
  (\autoref{lemma:schmidt}).
  And $\sum_i\lambda_i^2=1$ since $\psi_1$ is normalized.
  Since
  $\psi=\sum_i\lambda_i\psi_i^S\otimes\psi_i^Q\otimes \psi_2 \in
  (\spaceat S\phi)\otimes\elltwov{ Q'}$, and since the $\psi_i^Q$
  are orthogonal, it follows that all
  $\psi_i^S\otimes\psi_2\in \spaceat S\phi$.
  By definition of $\spaceat{}{}$ (\autoref{def:spaceat}),
  this implies
  $\psi_i^S\tensor\phi\tensor\psi_2=(\psi_i^S\tensor\psi_2)\tensor\phi \in S=\denotee{A}{\memuni{m_1m_2}}$.
  Thus $\pb\proj{\basis{\cl{V_1}\cl{V_2}}{\memuni{m_1m_2}}\tensor\psi_i^S\tensor\phi\tensor\psi_2}$ satisfies $A$, and thus
  \begin{equation*}
    \rho' := \sum_i\lambda_i^2\pb\proj{\basis{\cl{V_1}\cl{V_2}}{\memuni{m_1m_2}}\tensor\psi_i^S\tensor\phi\tensor\psi_2}
  \end{equation*}
  satisfies $A$. And $\rho'$ is separable.

  We have
  \begin{align}
    &\pb\denotc{\idx1(\Qinit{Q}{e})}\pB\paren{\pb\pointstate{\cl{V_1}}{m_1}\tensor\proj{\psi_1}}\notag\\
    &\qquad\starrel=
      \pb\pointstate{\cl{V_1}}{m_1}\tensor
      \underbrace{\partr{\qu{V_1}\setminus Q'}{Q'}\proj{\psi_1}}_{{}\starstarrel=\ \sum_i\lambda_i^2\proj{\psi_i^S}}
      \tensor \pb\proj{\,\underbrace{\Uvarnames {Q'}\denotee{\idx1 e}{m_1}}_{{}=\denotee{e'}{m_1}=\phi}\,}
    \notag\\
    &\qquad=
      \sum_i\lambda_i^2 \, 
      \pb\proj{\basis{\cl{V_1}}{m_1}\tensor
      \psi_i^S
      \tensor\phi} = \partr{V_1}{V_2}\rho'.
\label{eq:qinit.1}     
  \end{align}
  Here $(*)$
  follows from the semantics of quantum assignments. And $(**)$
  follows since $\psi_1=\sum_i\lambda_i\psi_i^S\otimes\psi_i^Q$
  with orthonormal $\psi_i^Q\in\elltwov{Q'}$.

  Furthermore,
  \begin{multline}
    \pb\denotc{\idx2\Skip}\pB\paren{\pb\pointstate{\cl{V_1}}{m_1}\tensor\proj{\psi_1}}
    \starrel=\pb\pointstate{\cl{V_1}}{m_1}\tensor\proj{\psi_1} \\
    \starstarrel=\sum_i\lambda_i\,\pb\proj{\basis{\cl{V_1}}{m_1}\tensor\psi_1}\cdot \tr\proj{\phi}
      =\partr{V_2}{V_1}\rho'.
      \label{eq:qinit.2}     
    \end{multline}
    Here $(*)$
    follows from the semantics of $\Skip$,
    and $(**)$
    from $\sum_i\lambda_i=1$
    and $\norm\phi=1$.

  Thus for any separable $\rho$
  that satisfies $A^*$, there is a $\rho'$
  that is separable and satisfies $A$
  such that and \eqref{eq:qinit.1} and \eqref{eq:qinit.2} hold. By
  \autoref{lemma:pure}, this implies $\rhl{A^*}{\Qinit Qe}\Skip A$.
\end{proof}

\subsection{Proof of \ruleref{Adversary}}

\begin{lemma}\label{rule-lemma:Adversary}
  \Ruleref{Adversary} is sound.
\end{lemma}

\begin{proof}
  We show \ruleref{Adversary} by induction over the structure of $C$.
  That is, we distinguish five cases: 
  $C=\Box_i$, $C=\seq{C'}{C''}$, $C=\langif e{C'}{C''}$, $C=\while e{C'}$, and $C=\bc$. 
  And we assume that \rulerefx{Adversary} holds for $C',C''$.

  \medskip

  \noindent\textbf{Case $C=\Box_i$:}
  In this case, we need to prove $\rhl A{\bc_i}{\bc'_i}A$.
  But that is already a premise of the rule, so there is nothing to
  prove.

  \medskip

  \noindent\textbf{Case $C=\seq{C'}{C''}$:}
  Note that since the premises of the rule hold for $C$, they also hold for the subterms $C',C''$. 
  Thus we have 
  $\rhl A{C'[\bc_1,\dots,\bc_n]}{C'[\bc_1,\dots,\bc_n]}A$
  and $\rhl A{C''[\bc_1,\dots,\bc_n]}{C''[\bc_1,\dots,\bc_n]}A$
  by \ruleref{Adversary} (which holds for $C',C''$ by induction hypothesis).
  By \ruleref{Seq}, we get
  $\rhl A{C'[\bc_1,\dots,\bc_n];C''[\bc_1,\dots,\bc_n]}{C'[\bc_1,\dots,\bc_n];C''[\bc_1,\dots,\bc_n]}A$
  which is the same as
  $\rhl A{(C';C'')[\bc_1,\dots,\bc_n]}{(C';C'')[\bc_1,\dots,\bc_n]}A$ 
  and thus finishes the proof in this case.
  
  \medskip

  \noindent\textbf{Case $C=\langif e{C'}{C''}$:}
  Since $C$ is $XY$-local, we have that $\fv(e)\subseteq X$. Thus 
  $A\subseteq\CL{X_1=X_2}\subseteq\CL{\idx1e=\idx2e}$.
  Since the premises of the rule hold for $C$, they also hold for the subterms $C',C''$. 
  Thus we have 
  $\rhl A{C'[\bc_1,\dots,\bc_n]}{C'[\bc_1,\dots,\bc_n]}A$
  and $\rhl A{C''[\bc_1,\dots,\bc_n]}{C''[\bc_1,\dots,\bc_n]}A$
  by \ruleref{Adversary} (which holds for $C',C''$ by induction hypothesis).
  By \ruleref{Conseq}, we get
  $\rhl {\CL{\idx 1e \land \idx 2e}\cap A}{C'[\bc_1,\dots,\bc_n]}{C'[\bc_1,\dots,\bc_n]}A$
  and $\rhl{\CL{\lnot\idx 1e \land \lnot\idx 2e}\cap A}{C''[\bc_1,\dots,\bc_n]}{C''[\bc_1,\dots,\bc_n]}A$.
  Then, by \ruleref{JointIf}, we get 
  \[
    \pb\rhl{A}{\langif e{C'[\bc_1,\dots,\bc_n]}{C''[\bc_1,\dots,\bc_n]}}
    {\langif e{C'[\bc'_1,\dots,\bc'_n]}{C''[\bc'_1,\dots,\bc'_n]}}A.
  \]
  This is the same as 
  $ \rhl{A}{C[\bc_1,\dots,\bc_n]}{C[\bc'_1,\dots,\bc'_n]} A$
  and thus proves the rule in this case.

  \medskip

  \noindent\textbf{Case $C=\while e{C'}$:}
  Since $C$ is $XY$-local, we have that $\fv(e)\subseteq X$. Thus 
  $A\subseteq\CL{X_1=X_2}\subseteq\CL{\idx1e=\idx2e}$.
  Since the premises of the rule hold for $C$, they also hold for the subterm $C'$. 
  Thus we have 
  $\rhl A{C'[\bc_1,\dots,\bc_n]}{C'[\bc_1,\dots,\bc_n]}A$
  by \ruleref{Adversary} (which holds for $C'$ by induction hypothesis).
  By \ruleref{Conseq}, we get
  $\rhl {\CL{\idx 1e \land \idx 2e}\cap A}{C'[\bc_1,\dots,\bc_n]}{C'[\bc_1,\dots,\bc_n]}A$.
  Then, by \ruleref{JointWhile}, we get 
  \[
    \pb\rhl{A}{\while e{C'[\bc_1,\dots,\bc_n]}}
    {\while e{C'[\bc'_1,\dots,\bc'_n]}}{\CL{\lnot\idx 1e \land \lnot\idx 2e} \cap A}.
  \]
  This is the same as 
  $ \rhl{A}{C[\bc_1,\dots,\bc_n]}{C[\bc'_1,\dots,\bc'_n]}{\CL{\lnot\idx 1e \land \lnot\idx 2e} \cap A}$.
  By \ruleref{Conseq}, we get 
  $ \rhl{A}{C[\bc_1,\dots,\bc_n]}{C[\bc'_1,\dots,\bc'_n]} A$
  and thus proved the rule in this case.

  \medskip

  \noindent\textbf{Case $C=\bc$:}
  Since $C$ is $XY$-local, $\bc$ is $XY$-local and thus $XYW$-local.
  By \ruleref{Equal} (with $X:=X$, $Y:=YW$, $\bc:=\bc$), we get
  \[
    \pb\rhl{\CL{X_1=X_2}\cap \paren{Y_1W_1\quanteq Y_2W_2}}\bc\bc{\CL{X_1=X_2}\cap \paren{Y_1W_1\quanteq Y_2W_2}}.
  \]
  Then, by \ruleref{Frame} (with $X_1,X_2:=XY$, $Y_1,Y_2:=\Tilde XZ$ with indices removed from $Z$, $\bc,\bd:=\bc$,
  $A,B:={\CL{X_1=X_2}\cap \paren{Y_1W_1\quanteq Y_2W_2}}$, and $R:=R$), we get:
  \[
    \pb\rhl{\CL{X_1=X_2}\cap \paren{Y_1W_1\quanteq Y_2W_2}\cap R}\bc\bc{\CL{X_1=X_2}\cap \paren{Y_1W_1\quanteq Y_2W_2\cap R}}.
  \]
  (The premises for \rulerefx{Frame} all follow from the premises of \rulerefx{Adversary} with $C=\bc$.)
  This is the same as 
  $\rhl{A}\bc\bc A$ and thus proves the rule in this case.
\end{proof}

\section{Supporting theorems}
\label{sec:supporting}

In this appendix, we prove several theorems that are not required for
the use of qRHL but support various motivational claims made in the
paper.

\subsection{Uniqueness of the quantum equality}
\label{sec:qeq.unique}

In \autoref{sec:intro.qrhl} and \autoref{sec:quantum.eq}, we mentioned that we conjecture
that the quantum equality defined there (\autoref{def:quanteq.simple})
is the only notion of quantum equality (represented as a subspace)
that has the following property:
\[
  \pb\rhl
  {\CL{X_1=X_2}\cap (Q_1\quanteq Q_2)}
  \bc\bc
  {\CL{X_1=X_2}\cap (Q_1\quanteq Q_2)}
\]
(If $X,Q$
are all the variables of $\bc$.)
In this section, we give formal evidence for it. (Namely, while we
cannot show it for our definition of qRHL, we can show it for the
variant from \autoref{def:qrhl.first}.)

We present this claim in two variants: \autoref{lemma:unique.qeq}
shows that our quantum equality is the only predicate that satisfies
this property with respect to \autoref{def:qrhl.first}. (Except for the predicate $0$,
which is always false and satisfies the property trivially.)  This
lemma is specific to the definition of qRHL from \autoref{def:qrhl.first}.

\autoref{lemma:unique.qeq1} shows the claim in a more abstract form
that is independent of the definition of qRHL judgments. It states
that our quantum equality is the only predicate such that:
\begin{compactenum}[(a)]
\item  for all vectors $\psi$, $\proj{\psi\otimes\psi}$
  satisfies this predicate, and
\item if $\rho:=\proj\phi$
  satisfies the predicate, then the two marginals $\rho_1,\rho_2$
  of $\rho$ are equal.
\end{compactenum} (We omitted some variable renamings in this
description.) This lemma may be useful when analyzing quantum equality
for other variants of qRHL, or in the context of completely different
logics. (And this lemma also constitutes the central argument used in the proof of
\autoref{lemma:unique.qeq1}.)

Note that for simplicity, we have derived those facts only for the
finite dimensional case.

\begin{lemma}[Symmetric and conjugate-normal matrices \cite{mo-conj-normal}]\label{lemma:conj-normal}
  Let $M$
  denote the space of all complex $(n\times n)$-matrices.
  ($M$ is a complex vector space.)
  Let $C\subseteq M$
  be the set of all conjugate-normal\index{conjugate-normal} matrices (i.e., matrices $A$
  with $A\adj A=\overline{\adj AA}$
  where $\overline{\vphantom A\,\cdot\,}$ denotes the element-wise complex conjugate).
  Let $S\subseteq C$ denote all symmetric matrices. (Then $S$ is a subspace of $M$.)
  Assume that $X$ is a subspace of $M$ such that $S\subseteq X\subseteq C$.
  Then $X=S$.  
\end{lemma}

\begin{lemma}\label{lemma:unique.qeq1}
  Let $Q_1$
  and $Q_2$
  be disjoint lists of distinct quantum variables with finite
  $\typel{Q_1}=\typel{Q_2}=:T$.

  Let $E$ be a subspace of $\elltwov{Q_1Q_2}$ such that:
  \begin{compactitem}
  \item 
    For any $\psi\in\elltwo T$,
    we have that $\Uvarnames{Q_1}\psi\otimes\Uvarnames{Q_2}\psi\in E$.
  \item For any $\phi\in E$,
    let $\rho_1:=\partr{Q_1}{Q_2}\proj{\phi}$
    and $\rho_2:=\partr{Q_2}{Q_1}\proj{\phi}$.
    Then $\adj{\Uvarnames{Q_1}}\rho_1\Uvarnames{Q_1}
    =
    \adj{\Uvarnames{Q_2}}\rho_2\Uvarnames{Q_2}
    $.
  \end{compactitem}
      Then $E=(Q_1\quanteq Q_2)$.
\end{lemma}

\begin{proof}\stepcounter{claimstep}%
  \newcommand\ket{\basis{}}%
  \newcommand\bra[1]{\parenthesis\langle\rvert{#1}}%
  \newcommand\inner[2]{\parenthesis\langle\rangle{#1\vert#2}}%
  Let $n:=\abs T$ and $T=:\{t_1,\dots,t_n\}$.
  We have that $\Uvarnames{Q_1}\basis{}{t_i}$ form a basis for $\elltwov{Q_1}$
  and $\Uvarnames{Q_2}\basis{}{t_i}$ form a basis for $\elltwov{Q_2}$.
  Thus 
  any $\phi\in\elltwov{Q_1Q_2}$
  can be written as
  \[
    \phi = \sum_{ij} a_{ij} \Uvarnames{Q_1}\basis{}{t_i}\otimes\Uvarnames{Q_2}\basis{}{t_j}
    \qquad\text{with }a_{ij}\in \setC
    .
  \]
  Then let $A_\phi$ denote the $(n\times n)$-matrix with entries $a_{ij}$.
  That is,
  \[
    A_\phi := \sum_{ij} a_{ij} \ket i\bra j
    =
    \sum_{ij} \pB\paren{\bra{t_i}\adj{\Uvarnames{Q_1}}\tensor\bra{t_j}\adj{\Uvarnames{Q_2}}}\phi
    \cdot\ket i\bra j
  \]
  where $\bra j$ denotes $\adj{\ket j}$.

  Let $X:=\{A_\phi:\phi\in E\}$. Let $S:=\{A_\phi:\phi\in(Q_1\quanteq Q_2)\}$.

  \begin{claim}\label{claim:X.is.ss}
    $X$ is a subspace of the set of complex $(n\times n)$-matrices.
  \end{claim}

  \begin{claimproof}
    Immediate from the facts that $E$ is a  subspace, and that $\phi\mapsto A_\phi$ is linear.
  \end{claimproof}

  \begin{claim}\label{claim:S.sym}
    $S$ is the set of symmetric matrices.
  \end{claim}

  \begin{claimproof}
    Since $\phi\mapsto A_\phi$
    is bijective, it suffices to show  for all $\phi\in\elltwov{Q_1Q_2}$ that $A_\phi$
    is symmetric iff $\phi\in(Q_1\quanteq Q_2)$.
    Fix some $\phi\in\elltwov{Q_1Q_2}$, and let
    $\phi =: \sum_{ij} a_{ij} \Uvarnames{Q_1}\basis{}{t_i}\otimes\Uvarnames{Q_2}\basis{}{t_j}$.
    Then $a_{ij}$ are the entries of $A_\phi$.
    We have
    \begin{align*}
      &\phi\in (Q_1\quanteq Q_2) 
      \\ \Longleftrightarrow\qquad&
                                    \phi
                                    = \paren{\Uvarnames{Q_2}\adj{\Uvarnames{Q_1}} \otimes
                                    \Uvarnames{Q_1}\adj{\Uvarnames{Q_2}}}\phi
      \\ \Longleftrightarrow\qquad&
                                    \sum_{ij} a_{ij} \Uvarnames{Q_1}\basis{}{t_i}\otimes\Uvarnames{Q_2}\basis{}{t_j}
                                    =
                                    \sum_{ij} a_{ij} \Uvarnames{Q_2}\basis{}{t_i}\otimes\Uvarnames{Q_1}\basis{}{t_j}
      \\ \Longleftrightarrow\qquad&
                                    \sum_{ij} a_{ij} \Uvarnames{Q_1}\basis{}{t_i}\otimes\Uvarnames{Q_2}\basis{}{t_j}
                                    =
                                    \sum_{ij} a_{ji} \Uvarnames{Q_1}\basis{}{t_j}\otimes\Uvarnames{Q_2}\basis{}{t_i}
      \\ \Longleftrightarrow\qquad&
                                    \forall i,j.\ a_{ij}=a_{ji}
      \\ \Longleftrightarrow\qquad&
                                    A_\phi\text{ symmetric}.
                                    \mathQED
    \end{align*}
  \end{claimproof}
  
  \begin{claim}\label{claim:X.subspace}
    $S\subseteq X$.
  \end{claim}
  
  \begin{claimproof}
    By assumption, $\phi:=\Uvarnames{Q_1}\psi\otimes\Uvarnames{Q_2}\psi\in E$ for any $\psi\in\elltwo T$.
    Thus
    \begin{align*}
      \sum_{ij}
      \bra{t_i}\psi\cdot \bra{t_j}\psi
      \cdot\ket i\bra j
      =
      \sum_{ij}
      \pB\paren{\bra{t_i}\adj{\Uvarnames{Q_1}}\tensor\bra{t_j}\adj{\Uvarnames{Q_2}}}
      (\Uvarnames{Q_1}\psi\otimes\Uvarnames{Q_2}\psi)
      \cdot\ket i\bra j
      =
      A_\phi
      \in
      X.
    \end{align*}
    In particular, with $\psi:=\ket{t_i}$, and with $\psi:=\ket{t_i}+\ket{t_j}$, we get
    \[
      \ket{i}\bra{i}\in X
      \qquad\text{and}\qquad
      \ket{i}\bra{i}+\ket{j}\bra{j}+
      \ket{i}\bra{j}+\ket{j}\bra{i}
      \in X
      \qquad\text{for all }i,j.
    \]
    And since $X$ is a subspace (\autoref{claim:X.is.ss}), also $\ket{i}\bra{j}+\ket{j}\bra{i}\in X$.

    Since the matrices $\ket{i}\bra{i}$
    and $\ket{i}\bra{j}+\ket{j}\bra{i}$
    form a basis for the space of symmetric matrices,
    it follows that all symmetric matrices are contained in $X$.
    By \autoref{claim:S.sym}, $S$ is the set of all symmetric matrices.
    Thus $S\subseteq X$.
  \end{claimproof}

  Let $U_T\ket i:=\ket{t_i}$ for $i=1,\dots,n$. Then $U_T$ is unitary.

  \begin{claim}\label{claim:tr1.phi}
    For any $\phi\in\elltwov{Q_1Q_2}$, we have
    $\adj{\Uvarnames{Q_1}} \pb\paren{ \partr{Q_1}{Q_2}\proj\phi } \Uvarnames{Q_1}=
    U_TA_\phi\adj {A_\phi}\adj{U_T}$
  \end{claim}

  \begin{claimproof}
    Let $a_{ij}$ be the entries of $A_\phi$. Then
    \[
      \phi = \sum_{ij} a_{ij} \Uvarnames{Q_1}\basis{}{t_i}\otimes\Uvarnames{Q_2}\basis{}{t_j}
    \]
    and thus
    \begin{align*}
      &\adj{\Uvarnames{Q_1}} \pb\paren{ \partr{Q_1}{Q_2}\proj\phi } \Uvarnames{Q_1} \\[-5pt]
      &=
        \adj{\Uvarnames{Q_1}}
        \pB\paren{
        \partr{Q_1}{Q_2}\sum_{ijkl} a_{ij} \adj{a_{kl}}\
        \Uvarnames{Q_1}\ket{t_i}\bra{t_k}\adj{\Uvarnames{Q_1}}
        \otimes
        \overbrace{\Uvarnames{Q_2}\ket{t_j}\bra{t_l}\adj{\Uvarnames{Q_2}}}
        ^{\tr(\cdot) = 1\text{ iff }j=l,\ \tr(\cdot)=0\text{ else}}
        }
        \Uvarnames{Q_1}
      \\
      &=
        \sum_{ijk} a_{ij} \adj{a_{kj}}\
        \ket{t_i}\bra{t_k}
      =
        U_T
        \pB\paren{\sum_{ijk} a_{ij} \adj{a_{kj}}\
        \ket{i}\bra{k}}\adj{U_T}
=     
    U_T A_\phi\adj{A_\phi}
    \adj{U_T}.
    \mathQED
    \end{align*}
  \end{claimproof}

  \begin{claim}\label{claim:tr2.phi}
    For any $\phi\in\elltwov{Q_1Q_2}$, we have
    $\adj{\Uvarnames{Q_2}} \pb\paren{ \partr{Q_2}{Q_1}\proj\phi } \Uvarnames{Q_2}=
    U_T \overline{\adj{A_\phi}A_\phi} \adj{U_T}$.
    (Where $\overline{\vphantom A\,\cdot\,}$ denotes the element-wise complex conjugate.)
  \end{claim}

  \begin{claimproof}
    Let $a_{ij}$ be the entries of $A_\phi$. Then
    \[
      \phi = \sum_{ij} a_{ij} \Uvarnames{Q_1}\basis{}{t_i}\otimes\Uvarnames{Q_2}\basis{}{t_j}
    \]
    and thus
    \begin{align*}
      \hskip1cm &\hskip-1cm \adj{\Uvarnames{Q_2}} \pb\paren{ \partr{Q_2}{Q_1}\proj\phi } \Uvarnames{Q_2} \\[-5pt]
      &=
        \adj{\Uvarnames{Q_2}}
        \pB\paren{
        \partr{Q_2}{Q_1}\sum_{ijkl} a_{ij} \adj{a_{kl}}\
        \overbrace{\Uvarnames{Q_1}\ket{t_i}\bra{t_k}\adj{\Uvarnames{Q_1}}}
        ^{\tr(\cdot) = 1\text{ iff }i=k,\ \tr(\cdot)=0\text{ else}}
        \otimes
        \Uvarnames{Q_2}\ket{t_j}\bra{t_l}\adj{\Uvarnames{Q_2}}
        }
        \Uvarnames{Q_2}
      \\
      &=
        \sum_{jil} a_{ij} \adj{a_{il}}\
        \ket{t_j}\bra{t_l}
      =
        U_T
        \pB\paren{\sum_{jil} a_{ij} \adj{a_{il}}\
        \ket{j}\bra{l}}\adj{U_T}
=     
    U_T \overline{\adj{A_\phi}}\ \overline{ A_\phi}
    \adj{U_T}
=     
    U_T \overline{\adj{A_\phi}{A_\phi}} \adj{U_T}.
    \mathQED
    \end{align*}
  \end{claimproof}

  \begin{claim}\label{claim.X.cn}
    Any $A\in X$ is conjugate-normal. (Cf.~\autoref{lemma:conj-normal}.)
  \end{claim}

  \begin{claimproof}
    If $A\in X$,
    then $A=A_\phi$
    for some $\phi\in E$ (by definition of $E$).

    Then
    $\adj{\Uvarnames{Q_1}}\paren{\partr{Q_1}{Q_2}\proj{\phi}}\Uvarnames{Q_1} =
    \adj{\Uvarnames{Q_2}}\paren{\partr{Q_2}{Q_1}\proj{\phi}}\Uvarnames{Q_2} $.
    (This is one of the  assumptions of the lemma.)
    By \autoref{claim:tr1.phi} and \autoref{claim:tr2.phi}, this implies
    $U_TA_\phi\adj{A_\phi}\adj{U_T}=U_T\overline{\adj{A_\phi}A_\phi}\adj{U_T}$. Since $U_T$ is unitary, this implies
    $A_\phi\adj {A_\phi}=\overline{\adj{A_\phi}{A_\phi}}$. Thus $A=A_\phi$ is conjugate-normal by definition.
  \end{claimproof}

  \begin{claim}\label{claim:X=S}
    $X=S$.
  \end{claim}

  \begin{claimproof}
    Let $M$
    denote the space of all complex $(n\times n)$-matrices,
    and $C$
    the set of all conjugate-normal matrices. By
    \autoref{claim:S.sym}, $S$
    is the set of all symmetric matrices. By
    \autoref{claim:X.subspace}, we have $S\subseteq X$,
    and by \autoref{claim.X.cn}, we have $X\subseteq C$.
    $X$
    is a subspace by \autoref{claim:X.is.ss}. Thus by
    \autoref{lemma:conj-normal}, $X=S$.
  \end{claimproof}

  Since $X=S$ by \autoref{claim:X=S}, we have
  $\{A_\phi:\phi\in E\}=\{A_\phi:\phi\in(Q_1\quanteq Q_2)\}$ by definition of $X$ and~$S$.
  Since $\phi\mapsto A_\phi$ is injective, this implies
  $E=(Q_1\quanteq Q_2)$.
\end{proof}

\autoref{lemma:unique.qeq1} only gives an abstract characterization of
$\quanteq$.
It leaves open whether it is the only notion of quantum equality that
satisfies natural rules when used in qRHL judgments. While we were
unable to show this for our notion of qRHL (\autoref{def:rhl}), we
were able to show it for the variant of qRHL presented in
\autoref{def:qrhl.first} which does not require the initial state and
final state to be separable:

\begin{lemma}\label{lemma:unique.qeq}
  Let $Q\subseteq\qu V$
  be a list of distinct quantum variables with finite
  $\typel{Q}$. Let $Q_i:=\idx iQ$.
  
  Let $E_{Q_1,Q_2}$
  be a $Q_1Q_2$-local
  predicate that satisfies \ruleref{Equal} with respect to the qRHL definition from \autoref{def:qrhl.first}.  That is, for
  $E_{Q_1,Q_2}$ the following rule holds:
  \[
    \inferrule{
      \text{$X$ is a list of classical variables} \\
      \text{$Q$ is a list of distinct quantum variables} \\
      \bc\text{ is $XQ$-local} \\
      X_i:=\idx iX\\
      Q_i:=\idx iQ
    }{
      \pb\rhlfirst
      {\CL{X_1=X_2}\cap E_{Q_1,Q_2}}
      \bc\bc
      {\CL{X_1=X_2}\cap E_{Q_1,Q_2}}
    }
  \]
  Then $E_{Q_1,Q_2}=(Q_1\quanteq Q_2)$ or $E_{Q_1,Q_2}=0$.
\end{lemma}

\begin{proof}\stepcounter{claimstep}
  Assume that $E_{Q_1,Q_2}\neq 0$.
  Since $E_{Q_1,Q_2}$
  is $Q_1Q_2$-local,
  we have that $E_{Q_1,Q_2}=E'\otimes\elltwov{\qu{V_1}\qu{V_2}\setminus Q_1Q_2}$ for some $E'\neq 0$.

  \begin{claim}\label{claim:phiphiE}
    $\Uvarnames{Q_1}\psi\otimes\Uvarnames{Q_2}\psi\in E'$
    for all $\psi\in\elltwo{\typel Q}$.
  \end{claim}

  \begin{claimproof}
    Without loss of generality, we can assume that $\psi$ is a unit vector.
    From the rule in the lemma, we get
    $\rhl{E_{Q_1,Q_2}}{\Qinit{Q}{\psi}}{\Qinit{Q}{\psi}}{E_{Q_1,Q_2}}$.
    Fix some $\phi\in E_{Q_1,Q_2}$.
    From the semantics of $\Qinit Q\psi$, we get:
    \begin{align*}
    \rho_1'&:=\denotc{\idx1(\Qinit Q\psi)}(\partr{V_1}{V_2}\proj\phi)
             = \proj{\Uvarnames{Q_1}\psi} \otimes \rho_1''
    \\
    \rho_2'&:=\denotc{\idx2(\Qinit Q\psi)}(\partr{V_2}{V_1}\proj\phi)
             = \proj{\Uvarnames{Q_2}\psi} \otimes \rho_2''
  \end{align*}
    for some positive $\rho_1'',\rho_2''$
    with $\tr\rho_1''=\tr\rho_2''=1$.
  
    By definition of qRHL judgments, there is a $\rho'$
    such that $\partr{V_1}{V_2}\rho'=\rho_1'$
    and $\partr{V_2}{V_1}\rho'=\rho_2'$
    and
    $\suppo\rho'\subseteq E_{Q_1,Q_2}=E'\otimes\elltwov{\qu{V_1}\qu{V_2}\setminus Q_1Q_2}$.
    From
    $\partr{V_1}{V_2}\rho'=\rho_1'$
    and $\partr{V_2}{V_1}\rho'=\rho_2'$,
    it follows that
    $\rho'=\proj{\Uvarnames{Q_1}\psi\otimes\Uvarnames{Q_2}\psi}\otimes\rho''$
    for some positive $\rho''$ with $\tr\rho''=1$.
    So $\suppo\rho'=\suppo\pb\paren{\proj{\Uvarnames{Q_1}\psi\otimes\Uvarnames{Q_2}\psi}\otimes\rho''}
    \subseteq E'\otimes\elltwov{\qu{V_1}\qu{V_2}\setminus Q_1Q_2}$.
    Hence $\Uvarnames{Q_1}\psi\otimes\Uvarnames{Q_2}\psi\in E'$.
  \end{claimproof}

  \begin{claim}\label{claim:same.tr}
    For any $\phi\in E'$,
    let $\tilde\rho_1:=\partr{Q_1}{Q_2}\proj{\phi}$
    and $\tilde\rho_2:=\partr{Q_2}{Q_1}\proj{\phi}$.
    Then
    $\adj{\Uvarnames{Q_1}}\tilde\rho_1\Uvarnames{Q_1} =
    \adj{\Uvarnames{Q_2}}\tilde\rho_2\Uvarnames{Q_2} $.
  \end{claim}

  \begin{claimproof}
    Assume for contradiction that for some
    $\phi\in E'$ the claim does not hold.
    Then
    $\adj{\Uvarnames{Q_1}}\tilde\rho_1\Uvarnames{Q_1} \neq
    \adj{\Uvarnames{Q_2}}\tilde\rho_2\Uvarnames{Q_2}$. Hence there is a
    projector $P$ such that 
    \begin{equation}
      p_1:=\tr P\adj{\Uvarnames{Q_1}}\tilde\rho_1\Uvarnames{Q_1} \neq
      \tr P\adj{\Uvarnames{Q_2}}\tilde\rho_2\Uvarnames{Q_2}=:p_2.
      \label{eq:p1.neq.p2}
    \end{equation}
    From the rule in the lemma, we get
    \begin{equation}
      \rhlfirst
      {\CL{\xx_1=\xx_2}\cap{E_{Q_1,Q_2}}}
      \bc\bc
      {\CL{\xx_1=\xx_2}\cap{E_{Q_1,Q_2}}}
      \label{eq:pres.eq.P}
    \end{equation}
    where $\bc:=      {\Qmeasure{\xx}{Q}M}$ and
       $\typev\xx:=\bool$
    and $M\in\Meas\bool{\typel{Q}}$
    is the measurement with $M(\true)=P$
    and $M(\false)=1-P$.

    Fix some normalized basis state $\psi_1\in\elltwov{V_1\setminus Q_1\xx_1}$,
    $\psi_2\in\elltwov{V_2\setminus Q_2\xx_2}$. Let $\rho:=\proj{\phi\otimes\psi_1\otimes\psi_2\otimes\basis{\xx_1}{\false}\otimes\basis{\xx_2}{\false}}$. Note that $\rho$ satisfies $ {\CL{\xx_1=\xx_2}\cap{E_{Q_1,Q_2}}}$.
    \begin{align}
      \rho_1' &:=
      \denotc{\idx1\paren{\Qmeasure\xx Q M}}(\partr{V_1}{V_2}\rho) \notag\\
      &=
        \denotc{\Qmeasure{\xx_1} {Q_1} M}(\tilde\rho_1\otimes\proj{\basis{\xx_1}{\false}\otimes\psi_1})\notag\\
        &\starrel=
        \proj{\basis{\xx_1}\true}
          \otimes
          \Uvarnames{Q_1}P\adj{\Uvarnames{Q_1}}\tilde\rho_1
          \Uvarnames{Q_1}P\adj{\Uvarnames{Q_1}}
          \otimes\proj{\psi_1}\notag\\
&+        \proj{\basis{\xx_1}\false}
          \otimes
          \Uvarnames{Q_1}(1-P)\adj{\Uvarnames{Q_1}}
          \tilde\rho_1
                    \Uvarnames{Q_1}(1-P)\adj{\Uvarnames{Q_1}}
          \otimes\proj{\psi_1}.
          \label{eq:rho1P}
    \end{align}
    Here $(*)$ follows from the semantics of ${\Qmeasure{\xx_1} {Q_1} M}$.
    
    And analogously
\begin{align*}
      \rho_2' &:=
      \denotc{\idx2\paren{\Qmeasure\xx Q M}}(\partr{V_2}{V_1}\rho) \\
      &=
        \proj{\basis{\xx_2}\true}
          \otimes
        \Uvarnames{Q_2}P\adj{\Uvarnames{Q_2}}\tilde\rho_2
          \Uvarnames{Q_2}P\adj{\Uvarnames{Q_2}}
          \otimes\proj{\psi_2}\notag\\
&+        \proj{\basis{\xx_2}\false}
          \otimes
          \Uvarnames{Q_2}(1-P)\adj{\Uvarnames{Q_2}}
          \tilde\rho_2
                    \Uvarnames{Q_2}(1-P)\adj{\Uvarnames{Q_2}}
          \otimes\proj{\psi_2}
    \end{align*}
    Since $\rho$ satisfies $\CL{\xx_1=\xx_2}\cap E_{Q_1,Q_2}$,
    from \eqref{eq:pres.eq.P} it follows that there is a $\rho'$
    such that $\partr{V_1}{V_2}\rho'=\rho_1'$
    and $\partr{V_2}{V_1}\rho'=\rho_2'$
    and $\suppo\rho'\subseteq {\CL{\xx_1=\xx_2}\cap{E_{Q_1,Q_2}}}$.
    Since $\suppo\rho'\subseteq\CL{\xx_1=\xx_2}$,
    it follows that
    $\rho'=\proj{\basis{\xx_1\xx_2}{\true,\true}}\otimes
    \rho'_{\true}+
    \proj{\basis{\xx_1\xx_2}{\false,\false}}\otimes
    \rho'_{\false}$
    for some $\rho'_{\true},\rho'_{\false}$.
    Hence
    \begin{align*}
      \rho_1' = \partr{V_1}{V_2}\rho =
      \proj{\basis{\xx_1}\true}\otimes\partr{V_1\setminus\xx_1}{V_2\setminus\xx_2}\rho'_{\true}
      +
      \proj{\basis{\xx_1}\false}\otimes\partr{V_1\setminus\xx_1}{V_2\setminus\xx_2}\rho'_{\false}
    \end{align*}
    With \eqref{eq:rho1P}, it follows that
    \begin{equation}
      \Uvarnames{Q_1}P\adj{\Uvarnames{Q_1}}\tilde\rho_1
      \Uvarnames{Q_1}P\adj{\Uvarnames{Q_1}}
      \otimes \proj{\psi_1}
      =
      \partr{V_1\setminus\xx_1}{V_2\setminus\xx_2}\rho'_{\true}
      \label{eq:Prho2}
    \end{equation}
    and hence
    \begin{align*}
      p_1
      & \eqrefrel{eq:p1.neq.p2}=
      \tr P\adj{\Uvarnames{Q_1}}\tilde\rho_1\Uvarnames{Q_1}
      =
      \tr \pb\paren{  \Uvarnames{Q_1}P\adj{\Uvarnames{Q_1}}\tilde\rho_1
        \Uvarnames{Q_1}P\adj{\Uvarnames{Q_1}}
        \otimes\proj{\psi_1} } \\
      &\eqrefrel{eq:Prho2}=
      \tr \partr{V_1\setminus\xx_1}{V_2\setminus\xx_2}\rho'_{\true}
      =
      \tr\rho'_{\true}.
    \end{align*}
    Analogously,
    \[
      p_2=\tr\rho'_{\true}.
    \]
    Thus $p_1= p_2$ in contradiction to \eqref{eq:p1.neq.p2}.
  \end{claimproof}

  From \autoref{claim:phiphiE} and \autoref{claim:same.tr}, it follows
  that the preconditions of \autoref{lemma:unique.qeq1} (with
  $Q_1:=Q_1$,
  $Q_2:=Q_2$,
  $E:=E'$) are satisfied. Thus $E'=(Q_1\quanteq Q_2)$.

  Then
  $E=E'\otimes\elltwov{\qu{V_1}\qu{V_2}\setminus Q_1Q_2}= (Q_1\quanteq
  Q_2)\otimes\elltwov{\qu{V_1}\qu{V_2}\setminus Q_1Q_2} =(Q_1\quanteq
  Q_2)$.  (Here the first $Q_1\quanteq Q_2$
  is the equality predicate over $\qu{V_1}\qu{V_2}\setminus Q_1Q_2$,
  while the second $Q_1\quanteq Q_2$ is the equality predicate over
   $\qu{V_1}\qu{V_2}$.)
\end{proof}

\autoref{lemma:unique.qeq} does not, in this form, hold for the
definition of qRHL that we actually use (\autoref{def:rhl}). For
example
$E_{\qq,\qq'}:=\SPAN\{\basis{\qq\qq'}{01}+\basis{\qq\qq'}{01}\}$
satisfies the rule from \autoref{lemma:unique.qeq} for that
definition. Why? $E_{\qq,\qq'}$ contains no separable states, therefore the only separable $\rho$ that satisfies 
$E_{\qq,\qq'}$ is $\rho=0$. Thus, in terms of which states satisfy it, 
$E_{\qq,\qq'}$ is equivalent to $0$, and thus
\begin{equation*}
  \rhl
  {\CL{\xx_1=\xx_2}\cap{E_{Q_1,Q_2}}}
  \bc\bc
  {\CL{\xx_1=\xx_2}\cap{E_{Q_1,Q_2}}}
\end{equation*}
holds trivially (since the precondition is not satisfiable). However,
that does not violate the spirit of \autoref{lemma:unique.qeq} because
$E_{\qq,\qq'}$
is essentially $0$
as far as \autoref{def:rhl} is concerned, not a different predicate.

More formally, we say $A\approx B$
iff for all separable $\rho\in\traceposcq{V_1V_2}$,
we have that $\rho$
satisfies $A$
iff $\rho$
satisfies $B$.
Then $E_{\qq,\qq'}\approx 0$.
With this notation, we can formulate the conjecture that with respect
to \autoref{def:rhl}, $\quanteq$
is still essentially the only equality notion:

\begin{conjecture}\label{conj:unique.qeq}
  Let $Q\subseteq\qu V$
  be a list of distinct quantum variables with finite
  $\typel{Q}$. Let $Q_i:=\idx iQ$.
  
  Let $E_{Q_1,Q_2}$
  be a $Q_1Q_2$-local
  predicate that satisfies \ruleref{Equal}.  That is, for
  $E_{Q_1,Q_2}$ the following rule holds:
  \[
    \inferrule{
      \text{$X$ is a list of classical variables} \\
      \text{$Q$ is a list of distinct quantum variables} \\
      \bc\text{ is $XQ$-local} \\
      X_i:=\idx iX\\
      Q_i:=\idx iQ
    }{
      \pb\rhl
      {\CL{X_1=X_2}\cap E_{Q_1,Q_2}}
      \bc\bc
      {\CL{X_1=X_2}\cap E_{Q_1,Q_2}}
    }
  \]
  Then $E_{Q_1,Q_2}\approx (Q_1\quanteq Q_2)$ or $E_{Q_1,Q_2}\approx 0$.
\end{conjecture}

We do not know whether this conjecture holds. If we try to use the
same proof as for \autoref{lemma:unique.qeq}, we run into the problem
that \autoref{claim:same.tr} holds only for separable $\phi$.
This in turns means that we need a stronger version of
\autoref{lemma:unique.qeq1} where the second condition holds only for
separable $\phi\in E$.
And this in turn means that \autoref{claim.X.cn} does not hold,
instead we only have that every $A\in X$
\emph{of rank 1} is conjugate-normal.
(Because $A_\phi$ has rank 1 iff $\phi$ is separable.)
And then we cannot apply
\autoref{lemma:conj-normal} any more. (A variant of
\autoref{lemma:conj-normal} that requires $X\cap R\subseteq C$
instead of $X\subseteq C$,
and that shows $X\cap R=S\cap R$
instead of $X=S$
would do the trick (where $R$ is the set of rank 1 matrices).)

We leave a proof or counterexample for \autoref{conj:unique.qeq} as an open problem.

\subsection{On predicates as subspaces}
\label{sec:pred.subspace}

In this section, we formally prove the claim made in
\autoref{sec:intro.qrhl} that predicates described by subspaces are
the only predicates that are closed under linear combinations and
linear decompositions (at least in the finite dimensional case):

\begin{lemma}
  Let $X$
  be finite.  Let $A\subseteq\tracepos X$.
  Assume that for all $\rho_1,\rho_2\in\tracepos X$,
  and $p_1,p_2\geq 0$,
  we have
  \begin{equation}
    p_1\rho_1+p_2\rho_2\in A\iff \rho_1,\rho_2\in A.
    \label{eq:decomp.A}
  \end{equation}
  
  Then there exists a subspace $S\subseteq\elltwo X$
  such that $A=\{\rho\in\tracepos X:\suppo\rho\subseteq S\}$.
\end{lemma}

\begin{proof}\stepcounter{claimstep}
  Let $S:=\bigcup_{\rho\in A}\suppo\rho\subseteq\elltwo X$.

  \begin{claim}\label{claim:mult}
    If $\psi\in S$,
    then $\alpha\psi\in S$
    for all $\alpha\in \setC$.
  \end{claim}
  
  \begin{claimproof}
    Immediate since $\suppo\rho$
    is a subspace for all $\rho$.
  \end{claimproof}

  \begin{claim}\label{claim:add}
    If $\psi_1,\psi_2\in S$, then $\psi_1+\psi_2\in S$.
  \end{claim}
  
  \begin{claimproof}
    Since $\psi_1\in S$,
    we have $\psi_1\in\suppo\rho_1$
    for some $\rho_1\in A$.
    Hence $\rho_1=\alpha\,\proj{\psi_1}+\rho_1'$
    for some  $\alpha\in\setC$, $\rho_1'\in\tracepos X$.
    Thus $\proj{\psi_1}\in A$ by \eqref{eq:decomp.A}. Analogously $\proj{\psi_2}\in A$.
    Furthermore,
    \[
      \proj{\psi_1+\psi_2}+\proj{\psi_1-\psi_2}
      = 2\proj{\psi_1}+2\proj{\psi_2}
      \eqrefrel{eq:decomp.A}\in A.
    \]
    Thus (again with \eqref{eq:decomp.A}), we have
    $  \proj{\psi_1+\psi_2}\in A$. Thus $\psi_1+\psi_2\in\suppo\proj{\psi_1+\psi_2}\in S$.
  \end{claimproof}

  \begin{claim}\label{claim:subspace}
    $S$ is a subspace.
  \end{claim}
  
  \begin{claimproof}
    Since $X$ is finite, this follows from \autoref{claim:mult} and \autoref{claim:add}.
  \end{claimproof}

  \begin{claim}\label{claim:AB}
    $A\subseteq \{\rho\in\tracepos X:\suppo\rho\subseteq S\} =: B$.
  \end{claim}

  \begin{claimproof}
    For any $\rho\in A$,
    we have $\suppo\rho\subseteq S$ by definition of $S$. Thus $\rho\in B$.
  \end{claimproof}

  \begin{claim}\label{claim:BA}
    $B\subseteq A$.
  \end{claim}

  \begin{claimproof}
    Fix some $\rho\in B$.
    Then $\suppo\rho\subseteq S$.
    Thus we can write $\rho=\sum_i \proj{\psi_i}$
    for finitely many $\psi_i\in S$.
    (Finitely many, since $X$
    is finite.)  Since $\psi_i\in S$,
    we have that there are $\rho_i\in A$
    with $\psi_i\in\suppo\rho_i$.
    Thus $\rho_i=\proj{\psi_i}+\rho_i'$
    for some $\rho_i'\in\tracepos X$.
    Hence $\proj{\psi_i}\in A$
    by~\eqref{eq:decomp.A}.
    Hence $\rho=\sum_i \proj{\psi_i}\in A$
    by~\eqref{eq:decomp.A}.
  \end{claimproof}
  
  The lemma follows from \autoref{claim:subspace}, \autoref{claim:AB},
  and \autoref{claim:BA}.
\end{proof}

\subsection{Non-existence of quantum equality for uniform qRHL}
\label{sec:qeq.uniform}

\begin{lemma}
  Let $\qq\in V$
  be a quantum variable with $\abs{\typev\qq}=3$.
  Assume that there is a classical variable $\xx\in V$
  with $\abs{\typev\xx}=1$.
  Let $E\subseteq\elltwov{\qu{V_1}\qu{V_2}}$
  be a $\qq_1\qq_2$-local predicate. Assume that the following rule holds:
  \[
    \inferrule{
      \bc\text{ is $\qq$-local} \\
    }{
      \rhlunif
      {E}
      \bc\bc
      {E}
    }
  \]
  Then $E=0$.
\end{lemma}

\begin{proof}\stepcounter{claimstep}
  Assume for contradiction that $E\neq0$.
  
  % Since $\abs{\typev\qq}=3$,
  % we can identify $\elltwov\qq$
  % with $\setC^3$,
  % and thus write elements of $\elltwov\qq$
  % as three-component vectors, and elements of $\boundedv\qq$
  % as $3\times 3$
  % matrices. Similarly, we identify $\elltwov{\qq_1\qq_2}$
  % with $\setC^{9}$. In particular, $E\subseteq\setC^9$.

  Given a vector $\psi\in\elltwov\qq$,
  we use the shorthand $\idx i\psi_i:=\Urename{\sigma_i}\psi$
  where $\sigma_i$
  maps $\qq$
  to $\qq_i$.
  That is, $\idx1\psi,\idx2\psi$
  are the same vector as $\psi$,
  except in $\elltwov{\qq_1},\elltwov{\qq_2}$, respectively.

  Since $E$
  is $\qq_1\qq_2$-local,
  there is an $E'\subseteq\elltwov{\qq_1\qq_2}$
  such that $E=E'\otimes\elltwov{V_1V_2\setminus\qq_1\qq_2}$.
  
  \begin{claim}\label{claim:psipsiE1}
    For all $\psi\in\elltwov\qq$,
    we have $\proj{\idx1\psi}\otimes\proj{\idx2\psi}\in E'$.
  \end{claim}

  \begin{claimproof}
    Without loss of generality, $\norm\psi=1$.
    
    Let $\bc_\psi:=(\Qinit\qq\psi)$.
    By assumption of the lemma, we have
    $\rhlunif E{\bc_\psi}{\bc_\psi} E$.
    Let $\calE$
    be a witness for this judgment.
    Since $E\neq 0$,
    there is a non-zero $\gamma\in E$.
    Let $\rho:=\proj\gamma$.
    Then $\rho$
    satisfies $E$.
    Let $\rho':=\calE(\rho)$.
    Then $\rho'$
    satisfies $E$
    and we have $\rho'_i=\denotc{\idx i\bc_\psi}(\rho_i)$
    where $\rho_1':=\partr{V_1}{V_2}\rho'$,
    $\rho_2':=\partr{V_2}{V_1}\rho'$,
    $\rho_1:=\partr{V_1}{V_2}\rho$,
    $\rho_2:=\partr{V_2}{V_1}\rho$.

    By definition of $\bc_\psi$,
    we have that
    $\rho_i'=\denotc{\idx i\bc_\psi}(\rho_i)=\proj{\idx
      i\psi}\otimes\rho_i^*$ for some non-zero $\rho_i^*$.
    Thus $\partr{V_1}{V_2}\rho'=\proj{\idx1\psi}\otimes\rho_1^*$
    and $\partr{V_2}{V_1}\rho'=\proj{\idx2\psi}\otimes\rho_2^*$.
    Hence
    $\rho'=\proj{\idx1\psi}\otimes\proj{\idx2\psi}\otimes\rho_{12}^*$
    for some nonzero positive $\rho_{12}^*$.
    Fix a non-zero $\phi\in\suppo(\rho_1^*\otimes\rho_2^*)$.
    Then
    $\idx1\psi\otimes\idx2\psi\otimes\phi\in\suppd\rho'\subseteq E$.
    Hence $\idx1\psi\otimes\idx2\psi\in E'$
  \end{claimproof}
  
  We define vectors $\tilde \psi_i\in \elltwov\qq$:
  \[
    \tilde \psi_1:={\footnotesize\begin{pmatrix}0 \\ 1 \\ 2\end{pmatrix}},\
    \tilde \psi_2:={\footnotesize\begin{pmatrix}2 \\ 1 \\ 0\end{pmatrix}},\
    \tilde \psi_3:={\footnotesize\begin{pmatrix}1 \\ 1 \\ 1\end{pmatrix}},\
    \tilde \psi_4:={\footnotesize\begin{pmatrix}1 \\ 2 \\ 3\end{pmatrix}},\
    \tilde \psi_5:={\footnotesize\begin{pmatrix}3 \\ 2 \\ 1\end{pmatrix}},\
    \tilde \psi_6:={\footnotesize\begin{pmatrix}3 \\ 4 \\ 5\end{pmatrix}}.
  \]
  (Here we identify $\elltwov\qq$ with $\setC^3$.)
  Let $\psi_i:=\tilde \psi_i/\norm{\psi_i}$.
  Let
  \[
    \alpha_{1} := 25,\
    \alpha_{2} := -15,\
    \alpha_{3} := -162,\
    \alpha_{4} := -147,\
    \alpha_{5} := 49,\
    \alpha_{6} := 250
    \qquad\text{and}\qquad
    P :=
    {\footnotesize\begin{pmatrix}1&0&0\\
      0&1&0\\
      0&0&0
    \end{pmatrix}}.
  \]

  \begin{claim}\label{claim:calc}
    $\sum_i\alpha_i\, \proj{\idx1 \psi_i}\otimes\proj{\idx2\psi_i} = 0$
    and
    $\sum_i\alpha_i\, \proj{\idx1 P\psi_i}\otimes\proj{\idx2 P\psi_i} / \norm{P\psi_i}^2 \neq 0$.
  \end{claim}

  \begin{claimproof}
    This is shown by explicit calculation. Sage \cite{sage} code for performing
    the calculation can be found here: \InputIfFileExists{no-equality.sagelink}{}{}{\texttt\SAGECELLLINKTINY}.
  \end{claimproof}

  Let $M\in\Meas{\typev{\xx}}{\typev{\qq}}$
  be the measurement with $M(z)=P$ for $z\in\typev\xx$.
  (Recall that $\typev\xx$
  has only one element.) Let $\bc_M:=(\Qmeasure\xx\qq M)$.
  Note that $\bc_M$
  is $\qq$-local,
  even though $\xx$
  occurs in $\bc_M$.
  This is because $\xx$
  can take only one possible value, so $\xx$
  is not actually modified by $\bc_M$.
  By assumption of the lemma, $\rhlunif E{\bc_M}{\bc_M} E$. Let $\calE$ be a witness for 
  $\rhlunif E{\bc_M}{\bc_M} E$.

  Let $\phi\in\elltwov{V_1V_2\setminus\qq_1\qq_2}$ be an arbitrary vector with $\norm\phi=1$.

  \begin{claim}\label{claim:epsi}
    For any non-zero $\psi\in\elltwov\qq$, we have
    \[
      \partr{\qq_1\qq_2}{V_1V_2\setminus\qq_1\qq_2} \calE\pb\paren{\proj{\idx1 \psi}\otimes\proj{\idx2\psi}\otimes\proj{\phi}}
      = 
      \proj{\idx1 P\psi}\otimes\proj{\idx2 P\psi} / \norm{P\psi}.
    \]
  \end{claim}
  
  \begin{claimproof}
    Let $\rho^*$ denote the lhs of the claim.
    
    By \autoref{claim:psipsiE1},
    $\proj{\idx1 \psi}\otimes\proj{\idx2\psi}\in E'$,
    thus
    $\rho:=\proj{\idx1 \psi}\otimes\proj{\idx2\psi}\otimes\proj{\phi}\in
    E$.
    Since $\calE$
    is a witness for $\rhlunif E{\bc_M}{\bc_M} E$, this implies that
    $\partr{V_1}{V_2}\calE(\rho) = \denotc{\idx1{\bc_M}}(\partr{V_1}{V_2}\rho)$.
    Thus
    \begin{align*}
      \partr{\qq_1}{\qq_2}\rho^*
      &=
        \partr{\qq_1}{V_1\setminus \qq_1}\partr{V_1}{V_2}\calE(\rho)
        =
        \partr{\qq_1}{V_1\setminus \qq_1}  \denotc{\idx1{\bc_M}}(\partr{V_1}{V_2}\rho)
      \\&
          =
        \partr{\qq_1}{V_1\setminus \qq_1} \denotc{\idx1{\bc_M}}\pb\paren{\proj{\idx1\psi}\otimes\partr{V_1\setminus\qq_1}{V_2\setminus\qq_2}\proj\phi}
      \\&
          \starrel=
          \partr{\qq_1}{V_1\setminus \qq_1} \pb\paren{
            \proj{\idx1 P\psi}\otimes\partr{V_1\setminus\qq_1}{V_2\setminus\qq_2}\proj\phi
          }
%      \\&
          =
           \proj{\idx1 P\psi}.
    \end{align*}
    Here $(*)$ follows from the semantics of $\Qmeasure\xx\qq M$.
    
    Analogously, $\partr{\qq_2}{\qq_1}\rho^*=\proj{\idx 2 P\psi}$. Thus
    $\rho^*=\proj{\idx1 P\psi}\otimes \proj{\idx2 P\psi}/\norm{P\psi}^2$.
  \end{claimproof}

  Finally, we have
  \begin{align*}
    0 &= \partr{\qq_1\qq_2}{V_1V_2\setminus\qq_1\qq_2} \calE(0)
        \quad \txtrel{\autoref{claim:calc}}=\quad 
        \partr{\qq_1\qq_2}{V_1V_2\setminus\qq_1\qq_2}
        \calE\pB\paren{\sum\nolimits_i\alpha_i\, \proj{\idx1 \psi_i}\otimes\proj{\idx2\psi_i}\otimes\proj\phi}
    \\&
      = 
        \sum\nolimits_i\alpha_i\,
        \partr{\qq_1\qq_2}{V_1V_2\setminus\qq_1\qq_2}
        \calE\pB\paren{\proj{\idx1 \psi_i}\otimes\proj{\idx2\psi_i}\otimes\proj\phi}
    \\&
        \txtrel{\autoref{claim:epsi}}=\quad
        \sum\nolimits_i\alpha_i\, \proj{\idx1P\psi_i}\otimes\proj{\idx2P\psi_i} / \norm{P\psi_i}^2
        \quad \txtrel{\autoref{claim:calc}}\neq\quad  0.
  \end{align*}
  Thus we have $0\neq 0$, a contradiction. Hence our initial assumption $E\neq0$ was false.
\end{proof}

\makeatletter
\let\oldsection\section
\def\nonumsectoc#1{\oldsection*{#1}\addcontentsline{toc}{section}{#1}}
\def\section{\@ifstar\nonumsectoc\oldsection}

\renewcommand\symbolindexentry[4]{\noindent\hbox{\hbox to 2in{$#2$\hfill}\parbox[t]{3.5in}{#3}\hbox to 1cm{\hfill #4}}\\[2pt]}

 \printsymbolindex

 \printindex  

 \listofrule

 \printbibliography

\end{document}